\documentclass{article}
\usepackage[a4paper,marginparwidth=1.5cm]{geometry}
\usepackage{amsmath}
\usepackage{amsthm}
\usepackage{tikz}
\usepackage{float}
\usetikzlibrary{arrows}
\usepackage[utf8]{inputenc}
\usepackage{amsfonts}
\usepackage{adjustbox}
\usepackage{enumitem}
\usepackage{breakurl}
\usepackage[hyperindex,breaklinks]{hyperref}
\usepackage{subcaption}
\usepackage{mathtools}
\usepackage{nicefrac}
\usepackage{amssymb, amsmath, bm,mleftright}
\usepackage{faktor}
\usepackage{amsmath}
\usepackage{amssymb}
\usepackage{bookmark}
\usepackage[toc,page,title]{appendix}
\usepackage{xparse}
\usepackage[
backend=biber,
style=alphabetic,
sorting=anyt,
maxbibnames=50,
]{biblatex}
\newcommand{\defeq}{\vcentcolon=}

\newcommand{\twist}[3]{\langle #1\,\, #2\,\, #3 \rangle}
\newcommand{\Twist}[1]{\langle #1 \rangle}

\newcommand{\orth}[1]{#1^\perp}

\newcommand{\Ext}{\mathrm{Ext}}
\newcommand{\BCFW}{\mathrm{BCFW}}
\newcommand{\Gr}[2]{\mathrm{Gr}_{#1, #2 }}
\newcommand{\Grnon}[2]{\mathrm{Gr}^{\geq }_{#1, #2 }}
\newcommand{\OG}[2]{\mathrm{OG}_{#1, #2 }}
\newcommand{\OGnon}[2]{\mathrm{OG}^{\geq }_{#1, #2 }}

\newcommand{\fracless}[2]{\genfrac{}{}{0pt}{}{#1}{#2}}

\newcommand*{\vertbar}{\rule[-1ex]{0.5pt}{2.5ex}}
\newcommand*{\horzbar}{\rule[.5ex]{2.5ex}{0.5pt}}
\NewDocumentCommand{\evalat}{sO{\big}mm}{
  \IfBooleanTF{#1}
   {\mleft. #3 \mright|_{#4}}
   {#3#2|_{#4}}
}

\makeatletter
\def\moverlay{\mathpalette\mov@rlay}
\def\mov@rlay#1#2{\leavevmode\vtop{
   \baselineskip\z@skip \lineskiplimit-\maxdimen
   \ialign{\hfil$\m@th#1##$\hfil\cr#2\crcr}}}
\newcommand{\charfusion}[3][\mathord]{
    #1{\ifx#1\mathop\vphantom{#2}\fi
        \mathpalette\mov@rlay{#2\cr#3}
      }
    \ifx#1\mathop\expandafter\displaylimits\fi}
\makeatother

\newcommand{\bigcupdot}{\charfusion[\mathop]{\bigcup}{\cdot}}

\begin{document}

\title{The BCFW Tiling of the ABJM Amplituhedron}
\author{Michael Oren Perlstein\thanks{Department of Mathematics, Weizmann Institute of Science, \texttt{michael.oren-perlstein@weizmann.ac.il}}
\and Ran J.~Tessler\thanks{Department of Mathematics, Weizmann Institute of Science, \texttt{ran.tessler@weizmann.ac.il}}}
\date{December 2025}

\maketitle
\newtheorem{thm}{Theorem}[subsection]
\newtheorem{lem}[thm]{Lemma}
\newtheorem{claim}[thm]{Claim}
\newtheorem{prop}[thm]{Proposition}
\newtheorem{obs}[thm]{Observation}
\newtheorem{coro}[thm]{Corollary}
\newtheorem{rmk}[thm]{Remark}
\newtheorem{con}[thm]{Cconjecture}
\newtheorem*{con*}{Cconjecture}

\theoremstyle{definition}

\newtheorem{dfn}[thm]{Definition}

\newtheorem{exm}[thm]{Example}

\begin{abstract}
The orthogonal momentum amplituhedron \(\mathcal{O}_k\) was introduced simultaneously in 2021 by Huang, Kojima, Wen, and Zhang in \cite{Huang_2022}, and by He, Kuo, Zhang in \cite{he2022momentum}, in the study of scattering amplitudes of ABJM theory. It was conjectured that it admits a decomposition into BCFW cells. We prove this conjecture.  
\end{abstract}

\tableofcontents

\section{Introduction}
The \emph{(tree) amplituhedron} is a geometric space that was introduced in 2013 by Arkani-Hamed and Trnka \cite{Arkani_Hamed_2014} in their study of scattering amplitudes in quantum field theories, specifically planar \(\mathcal N = 4\) supersymmetric Yang-Mills theory (SYM). It was conjectured to admit a decomposition into images of \emph{BCFW positroid cells}, a conjecture proven by \cite{amptriag,even2023cluster}. The motivation for this conjecture, and to some extent the motivation for defining the amplituhedron itself, came from physics: This decomposition is the geometric manifestation of the BCFW recursions \cite{britto2005direct,britto2005new} for planar $\mathcal{N}=4$ SYM.

The \emph{orthogonal momentum amplituhedron}, or the \emph{ABJM amplituhedron} was introduced  simultaneously in 2021 by Huang, Kojima, Wen, and Zhang in \cite{Huang_2022}, and by He, Kuo, and Zhang in \cite{he2022momentum},
as a space that should encode the scattering amplitudes for $\mathcal{N}=6$ Aharony-Bergman-Jafferis-Maldacena (ABJM) \cite{abjm1}, following an earlier Grassmannian picture \cite{PosOG,BCFWrec}.  Based on the earlier picture, it was conjectured that the ABJM amplituhedron also admits a decomposition into images of BCFW orthitroid cells, which are the ABJM analogue of BCFW positroid cells \cite{BCFWrec}. 
This manuscript reviews the definition of the latter objects and proves the conjecture.

\subsection{A Speed of Light Review of Positive Grassmannians, the Positive Orthogonal Grassmannian and the ABJM Amplituhedron}
We start with a concise review of the non-negative Grassmannian, its orthogonal cousin, and the ABJM amplituhedron. We refer the reader to Appendix~\ref{app:OG} for a more complete review.

The (real) Grassmannian $\Gr{k}{n}$ is the space of $k$-dimensional linear subspaces of $\mathbb{R}^n.$ We represent an element $V\in \Gr{k}{n}$ as the row-span of a \(k\times n\) non-degenerate matrix, and will often identify them in writing for brevity. A natural coordinate system on this space is given by the \emph{Pl\"ucker coordinates}. If $C$ is a matrix representative of $V\in\Gr{k}{n},$ that is, a $k\times n$ matrix whose rows span $V,$ we define $\Delta_I(C),$ for $I\in\binom{[n]}{k}$, as the minor whose columns are indexed $I.$ While each coordinate separately depends on the choice of $C,$ the collection of all coordinates, for $I\in\binom{[n]}{k},$ depends on $C$ only up to a common scaling, thus giving rise to \emph{projective coordinates}. We refer to them as Pl\"ucker coordinates, and sometimes abuse notation and refer also to the maximal minors of the matrices as the matrices' Pl\"ucker coordinates. The \emph{non-negative Grassmannian} \cite{postnikov} $\Grnon{k}{n}$ is the subspace of $\Gr{k}{n}$ made of vector spaces with a representative with all non-negative Pl\"ucker coordinates. In his seminal work, Postnikov \cite{postnikov} had proved that this space is a stratified space, where each stratum, called a \emph{positroid cell}, is the subspace defined by the vanishing of a certain collection of Pl\"ucker coordinates. He showed that the strata are homeomorphic to open balls. He also found several ways to label the different strata by various combinatorial objects, which include \emph{plabic graphs} and \emph{decorated permutations}. 

The theory of the non-negative Grassmannian is a particularly nice instance of Lusztig's theory of positivity for algebraic groups and partial flag varieties \cite{lusztig1994total}. It was further developed by Rietsch, Marsh, Fomin, Zelevinsky, Postnikov, and others \cite{rietsch1998total,rietsch2006closure,marsh2004parametrizations,fomin1999double, fomin2002cluster, fomin2003cluster,postnikov}.
The positive Grassmannian was the subject of many researches in cluster algebras, tropical geometry, integrable systems, and recently also scattering amplitudes  \cite{speyer2005tropical, kodama2013combinatorics, kodama2014kp, lukowski2023positive, speyer2021positive,arkanihamed2014scattering,Arkani_Hamed_2014}. 

In order to define the ABJM amplituhedron, we need to define the orthogonal group analog of the positive Grassmannian. There are various equivalent definitions, we will use the one given by Galashin and Pylyavskyy \cite{OG_source}: the orthogonal Grassmannian $\OG{k}{2k}$ is the space of is the space of $k$-dimensional linear subspaces of $\mathbb{R}^{2k}$ that are self-orthogonal with respect to the inner product defined by \(\eta\), a \(2k\times 2k\) diagonal matrix with alternating \(\pm 1\) on the diagonal. The non-negative orthogonal Grassmannian \(\OGnon{k}{2k}\) is then similarly defined as the subset of \(\OGnon{k}{2k}\) with all non-negative Pl\"ucker coordinates. This object is less studied than its $GL_n$ cousin, and we refer the reader to \cite{OG_source,PosOG,kim_lee,companion} and Section~\ref{sec:background} for further reading. Importantly, this is also a stratified space, and will refer to its strata as \emph{orthitroid cells}. 

Orthitroid cells are specified by a collection of vanishing Pl\"ucker coordinates and are in bijection with several combinatorial objects, e.g. fixed-point–free involutions on \([2k]\) and \(k\)-OG graphs, that is four-regular graphs embedded in a disc with \(2k\) cyclically-labeled external legs (see Appendix~\ref{app:OGgraphs} for details). The orthitroid cell corresponding to an OG graph \(\Gamma\) will be labeled \(\Omega_\Gamma\). Moreover, an orientation of the OG graph gives rise to a parameterization of the orthitroid cell (\ref{def:param}). Such graphs can be built inductively using \emph{local moves}:
\emph{The \(\mathrm{Inc}\) move}, which adds two new external vertices to the boundary of the disk and connects them with an edge, and  \emph{the \(\mathrm{Rot}\) move}, which adds an internal vertex by braiding two adjacent external legs connecting to the boundary of the disk. Another useful operation is \emph{the \(\mathrm{Cyc}\) move}, which rotates the graph inside the disk (see Appendix~\ref{app:local moves} for details).

We now briefly define the ABJM amplituhedron (\cite{Huang_2022,he2022momentum}). For a fixed \(2k\times (2k+2)\) matrix \(\Lambda\) the \emph{orthogonal momentum (ABJM) amplituhedron} \(\mathcal{O}_k (\Lambda)\subset \Gr{k}{k+2}\) is defined as the image of \(\OGnon{k}{2k}\) under the map \(C \mapsto C\Lambda\), when we again identify an element of the Grassmannian with its matrix representative. See Section~\ref{sec:The ABJM Amplituhedron and Its Natural Coordinates} for a precise definition. The geometry of the ABJM amplituhedron is conjectured to be independent of \(\Lambda\). For this reason we we will occasionally omit \(\Lambda\) from the notations. 
The work \cite{companion} will survey and summarize many basic properties of the non-negative orthogonal Grassmannian and the ABJM amplituhedron.

In Section~\ref{sec:bcfw} we define a particularly nice collection of orthitroid cells in $\OGnon{k}{2k}$. These cells can be defined in a recursive manner, and can be labeled nicely via \emph{trees of triangles}. We denote this collection by $\BCFW_k,$ and refer to them as \emph{($k$-)BCFW cells}. We refer the reader to Section~\ref{sec:bcfw} for their graphical and recursive descriptions.

\begin{figure}[H]
    \centering
                \begin{center}
\begin{tikzpicture}[scale = 0.8]
\draw (0,0) circle (2cm);

\draw(0.517638, -1.93185)node[anchor=north]{\(1\)}--(0.517638, 1.93185)node[anchor=south]{\(6\)};

\draw(1.41421, -1.41421)node[anchor=north]{\(2\)}--(-1.93185, 0.517638)node[anchor=east]{\(9\)};

\draw(1.93185, -0.517638)node[anchor=west]{\(3\)}--(-0.517638, -1.93185)node[anchor=north]{\(12\)};

\draw(1.93185, 0.517638)node[anchor=west]{\(4\)}--(-0.517638, 1.93185)node[anchor=south]{\(7\)};

\draw(1.41421, 1.41421)node[anchor=south]{\(5\)}--(-1.93185, -0.517638)node[anchor=east]{\(10\)};

\draw(-1.41421, 1.41421)node[anchor=south]{\(8\)}--(-1.41421, -1.41421)node[anchor=north]{\(11\)};

\filldraw[blue] (0,0) circle (1pt);

\draw[blue](0.643951, -1.11536)--(0,0);
\filldraw[blue] (0.643951, -1.11536) circle (1pt);

\draw[blue](0.135868, -1.99538)--(0.643951, -1.11536);
\filldraw[blue] (0.135868, -1.99538) circle (1pt);

\draw[blue](1.66012, -1.11536)--(0.643951, -1.11536);
\filldraw[blue] (1.66012, -1.11536) circle (1pt);

\draw[blue](-1.2879,0)--(0,0);
\filldraw[blue] (-1.2879,0) circle (1pt);

\draw[blue](-1.79598, 0.880025)--(-1.2879,0);
\filldraw[blue] (-1.79598, 0.880025) circle (1pt);

\draw[blue](-1.79598, -0.880025)--(-1.2879,0);
\filldraw[blue] (-1.79598, -0.880025) circle (1pt);

\draw[blue](0.643951, 1.11536)--(0,0);
\filldraw[blue] (0.643951, 1.11536) circle (1pt);

\draw[blue](0.135868, 1.99538)--(0.643951, 1.11536);
\filldraw[blue] (0.135868, 1.99538) circle (1pt);

\draw[blue](1.66012, 1.11536)--(0.643951, 1.11536);
\filldraw[blue] (1.66012, 1.11536) circle (1pt);

\end{tikzpicture}
\end{center}
    \caption{a BCFW graph with its triangle tree superimposed in blue}
    \label{fig:tirag tree}
\end{figure}

\subsection{Tilings and BCFW Tilings}
We follow the definition in \cite{m2triangulation} of \emph{triangulations} or \emph{tilings}, generalized to meet our needs.
\begin{dfn}
Let $X,Y,M$ be topological spaces, and let $F:X\times M\to Y$ be a continuous function. For $m\in M,$ write $F_m$ for $F(-,m).$
We say that the images of a collection of subspaces $S_1,\ldots, S_N\subseteq X$ \emph{tile}, or \emph{triangulate} $Y$ for a given $m,$ if the following conditions are met
\begin{itemize}
\item \emph{Injectivity}: $S_i \to F_m(S_i)$ is injective for all $i.$
\item \emph{Separation}: $F_m(S_i)$ and $F_m(S_j)$ are disjoint for every two $i\neq j.$
\item \emph{Surjectivity}: $\bigcup_{i\in[N]} F_m(S_i)$ is an open dense subset of $\text{Im}(F_m)$.
\end{itemize}
We say that the images of $S_1,\ldots, S_N$ tile $Y_m$ for all $m$ if...well, if they tile $Y_m$ for all $m\in M.$
\end{dfn}
Arkani-Hamed and Trnka \cite{Arkani_Hamed_2014} conjectured, following the Grassmannian picture of planar $\mathcal{N}=4$ SYM \cite{arkanihamed2014scattering}, that the images of a certain collection of positroid cells, the (SYM) BCFW cells tile $\mathcal{A}_{n,k,4}(Z)$ for every positive $Z,$ a proof was later found in \cite{amptriag,even2023cluster}.

The work \cite{BCFWrec} suggests a Grassmannian realization of the BCFW recursion for calculating ABJM amplitudes. In analogy to $\mathcal{N}=4$ SYM case, this realization was translated in ABJM-amplituhedron means to the following conjecture.
\begin{con*}
    For every $k\geq 3,$ the images of orthitroid BCFW cells $\Omega_\Gamma,~\Gamma\in\BCFW_k$ (see Section~\ref{sec:bcfw}), tile $\mathcal{O}_k(\Lambda)$, for every $2k\times(k+2)$ matrix $\Lambda$ with all positive maximal minors.
\end{con*} 

\subsection{Main Results}
Our main results are the following two theorems.
\begin{thm}\label{thm:inj}
For every positive $\Lambda\in \mathrm{Mat}^>_{2k \times (k+2)}$ the BCFW cells $\Omega\in\BCFW_k$ map injectively to the amplituhedron. Moreover, two BCFW cells whose closures share a common codimension $1$ boundary are locally separated near that boundary.    
\end{thm}
For the accurate statements see Theorem~\ref{BCFW inj} and Theorem~\ref{thm:local sep}. 
In Section~\ref{sec:immanant_pos} we define, following an idea of Galashin \cite{origami}, a collection \emph{strongly positive matrices} $\Lambda,$ which form an open subset of the set of all $2k\times(k+2)$ matrices. For them we can say much more.
\begin{thm}\label{thm:tilings}
    For every $k$, and every strongly positive $\Lambda$ the images of the BCFW cells $\Omega\in\BCFW_k$ tile the ABJM amplituhedron. 
\end{thm}
This theorem is proven in Section~\ref{sec:tiling}.
In addition, there are several other central results concerning the the structure of boundary strata of BCFW cells which can be found is Section~\ref{sec:BCFW Cells and their Boundaries}. These results allow us to prove that all orthitroid cells that are boundary strata of BCFW cells map injectively to the amplituhedron in Section~\ref{sec:inj}  (see Theorem~\ref{BCFW inj} for the precise statement).
\subsection{Relation with Existing Literature}
The motivation for the definition of the ABJM amplituhedron came from its two cousins, the original amplituhedron \cite{Arkani_Hamed_2014} of Arkani-Hamed and Trnka, and the momentum amplituhedron \cite{momAmpli} defined by Damgaard, Ferro, Lukowski, Tomasz and Parisi. The ABJM amplituhedron was also studied in \cite{lukowski2022geometry,lukowski2023momentum,he2023abjm}.
The BCFW tiling conjecture for the original amplituhedron was proven in \cite{amptriag,even2023cluster}, by by Even-Zohar, Lakrec, and the second named author, and Even-Zohar, Lakrec, Parisi, Sherman-Bennet, Williams and the second named author. For the momentum amplituhedron it was proven by Galashin in \cite{origami}.
In \cite{origami}, as in this paper, the conjecture is proven under the slight simplification of requiring stronger positivity requirements on the external data, which is in our case $\Lambda.$ In Galashin's work and in this work this extra condition has the same origin, the need to guarantee that Mandelstam variables will not change sign. Still, our approach is closer to that of \cite{amptriag}, which relies on the notion of \emph{promotions}. Both this work and \cite{amptriag} works follow a similar high-level strategy, but the different geometries require disparate treatments. 
The common, to some extent, strategy is:
\begin{itemize}
\item \emph{Injectivity}: Both works construct the BCFW cells using iterations of simple operations, and use promotions to show that the resulting cells map injectively. 
\item \emph{Separation}: In \cite{amptriag} the separation was proven by showing it in simple base cases, and then showing it is preserved under the promotions. Here we do the same, but for \emph{local separation}, showing that BCFW cells which share a common codimension $1$ boundary are locally separated along this boundary. 
\item \emph{Surjectivity} then follows from a topological argument. In our case it is more entangled, since the separation is only known to be local at the time we apply the (refined) topological argument.
\end{itemize}
The main differences between the proofs come from the different geometry and the weaker positivity in the ABJM setting.
\begin{itemize}
\item The injectivity in \cite{amptriag} morally comes from the fact that the algebraic degree of the map from (the complexified) BCFW cells to the amplituhedron is of degree $1.$ In this work the map is of degree $2^{k-2},$ and non-trivially becomes $1$ only upon restricting to the positive real part.
\item In \cite{amptriag,even2023cluster}, promotion involved only rational functions, while in the ABJM has functions involving square roots (see Appendix~\ref{app:local moves} and Definition~\ref{def:arc}).
\item The recursive structure of the cells is different.
\item Boundary defining functions, that is, functions whose zero loci define the boundaries of BCFW cells or the whole amplituhedron, may not have a definite sign on the image of the whole cell or the amplituhedron, respectively. For this we first show local separation, and also the topological argument for surjectivity becomes more complex.
\item It is more intricate to treat the external boundaries of the whole amplituhedron. For this we introduce the strongly positive matrices, and carefully study them.
\end{itemize}
An additional source of difficulties is that the orthogonal Grassmannian is less studied than its $GL_n$ cousin, which sometimes requires finding or developing technical bypasses.

An amusing fact regarding the different promotions appearing in the $\mathcal{N}=4$ SYM picture and here, is that in SYM they are related to intersection of planes, while here we need to intersect planes and spheres. One can think of these two scenarios as constructions which use only straightedge and constructions which use straightedge and compass.
\subsection{Plan of the Paper}
This paper is organized as follows. Section~\ref{sec:background} reviews the basic definitions and results regarding the BCFW graphs and the ABJM amplituhedron.

A key technical tool we develop here, following an $\mathcal{N}=4$ Super Yang Mills analog, called \emph{promotion}, which is roughly speaking an amplituhedron-friendly way to manipulate orthitroid cells and functions. This is the subject of Section~\ref{sec:prom}. 

In Section~\ref{sec:BCFW Cells and their Boundaries} we examine BCFW cells and their codimension \(1\) boundary cells in greater detail. We describe their structure and combinatorial properties, and provide a characterization of the different types of codimension \(1\) boundary cells associated with BCFW cells. 

Section~\ref{sec:inj} proves that BCFW cells map injectively to the ABJM amplituhedron, and Section~\ref{separation section} shows that they are locally separated in the sense of Theorem~\ref{thm:inj}. Thus, these two sections together provide a proof for this theorem.

In Section~\ref{sec:immanant_pos} we restrict to strongly positive matrices, and prove that they give rise to \emph{non-negative Mandelstam variables}. We also study boundaries of the ABJM amplituhedron for strongly positive $\Lambda.$ Finally, in Section~\ref{sec:tiling} we prove Theorem~\ref{thm:tilings}.

A review of the basic notions of the theory of the non-negative orthogonal Grassmannian can be found in Appendix~\ref{app:OG}. Appendix~\ref{apx:calcs} contains a few necessary technical calculations as well as some practical examples of the tools and notions developed in this paper.  
\subsection*{Acknowledgments}
The authors were supported by the ISF grants no.~335/19 and 1729/23. 

R.T.~thanks Evgeniya Akhmedova, Chaim Even-Zohar, Yoel Groman, Song He, Tsviqa Lakrec, Matteo Parisi, Melissa Sherman-Bennett, and Lauren Williams for helpful discussions related to this work.

\section{BCFW Graphs and the ABJM Amplituhedron}

\label{sec:background}
Here we review key background on the ABJM amplituhedron and the BCFW graphs relevant to ABJM theory.

\subsection{BCFW Graphs}\label{sec:bcfw}
\begin{dfn}[\cite{BCFWrec}, Section~4.4]
    The \emph{BCFW graphs} are defined recursively as follows:

    \underline{\textbf{Base case:}}
    The only BCFW graph for \(k=2\) is:

    
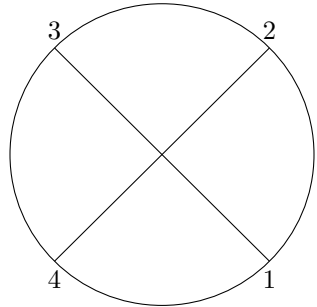
\begin{center}
\begin{tikzpicture}[scale = 0.5]
\draw (0,0) circle (2cm);
\draw(-1.41421, 1.41421)node[anchor=south]{\(3\)}--(1.41421, -1.41421)node[anchor=north]{\(1\)};
\draw(-1.41421, -1.41421)node[anchor=north]{\(4\)}--(1.41421, 1.41421)node[anchor=south]{\(2\)};
\end{tikzpicture}
\end{center}

\underline{\textbf{Recursion:}}

For \(k>1\) the BCFW graphs are constructed from two previous BCFW graphs, \(L\) and \(R\) with \(2k_L\) and \(2k_R\) external vertices resp., such that \(k=k_L+k_R -1\). The new graph is as follows:

\begin{center}
\begin{tikzpicture}[scale = 0.6,every node/.style={scale=0.8}]
\draw (0,0) circle (3cm);

\filldraw[lightgray] (-1.25,1/2) circle (1cm);
\node[scale=3] (c) at (-1.25,1/2)  {\(L\)};
\draw (-1.25,1/2) circle (1cm);

\filldraw[lightgray] (1.25,1/2) circle (1cm);
\node[scale=3] (c) at (1.25,1/2)  {\(R\)};
\draw (1.25,1/2) circle (1cm);

\draw(0.25,1/2)--(-0.25,1/2);

\draw (1.75, -0.366025)--(2.59808, -1.5);
\draw (1.25, 1.5)--(1.25, 2.72718);

\draw (-1.75, -0.366025)--(-2.59808, -1.5);
\draw (-1.25, 1.5)--(-1.25, 2.72718);

\filldraw[black] (2.44828, 1.01411) circle (2pt);
\filldraw[black] (2.02305, 1.55234) circle (2pt);
\filldraw[black] (2.57776, 0.339368) circle (2pt);

\filldraw[black] (-2.44828, 1.01411) circle (2pt);
\filldraw[black] (-2.02305, 1.55234) circle (2pt);
\filldraw[black] (-2.57776, 0.339368) circle (2pt);

\draw (-0.75, -0.366025)node[anchor=north]{\(2k_L\)} --(1.5, -2.59808)node[anchor=north]{\(1\)};
;

\draw (0.75, -0.366025)  --(-1.5, -2.59808)node[anchor=north]{\(2k\)};
\draw (0.75, -0.366025)node[anchor=north]{\(1\)};

\end{tikzpicture}
\end{center}

Denote the set of cells represented by BCFW graphs with \(2k\) external vertices as \(\BCFW_k\). A BCFW graph with \(2k\) external vertices will be called a \(k\)-BCFW graph.

\end{dfn}
 
\begin{dfn}
   Let \(G\) be a planar graph embedded in a disk with \(k\) external vertices. Consider external edges are composed of two external half edges, and supose the half edges are numbered \(1\) to \(2k\) counter-clockwise. We will call \(G\) a \emph{disk graph}. We define the \emph{medial graph} \(M(G)\) of \(G\) as follows (this is essentially the same as a regular medial graph but we take specific care of the embedding disk):
   \begin{enumerate}
        
       \item The internal vertices of \(M(G)\) are in bijection with the internal edges of \(G\).
       \item On an external half edge of \(G\) that is numbered \(i\) add a vertex of \(M(G)\) and number number it \(i\).
       \item Let \(f\) be a face of \(G\), and \(v\) vertex of \(G\) in \(f\). \(v\) is incident to two edges \(e_1\) and \(e_2\) of \(G\) in \(f\). For each such \(f\) and \(v\) in \(f\), add an edge in \(M(G)\) between the vertices corresponding to \(e_1\) and \(e_2\) in \(M(G)\). If \(e_1\) or \(e_2\) are external and thus correspond to two vertices in \(M(G)\), use the vertices that correspond to the half edges incident to \(v\) (see Figure~\ref{fig:tirag tree}).
       
   \end{enumerate}
\end{dfn}
\begin{prop}[\cite{companion}]
\label{prop:trigonometric from disk graph}
    For \(\Gamma\) an OG graph with \(2k\) external vertices, a choice of of trigonometric orientation is equivalent to a choice of a disk graph \(G\) with \(k\) external vertices such that \(M(G)=  \Gamma\).
\end{prop}
\begin{dfn}
 \label{def:ToT}

    An OG graph $\Gamma$ is a \emph{tree of triangles (or ToT)} if there exist a disk graph \(G\) with \(\Gamma = M(G)\), such that \(G\) is a tree contained in a disk with all 3-regular vertices. We will call \(G\) the \emph{tree of triangles} or \emph{triangle tree} of \(\Gamma\) and write \(\Delta(\Gamma)\) (see Figure 1).

    Note as a graph and its dual both have the same medial graph this is not necessarily well-defined. However in all but one case at most one of them would be a tree.

    An orthitroid cell corresponding to a ToT graph (rep. BCFW graph) will be called a \emph{ToT cell} (resp. \emph{BCFW cell}).
\end{dfn}
\begin{thm}[\cite{PosOG}, Section~4.4, \cite{companion}] 
\label{thm:BCFW are trees}
    For \(k\geq 3\) the BCFW graphs are ToT graphs such that the tree of triangles has a vertex between the \(1\) and \(2k\) external half edges.

\end{thm}

\begin{prop}[\cite{companion}]
\label{prop:ToT reduced}
    ToT graphs with \(2k\) external vertices are reduced and contains no arcs of the form \(\{i,i+1\}\) (considered mod \(2k\)).
\end{prop}

\begin{prop}[\cite{companion}]
\label{prop:uniquness of trees}
    Different BCFW graphs are never equivalent. BCFW cells with \(2k\) external vertices for \(k\geq 3\) are in bijection with disk graphs  with \(k\) external vertices and a vertex between the \(2k\) and \(1\) external half edges. 
\end{prop}

\subsection{The ABJM Amplituhedron and Its Natural Coordinates}
\label{sec:The ABJM Amplituhedron and Its Natural Coordinates}
We define the ABJM amplituhedron.
    \begin{thm}[\cite{Arkani_Hamed_2014}, Section~4]
    \label{amp map well-defined}   
    Let \(\mathrm{Mat}_{k\times n}^>\) denote the set of \(k\times n\) matrices with all maximal minors being positive.  Let  \(C\) be a matrix representative of an element of \( \Grnon{k}{n}\) and \(\Lambda\in \mathrm{Mat}_{n\times(k+m)}^>\). Then \(C \Lambda\) is of full rank.
    \end{thm}
\begin{dfn}[\cite{Huang_2022}, Section~3.1,\cite{he2022momentum}, Section~3.1]
        For \(\Lambda \in \mathrm{Mat}_{2k\times(k+2)}\),   the \emph{amplituhedron map} \(\widetilde\Lambda:\OGnon{k}{2k}\rightarrow \Gr{k}{k+2}\) is defined by  \(C \mapsto C\Lambda\), which represents a proper element of \(\Gr{k}{k+2}\) by Theorem~\ref{amp map well-defined}.
        The \emph{orthogonal momentum (ABJM) amplituhedron} \(\mathcal{O}_k (\Lambda)\) for \(\Lambda \in \mathrm{Mat}_{2k\times(k+2)}^>\) is defined as the image of \(\OGnon{k}{2k}\) under the amplituhedron map.
\end{dfn}
\begin{thm}[\cite{Huang_2022}, Section~3.1,\cite{he2022momentum}, Section~3.1,\cite{companion}]
    \label{dimension thm}
    The amplituhedron \(\mathcal O_k(\Lambda)\) is of dimension \(2k-3\).
\end{thm}
    \begin{prop}[\cite{companion}]
    \label{prop:extends to a nbhd}
        For \(\Lambda \in \mathrm{Mat}_{2k\times(k+2)}^>\)  The amplituhedron map \(\widetilde\Lambda\) extends to an open neighborhood \(U\subset \OG{k}{2k}\) of \(\OGnon{k}{2k}\), as a submersion. Thus, the interior of the \(\mathcal O_k(\Lambda)\) and also \(\widetilde\Lambda(U)\) are both submanifolds of the respective spaces of dimension \(2k-3\).
    \end{prop}

\subsubsection{Twistors Variables}
An important set of coordinates on \(\mathcal O_k(\Lambda)\) are the twistor variables.
    \begin{dfn}
    \label{def twist}
        Let \(Y \in \Gr{k}{k+2},\) and pick a matrix representative $M$ for it. For every $i\neq j\in[2k],$ we define the \emph{twistor variable} $\twist{Y}{i}{j}_\Lambda$ by appending to  $M$ of $Y$ the two row vectors $\Lambda_i,\Lambda_j,$ and calculating the determinant.
        The subscript \(\Lambda\) will be omitted when \(\Lambda\) is clear from context. While this definition depends on the choice of $M$, different choices of $M$ differ by an element \(g\in\mathrm{GL(k)}\). Changing $M$ to $gM$ changes the above determinant by $\deg(g).$ Thus, the \emph{projective vector} $(\twist{Y}{i}{j}_\Lambda)_{i< j\in[2k]}\in\mathbb{P}^{\binom{[2k]}{2}-1}$ is well-defined.
        Note also that the twistors are anti-symmetric.       
        
    \end{dfn}
Since we wish to regard points in the amplituhedron through their twistor coordinates, we introduce a space that is invariant under the transformations preserving these coordinates.
\begin{dfn}
\label{def:mathcal u}
Write
\(
W_k \defeq \mathrm{Mat}_{k\times 2k}\times \mathrm{Mat}_{2k\times(k+2)}^> \times \mathrm{Mat}_{k\times (k+2)},
\)
\[
U_k \defeq \{ (C,\Lambda,Y) \in W_k : C \,\eta\,C^\intercal = 0,\, \mathrm{rank}(C) = k,\,C\,\Lambda =Y\},
\]
and
\(
U_{k,\,\Lambda} \defeq \{(C,\Lambda',Y) \in U_k : \Lambda' = \Lambda\}.
\)

Define the following group actions on those spaces:
\begin{enumerate}
    \item  Left \(\mathrm{GL}_k (\mathbb{R})\) action: \(g(C,\Lambda,Y) \defeq (g C,\Lambda,g Y) \).

    \item Right \(\mathrm{GL}_{k+2}(\mathbb{R})\) action: \((C,\Lambda,Y)g \defeq (C,\Lambda g,Y g) \).
 \end{enumerate}   
We can now write
\[
\mathcal{U}_{k}\defeq\raisebox{-0.5ex}{\({\mathrm{GL}_k}(\mathbb R)\)} \backslash \raisebox{0.5ex}{\(U_{k}\)} / \raisebox{-0.5ex}{\({\mathrm{GL}_{k+2}}(\mathbb R)\)}, \qquad
\mathcal{U}_{k,\,\Lambda}\defeq\raisebox{-0.5ex}{\({\mathrm{GL}_k}(\mathbb R)\)} \backslash \raisebox{0.5ex}{\(U_{k,\Lambda}\)} / \raisebox{-0.5ex}{\({\mathrm{GL}_{k+2}}(\mathbb R)\)}.
\]
Let \(\mathcal{U}_{k,\,\Lambda}^\geq \) be the points of \(\mathcal{U}_{k,\,\Lambda} \) with \(C\in \OGnon{k}{2k}\), and define \(\mathcal{U}_{k}^\geq \) similarly.
\end{dfn}
This construction imitates the universal amplituhedron defined in \cite{universal}.
Note that the twistor coordinates, as a projective vector, are invariant under those actions.

Consider the projection on the third coordinate,
\(
\mathcal{U}_{k,\,\Lambda}^\geq \xrightarrow{\pi_Y} \Gr{k}{k+2},\quad
[C,\Lambda,Y] \mapsto Y
\)
which is well-defined as the choice of \(\Lambda\) fixes the right action. The image of the projection is clearly \(\mathcal{O}_k(\Lambda)\). Indeed, we have a natural bijection
\(
\mathcal{O}_k(\Lambda) \cong \faktor{\mathcal{U}_{k,\,\Lambda}^\geq}{\sim}
\)
where two triplets are considered equivalent if they have the same \(Y\). 

A canonical choice of representatives for these equivalence classes is those  with the second matrix being \(\Lambda\). Since \(Y\) is of full rank, we may use the right action to find a representative with \(Y_{k \times (k+2)} = \left(0_{k\times 2} \,\,\, \mathrm{id}_{k\times k }\right)\).
More concretely, for \((C, \Lambda, Y)\in U_k\) we may choose  \(g\in\mathrm{GL}_{k+2}\) such that \((Yg)^{\{1,\,2\}} = 0\) and \((Yg)^{\{3,\,4,\,...,\,k+2\}}\) is a full-rank square matrix. \(g\) is uniquely determined up to \(\mathrm{GL_2}\) on the first two columns, and \(\mathrm{GL_k}\) on the remaining \(k\). Since \(GL_k\) acts freely and transitively on \(k\times k\) full-rank matrices both from the left and from the right, choosing a representative of \(C\in \OGnon{k}{2k}\) and setting \((Yg)^{\{3,\,4,\,...,\,k+2\}} = \mathrm{id}_k\) fixes this \(\mathrm{GL}_k\) freedom. That is, by fixing \(Y = Y_0 \defeq \left(0_{k\times 2} \,\,\, \mathrm{id}_{k\times k }\right)\), we are left with:

\begin{enumerate}
    \item  Left \(\mathrm{GL}_k (\mathbb{R})\) action: \(g(C,\Lambda,Y_0) = (g C,\Lambda,g Y_0) \sim  (g C,\Lambda g^{\prime-1}, Y_0) \) where \(g^\prime\) is the square \(k+2\) block matrix with identity for the \(\{1,\,2\}\) block and \(g\) for the \(\{3,\,...,\,k+2\}\) block.

    \item Right \(\mathrm{GL}_{2}(\mathbb{R})\) action: \((C,\Lambda,Y_0)g \defeq (C,\Lambda g^\prime,Y_0 \,g^\prime) = (C,\Lambda g^\prime,Y_0 )\) where \(g^\prime\) is the square \(k+2\) block matrix with \(g\)  for the \(\{1,\,2\}\) block and identity for the \(\{3,\,...,\,k+2\}\) block.
\end{enumerate}
This gives rise to the following definition:
\begin{dfn}
\label{def lambda}
    For \([C,\Lambda,Y]\in \mathcal{U}_k^\geq\), define
    \(\lambda \defeq C^\perp\cap\Lambda^\intercal,\)
    as linear spaces.
\end{dfn}
This view of the amplituhedron is based on the projection through \(Y\) for the momentum twistor (standard) amplituhedron in \cite{Unwinding}, and a similar definition employed in \cite{Karp_2017}.
\begin{prop} [\cite{Unwinding}, Section~2]
\label{prop:twist pluckers}
    For \([C,\Lambda,Y]\in \mathcal{U}_k^\geq\), one has
    \(
    \Delta_{\{i,\,j\}}(\lambda) = \twist{Y}{i}{j}.
    \)

    As a consequence, by Proposition~\ref{prop:plucker relations}, the twistors satisfy the Pl\"ucker relations:
    \[\forall i,\,j_1<j_2<j_3 \in [2k],~
     \twist{Y}{i}{j_1} \twist{Y}{j_2}{j_3}+\twist{Y}{i}{j_3} \twist{Y}{j_1}{j_2} = \twist{Y}{i}{j_2} \twist{Y}{j_1}{j_3}
    \]
    
\end{prop}

Denote by \(\Twist{Y}\) the real \(2k \times 2k\) matrix defined by 
    \(\Twist{Y}_{i,\,j} \defeq \twist{Y}{i}{j} .\)

\begin{prop}[\cite{companion}]
\label{prop:twistor lambda}
    \(\mathrm{Span}(\Twist{Y})=\lambda\) and both are two dimensional. Specifically, for any \(i,\,j\) with \(\twist{Y}{i}{j} \neq 0\) we have that \(\lambda\) is spanned by the \(i\) and \(j\) rows of \(\Twist{Y}\).
\end{prop}
\begin{prop}[\cite{companion}]
\label{prop:Consecutive Twistors}
For \([C,\Lambda,Y]\in \mathcal{U}_k^\geq\), and \(i  = 1,\,...,\, 2k-1\) we have
 \(\twist{Y}{i}{i+1} \geq 0,\)
and
\((-1)^k \twist{Y}{1}{\left(2k\right)} \geq 0.\)
Furthermore, all of the inequalities above are strict when \(C\) is in the interior of the positive Orthogonal Grassmannian.
\end{prop}

\begin{prop}[\cite{Huang_2022}, Section~3.1,\cite{he2022momentum}, Section~3.1,\cite{companion}]
    \label{prop:mom cons}
    For \([C,\Lambda,Y]\in \mathcal{U}^\geq_k\), we have
    \(
    \lambda \, \eta \, \lambda^\intercal = 0
    \)
    This equation will be referred to as \emph{momentum conservation}.
\end{prop}

\begin{lem}[\cite{companion}]
\label{lem:twist in C}
We have \(C \subset \orth {\lambda},\) \(\lambda\eta \subset C,\) and \(\mathrm{dim}(\lambda \eta) = 2\).
\end{lem}

\begin{prop}[\cite{companion}]
\label{prop:Bdiffeo} Define
        \(
    \mathrm{OG}_2( \Lambda^\intercal) \defeq  \{\lambda\subset \Lambda^\intercal\,:\, \mathrm{dim}(\lambda) = 2,\,\lambda \eta \lambda^\intercal=0 \}
    .\)
    We have an injection
    \(\mathcal{O}_k(\Lambda) \xhookrightarrow{} \mathrm{OG}_2( \Lambda^\intercal),\)
    by \(Y \mapsto \lambda,\)
    which is a diffeomorphism onto its image.
\end{prop}

Following the definition of the \(\mathcal{ B} \)-amplituhedron in \cite{Karp_2017}, we define
    The \emph{orthogonal momentum B-amplituhedron} is the image of \(\mathcal{O}_k(\Lambda)\) under the injection above.

\begin{dfn}
\label{def:Ylk}
    For \(\Lambda\in \mathrm{Mat}^>_{2k\times(k+2)}\) write \(\lambda_\Lambda:\Gr{k}{k+2}\rightarrow \Gr{2}{2k}\), by \(Y\mapsto \lambda\) as defined in Definition~\ref{def lambda}, which is well-defined as by Proposition~\ref{prop:twistor lambda}. Define 
    \[\mathbf Y^k_\Lambda \defeq  \Bigl\{Y\in \Gr{k}{k+2}\Big| \lambda_\Lambda(Y)\eta\lambda_\Lambda(Y)^\intercal = 0 \Bigl\}.\]
 \end{dfn}
\begin{prop}[\cite{companion}]
\label{prop:zerolocus submfld}
    \(\mathbf Y^k_\Lambda\subset \Gr{k}{k+2}\) is a smooth submanifold of dimension \(2k-3\).
\end{prop}

\section{Promotion}\label{sec:prom}
The central constructions of \cite{amptriag,even2023cluster} involve an operation termed \emph{promotion} at the level of functions, whose geometric counterparts are the BCFW map and the upper BCFW map. Interestingly, analogous maps appear in the orthogonal amplituhedron setting, although their explicit forms differ substantially. We now turn to these analogues.

Our approach is to analyze points in \(\mathcal{O}_k(\Lambda)\) inductively by applying local moves on OG graphs. To this end, we must first specify how these local moves act on \(\mathcal{U}_k\) and determine their induced transformations on the twistor variables. A precise description of the local moves on the orthogonal Grassmannian and the associated combinatorial objects is provided in Appendix~\ref{app:local moves}.

\subsection{The Basic Moves}
\subsubsection{The \texorpdfstring{\(\mathrm{Rot}\)}{Rot} Move}

\begin{dfn}
\label{def:rot of lambda}
For \(\Lambda \in \mathrm{Mat}^>_{2k\times(k+2)}\) define \(\mathrm{Rot}_{i,i+1}(\alpha)(\Lambda) \defeq R_{i,i+1}(\alpha)^{-1}\Lambda  \) for \(i<2k\) and  \(\mathrm{Rot}_{2k,1}(\alpha)(\Lambda) = R_{1,2k}(-(-1)^k\alpha) ^{-1}\Lambda \). Notice we have that for \(C\in \OGnon{k}{2k}\), 
\[C\Lambda =\mathrm{Rot}_{i,i+1}(\alpha)(C)\mathrm{Rot}_{i,i+1}(\alpha)(\Lambda).\] 
    For \([C,\Lambda,Y]\in \mathcal{U}_k\), define
    \(
    \mathrm{Rot}_{i,i+1}^{-1} (\alpha) [C,\Lambda,Y] = [\mathrm{Rot}_{i,i+1}^{-1}(\alpha) C,\mathrm{Rot}_{i,i+1}^{-1}(\alpha)\Lambda,Y].
    \)
\end{dfn}

Just as \(\mathrm{Rot}_{i,i+1}\) preserves the positivity of \(C\), \(\mathrm{Rot}_{i,i+1}^{-1}\) preserves the positivity of \(\Lambda\). It clearly commutes with the right and left actions so it is well-defined.

\begin{prop}[\cite{companion}]
\label{prop:promoting twistors}
For \([C,\Lambda,Y]\in 
\mathcal{U}_k\), \([C^\prime,\Lambda^\prime,Y^\prime] = \mathrm{Rot}_{i,i+1}^{-1}(\alpha)[C,\Lambda,Y]\) we have
    \(
  \Twist{Y^\prime}_{\Lambda^\prime} =  \Twist{Y}_\Lambda\,R^{-1},
    \)
    where \(R = R_{i,i+1}(\alpha)\) for \(i<2k\) and  \(R = R_{1,2k}(-(-1)^k\alpha)\) for \(i=2k\).
\end{prop}

\subsubsection{The \texorpdfstring{\(\mathrm{Cyc}\)}{Cyc} Move}

\begin{obs}
\label{cyc plu}
By definition it is clear that we have
    \(\Delta_I(C) = \Delta_{\mathrm{Cyc}(I)}(\mathrm{Cyc}(C)).\)
    
\end{obs}

\begin{dfn}
    Define the action of \(\mathrm{Cyc}\) on \(\Lambda\) to be the same as on \(C\) but cycling the rows instead of the columns. That is, for \(\Lambda\in\mathrm{Mat}_{2k\times(k+2)}\) with rows \(\Lambda_i\)
    define
        \[
     \mathrm{Cyc}_k(\Lambda) =\scalebox{0.8}{\(
         \left(
    \begin{array}{ccc}
    \horzbar&-(-1)^k\Lambda_{2k}&\horzbar\\
         \horzbar&\Lambda_1&\horzbar  \\
         \horzbar&\Lambda_2&\horzbar\\
         &\vdots&\\
         \horzbar&\Lambda_{2k-1}&\horzbar\\
    \end{array}
    \right)\)}.
    \]
     The resulting matrix is positive by the previous observation.
\end{dfn}
\begin{prop}[\cite{companion}]
\label{cyc twist}
    We have
    \[
    \twist{(C\,\Lambda)}{i}{j}_\Lambda = (-(-1)^{k})^{\delta_{i,\,2k}+\delta_{j,\,2k}}\twist{(\mathrm{Cyc}(C)\,\mathrm{Cyc}(\Lambda))}{i+1}{j+1}_{\mathrm{Cyc}(\Lambda)},
    \]
    when the indices are considered modulo \(2k\). Thus we can write
    \(
    \mathrm{Cyc}(\Twist{C\,\Lambda}) = \Twist{\mathrm{Cyc}(C)\,\mathrm{Cyc}(\Lambda)}
    \).
\end{prop}

\begin{dfn}
        For \([C,\Lambda, Y]  \in \mathcal{U}_k\) define
    \(
    \mathrm{Cyc}[C,\Lambda, Y] = [\mathrm{Cyc}(C),\mathrm{Cyc}(\Lambda), \mathrm{Cyc}(C)\mathrm{Cyc}(\Lambda)]
    \).
    As the move permutes columns of \(C\) and rows of \(\Lambda\), it is clearly invariant under the left and right actions. 
\end{dfn}

\subsubsection{The \texorpdfstring{\(\mathrm{Inc}\)}{Inc} Move}
We now describe how the \(\mathrm{Inc}_i\) move, defined in~\ref{inc subsubsection}, acts on \(\mathcal{U}_k\). To do so, we must first specify its action on \(\Lambda\). Throughout this section, the indices will be viewed as elements of an arbitrary set of integers, not necessarily consecutive.

While the inverse of the \(\mathrm{Inc}\) move is straightforward to obtain when applied to matrices of the form \(\mathrm{Inc}_i(C)\), it is less clear how to modify the corresponding \(\Lambda\) matrices so as to preserve the twistor variables while simultaneously changing the value of \(k\). We therefore introduce the \(\mathrm{Inc}_i\) operation on \(\Lambda\) matrices as follows:

Since we are discussing this for the purposes of \(\mathcal{U}_k\), using the right action, we can assume without loss of generality that \(\Lambda_i = e_i\) and \(\Lambda_{i+1} = e_{i+1}\),  where \(e_j\) are the standard basis vectors.

For sets of indices \(A\) and \(B\) let \( \mathrm{Mat}_{A\times B}\) the set of \(A\times B\) indexed real matrices, that is, a choice of real number of each pair in \(A\times B\). Let \(N = [2k]\), \(K = [k+2]\), \(I = \{i,\,i+1\}\), and \(\Lambda^\prime\in \mathrm{Mat}_{(N\setminus I )\times (K \setminus \{i+1\}}\) be such that
\begin{align*}
    &\Lambda^{\prime j} = \Lambda^{j}_{N\setminus I} \,\, &\text{for} \, j<i\\
    &\Lambda^{\prime i} = \Lambda^{i}_{N \setminus I} - \Lambda^{i+1}_{N \setminus I} \,\, &\text{for} \, j=i\\
    &\Lambda^{\prime j} = \Lambda^{j}_{N \setminus I} \,\, &\text{for} \, j>i+1.
\end{align*}
\begin{dfn}
Let us now define
    \(\mathrm{Inc}_i^{-1}(\Lambda) = \Lambda^\prime\)
\end{dfn}
\begin{prop}[\cite{companion}]
    The action above is well-defined for general \(\Lambda\in \mathrm{Mat}^>_{2k\times(k+2)}\).
\end{prop}

\begin{prop}[\cite{companion}]
\label{prop:inc plu}
For every \(J \in \binom{N\setminus I}{k+1}  \),

\[\mathrm{det}(\mathrm{Inc}_i^{-1}(\Lambda)_J) = \mathrm{det}(\Lambda_{J\cup \{i\}}) + \mathrm{det}(\Lambda_{J\cup \{i+1\}})
\]
\end{prop}

\begin{coro}
\label{coro:inc pos}
    Thus if \(\Lambda\in \mathrm{Mat}^>_{2k \times (k+2)}\), then \(\mathrm{Inc}_i^{-1}(\Lambda)\in \mathrm{Mat}^>_{2k \times (k+2)}\). 
\end{coro}

\begin{prop}[\cite{companion}]
\label{prop:TaylorTwist}
For \(C\in\OGnon{k}{2k}\), \(\Lambda^\prime\in\mathrm{Mat}^>_{(2k+2)\times(k+3)}\), \( C^\prime = \mathrm{Inc}_i(C)\), \( \Lambda = \mathrm{Inc}_i^{-1} (\Lambda^\prime)\), \(Y = C\Lambda\), \(Y^\prime = C^\prime \Lambda^\prime\) and \(j_1,j_2\in[2k+2]\), we have that 
\[
\twist{Y^\prime}{j_1}{j_2}_{\Lambda^\prime}=
\begin{cases}
    \twist{Y}{j_1}{j_2}_{\Lambda}&j_1,j_2<i\\
    -\twist{Y}{j_1}{j_2-2}_{\Lambda}& j_1<i,i+1< j_2\\
    -\twist{Y}{j_1-2}{j_2}_{\Lambda}& j_2<i,i+1< j_1\\
\twist{Y}{j_1-2}{j_2-2}_{\Lambda}& i+1< j_1,j_2.
\end{cases}
\]
We also have that for \(j\in[2k+2]\),
\(
\twist{Y^\prime }{j}{i}_{\Lambda^\prime} =-(-1)^{\delta_{i,\,2k}}\twist{Y^\prime }{j}{i+1}_{\Lambda^\prime}
\) and
\(
\twist{Y^\prime}{i}{i+1}_{\Lambda^\prime} = 0.
\)
\end{prop}
\begin{coro}
\label{coro:inv inc inj}
    Let \(C_1^k, C_2^k\in \OGnon{k}{2k}\) and \(\Lambda^{k+1}\in\mathrm{Mat}^>_{(2k+2)\times(k+3)}\) be such that 
    \[\widetilde{\Lambda}^{k+1}(\mathrm{Inc}_i(C_1^k))=\widetilde{\Lambda}^{k+1}(\mathrm{Inc}_i(C_2^k)),\]
     for some \(i\in[2k]\). Then we have that \(\widetilde{\Lambda}^{k}(C_1^k)=\widetilde{\Lambda}^{k}(C_2^k),\)
     where \(\Lambda^k\defeq \mathrm{Inc}_i^{-1}(\Lambda^{k+1})\).
\end{coro}
\begin{proof}
    Write \(Y_1 \defeq \widetilde{\Lambda}^{k}(C_1^k)\) and \(Y_1 \defeq \widetilde{\Lambda}^{k}(C_2^k)\). By Proposition
   ~\ref{prop:TaylorTwist} we have that \(\twist{Y_1}{n}{m}=\twist{Y_2}{n}{m}\) for all \(n,m\in[2k]\). By Propositions~\ref{prop:Bdiffeo} and~\ref{prop:twist pluckers}, this means \(Y_1=Y_2\).
\end{proof}

\begin{dfn}
    For \([\mathrm{\mathrm{Inc}}_i(C),\Lambda,Y]\in \mathcal{U}_{k+1}^\geq\), define \(
    \mathrm{\mathrm{Inc}}_i^{-1} [\mathrm{\mathrm{Inc}}_i(C),\Lambda,Y] = [C,\,\mathrm{\mathrm{Inc}}_i^{-1}(\Lambda),Y^\prime],
    \)
    where \(Y^\prime = C \,\mathrm{\mathrm{Inc}}_i^{-1}(\Lambda)\).
\end{dfn}
\begin{rmk}
    When we write
    \( \mathrm{Inc}_i[C,\Lambda, Y] = [C^\prime,\Lambda^\prime, Y^\prime]\) or \( [C,\Lambda, Y] = \mathrm{Inc}_i^{-1}[C^\prime,\Lambda^\prime, Y^\prime],\)
    we mean to say that \(C^\prime = \mathrm{Inc}_i (C)\), \(\Lambda =\mathrm{Inc}_i^{-1}(\Lambda^\prime)\), \(C^\prime\Lambda ^\prime = Y^\prime\) and \(C \Lambda = Y\).
\end{rmk}
\begin{rmk}
\label{rmk:inv inc abuse}
    It is important to note that there is an abuse of notation at play. Note that \(\mathrm{Inc}_i^{-1}\) is not explicitly defined for all of \(\OGnon{k}{2k}\), but only on those matrices that are the in the image of \(\mathrm{Inc}_i\). As such, \(\mathrm{Inc}_i^{-1}\) is not the true inverse of \(\mathrm{Inc}_i\), but only a left inverse.
\end{rmk}

\subsubsection{Summary}

We have defined the \(\mathrm{Rot}\), \(\mathrm{Inc}\), and \(\mathrm{Cyc}\) moves on \(\mathcal{U}_k\) in a way that corresponds to their action on \(\OGnon{k}{2k}\) and OG graphs, and have seen their effect on the twistors. To summarize:
\begin{prop}
\label{prop:prom rules}
    The effects of the \(\mathrm{Rot}\), \(\mathrm{Cyc}\), and \(\mathrm{Inc}\) moves are as follows:

\begin{itemize}

    \item  The \(\mathrm{Rot}_{i,i+1}(\alpha)\) move adds a vertex with the angle \(\alpha\) between the vertices \(i\) and \(i+1\), and acts both on \(C\) and \(\lambda\), and the \(\Twist{Y}\) by the corresponding hyperbolic rotation.
    \item  The \(\mathrm{Cyc}\) move rotates the index labels on the graph anticlockwise and keeps the same angle for the internal vertices. It cycles the columns of \(C\) to the left and then adds a sign to the last  column if \(k\) is even. It cycles the columns (and rows) of \(\Twist{Y}\) to the left (top) and a sign to the last  column (row) if \(k\) is even. 
    \item The \(\mathrm{Inc}_i\) move adds two new vertices between the \(i-1\) and \(i\) vertices and connects them with an edge. If we have an edge going form \(i\) to \(i+1\) that is disconnected from the rest of the graph, \(\mathrm{Inc}_i^{-1}\) removes the edge, keeps the same angles for the vertices, and acts appropriately on \(C\). The twistors that do not contain \(i,\,i+1\) remain the same except for the indices greater then \(i+1\) going down by 2, and a minus sign added for each such index. In the remaining twistors, the \(i\) and \(i+1\) indices are equivalent (except up to a \((-1)^k\) sign when \(i=2k\)), and the \(i,\,i+1\) twistor is zero.
\end{itemize}
\end{prop}

\subsection{Using the Moves}

\subsubsection{Mandelstam Variables}
Another set of useful variables on the amplituhedron are the Mandelstam variables, which in connection to the ABJM amplituhedron are discussed in \cite{Huang_2022,he2022momentum}.
\begin{dfn}
\label{def mand}
    The \emph{Mandelstam variables} for \(I \subset [2k]\) on \(\mathcal{U}_k\) are defined as

    \[
    S_I(\Lambda,Y) \defeq \sum_{\{i,\,j\} \subset I} (-1)^{i-j+1} \twist{Y}{i}{j}^2
    \]
\end{dfn}

\begin{prop}[\cite{companion}]
\label{s rot}
    For \([C^\prime,\Lambda ^\prime, Y ^\prime] = \mathrm{\mathrm{Rot}_{i,i+1}(\alpha)}[C,\Lambda,Y]\), with \(\alpha>0\) and \\\([C,\Lambda,Y]\in\mathcal{U}_k^\geq\), \(I\subset [2k]\), we have that if \(|\{i,\,i+1\}\cap I |=0,\,2\) then 
    \(
    S_I(\Lambda^\prime,Y^\prime) = S_I(\Lambda,Y).
    \)
\end{prop}

\begin{prop}[\cite{companion}]
\label{s inc}
    For \(\mathrm{Inc}_i^{-1}[C^\prime,\Lambda ^\prime, Y ^\prime] = [C,\Lambda,Y]\), with \(C^\prime =\mathrm{Inc}_i( C)\) and \\\([C,\Lambda,Y]\in\mathcal{U}_{k+1}^\geq\), \(I\subset [2k]\), \(I^\prime = I\setminus\{i,\,i+1\}\) 
    for \(|\{i,\,i+1\}\cap I |=0,\,2,\)
    we have \(
    S_I(\Lambda^\prime,Y^\prime) = S_{I^\prime}(\Lambda^\prime,Y^\prime)= S_{I^\prime}(\Lambda,Y).
    \)
   
\end{prop}

\begin{prop}[\cite{companion}]
\label{s cyc}
    For \([C^\prime,\Lambda ^\prime, Y ^\prime] = \mathrm{Cyc}[C,\Lambda,Y]\), and \([C,\Lambda,Y]\in\mathcal{U}_{k+1}^\geq\), \(I\subset [2k]\), we have that
    for \(|\{i,\,i+1\}\cap I |=0,\,2,\)
    \(
    S_I(\Lambda^\prime,Y^\prime) = S_{\mathrm{Cyc}(I)}(\Lambda,Y).
    \)
\end{prop}

\begin{prop}[\cite{companion}]
    \label{prop:S orth}
    For \([C,\Lambda,Y]\in\mathcal{U}_k^\geq\), \(I\subset [2k]\),  we have \(
    S_I(\Lambda,Y) = S_{[2k]\setminus I}(\Lambda,Y).
    \)
    Morever, if \(C\) belongs to the orthitroid cell \(\Omega_\Gamma\) which corresponds to the permutation \(\tau\), and if \(|I\setminus \tau(I)|<2\), then
    \(
    S_I(\Lambda,Y) = 0.
    \)
\end{prop}

\subsubsection{Abstract Twistors}
Let us first set up notations that would help us discus the algebra of twistor variables without referring to a specific point in the amplituhedron.
\begin{dfn}

Let 
\(\Gr{k}{\mathbb N} := \,_{\mathrm{GL}_{k}(\mathbb R )}\mkern-.5mu\backslash\mkern-2mu^{\mathrm{Mat}_{[k] \times \mathbb N}^*(\mathbb R )}, \)
where \(\mathrm{Mat}_{[k] \times \mathbb N}^*(\mathbb R )\) is the space of full rank real matrices with \(k\) rows and countably many columns. For an element \(\widetilde\lambda\in \Gr{2}{\mathbb N}\) write \(\langle i \, j\rangle(\widetilde\lambda) \defeq \Delta_{\{i,j\}}(\widetilde\lambda)\), the \(i,j\)-th Pl\"ucker coordinate. The functions \(\langle i \, j\rangle\) would be referred to as \emph{abstract twistors}.
\end{dfn}
\begin{obs}
The abstract twistors  form a commutative algebra satisfying the following relations:
\begin{itemize}
    \item \textbf{Anti-symmetry:} \(\langle i \, j \rangle = -\langle j \, i \rangle\).

    \item  \textbf{Pl\"ucker relations:} For any \(i,\,j_1,\,j_2,\,j_3 \in \mathbb N\) with \(j_1<j_2<j_3\),

    \[
    \langle i \, j_1 \rangle \langle j_2 \, j_3 \rangle - \langle i \, j_2 \rangle \langle j_1 \, j_3 \rangle + \langle i \, j_3 \rangle \langle j_1 \, j_2 \rangle = 0 .
    \]
\end{itemize}
\end{obs}
Notice there is a natural embedding \(\iota:\mathrm{OG}_2( \Lambda^\intercal) \xhookrightarrow{} \Gr{2}{\mathbb N}\) that is induced by the embedding of \(\mathrm{Span }(\Lambda^\intercal)\subset \mathbb R^{2k}\subset \mathbb R^{\mathbb N}\). By Proposition~\ref{prop:mom cons} we have an injection  \(\varphi_\Lambda :\mathcal{O}_k(\Lambda) \xhookrightarrow{} \mathrm{OG}_2( \Lambda^\intercal),\) by \(Y \mapsto \lambda\).
\begin{obs}
Those combine into an injection \(\iota\circ\varphi_\Lambda : \mathcal{O}_k(\Lambda)\rightarrow \Gr{2}{\mathbb N}\) by \(Y\mapsto \widetilde\lambda\), where \(\widetilde\lambda\) is the unique element \(\widetilde\lambda\in\Gr{2}{\mathbb N}\) defined by
\(
\langle i\,j\rangle (\widetilde\lambda) =  
\twist{Y}{i}{j}_\Lambda\) for \(i,j\in[2k]\) and \(0\) otherwise.
\end{obs}
\begin{dfn}
Let  let \(\mathcal F\) be the space of functions \(U\rightarrow\widehat{\mathbb C }\) where \(\widehat{\mathbb C }\) denotes the Riemann sphere, for some  \(U\subset\Gr{k}{\mathbb N}\). Those can be expressed as formal expressions in abstract twistors.

For \(f\in \mathcal F\), the \emph{index-support} \(\mathcal I(f)\) of \(f\), is defined to be the set of indices appearing in the abstract twistors on which $f$ depends. The index-support of \( \mathcal F\)-valued matrices and vectors is defined as the union of the index-supports of their entries.  Let  \(\mathcal F_I\) for \(I\subset \mathbb N\) be the space of such functions with index-support contained in \(I\).

\begin{dfn}
\label{abs substit def}
    For \(f\in \mathcal F _{[2k]}\) and \([C,\Lambda,Y]\in \mathcal U _k ^\geq \) write \(f(\Lambda, Y)\defeq f(\iota\circ\varphi_\Lambda(Y))\in\widehat{\mathbb C }\) . We say that \(f\in\mathcal F\) is defined on \([C,\Lambda,Y]\in\mathcal U_k^\geq\) if \(f(\Lambda, Y) \) is defined and finite.
\end{dfn}

We define the action of the three moves on \(\mathcal F\), in analogy with the action on the twistors:
\begin{itemize}
    \item \underline{\(\mathrm{Rot}\):}
    
    \[
\mathrm{Rot}_{i,i+1}(\alpha)\langle{n\,m}\rangle =
\begin{cases} 
     \langle{n\,m}\rangle & n,\,m\neq i,\,i+1 \\
     \langle{i\,m}\rangle \cosh \alpha -\epsilon_i\,  \langle{i+1\,m}\rangle \sinh \alpha & n=i,\, m\neq i+1\\
     \langle{i+1\,m}\rangle \cosh \alpha -\epsilon_i\,  \langle{i\,m}\rangle \sinh \alpha & n=i+1,\, m\neq i\\
     \langle{n\,i+1}\rangle \cosh \alpha -\epsilon_i\,  \langle{n\,i}\rangle \sinh \alpha  & n\neq i, \, m =i+1\\
     \langle{n\,i}\rangle \cosh \alpha -\epsilon_i\,  \langle{n\,i+1}\rangle \sinh \alpha  & n\neq i+1, \, m =i\\
    \langle{n\,m}\rangle & n,\,m= i,\,i+1 \\
   \end{cases}
\]
where \(\epsilon_i \defeq -(-1)^k\) if \(i=2k\) and \(1\) otherwise.

\item \underline{\(\mathrm{Cyc}\):}
\(
\mathrm{Cyc}_k \langle n\,m\rangle = 
\epsilon_n\epsilon_m \langle n+1\,m+1\rangle
\)
where \( n+1,\,m+1\) are considered mod \(2k\). 
\item \underline{\(\mathrm{Inc}\):}
\begin{align*}
    &\mathrm{Inc}_i \langle n\,m\rangle = 
\begin{cases} 
     -\langle{n\,m+2}\rangle & n<i \leq m \\
     -\langle{n+2\,m}\rangle  & m<i \leq n\\
     \langle{n\,m}\rangle  & n,m<i\\
     \langle{n+2\,m+2}\rangle  & i\leq n,m,\\
   \end{cases}\\
   &\mathrm{Inc}_i^{-1} \langle n\,m\rangle = 
\begin{cases} 
     0 & \{n,m\}\cap \{i,i+1\} \neq\emptyset\\
     -\langle{n\,m-2}\rangle & n<i,i+1 < m \\
     -\langle{n-2\,m}\rangle  & n<i,i+1 < m\\
     \langle{n\,m}\rangle  & n,m<i\\
     \langle{n-2\,m-2}\rangle  & i+1< n,m.\\
   \end{cases}
\end{align*}
\end{itemize}
\end{dfn}
\begin{rmk}
\label{rmk:abuse of notation 2}
    In the spirit of Remark~\ref{rmk:inv inc abuse}, notice \(\mathrm{Inc}_i^{-1}\) not the true inverse of \(\mathrm{Inc}_i\), but only a left inverse.
\end{rmk}
\begin{obs}
    For \(\Lambda \in \mathrm{Mat}_{k\times2k}\), \(C\in \OGnon{k}{2k}\), and \(Y = C\Lambda\), \(\langle i \, j \rangle(\Lambda ,Y)\) are smooth function in \(\Lambda\), \(C\), and \(Y\).
\end{obs}
\begin{dfn}
    Let \(a\in \mathcal F_{[2k]}\), \(\Lambda \in \mathrm{Mat}_{2k\times(k+2)}\), \(C\in \OGnon{k}{2k}\),\(Y = C\Lambda\). If \(a(\Lambda ,Y)\) is a smooth (positive) function in \(\Lambda\), \(C\), and \(Y\)  we would say \(a\) is \emph{smooth (positive)}.

    Let \(\Gamma\) be a \(k\)-OG graph, and  \(C\in \Omega_\Gamma \subset  \OGnon{k}{2k}\). If \(a(\Lambda ,Y)\) is a smooth function in \(\Lambda\), \(C\), and \(Y\), then we would say \(a\) is \emph{smooth (positive)} on \(\Omega_\Gamma\) or \(\Gamma\).
\end{dfn}
\begin{obs}
    \label{obs:somth prom of smooth is smooth}
    Let \(a\in \mathcal F\) is smooth on \(\Gamma\), an OG graph. Suppose \(\Gamma^\prime = G(\Gamma)\) where \(G = \mathrm{Cyc},\mathrm{Inc_i}\), or \(\mathrm{Rot}_{i,i+1}(b)\) where \(b\in\mathcal{F}\) is smooth and positive on \(\Gamma^\prime\). Then \(G(a)\) is smooth on \(\Gamma^\prime\). 
\end{obs}
\begin{proof}
    Immediate from the fact that for \(C\in \OGnon{k}{2k}\), \(G(C)\) is a smooth function in \(C\) and \(b\).
\end{proof}
\begin{dfn}
    For brevity, we will define for \(i,j\in[2k]\) and \(n,m\in\mathbb Z\),
    \[
    \langle (i + 2kn)  \,(j+2km)\rangle_{(2k)} \defeq (-(-1)^k)^{n+m}\langle i  \,j\rangle,
    \]
    and
    \[
    \twist{Y}{(i + 2kn) }{(j+2km)}_{(2k)}\defeq (-(-1)^k)^{n+m}\twist{Y}{i}{j}.
    \]
We also define, for \(I\subset\mathbb N\),
    \(
    S_I \defeq \sum_{\{i,\,j\} \subset I} (-1)^{i-j+1}\langle i\,j\rangle^2
    \).
\end{dfn}
\begin{obs}
\label{obs:s rot}
    For  \(I\subset[2k]\) and \(\left|\{i,i+1\}\cap I\right| =0,2\) considered mod \(2k\), we have that
    \(
    \mathrm{Rot}_{i,i+1} S_I = S_{I}.
    \)
\end{obs}
\begin{proof}
    This evident by a trivial computation from~\ref{abs substit def}.
\end{proof}
\begin{dfn}
\label{mat vec abs subst def}
    We write \(\mathcal F^{n}\) for the set of \(\mathcal F_I\)-valued \(n\)-vectors
    and \(\mathcal F^{j\times n}\) for the set of \(\mathcal F_I\)-valued \(j\times n\) matrices.

    For elements \(v\in\mathcal F^{n}\) and \(M\in\mathcal F^{j\times n}\) and a point \([C,\Lambda, Y]\in\mathcal U^\geq_k\) define \(v(\Lambda, Y)\) and \(M(\Lambda,Y)\) to be the vector or matrix achieved by preforming the evaluation entry-wise as in Definition~\ref{abs substit def}. Those may interchangeably refer to projective vectors or elements of the Grassmannian. \(v(\Lambda, Y)\) and \(M(\Lambda,Y)\) are said to be defined on \([C,\Lambda, Y]\in\mathcal U^\geq_k\) if all of their entries are defined on \([C,\Lambda, Y]\).
\end{dfn}
\begin{dfn}
   On \(\mathrm{Mat}_{j\times2k}\),  we define the action of the moves as on \(C\in \OGnon{k}{2k}\). On \(\mathbb R^{2k}\),  we define the action of the moves as acting on rows of elements from \(\mathrm{Mat}_{j\times2k}\), that is, like on \(\mathrm{Mat}_{1\times2k}\)  except for the \(\mathrm{Inc_i}\) move where we define \(\mathrm{Inc_i}:\mathbb R^{2k} \rightarrow \mathbb R^{2k+2}\) by adding two entries of padding zeros after the \(i\)-th position. For the sake of completion we will define the moves as acting on \(\mathbb R \) as the identity.

   On \(\mathcal F_I^{j\times 2k}\),  we define the action of the moves as follows: first act on the entries, then act on the resulting matrix as on \(\mathrm{Mat}_{j\times2k}\). On \(\mathcal F_I^{2k}\),  we define the action of the moves as follows: first act on the entries, then act on the resulting vector as elements of \(\mathbb R^{2k}\).
\end{dfn}

\begin{obs}
      \label{prom index supp obs}
      Let \(f\) be an element of \(\mathcal F\), a \(\mathcal F\)-valued matrix or a \(\mathcal F\)-valued vector, and consider a series \(G\) of moves, such that for any rotation in the series the angles themselves have index-support contained in \(J\). Then \(\mathcal{I}(G(f))\subseteq G(\mathcal I( f))\cup J\).
\end{obs}
\begin{prop}
\label{prop:com diag}
     Let \(G\) be a series of moves, \(I\subset [2k]\), and \([C,\Lambda,Y] \in U^\geq_k\). Write \(G [C,\Lambda,Y] = [C_0,\Lambda_0,Y_0]\). For \(f \in \mathcal F _{I},   \mathcal F_I^{ 2k},\) or \(\mathcal F_I^{j\times 2k}\), we have   \((G f)(\Lambda_0,Y_0) = G (f(\Lambda,Y)).\)
     
\end{prop}
\begin{proof}
    Immediate form Proposition~\ref{prop:prom rules}.
\end{proof}

\begin{coro}
    If \(M\in \mathcal F_I^{j\times 2k}\) is defined on \([C,\Lambda,Y]\in\mathcal U_k^\geq\), then \(G M\) is defined on \(G[C,\Lambda,Y]\). If \(\mathbf v\in \mathcal F_I^{ 2k}\) is defined on \([C,\Lambda,Y]\in\mathcal U_k^\geq\), then \(G \mathbf v\) is defined on \(G[C,\Lambda,Y]\).
\end{coro}

Our strategy for inverting the amplituhedron map on orthitroid cells is to associate to each cell an element \(M \in \mathcal{F}_{[2k]}^{k \times 2k}\) such that for any \([C,\Lambda,Y] \in \mathcal{U}_k^\geq\) with \(C\) lying in that cell, we have
\(
C = M(\Lambda,Y)
\)
as elements of \(\Gr{k}{2k}\). The space \(\mathcal{F}_{[2k]}\) denotes a natural collection of functions on the amplituhedron that arise in the inversion of the amplituhedron map on BCFW cells and their boundaries. In practice, we do not require all such functions, but only those that can be expressed using quotients and square roots. See Section~\ref{sec:inj} for further details.

\subsubsection{The \texorpdfstring{\(\mathrm{Arc}\)}{Arc} Move}
Having introduced the basic local moves in Appendix~\ref{app:local moves}, we now define a compound move that plays a central role in the study of the amplituhedron and the BCFW cells. The significance of this move will become clear in Section~\ref{sec:inj}. This operation can be viewed as the ABJM analogue of the \emph{upper embedding} introduced in \cite{amptriag} for the planar \(\mathcal{N}=4\) SYM amplituhedron. We begin by introducing some useful expressions.

\begin{dfn}
\label{def:vec angle exp}
For \(n =2,3,4\) and \(\ell\in[2k]\), define \(\mathrm v_{n,\ell,k}\in\mathcal F^{2k}\) as follows (with indices considered mod \(2k\)):
\begin{align*}
    \mathbf{v}_{2,\ell,k}^i& \defeq \begin{cases}
        1&i=\ell,\ell+1\\
        0 &\mathrm{otherwise}\\
    \end{cases}\\
    \mathbf{v}_{3,\ell,k}^i& \defeq \begin{cases}
        \langle i_2\,i_3\rangle_{(2k)}&i=i_1\\
        -\varepsilon_{i}\langle i_1\,i_3\rangle_{(2k)}&i=i_2\\
        \varepsilon_{i}\langle i_1\,i_2\rangle_{(2k)}&i=i_3\\
        0 &\mathrm{otherwise}\\
    \end{cases}\\
      \mathbf{v}_{4,\ell,k}^i& \defeq
       \begin{cases}
        S_{\{i_2,i_3,i_4\}}&i=i_1\\
        \varepsilon_{i}\left(\langle i_1\,i_4\rangle_{(2k)}\langle i_2\,i_4\rangle_{(2k)} - \langle i_1\,i_3\rangle_{(2k)}\langle i_2\,i_3\rangle_{(2k)}+\langle i_3\,i_4\rangle_{(2k)}S\right)&i=i_2\\
        \varepsilon_{i}\left(\langle i_1\,i_2\rangle_{(2k)}\langle i_2\,i_3\rangle_{(2k)} - \langle i_1\,i_4\rangle_{(2k)}\langle i_3\,i_4\rangle_{(2k)}-\langle i_2\,i_4\rangle_{(2k)}S\right)&i=i_3\\
        \varepsilon_{i}\left(\langle i_1\,i_3\rangle_{(2k)}\langle i_3\,i_4\rangle_{(2k)} - \langle i_1\,i_2\rangle_{(2k)}\langle i_2\,i_4\rangle_{(2k)}+\langle i_2\,i_3\rangle_{(2k)}S\right)&i=i_4\\
        0 &\mathrm{otherwise},\\
    \end{cases} 
\end{align*}
    where \(i_j = \ell+j-1\), \(\varepsilon_{i}= 1\) if \(\ell\leq i\) and \(-1\) otherwise, and \(S = \sqrt{S_{\{i_1,i_2,i_3,i_4\}}}\). Additionally, define \(\alpha_{n,\ell,k,i}\in\mathcal F\) recursively as follows (with the empty product equaling \(1\)):

    \[
     \alpha_{n,\ell,k,i}= \mathrm{arccosh}\left(\frac{\mathrm{Cyc}^{\ell-1}_k (\mathbf v_{n,1,k}^{i+1})}{\mathbf \mathrm{Cyc}^{\ell-1}_k (\mathbf v_{n,1,k}^{1})\prod_{j=1}^{i-1}\mathcal \sinh(\alpha_{n,\ell,k,j})}\right)
    .\]
    \end{dfn}
\begin{dfn}
\label{def:arc}

     For \(\Gamma_0\) a reduced \(k\)-OG graph, define \(\Gamma = \mathrm{Arc}_{n,\ell,k}(\Gamma_0)\) to be the \((k+1)\)-OG graph resulting from adding an external arc starting at \(\ell\)  going anti-clockwise with support of length \(n\) (for \(n>1\)). That is, for \(I=\{\ell=i_1,i_2,..,i_n\}\) consecutive mod \(2k+2\), define
    \[
    \Gamma = \mathrm{Arc}_{n,\ell,k}(\Gamma_0)\defeq \mathrm{Rot}_{i_{n-1},i_n}\cdot...\cdot\mathrm{Rot}_{i_3,i_4}\mathrm{Rot}_{i_2,i_3}\mathrm{Inc}_{\ell}(\Gamma_0).
    \]
    and define the actions on arcs, indices, and sets of indices in the same way (see Figure~\ref{fig:ext arc} for an example of \(\mathrm{Arc}_{4,i_1}(\Gamma_0)\)).  If \(\tau_\ell\) is an external arc, we can choose a \(\tau_\ell\)-proper orientation for \(\Gamma\) (see Definition~\ref{prop orientation def}). We get that \(i_2,i_3,...,i_{n}\) are sinks. Thus,  we can label the new oriented vertices oriented with the \(\tau_\ell\)-proper orientation:
    \[
    \Gamma = \mathrm{Arc}_{n,\ell,k}(v^\omega)(\Gamma_0)\defeq \mathrm{Rot}_{i_{n-1},i_n}(v^\omega_{n-2})\cdot...\cdot\mathrm{Rot}_{i_3,i_4}(v^\omega_{2})\mathrm{Rot}_{i_2,i_3}(v^\omega_{1})\mathrm{Inc}_{\ell}(\Gamma_0).
    \]
     
     Similarly, for \(n\leq 4\), and  define
    \[
    \mathrm{Arc}_{n,\ell,k}(M)\\
    \defeq \mathrm{Rot}_{i_{n-1},i_n}(\alpha_{n,\ell,k,n -2})\cdot...\cdot\mathrm{Rot}_{i_3,i_4}(\alpha_{n,\ell,k,2})\mathrm{Rot}_{i_2,i_3}(\alpha_{n,\ell,k,1})\mathrm{Inc}_{l}(M),
    \]
    and for \(\theta\in\mathbb R^{n-2}\), write
    \[
    \mathrm{Arc}_{n,\ell,k}(\theta)(M)\\
    \defeq \mathrm{Rot}_{i_{n-1},i_n}(\theta_{n-2})\cdot...\cdot\mathrm{Rot}_{i_3,i_4}(\theta_2)\mathrm{Rot}_{i_2,i_3}(\theta_1)\mathrm{Inc}_{\ell}(M),
    \]
    for the induced action on \(M \in \mathcal F^{j\times 2k}\), \(M \in \mathcal F^{2k}\), or \(M\in \mathcal{F}\). We will omit \(k\) when it is clear from context.
\end{dfn}
\begin{obs}
\label{obs:inv arc inj}
    Let \(C_1^k, C_2^k\in \OGnon{k}{2k}\) and \(\Lambda^{k+1}\in\mathrm{Mat}^>_{(2k+2)\times(k+3)}\) be such that 
    \[\widetilde{\Lambda}^{k+1}(\mathrm{Arc}_{n,\ell}(\theta)(C_1^k))=\widetilde{\Lambda}^{k+1}(\mathrm{Arc}_{n,\ell}(\theta)(C_2^k)),\]
     for some \(n\geq 2\), \(\theta \in \mathbb R^{n-2}\), and \(i\in[2k]\). Then we have that \(\widetilde{\Lambda}^{k}(C_1^k)=\widetilde{\Lambda}^{k}(C_2^k),\)
     where \(\Lambda^k\defeq \mathrm{Arc}_{n,\ell}(\theta)^{-1}(\Lambda^{k+1})\).
\end{obs}
\begin{proof}
    Write \(Y_1 \defeq \widetilde{\Lambda}^{k}(C_1^k)\) and \(Y_1 \defeq \widetilde{\Lambda}^{k}(C_2^k)\). We will prove by induction on \(n\). For the base case of \(n=2\), we have that \(\mathrm{Arc}_{2,\ell}=\mathrm{Inc}_\ell\) and the claim reduces to Corollary~\ref{coro:inv inc inj}. For the step, it enough to prove the stronger statement that for \(R = \mathrm{Rot}_{i,i+1}(\alpha)\), \(\alpha\in\mathbb R\), \(C_1,C_2\in\OGnon{k}{2k}\), and \(\Lambda\in\mathrm{Mat}^>_{2k\times(k+2)}\), we have that \(R(C_1) \Lambda = R(C_2) \Lambda\) iff \(C_1 R^{-1}(\Lambda)=C_2 R^{-1}(\Lambda)\). But that is trivially true by Definition~\ref{def:rot of lambda}.
\end{proof}
\begin{dfn}
\label{def:O graph}
   Let us denote by \(O\) the \((k=0)\)-OG graph. That is, the graph with no internal or external vertices. We will consider \(\Omega_{O}\) to be containing a single element, which is represented by the empty \(0\times0\) matrix.
\end{dfn}
\begin{rmk}
We can now write complex graphs algebraically, as seen in Figure~\ref{fig:building with arcs}.
\begin{figure}[H]
    \centering
                \begin{center}
\begin{tikzpicture}[scale = 0.85]

\draw(-1.1,0)node[anchor=east]{\(\textcolor{blue}{\mathrm{Arc}_{2,1}}(O) =\)};
\draw (0,0) circle (1cm);
\draw[blue](0,-1)--(0,1);
\draw(0,-1)node[anchor=north]{\(1\)};
\draw(0,1)node[anchor=south]{\(2\)};

\draw(-1.1+6,0)node[anchor=east]{\(\textcolor{blue}{\mathrm{Arc}_{3,2}}\mathrm{Arc}_{2,1}(O) =\) };
\draw (0+6,0) circle (1cm);
\draw[blue](-1.41421/2+6, -1.41421/2)--(1.41421/2+6, 1.41421/2);
\draw(-1.41421/2+6, -1.41421/2)node[anchor=north]{\(4\)};
\draw(1.41421/2+6, 1.41421/2)node[anchor=south]{\(2\)};
\draw(-1.41421/2+6, 1.41421/2)node[anchor=south]{\(3\)}--(1.41421/2+6, -1.41421/2)node[anchor=north]{\(1\)};

\draw(-1.1+3,0-3)node[anchor=east]{\(\textcolor{blue}{\mathrm{Arc}_{4,2}}\mathrm{Arc}_{3,2}\mathrm{Arc}_{2,1}(O) =\)};
\draw (0+3,0-3) circle (1cm);
\draw({-0.517638/2+3, -1.93185/2-3})node[anchor=north]{\(6\)}--(1.41421/2+3, 1.41421/2-3)node[anchor=south]{\(3\)};
\draw[blue](-1.93185/2+3, 0.517638/2-3)--(1.93185/2+3, 0.517638/2-3);
\draw(-1.93185/2+3, 0.517638/2-3)node[anchor=east]{\(5\)};
\draw(1.93185/2+3, 0.517638/2-3)node[anchor=west]{\(2\)};
\draw(-1.41421/2+3, 1.41421/2-3)node[anchor=south]{\(4\)}--(0.517638/2+3, -1.93185/2-3)node[anchor=north]{\(1\)};
\end{tikzpicture}
\end{center}
    \caption{building OG graphs using the \(\mathrm{Arc}\) move\\ (the added arcs for each step are in blue)}
    \label{fig:building with arcs}
\end{figure}

\end{rmk}
\begin{rmk}
     For the sake of completion, let us write the effect of \(\mathrm{Arc}_{4,2}(\alpha_1,\alpha_2)\)  on matrices explicitly (the effect for other \(i\) is the same up to some \(\mathrm{Cyc}\) moves):

\begin{align*}
&\left(
\begin{array}{cccc}
\vertbar&\vertbar&&\vertbar\\
C^1&C^2&\cdots&C^{2k}\\
\vertbar&\vertbar&&\vertbar
\end{array}
\right)
\mapsto\\
&\mapsto 
\left(
\begin{array}{cccccccc}
\vertbar&\vertbar&\vertbar&\vertbar&\vertbar&\vertbar&&\vertbar\\
C^1&0&s_1C^{2}&c_1c_2C^{2}+s_2C^{3}&c_1s_2C^{2}+c_2C^{3}&C^{4}&\cdots&C^{2k}\\
\vertbar&\vertbar&\vertbar&\vertbar&\vertbar&\vertbar&&\vertbar\\
0&1&c_1&s_1c_2&s_1s_2&0&\cdots&0\\
\end{array}
\right),
\end{align*}
where \(s_i = \mathrm{sinh}(\alpha_i)\), and \(c_i = \mathrm{cosh}(\alpha_i)\) for \(i=1,2\).
The effect of the \(\mathrm{Arc}_{4,2}\) move on twistor variables is as follows.
For brevity, we will write the effect on indices, which then can be expanded linearly for the effect on twistors.
\[
i\mapsto
\begin{cases}
 1&i=1\\
 c_14+s_13&i=2\\
 c_25+s_2c_14+s_2s_13&i=3\\
 i+2&\text{otherwise}.\\
\end{cases}
\]
\end{rmk}

See Example~\ref{exm:s234} for a calculation of an explicit expression for \(\mathrm{Arc_{4,4}}(S_{\{2,3,4\}})\).

\section{BCFW Cells and Their Boundaries} \label{sec:BCFW Cells and their Boundaries}
The goal of the following section is to describe the BCFW cells and their boundaries in the language of the \(\mathrm{Arc}\) move, defined in Definition~\ref{def:arc}. We have seen that BCFW cells are three-regular trees of triangles with a triangle between the \(1\) and \(2k\) vertex (Theorem~\ref{thm:BCFW are trees}). We will show these graphs can be built by repeatedly adding external arcs of length \(4\) starting on even indices, to the top cell of \(\OGnon{3}{6}\) (which is represented by the graph \(\Gamma = \mathrm{Arc}_{4,2}\mathrm{Arc}_{3,2}\mathrm{Arc}_{2,1}(O)\)).

Start by considering the effect of the \(\mathrm{Arc}_{4,2i}\) move on the triangle tree (recall Definition~\ref{def:ToT}). Applying an \(\mathrm{Arc}_{4,2i}\) will add two additional triangle leaves on the \(2i\), \(2i+1\) triangle leaf (see Figures~\ref{fig:tirag tree} and~\ref{fig:arc triag}).
\begin{figure}[H]
    \centering
\begin{center}
\begin{tikzpicture}[scale = 0.8]

\draw[very thick, -latex, shorten <= 0] (2.3,0)   --(3.7,0);
\draw (3,0)node[anchor=south]{\(\mathrm{Arc}_{4,2i}\)};

\draw (0,0) circle (2cm);
\filldraw[lightgray] (0,1/2) circle (1cm);
\node[scale=3] (c) at (0,1/2)  {\(\Gamma_0\)};
\draw (0,1/2) circle (1cm);
\draw (1.96962/2, 1.3473/2) --(3.64697/2, 1.64306/2);
\draw (-1.96962/2, 1.3473/2) --(-3.64697/2, 1.64306/2);
\filldraw[black] (0,3.5/2) circle (2pt);

\filldraw[black] (1.27565/2, 3.20949/2) circle (2pt);
\filldraw[black] (-1.27565/2, 3.20949/2) circle (2pt);

\draw (0.68404/2, -0.879385/2) --(1.68446/2, -3.62803/2);
\draw (1.68446/2+0.3, -3.62803/2)node[anchor=north]{\(2i+1\)};

\draw (-0.68404/2, -0.879385/2)  --(-1.68446/2, -3.62803/2);
\draw (-1.68446/2-0.3, -3.62803/2)node[anchor=north]{\(2i\)};

\draw[blue](0, -1/2)-- (0,-2);
\filldraw[blue] (0,-2) circle (1pt);

\draw (0+6,0) circle (2cm);
\filldraw[lightgray] (0+6,1/2) circle (1cm);
\node[scale=3] (c) at (0+6,1/2)  {\(\Gamma_0\)};
\draw (0+6,1/2) circle (1cm);
\draw (1.96962/2+6, 1.3473/2) --(3.64697/2+6, 1.64306/2);
\draw (-1.96962/2+6, 1.3473/2) --(-3.64697/2+6, 1.64306/2);
\filldraw[black] (0+6,3.5/2) circle (2pt);

\filldraw[black] (1.27565/2+6, 3.20949/2) circle (2pt);
\filldraw[black] (-1.27565/2+6, 3.20949/2) circle (2pt);

\draw (0.68404/2+6, -0.879385/2) --(1.68446/2+6, -3.62803/2);
\draw (1.68446/2+0.3+6, -3.62803/2)node[anchor=north]{\(2i+2\)};

\draw (-0.68404/2+6, -0.879385/2)  --(-1.68446/2+6, -3.62803/2);
\draw (-1.68446/2-0.3+6, -3.62803/2)node[anchor=north]{\(2i+1\)};

\draw (-3.27261/2+6, -2.3/2) node[anchor=north east]{\(2i\)}-- (3.27261/2+6, -2.3/2)node[anchor=north west ]{\(2i+3\)} ;

\draw[blue](0+6, -1/2)-- (0+6,-0.75);
\filldraw[blue] (0+6,-0.75) circle (1pt);
\filldraw[blue] (-1.23539+6, -1.57284) circle (1pt);
\draw[blue](-1.23539+6, -1.57284)-- (0+6,-0.75);
\filldraw[blue] (1.23539+6, -1.57284) circle (1pt);
\draw[blue](1.23539+6, -1.57284)-- (0+6,-0.75);
 
\end{tikzpicture}
\end{center}
    \caption{The effect of \(\mathrm{Arc}_{4,2i}\) on the triangle graph}
    \label{fig:arc triag}
\end{figure}

The conditions necessary for getting a reduced graph after applying the \(\mathrm{Arc}_{4,2i}\) are that the original graph is reduced (true by Proposition~\ref{prop:ToT reduced}), and that there is no arc contained in \(\{2i,2i+1\}\) in the original graph (by Observation~\ref{obs:reduce arc move}), that is, that \(\{2i,2i+1\}\) is not an external \(2\)-arc (see Definition~\ref{def:arc sup}). This is true for any ToT graph because ToT graphs have no external \(2\)-arcs, only external \(4\)-arcs. 

\begin{prop}
\label{prop:BCFW arcs}
    An OG graph \(\Gamma\) is a BCFW graphs iff it can be expressed in the form 
    \[
    \mathrm{Arc}_{4,2i_{n}}\mathrm{Arc}_{4,2i_{n-1}}\cdot...\cdot\mathrm{Arc}_{4,2i_3}\mathrm{Arc}_{4,2i_2}\mathrm{Arc}_{3,2i_1}\mathrm{Arc}_{2,1}(O) ,
    \]
    for \(n\geq 1 \), where \(i_j \in [j]\), with the result being reduced after each \(\mathrm{Arc}\) move.
\end{prop}
\begin{proof}
    This follows from the previous discussion and the fact that any 3-regular tree can be created by taking the graph with a single 3-regular vertex and repeatedly replacing leaves with 3-regular vertices.
\end{proof}

This kind of description will allow us to invert the amplituhedron map on BCFW cells in Section~\ref{sec:inj}. The rest of this section will be devoted showing a similar description for the boundary cells of BCFW cells . We will start by describing boundary cells of BCFW cells in Section~\ref{sec:boundary cells}, and then prove an analogous result for them in Section~\ref{sec:arc seq} (Proposition~\ref{prop:arc limits boundary sub BCFW}).

\subsection{Boundary Cells}
\label{sec:boundary cells}
In this section we aim to characterize all of the boundary cells of BCFW cells and study their properties. The structure we describe for such cells will be instrumental to the rest of the paper.

Consider \(\Gamma\) a BCFW graph with a trigonometric orientation \(\omega\) and angles \(\{\alpha_i\}_{i=1}^n\). By Theorem~\ref{param bij} we have a bijection \((0,\frac{\pi}{2})^n \rightarrow \Omega_\Gamma\) by 
\(
\alpha =(\alpha_1,...,\alpha_{n+m})\mapsto \Gamma^\omega(\alpha) 
\) 
according to the parametrization we defined earlier. If \(I\) is the set of sources in \(\omega\) we have that \(\left|I\right| =k\) and \((\Gamma^\omega(\alpha))^I =\mathrm{Id}\) for any \(\alpha\in\mathbb R^n\). As the entries are bounded by Corollary~\ref{boundary measurements bounded coro}, we can conclude that \(\lim_{\alpha_i\rightarrow0}\Gamma^\omega(\alpha)\) and \(\lim_{\alpha_i\rightarrow\frac{\pi}{2}}\Gamma^\omega(\alpha)\) are well-defined elements of \(\OGnon{k}{2k}\). As \(\Gamma^\omega(\alpha)\) is continuous, we actually just have a map \([0,\frac{\pi}{2}]^n \rightarrow \overline\Omega_\Gamma\), the closure of \(\Omega_\Gamma\) in \(\OGnon{k}{2k}\), although this need not be a bijection. Thus, we can study the boundaries of BCFW cells by taking limits of angles in their parameterizations.

Consider now \(\Gamma\) a BCFW graph with some perfect orientation and angles \(\alpha_i\), with \(\alpha_1,...,\alpha_n\) corresponding to vertices with trigonometric orientation and \(\alpha_{n+1},...,\alpha_{n+m}\) corresponding to vertices with hyperbolic orientations. By Theorem~\ref{param bij} we have a bijection \((0,\frac{\pi}{2})^n\times(0,\infty)^m \rightarrow \Omega_\Gamma\) by \((\alpha_1,...,\alpha_{n+m})\mapsto \Gamma^\omega(\alpha)\) according to the parametrization we defined earlier. 

The boundaries of \(\Omega_\Gamma\) again will be achieved by taking the limit of some angle to \(0\), \(\frac{\pi}{2}\) if it is trigonometric, or \(\infty\) if it is hyperbolic (that is, an angle that corresponds to a vertex with a trigonometric/hyperbolic orientation resp.). However, as the entries are not bounded, taking the limit is a more delicate process that requires some care.

We can solve this in the following way: By Proposition~\ref{prop:trigonometric existence} there exist a trigonometric orientation \(\omega^\prime\) for \(\Gamma\), with a corresponding parametrization \((0,\frac{\pi}{2})^{n+m} \rightarrow \Omega_\Gamma\) by \((\alpha^\prime_1,...,\alpha^\prime_{n+m})\mapsto \Gamma^{\omega^\prime}(\alpha^\prime)\), where the angles \(\alpha^\prime_i\) and \(\alpha_i\) correspond to the same vertex. Now all of the boundaries are of the form \(\lim_{\alpha_i^\prime\rightarrow0} \Gamma^{\omega^\prime}(\alpha^\prime)\) and \(\lim_{\alpha^\prime_i\rightarrow\frac{\pi}{2}} \Gamma^{\omega^\prime}(\alpha^\prime)\). 

By Proposition~\ref{prop:angles repara} we are not missing any of the boundaries by changing the parametrization. Indeed, we get that for a hyperbolic angle \(\alpha_i\), the boundaries \(\lim_{\alpha_i\rightarrow0}\Gamma^\omega(\alpha)\) and  \(\lim_{\alpha_i\rightarrow\infty}\Gamma^\omega(\alpha)\) just become \(\lim_{\alpha_i^\prime\rightarrow0} \Gamma^{\omega^\prime}(\alpha^\prime)\) and \(\lim_{\alpha^\prime_i\rightarrow\frac{\pi}{2}} \Gamma^{\omega^\prime}(\alpha^\prime)\) in the new parameterization. By changing parameterization all boundaries are still manifest, and the infinite boundary becames finite. 

This is the correct way of taking limits of hyperbolic angles going to infinity, and indeed any perfect orientation: when we write \(\lim_{\alpha_i\rightarrow L}\Gamma^\omega(\alpha)\) for \(\alpha_i\) corresponding to some internal vertex \(v_i\) with \(\omega\) a non-trigonometric perfect orientation, we mean that we first re-parametrize to a trigonometric orientation (possible by Proposition~\ref{prop:trigonometric existence}), get a new angle \(\alpha^\prime_i\) corresponding \(v_i\), and then take the new angle to whichever finite boundary corresponds to \(\alpha_i \rightarrow L\), be it \(\alpha^\prime_i\rightarrow0\) or \(\alpha^\prime_i\rightarrow\frac{\pi}{2}\), in accordance with Proposition~\ref{prop:angles repara}. This limit is a well-defined element of \(\OGnon{k}{2k}\).

Now we will discuss what happens to the OG graphs. Let us consider a vertex with a trigonometric orientation. The boundaries of the corresponding cell will correspond to taking some angles to \(0\) or \(\frac{\pi}{2}\) degrees. Consider a vertex with its associated angle \(\alpha\):

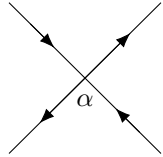
    \begin{center}
    \scalebox{0.8}{
            \begin{tikzpicture}

            \draw (-1,1)--(1, -1);
            \draw (1, 1)--(-1, -1);

            \draw [very thick, -latex, shorten <= 10] (-0.71967, 0.71967)--(-0.4, 0.4);
            \draw [very thick, -latex, shorten <= 10] (0.71967, -0.71967)--(0.4,-0.4 );
            \draw [very thick, -latex, shorten <= 20] (0,0)--(0.6, 0.6);
            \draw [very thick, -latex, shorten <= 20] (0, 0)--(-0.6,-0.6);

            \draw (0,-0.5)node[anchor = south]{\(\alpha\)};
        \end{tikzpicture}
        }
        \end{center}
By the definition of the weights on paths, \(\alpha = 0\) corresponds to a left turn having weight 1 and right turn having weight 0, and \(\alpha = \frac{\pi}{2}\) corresponds to the opposite. We thus get the same weights for paths by replacing the vertex with the following configurations for \(\alpha = 0\) and \(\frac{\pi}{2}\) respectively:

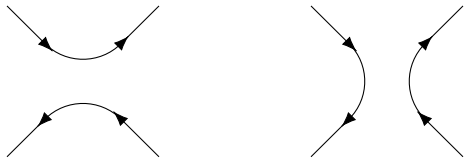
\begin{center}
\scalebox{0.8}{
            \begin{tikzpicture}

            \draw (0,2)--(0.5,1.5);
            \draw (1.5,0.5)--(2,0);
            \draw (2,2)--(1.5,1.5);
            \draw (0.5,0.5)--(0,0);

            \draw[black] (0.5,1.5) arc (-45-90:-45:1.41421/2);
            \draw[black] (1.5,0.5) arc (45:135:1.41421/2);

            \draw [very thick, -latex, shorten <= 2] (0.6-0.001, 1.4+0.001)--(0.6, 1.4);
            \draw [very thick, -latex, shorten <= 2] (1.4+0.001,0.6-0.001 )--(1.4,0.6 );
            \draw [very thick, -latex, shorten <= 2] (1.6-0.001, 1.6-0.001)--(1.6, 1.6);
            \draw [very thick, -latex, shorten <= 2] (0.4+0.001,0.4+0.001)--(0.4,0.4);

            \draw (0+4,2)--(0.5+4,1.5);
            \draw (1.5+4,0.5)--(2+4,0);
            \draw (2+4,2)--(1.5+4,1.5);
            \draw (0.5+4,0.5)--(0+4,0);

            \draw[black] (1.5+4,1.5) arc (135:360-135:1.41421/2);
            \draw[black] (0.5+4,0.5) arc (-45:45:1.41421/2);

            \draw [very thick, -latex, shorten <= 2] (0.6-0.001+4, 1.4+0.001)--(0.6+4, 1.4);
            \draw [very thick, -latex, shorten <= 2] (1.4+0.001+4,0.6-0.001 )--(1.4+4,0.6 );
            \draw [very thick, -latex, shorten <= 2] (1.6-0.001+4, 1.6-0.001)--(1.6+4, 1.6);
            \draw [very thick, -latex, shorten <= 2] (0.4+0.001+4,0.4+0.001)--(0.4+4,0.4);

        \end{tikzpicture}
        }
        \end{center}
\begin{dfn}
\label{def opening the vertex}
We would refer to this as \emph{opening the vertex}.   
\end{dfn}
This allows us to extend the extend the talk of limits to OG graphs.
\begin{dfn}
\label{def:boundary graph}
    For an OG graph \(\Gamma\)  with an oriented vertex \(v^\omega\) with a trigonometric orientation, define \(\lim_{v^\omega\rightarrow0} \Gamma\) and \(\lim_{v^\omega\rightarrow\frac{\pi}{2}} \Gamma\) to be graph \(\Gamma\) with the vertex \(v\) replaced by the configurations as seen in Figure~\ref{fig:opening a vertex}. Define the inherited orientations \(\lim_{v^\omega\rightarrow0}\omega\) and \(\lim_{v^\omega\rightarrow\frac{\pi}{2}} \omega\) to be the orientations on the respective graphs in Figure~\ref{fig:opening a vertex}. 

      For \(\Gamma\) an OG graph with an oriented vertex \(v^{\omega^\prime}\) with hyperbolic orientation, define \(\lim_{v^{\omega^\prime}\rightarrow0} \Gamma\) and \(\lim_{v^{\omega^\prime}\rightarrow\infty} \Gamma\) to be graph \(\Gamma\) with the vertex \(v\) replaced by the configurations as seen in Figure~\ref{fig:opening a vertex}.
    
    \begin{figure}[H]
    \centering
                \begin{center}
                \scalebox{0.8}{
\begin{tikzpicture}

                \draw (0,2+3)--(2, 0+3);
            \draw (2, 2+3)--(0, 0+3);

            \draw [very thick, -latex, shorten <= 10] (1-0.71967, 1+0.71967+3)--(1-0.4,1+ 0.4+3);
            \draw [very thick, -latex, shorten <= 10] (1+0.71967, 1-0.71967+3)--(1+0.4,1-0.4 +3);
            \draw [very thick, -latex, shorten <= 20] (1,1+3)--(1.6, 1.6+3);
            \draw [very thick, -latex, shorten <= 20] (1, 1+3)--(1-0.6,1-0.6+3);
            \draw (1,.4+3)node[anchor = south]{\(v^\omega\)};

                            \draw (0+4,2+3)--(2+4, 0+3);
            \draw (2+4, 2+3)--(0+4, 0+3);

            \draw [very thick, -latex, shorten <= 10] (1-0.71967+4, 1+0.71967+3)--(1-0.4+4,1+ 0.4+3);
            \draw [very thick, -latex, shorten <= 20] (1+4,1+3)--(1.6+4, 1.6+3);

             \draw [very thick, -latex, shorten <= 10] (1-0.71967+4, 1-0.71967+3)--(1-0.4+4,1- 0.4+3);
            \draw [very thick, -latex, shorten <= 20] (1+4,1+3)--(1.6+4, 1-.6+3);
            
            \draw (1+4,.3+3)node[anchor = south]{\(v^{\omega^\prime}\)};

            \draw (0,2)--(0.5,1.5);
            \draw (1.5,0.5)--(2,0);
            \draw (2,2)--(1.5,1.5);
            \draw (0.5,0.5)--(0,0);

            \draw[black] (0.5,1.5) arc (-45-90:-45:1.41421/2);
            \draw[black] (1.5,0.5) arc (45:135:1.41421/2);

            \draw (1,0.4)node[anchor = north]{\(v^\omega\rightarrow 0
           \)};
            \draw (1,0.15)node[anchor = north]{\( v^{\omega^\prime}\rightarrow 0\)};

            \draw [very thick, -latex, shorten <= 2] (0.6-0.001, 1.4+0.001)--(0.6, 1.4);
            \draw [very thick, -latex, shorten <= 2] (1.4+0.001,0.6-0.001 )--(1.4,0.6 );
            \draw [very thick, -latex, shorten <= 2] (1.6-0.001, 1.6-0.001)--(1.6, 1.6);
            \draw [very thick, -latex, shorten <= 2] (0.4+0.001,0.4+0.001)--(0.4,0.4);

            \draw (1+4,0.4)node[anchor = north]{\(v^\omega\rightarrow \frac{\pi}{2}
           \)};
            \draw (1+4,0.15)node[anchor = north]{\( v^{\omega^\prime}\rightarrow \infty\)};

            \draw (0+4,2)--(0.5+4,1.5);
            \draw (1.5+4,0.5)--(2+4,0);
            \draw (2+4,2)--(1.5+4,1.5);
            \draw (0.5+4,0.5)--(0+4,0);

            \draw[black] (1.5+4,1.5) arc (135:360-135:1.41421/2);
            \draw[black] (0.5+4,0.5) arc (-45:45:1.41421/2);

            \draw [very thick, -latex, shorten <= 2] (0.6-0.001+4, 1.4+0.001)--(0.6+4, 1.4);
            \draw [very thick, -latex, shorten <= 2] (1.4+0.001+4,0.6-0.001 )--(1.4+4,0.6 );
            \draw [very thick, -latex, shorten <= 2] (1.6-0.001+4, 1.6-0.001)--(1.6+4, 1.6);
            \draw [very thick, -latex, shorten <= 2] (0.4+0.001+4,0.4+0.001)--(0.4+4,0.4);

\end{tikzpicture}
}
\end{center}
    \caption{opening a trigonometric or hyperbolic oriented vertex}
    \label{fig:opening a vertex}
\end{figure}

    When we write \(\lim_{v^\omega\rightarrow L}\Gamma\) we implicitly assume that if \(v^\omega\) has a trigonometric (hyperbolic) orientation we have that \(L\in\{0,\frac{\pi}{2}\}\) (resp.  \(L\in\{0,\infty\}\)).

    Given an OG graph \(\Gamma\), a graph resulting from a series of limit operations on the graph \(\Gamma\) will be called a \emph{boundary graph} of \(\Gamma\). If an OG graph correspond to a codimension \(n\) boundary cell of \(\Omega_\Gamma\), it will be called a codimension \(n\) boundary graph of \(\Gamma\).
\end{dfn}

    \begin{prop}[\cite{companion}]
\label{prop:limits commute with para}
    Let \(\Gamma\) be an OG graph with a trigonometric orientation \(\omega\), and \\ \(\Gamma^{\omega}(\alpha_1,...,\alpha_n)\) be the corresponding parametrization of \(\Omega_\Gamma\), with an enumeration of its oriented internal vertices \(v_1^\omega,...,v_n^\omega\). Let \(L\in\{0,\frac{\pi}{2}\}\) and write \(\Gamma_0 = \lim_{v_{n}^\omega\rightarrow L}\Gamma\). Then
    \[
    \lim_{\alpha_{n}\rightarrow L} \Gamma^{\omega}(\alpha)= \Gamma^{\omega}(\alpha_1,...,\alpha_{n-1},\ell) \in\Omega_{\Gamma_0}.
    \]

    Furthermore, let \(\omega_0\) be the inherited orientation from \(\Gamma^\omega\), and \(\Gamma_0^{\omega_0}(\alpha_1,...,\alpha_{n-1})\) the corresponding parameterization. If \(\Gamma_0\) is reduced, then \(\Gamma^\omega:(0,\frac{\pi}{2})^n\rightarrow\Omega_\Gamma\) can be extended to a map \( \Gamma^\omega:U\rightarrow\OG{k}{2k}\) which is a diffeomorphism onto its image, where \(U\subset \mathbb R^n\) is an open neighborhood containing \((0,\frac{\pi}{2})^{n-1}\times\left((0,\frac{\pi}{2})\cup \{L\}\right)\), with \(
    \Gamma^\omega(\alpha_1,...,\alpha_{n-1},\ell)=\Gamma_0^{\omega_0}(\alpha_1,...,\alpha_{n-1}).\)

\end{prop}

    \begin{coro}
    \label{colosure continuous coro}
    For \(\Gamma\) an OG graph with \(n\) internal vertices and \(\omega\) a trigonometric orientation, the parametrization of \(\Omega_\Gamma\) is a continuous map
    \(
    \Gamma^\omega:[0,\frac{\pi}{2}]^n \rightarrow\ \overline\Omega_\Gamma.
    \)
\end{coro}
\begin{proof}
    Let \(I\) be the set of sources. By definition we have that \((\Gamma^\omega(\alpha))^I =\mathrm{Id_{k\times k}}\) for any \(\alpha\), and the entries of the matrix are continuous bounded functions of the angles by Corollary~\ref{boundary measurements bounded coro}. Therefore, we have a well defined continuous map to \(\OGnon{k}{2k}\).
\end{proof}

\begin{coro}[\cite{companion}]
    \label{order of limits coro}
    Let \(\Gamma^\omega\) be an OG graph with a trigonometric orientation, and \(\{v^\omega_i\}_{i=1}^n\)  enumeration of its internal vertices. Select some \(L_i\in \{0,\frac{\pi}{2}\}\) for some \(i\in[n]\).
    Then the equivalence class of the OG graph
    \(
    \lim_{v^\omega_1\rightarrow L_1}\lim_{v^\omega_2\rightarrow L_2}...\lim_{v^\omega_m\rightarrow L_m}\Gamma
    \)
    does not depend on the ordering of the limit operations.
\end{coro}

\begin{prop}[\cite{companion}]
\label{prop:limits are boundaries}
    Let \(\Gamma\) be a reduced OG graph. Denote by \(\partial \Gamma\) be the set of all graphs resulting from applying some limit operations to \(\Gamma\), that is, the boundary graphs of \(\Gamma\). Then
    \(
    \partial \Omega_\Gamma = \bigcupdot_{\Gamma^\prime \in \partial \Gamma} \Omega_{\Gamma^\prime}
    \)
\end{prop}
We see that \(\Omega_{\Gamma^\prime}\) for \(\Gamma^\prime \in \partial \Gamma\) are the boundary strata of \(\Omega_{\Gamma}\). This justifies thinking of graphs of the form \(\partial \Gamma\) as boundaries of \(\Gamma\).

In Proposition~\ref{prop:BCFW arcs} we saw that we can represent BCFW graphs a sequence of \(\mathrm{Arc}\) moves with certain properties. We would now like to study the effect of the limit operations on \(\mathrm{Arc}\)-moves, with the aim of showing a similar result for boundaries of BCFW graphs.

Consider the graph \(\mathrm{Arc}_{3,\ell}(\Gamma)\):

\begin{center}
\begin{tikzpicture}[scale = 0.6]
\draw (0,0) circle (2cm);
\filldraw[lightgray] (0,1/2) circle (1cm);
\node[scale=3] (c) at (0,1/2)  {\(\Gamma\)};
\draw (0,1/2) circle (1cm);
\draw (1.96962/2, 1.3473/2) --(3.64697/2, 1.64306/2);
\draw (-1.96962/2, 1.3473/2) --(-3.64697/2, 1.64306/2);
\filldraw[black] (0,3.5/2) circle (2pt);

\filldraw[black] (1.27565/2, 3.20949/2) circle (2pt);
\filldraw[black] (-1.27565/2, 3.20949/2) circle (2pt);

\draw  (0, -1/2)--(0,-2);
\draw (0,-2)node[anchor=north]{\(\ell+1\)};

\draw (-3.27261/2, -2.3/2) node[anchor=north east]{\(\ell\)}-- (3.27261/2, -2.3/2)node[anchor=north west ]{\(\ell+2\)} ;
\end{tikzpicture}
\end{center}

The two boundaries the are achieved by taking opening the internal vertex correspond to the following graphs:
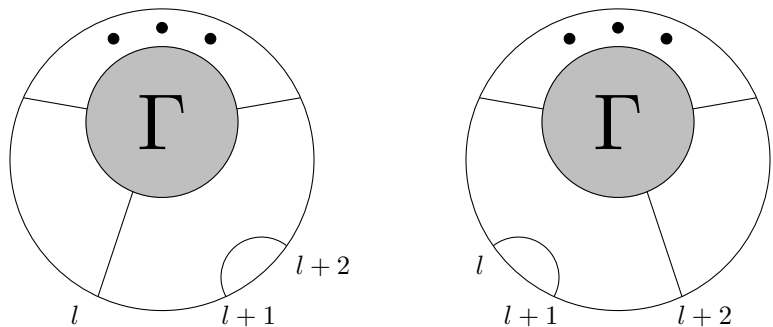
\begin{center}
\begin{tikzpicture}[scale = 0.6]
\draw (0,0) circle (2cm);
\filldraw[lightgray] (0,1/2) circle (1cm);
\node[scale=3] (c) at (0,1/2)  {\(\Gamma\)};
\draw (0,1/2) circle (1cm);
\draw (-1.96962/2, 1.3473/2) --(-3.64697/2, 1.64306/2);
\draw (1.96962/2, 1.3473/2) --(3.64697/2, 1.64306/2);
\filldraw[black] (-0,3.5/2) circle (2pt);

\filldraw[black] (-1.27565/2, 3.20949/2) circle (2pt);
\filldraw[black] (1.27565/2, 3.20949/2) circle (2pt);

\draw (-0.382683, -0.42388) --(-1.68446/2, -3.62803/2);
\draw (-1.68446/2-0.3, -3.62803/2)node[anchor=north]{\(\ell\)};
\draw (1.68446/2+0.3, -3.62803/2)node[anchor=north]{\(\ell+1\)};
\draw (3.27261/2, -2.3/2)node[anchor=north west]{\(\ell+2\)};

\draw (1.68446/2, -3.62803/2) arc (180+28.7936 :180+28.7936-157.5:0.53254);

\draw (0+6,0) circle (2cm);
\filldraw[lightgray] (0+6,1/2) circle (1cm);
\node[scale=3] (c) at (0+6,1/2)  {\(\Gamma\)};
\draw (0+6,1/2) circle (1cm);
\draw (1.96962/2+6, 1.3473/2) --(3.64697/2+6, 1.64306/2);
\draw (-1.96962/2+6, 1.3473/2) --(-3.64697/2+6, 1.64306/2);
\filldraw[black] (0+6,3.5/2) circle (2pt);

\filldraw[black] (1.27565/2+6, 3.20949/2) circle (2pt);
\filldraw[black] (-1.27565/2+6, 3.20949/2) circle (2pt);

\draw (0.382683+6, -0.42388) --(1.68446/2+6, -3.62803/2);
\draw (1.68446/2+0.3+6, -3.62803/2)node[anchor=north]{\(\ell+2\)};
\draw (-1.68446/2-0.3+6, -3.62803/2)node[anchor=north]{\(\ell+1\)};
\draw (-3.27261/2+6, -2.3/2)node[anchor=north east]{\(\ell\)};

\draw (-1.68446/2+6, -3.62803/2) arc (-28.7936 :-28.7936+157.5:0.53254);

\end{tikzpicture}
\end{center}

Since the original graph was reduced those are reduced as well by Proposition~\ref{prop:reduceability}. It is thus evident that 
\begin{coro}
\label{arc 3 boundary coro}
For a reduced graph $\Gamma$,
\[
\lim_{v^\omega\rightarrow \infty}\mathrm{Arc}_{3,\ell}(v^\omega)(\Gamma) = \mathrm{Arc}_{2,\ell+1}(\Gamma),\quad\quad\lim_{v^\omega\rightarrow 0}\mathrm{Arc}_{3,\ell}(v^\omega)(\Gamma) = \mathrm{Arc}_{2,\ell}(\Gamma),
\]
and the resulting graphs are reduced as the result of the \(\mathrm{Arc}_{\ell,n}\) on a reduced graph are always reduced for \(n<4\) by Observation~\ref{obs:reduce arc move}.
\end{coro}
Let \(\Gamma\) be a BCFW graph. By Theorem~\ref{thm:BCFW are trees}, the triangle graph has a leaf between the external vertices \(2i\) and \(2i+1\) considered mod \(2k\).

\begin{center}
\begin{tikzpicture}[scale = 0.6]
\draw (0,0) circle (2cm);
\filldraw[lightgray] (0,1/2) circle (1cm);
\node[scale=3] (c) at (0,1/2)  {\(\Gamma\)};
\draw (0,1/2) circle (1cm);
\draw (1.96962/2, 1.3473/2) --(3.64697/2, 1.64306/2);
\draw (-1.96962/2, 1.3473/2) --(-3.64697/2, 1.64306/2);
\filldraw[black] (0,3.5/2) circle (2pt);

\filldraw[black] (1.27565/2, 3.20949/2) circle (2pt);
\filldraw[black] (-1.27565/2, 3.20949/2) circle (2pt);

\draw  (-0.68404/2, -0.879385/2)--(1.68446/2, -3.62803/2);
\draw (1.68446/2+0.3, -3.62803/2)node[anchor=north]{\(2i+1\)};

\draw (0.68404/2, -0.879385/2)  --(-1.68446/2, -3.62803/2);
\draw (-1.68446/2-0.3, -3.62803/2)node[anchor=north]{\(2i\)};
\end{tikzpicture}
\end{center}
Now consider the graph \(\mathrm{Arc}_{4,2i}(\Gamma)\):
\begin{center}
\begin{tikzpicture}[scale = 0.6]
\draw (0,0) circle (2cm);
\filldraw[lightgray] (0,1/2) circle (1cm);
\node[scale=3] (c) at (0,1/2)  {\(\Gamma\)};
\draw (0,1/2) circle (1cm);
\draw (1.96962/2, 1.3473/2) --(3.64697/2, 1.64306/2);
\draw (-1.96962/2, 1.3473/2) --(-3.64697/2, 1.64306/2);
\filldraw[black] (0,3.5/2) circle (2pt);

\filldraw[black] (1.27565/2, 3.20949/2) circle (2pt);
\filldraw[black] (-1.27565/2, 3.20949/2) circle (2pt);

\draw  (-0.68404/2, -0.879385/2)--(1.68446/2, -3.62803/2);
\draw (1.68446/2+0.3, -3.62803/2)node[anchor=north]{\(2i+2\)};

\draw (0.68404/2, -0.879385/2)  --(-1.68446/2, -3.62803/2);
\draw (-1.68446/2-0.3, -3.62803/2)node[anchor=north]{\(2i+1\)};

\draw (-3.27261/2, -2.3/2) node[anchor=north east]{\(2i\)}-- (3.27261/2, -2.3/2)node[anchor=north west ]{\(2i+3\)} ;
\end{tikzpicture}
\end{center}

Looking at the bottom-left internal vertex, corresponding to \(\alpha_1\), the two boundaries correspond to the following graphs after reduction: 

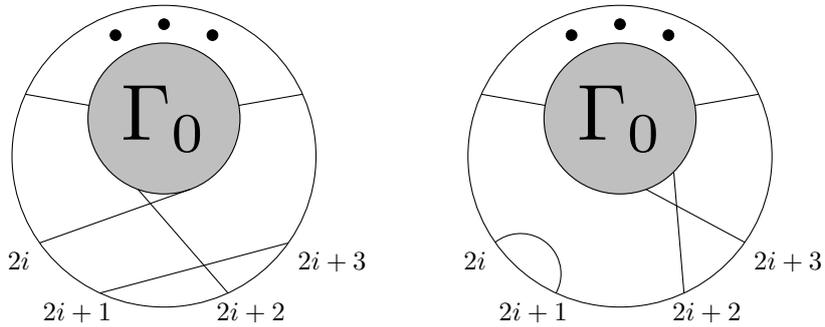
\begin{center}
\begin{tikzpicture}[scale = 0.6]
\draw (0,0) circle (2cm);
\filldraw[lightgray] (0,1/2) circle (1cm);
\node[scale=3] (c) at (0,1/2)  {\(\Gamma\)};
\draw (0,1/2) circle (1cm);
\draw (1.96962/2, 1.3473/2) --(3.64697/2, 1.64306/2);
\draw (-1.96962/2, 1.3473/2) --(-3.64697/2, 1.64306/2);
\filldraw[black] (0,3.5/2) circle (2pt);

\filldraw[black] (1.27565/2, 3.20949/2) circle (2pt);
\filldraw[black] (-1.27565/2, 3.20949/2) circle (2pt);

\draw (-0.68404/2, -0.879385/2) --(1.68446/2, -3.62803/2);
\draw (1.68446/2+0.3, -3.62803/2)node[anchor=north]{\(2i+2\)};

\draw   (3.27261/2, -2.3/2)node[anchor=north west ]{\(2i+3\)} --(-1.68446/2, -3.62803/2);
\draw (-1.68446/2-0.3, -3.62803/2)node[anchor=north]{\(2i+1\)};

\draw (-3.27261/2, -2.3/2)node[anchor=north east]{\(2i\)}--(0.68404/2, -0.879385/2);

\draw (0+10,0) circle (2cm);
\filldraw[lightgray] (0+10,1/2) circle (1cm);
\node[scale=3] (c) at (0+10,1/2)  {\(\Gamma\)};
\draw (0+10,1/2) circle (1cm);
\draw (1.96962/2+10, 1.3473/2) --(3.64697/2+10, 1.64306/2);
\draw (-1.96962/2+10, 1.3473/2) --(-3.64697/2+10, 1.64306/2);
\filldraw[black] (0+10,3.5/2) circle (2pt);

\filldraw[black] (1.27565/2+10, 3.20949/2) circle (2pt);
\filldraw[black] (-1.27565/2+10, 3.20949/2) circle (2pt);

\draw ({1.41421/2+10, -0.414214/2}) --(1.68446/2+10, -3.62803/2);
\draw (1.68446/2+0.3+10, -3.62803/2)node[anchor=north]{\(2i+2\)};
\draw (-1.68446/2-0.3+10, -3.62803/2)node[anchor=north]{\(2i+1\)};
\draw (-3.27261/2+10, -2.3/2)node[anchor=north east]{\(2i\)};
\draw   (3.27261/2+10, -2.3/2)node[anchor=north west ]{\(2i+3\)};

\draw  (0.68404/2+10, -0.879385/2)  --({3.28388/2+10, -2.28388/2});

\draw (-1.68446/2+10, -3.62803/2) arc (-28.7936 :-28.7936+157.5:0.53254);

\end{tikzpicture}
\end{center}

Since the original graph was reduced these are reduced as well by Proposition~\ref{prop:reduceability}. The picture is of course the mirror image for \(\alpha_2\). It is thus evident that 
\begin{coro}
\label{arc 4 boundary coro}
for \(\Gamma\) a reduced BCFW graph,

\begin{align*}
    &\lim_{v^\omega_1\rightarrow \infty}\mathrm{Arc}_{4,2i}(v^\omega_1,v^\omega_2)(\Gamma) = \mathrm{Arc}_{3,2i+1}(\Gamma), &
    &\lim_{v^\omega_1\rightarrow 0}\mathrm{Arc}_{4,2i}(v^\omega_1,v^\omega_2)(\Gamma) = \mathrm{Arc}_{2,2i}(\Gamma), \\
    &\lim_{v^\omega_2\rightarrow 0}\mathrm{Arc}_{4,2i}(v^\omega_1,v^\omega_2)(\Gamma) = \mathrm{Arc}_{2i,3}(\Gamma), &
    &\lim_{v^\omega_2\rightarrow \infty}\mathrm{Arc}_{4,2i}(v^\omega_1,v^\omega_2)(\Gamma) = \mathrm{Arc}_{2,2i+2}(\Gamma),
\end{align*}
and the resulting graphs are reduced as the result of the \(\mathrm{Arc_{\ell,n}}\) on a reduced graph are always reduced for \(n<4\) by Observation~\ref{obs:reduce arc move}.
\end{coro}

Since moves \(\mathrm{Rot}\), \(\mathrm{Inc}\), and thus \(\mathrm{Arc}\), always act on OG graphs right next to the boundary disc and never interact with internal vertices, they commute with the taking of a limit.
\begin{coro}
    \label{limit comute for graphs coro}
    Let \(\Gamma\) be an OG graph and  \(v^\omega\) an oriented vertex in \(\Gamma\) and let \(G\) be a move, we have that
    \(
    \lim_{v^\omega\rightarrow L}G(\Gamma) = G(\lim_{v^\omega\rightarrow L}\Gamma).
    \)
\end{coro}

 We will show that external arcs with sufficiently small support exhibit particularly good behavior (see Corollary~\ref{n<=4 twist solv coro}). Observe that taking limits can only decrease the support of external arcs. This suggests that boundaries of BCFW graphs—much like the BCFW graphs themselves—can be obtained through a sequence of \(\mathrm{Arc}\) moves, albeit with simpler types of \(\mathrm{Arc}\) moves. In the remainder of this section we make this intuition precise.

 \subsection{Arc-Sequences}
 \label{sec:arc seq}
Having understood the boundary cells of BCFW cells, our goal in this section is to establish an analogue of Proposition~\ref{prop:BCFW arcs}. Specifically, we show that OG graphs representing boundaries of BCFW cells can be constructed by repeatedly applying \(\mathrm{Arc}_{n,\ell}\) moves with \(n \leq 4\). A precise formulation is given in Proposition~\ref{prop:arc limits boundary sub BCFW}. Readers willing to accept this statement on faith may wish to skip this (relatively technical) section on a first pass.

\begin{dfn}
\label{def:sequences}
    Consider the following representation for an OG graph:
    \[
    \Gamma =\mathrm{Arc}_{n_m,i_{m}}\mathrm{Arc}_{n_{m-1},i_{m-1}}\cdot...\cdot\mathrm{Arc}_{n_3,i_3}\mathrm{Arc}_{n_2,i_2}\mathrm{Arc}_{n_1,i_1}\mathrm{Arc}_{2,1}(O) ,
    \]
    for \(m\geq1 \), where \(i_j \in [2j]\). Write 
    \[
    \Xi =\mathrm{Arc}_{n_m,i_{m}}\mathrm{Arc}_{n_{m-1},i_{m-1}}\cdot...\cdot\mathrm{Arc}_{n_3,i_3}\mathrm{Arc}_{n_2,i_2}\mathrm{Arc}_{n_1,i_1}\mathrm{Arc}_{2,1}.
    \]
    We call \(\Xi\) an \emph{arc-sequence of \(\Gamma\) of length \(m\) with index-sequence} \(\{(n_j,i_j)\}_{j=1}^m\). We will say that \(\Gamma = \Xi (O)\) is an arc-sequence representation of \(\Gamma\). We will refer to the \(\mathrm{Arc}\) move \(\mathrm{Arc}_{n_r,i_r}\) as the \(r\)-th numbered move. If the graph is reduced after each move in the sequence, we will call such \(\Xi\) a \emph{reduced} arc-sequence. 

    Recall that by Proposition~\ref{prop:BCFW arcs}, we have that an OG graph \(\Gamma\) is a BCFW graphs iff it has a representation as a reduced arc-sequence with index-sequence \(\{(n_j,i_j)\}_{j=1}^m\) with  all \(i_j\) being even and with \(n_j = 3, 4, 4,...,4\). We will call such an arc-sequence a \emph{BCFW arc-sequence}.

    We define the following partial orders on arc-sequences: If \(\Xi\) is an arc-sequence with an index-sequence \(\{(n_j,i_j)\}_{j=1}^m\), and \(\Xi^\prime\) is a different arc-sequence with an index-sequence  \(\{(i_j^\prime,n_j^\prime)\}_{j=1}^m\), such that \(n_j\leq n_j^\prime\) for every \(j\in[m]\) we say that \(\Xi^\prime\) is smaller or equal to $\Xi$ and write \(\Xi^\prime \leq \Xi\). $\Xi^\prime$ is smaller than $\Xi,$ denoted $\Xi^\prime<\Xi$ or $\Xi^\prime,$ if $\Xi^\prime\leq\Xi$ and for at least one $j,~n^\prime_j<n_j.$
    If for \(\Xi^\prime \leq \Xi\) there exists \(r\geq 0\) such that for every \(j<r\) it holds that \((n_j,i_j) = (i_j^\prime,n_j^\prime)\), we write \(\Xi^\prime \leq_r \Xi\).

    If \(\Xi\) is an arc-sequence such that there exist a BCFW arc-sequence \(\Xi^\prime\) with \(\Xi\leq \Xi^\prime\), we would call \(\Xi\) a \emph{sub-BCFW} sequence. This is equivalent to saying that \(n_1\leq3\) and for \(j>1\) we have \(n_j\leq 4\). If  \(\Xi\leq_r \Xi^\prime\) we call \(\Xi\) an \emph{\(r\)-sub-BCFW} sequence, which is equivalent to \(n_1\leq3\), for \(j>1\) we have \(n_j\leq 4\), for \(j<r\) the previous inequalities are equalities, and also \(i_j\) are all even.
    \end{dfn}

    This will allow us to represent graphs that are “simpler than BCFW’’ in a form amenable to twistor-solving as well. We will show that all limits of BCFW graphs—namely, the graphs corresponding to boundaries of BCFW cells—are indeed sub-BCFW. Some obstacles remain, however: opening a vertex may produce a non-reduced graph. We will therefore need to show how to systematically reduce such graphs so that Proposition~\ref{prop:sol induction prop} can be applied to obtain their twistor-solutions.

    \begin{prop}
    \label{prop:sequence reduction}
        If \(\Gamma\) an OG graph has a sub-BCFW sequence \(\Xi\) of length \(m\) such that the graph is reduced after each move up to the \(r\)-th move (after which the graph is no longer reduced), then there exists  a reduced arc-sequence \(\hat\Xi\) such that \(\hat \Xi\leq_r\Xi\) with the resulting graph being equivalent to \(\Gamma\). 
    \end{prop}
    \begin{proof}

    Consider the sequence \(\Xi\) with index-sequence \(\{i_j,n_j\}_{j=1}^{m}\)
         \[
     \Gamma = \mathrm{Arc}_{n_m,i_{m}}...\mathrm{Arc}_{n_r,i_{r}} \mathrm{Arc}_{n_{r-1},i_{r-1}}...\mathrm{Arc}_{n_1,i_1}\mathrm{Arc}_{2,1}(O),
    \]
     In order to obtain a reduced graph after applying \(\mathrm{Arc}_{n,i}\) on a \(k\)-OG graph, the original graph has to be reduced, and there must not be an arc contained in \(\{i,i+1,...,i+n-1\}\) (considered mod \(2k\)) in the original graph, by Observation~\ref{obs:reduce arc move}. Since it is a sub-BCFW sequence, we know that \(2 \leq n_j \leq 4\). An arc is composed of two indices, thus the only case in which \(\{i,i+1,...,i+n-1\}\) may contain an arc is \(n_j = 4\).Therefore \(n_r=4\).

    Consider the \((k=r)\)-OG graph defined by the arc-sequence truncated at \(r\):
    \[
    \Gamma_0 =  \mathrm{Arc}_{n_{r-1},i_{r-1}}...\mathrm{Arc}_{n_1,i_1}\mathrm{Arc}_{2,1}(O).
    \]
     By assumptions \(\Gamma_0\) is reduced. Moreover, the above arc-sequence for \(\Gamma_0\) is a reduced sub-BCFW of length \(r\), by definition.

     By assumptions we also have that the graph
     \(
     \mathrm{Arc}_{4,i_{r}}(\Gamma_0)
     \)
     is not reduced. Thus, \(\{i_r,i_r+1\}\) (mod \(2r\)) is an external \(2\)-arc in \(\Gamma_0\), a reduced graph.

     Since there are no external vertices between \(\{i_r,i_r+1\}\), we have that any arcs that cross the arc \(\{i_r,i_r+1\}\) must cross it twice; by Proposition~\ref{prop:reduceability}, that is impossible. Hence, no arcs cross the arc \(\{i_r,i_r+1\}\) in \(\Gamma_0\). This means that \(
     \Gamma_0 = \mathrm{Inc}_{i_r}(\Gamma_0^\prime) = \mathrm{Arc}_{2,i_{r}}(\Gamma_0^\prime),
     \)
     which implies
     \(
     \mathrm{Arc}_{4,i_{r}}(\Gamma_0) =\mathrm{Arc}_{4,i_{r}}\mathrm{Arc}_{2,i_{r}}(\Gamma_0^\prime).
     \)

     We can now apply equivalence move~\ref{eqmove:2}  on \( \mathrm{Arc}_{4,i_{r}}\mathrm{Arc}_{2,i_{r}}(\Gamma_0^\prime)\), to obtain \( \mathrm{Arc}_{3,i_{r}}\mathrm{Arc}_{2,i_{r}}(\Gamma_0^\prime)\):
    \begin{center}
\begin{tikzpicture}[scale = 0.8]
\draw (0,0) circle (2cm);
\filldraw[lightgray] (0,1/2) circle (1cm);
\node[scale=3] (c) at (0,1/2)  {\(\Gamma_0^\prime\)};
\draw (0,1/2) circle (1cm);
\draw (1.96962/2, 1.3473/2) --(3.64697/2, 1.64306/2);
\draw (-1.96962/2, 1.3473/2) --(-3.64697/2, 1.64306/2);
\filldraw[black] (0,3.5/2) circle (2pt);

\filldraw[black] (1.27565/2, 3.20949/2) circle (2pt);
\filldraw[black] (-1.27565/2, 3.20949/2) circle (2pt);

 \draw [black] plot [smooth, tension=1.5] coordinates { (1.68446/2, -3.62803/2) (0, -0.75) (-1.68446/2, -3.62803/2)};

\draw (1.68446/2+0.3, -3.62803/2)node[anchor=north]{\(i_r+2\)};

\draw (-1.68446/2-0.3, -3.62803/2)node[anchor=north]{\(i_r+1\)};

\draw (-3.27261/2, -2.3/2) node[anchor=north east]{\(i_r\)}-- (3.27261/2, -2.3/2)node[anchor=north west ]{\(i_r+3\)} ;

\draw[very thick, -latex, shorten <= 0] (2.3,0)   --(3.7,0);

\draw (0+6,0) circle (2cm);
\filldraw[lightgray] (0+6,1/2) circle (1cm);
\node[scale=3] (c) at (0+6,1/2)  {\(\Gamma_0^\prime\)};
\draw (0+6,1/2) circle (1cm);
\draw (1.96962/2+6, 1.3473/2) --(3.64697/2+6 ,1.64306/2);
\draw (-1.96962/2+6, 1.3473/2) --(-3.64697/2+6 ,1.64306/2);
\filldraw[black] (0+6,3.5/2) circle (2pt);

\filldraw[black] (1.27565/2+6, 3.20949/2) circle (2pt);
\filldraw[black] (-1.27565/2+6, 3.20949/2) circle (2pt);

\draw (1.68446/2+0.3+6, -3.62803/2)node[anchor=north]{\(i_r+2\)};

\draw (-1.68446/2-0.3+6, -3.62803/2)node[anchor=north]{\(i_r+1\)};

\draw (-3.27261/2+6, -2.3/2) node[anchor=north east]{\(i_r\)}-- (1.68446/2+6, -3.62803/2) ;

\draw  (-1.68446/2+6, -3.62803/2)-- (3.27261/2+6, -2.3/2)node[anchor=north west ]{\(i_r+3\)} ;

\end{tikzpicture}
\end{center}
By Proposition~\ref{prop:reduceability} the resulting graph \( \mathrm{Arc}_{3,i_{r}}\mathrm{Arc}_{2,i_{r}}(\Gamma_0^\prime)\) is reduced. We can now write that the following graphs \(
     \mathrm{Arc}_{4,i_{r}}(\Gamma_0) \) and \(\mathrm{Arc}_{3,i_{r}}(\Gamma_0)
\) are equivalent. Now we replace the problematic \(\mathrm{Arc}_{4,i_r}\) with \(\mathrm{Arc}_{3,i_r}\) and then continue the sequence as usual. We will have defined a new arc-sequence \(\Xi^1\), with index-sequence \(\{(i_j^1,n_j^1)\}_{j=1}^m\) such that \(i_j^1 = i_j\) for any \(j\), and \(n_j^1 = n_j\) for any \(j\neq r\) and \(n_r^1 = 3 < 4 =n_r\), thus we have  \(\Xi^1\leq_r\Xi\). By Corollary~\ref{moves commute with equiv coro} preforming the same moves on equivalent graphs result in equivalent graphs, and thus the resulting graphs are equivalent. 

Since \(\Xi\) is reduced after each move up to the \(r\)-th move, and we have that \(\Xi^1\) is reduced after the \(r\)-th move, we have that \(\Xi^1\) is reduced after each move at least up to the \(r+1\)-th move. Suppose it is reduced up to the \(r^1\)-th move, with \(r^1>r\). 

We will now continue by induction, and use \(\Xi^0 = \Xi\) as our base case.

For the step, take \(\Xi^q\), a sub-BCFW sequence with \(\Xi^q \leq_r \Xi\) (meaning the sequences are the same up to the \(r^q\) move, and afterwards \(n_j^{q}\leq n_j\)) that is also reduced up to the \(r^q\)-th move. Apply the same algorithm described above to get \(\Xi^{q+1}\), a sub-BCFW sequence that is reduced up to the \(r^{q+1}\)-th move (with \(r^{q+1}>r^q\)), whose graph is equivalent to the graph of \(\Xi^q\), and that  \(\Xi^{q+1}\leq_{r^q}\Xi^q\) (meaning the sequences are the same up to the \(r^q\) move, and afterwards \(n_j^{q+1}\leq n_j^q\)). 

Since \(r^{q+1}>r^q>...>r^1>r\), we also have \(\Xi^{q+1}\leq_r\Xi^q\) and thus \(\Xi^{q+1}\leq_r\Xi\).

Continue applying this algorithm by induction until \(r^q=m\). This will result in \(\hat \Xi \) --  a sequence that is reduced after each move and that has \(\hat \Xi\leq_r\Xi \), and whose graph is equivalent to that of \(\Xi\), finishing the proof.
    \end{proof}

Recall how we label the vertices in a graph added by the \(\mathrm{Arc}\) move in Definitions~\ref{def:arc} and~\ref{vertex label}.
    \begin{dfn}
        An arc-sequence \(\Xi\)  of a graph \(\Gamma\)  gives us a natural ordering of the internal vertices in the graph. Let us label them \(\{v_\ell\}_{\ell=1}^{2m-1}\) in the following way:
    \[
    \Gamma =\mathrm{Arc}_{n_m,i_{n_m}}(v_{2m-1},v_{2m-2})\cdot...\cdot\mathrm{Arc}_{n_2,i_2}(v_{3},v_2)\mathrm{Arc}_{n_1,i_1}(v_1)\mathrm{Arc}_{2,1}(O) ,
    \]
    and call \(\{v_\ell\}_{\ell=1}^{2m-1}\) \emph{vertex-sequence} associated to \(\Xi\). 
    
    Recall that by Definition~\ref{def:arc} the vertices the arc-sequence induces a natural orientation for each vertex; this orientation often does not combine into a well-defined orientation on the whole graph. When referring to a vertex in a vertex-sequence that is associated to an arc-sequence, we implicitly mean for them to be oriented with the orientation induced by the arc-sequence.
    \end{dfn}
\begin{rmk}
  Notice the vertices added by move numbered \(r\) in the sequence are numbered \(v_{2r-2}\) and \(v_{2r-1}\).
  \end{rmk}
We will now work up to show that graphs that correspond to limits of BCFW graphs, that is, to cells that are boundaries of BCFW cells, are indeed simpler then BCFW cells.
\begin{prop}

    For \(\Gamma\) a perfectly oriented BCFW graph with a reduced BCFW sequence \(\Xi\) and an associated vertex-sequence \(\{v_\ell\}_{\ell=1}^{2m-1}\). Let \(v\) some internal vertex \(v_{2r-2}\) or \(v_{2r-1}\) (that is, a vertex added by move numbered \(r\) in the sequence \(\Xi\)). 
    Then in the equivalence class of the graph
    \(
    \lim_{v\rightarrow L}\Gamma
    \)
    there is a representative with a reduced sub-BCFW arc-sequence \(\hat\Xi\), and \(\hat\Xi \leq_r \Xi\).
\end{prop}
Here we have kept track of graphs that are equivalent but not identical. Since constructing graphs explicitly via sequences of moves makes the distinction between equivalent representatives especially visible, it is convenient to maintain this distinction for the purposes of this proof, while we are still becoming comfortable with these techniques. However, as we ultimately care only about equivalence classes of OG graphs, we will later revert to not distinguishing between equivalent representatives.

\begin{proof}

    By Proposition~\ref{prop:BCFW arcs} \(\Gamma\) has a reduced BCFW arc-sequence representation
     \[
    \Gamma =\mathrm{Arc}_{n_m,i_{m}}\mathrm{Arc}_{n_m,i_{m-1}}\cdot...\cdot\mathrm{Arc}_{n_3,i_3}\mathrm{Arc}_{n_2,i_2}\mathrm{Arc}_{n_1,i_1}\mathrm{Arc}_{2,1}(O) ,
    \]
    for \(m\geq1 \), where \(i_j \in [2j]\) and are even, \(n_1 = 3\) and \(n_j = 4\) for \(j>1\), with the result being reduced after each \(\mathrm{Arc}\) move.

    By Corollary~\ref{limit comute for graphs coro} we can commute the limit operation up to the \(\mathrm{Arc}\) move that adds the vertex \(v\), which is the \(r\)-th one. By Proposition~\ref{prop:BCFW arcs}, we now have an expression of one of the following forms:
     \[
    \lim_{v_{2r-1}\rightarrow L}\Gamma = \mathrm{Arc}_{4,2i_{m}}\cdot...\cdot\mathrm{Arc}_{n_{r+1},i_{r+1}} \lim_{v_{2r-1}\rightarrow L}\mathrm{Arc}_{n_r,i_r}(v_{2r-1},v_{2r-2})(\Gamma_0),
    \]
    for some \(1<r\leq m\) with \(\Gamma_0\) a BCFW graph,
      \[
    \lim_{v_{2r-2}\rightarrow L}\Gamma = \mathrm{Arc}_{4,2i_{m}}\cdot...\cdot\mathrm{Arc}_{n_{r+1},i_{r+1}} \lim_{v_{2r-2}\rightarrow L}\mathrm{Arc}_{n_r,i_r}(v_{2r-1},v_{2r-2})(\Gamma_0),
    \]
    for some \(1<r\leq m\) with \(\Gamma_0\) a BCFW graph, or if \(r=1\)
    \[
    \lim_{v_1\rightarrow L}\Gamma = \mathrm{Arc}_{n_m,i_{m}}\cdot...\cdot\mathrm{Arc}_{n_{1},i_1} \lim_{v_1\rightarrow L}\mathrm{Arc}_{3,2}(v_1)\mathrm{Arc}_{2,1}(O).
    \]
    Since \(i_j\) are even and \(n_j = 4\) for \(j>1\), by Corollaries~\ref{arc 3 boundary coro} and~\ref{arc 4 boundary coro}, we can replace \\\(\lim_{v_{2r-1}\rightarrow L}\mathrm{Arc}_{n_r,i_r}(v_{2r-1},v_{2r-2}),\) \(\lim_{v_{2r-2}\rightarrow L}\mathrm{Arc}_{n_r,i_r}(v_{2r-1},v_{2r-2}),\) or \(\lim_{v_1\rightarrow L}\mathrm{Arc}_{3,i_r}(v_1)\), by \(\mathrm{Arc}_{n,\ell}\) with \(n< n_r\) and some \(l \in [2j]\). This defines a new arc-sequence \(\Xi^\prime\) with index-sequence \(\{(i_j^\prime,n_j^\prime)\}_{j=1}^m\) such that \(i_j^\prime = i_j\) and \(n_j^\prime = n_j\) for \(j\neq r\) and \(n_r^\prime \leq n_r\), thus  \(\Xi^\prime\leq_r\Xi\). 
    
    Since \(\Xi\) is reduced and the sequence are identical up to the \(r\)-th move, we know that \(\Xi^\prime\) is reduced up to the \(r\)-th move. Now we can apply Proposition~\ref{prop:sequence reduction} to get \(\hat \Xi\) a reduced arc-sequence with \(\hat \Xi\leq_r\Xi^\prime\) and thus \(\hat \Xi\leq_r\Xi\). Since the graph of \(\hat \Xi\) is equivalent to that of \(\Xi^\prime\), we get that it is equivalent to \(
    \lim_{v\rightarrow L}\Gamma\). Finishing the proof.
\end{proof}

\begin{lem}
    Let \(\Gamma\) be a perfectly oriented graph with a reduced \(q\)-sub-BCFW sequence \(\Xi\)  and vertex-sequence \(\{v_l\}_{l=1}^{2m-1}\). Let \(v\) be the vertex \(v_{2r-2}\) or \(v_{2r-1}\) (that is, a vertex added by move numbered \(r\) in the sequence \(\Xi\)) with \(r\leq q\). Then 
    \(
    \lim_{v\rightarrow L}\Gamma
    \)
    has a representation as a reduced \(r\)-sub-BCFW arc-sequence \(\hat\Xi\)  with \(\hat\Xi \leq_r \Xi\).
\end{lem}
\begin{proof}
    Consider the sequence \(\Xi\) 
     \[
     \Gamma = \mathrm{Arc}_{n_m,i_{m}}...\mathrm{Arc}_{n_q+1,i_{q+1}} \mathrm{Arc}_{n_{q},i_{q}}...\mathrm{Arc}_{n_1,i_1}\mathrm{Arc}_{2,1}(O),
    \]
    and define 
    \(
    \Gamma^0 = \mathrm{Arc}_{n_{q},i_{q}}...\mathrm{Arc}_{n_1,i_1}\mathrm{Arc}_{2,1}(O).
    \)
    As \(\Xi\) is a reduced \(q\) sub-BCFW sequence, we have that the above sequence \(\Xi^0\) for \(\Gamma^0\) is a reduced BCFW sequence. Thus, by Proposition~\ref{prop:BCFW arcs} \(\Gamma^0\) is a BCFW graph.

    As \(r\leq q\),  by Corollary~\ref{limit comute for graphs coro}  we have that
    \begin{align*}
        \lim_{v\rightarrow L}\Gamma &= \lim_{v\rightarrow L}\mathrm{Arc}_{n_m,i_{m}}...\mathrm{Arc}_{n_{q+1},i_{q+1}} \mathrm{Arc}_{n_{q},i_{q}}...\mathrm{Arc}_{n_1,i_1}\mathrm{Arc}_{2,1}(O)\\
        &= \mathrm{Arc}_{n_m,i_{m}}...\mathrm{Arc}_{n_{q+1},i_{q+1}} \lim_{v\rightarrow L}\mathrm{Arc}_{n_{q},i_{q}}...\mathrm{Arc}_{n_1,i_1}\mathrm{Arc}_{2,1}(O)\\
        &= \mathrm{Arc}_{n_m,i_{m}}...\mathrm{Arc}_{n_{q+1},i_{q+1}} \lim_{v\rightarrow L}\Gamma^0.
    \end{align*}

    By the previous claim, as we saw \(\Gamma^0\) is a BCFW graph with a reduced BCFW sequence \(\Xi^0\), we have a reduced \(r\)-sub-BCFW sequence \(\Xi^\prime\) with \(\Xi^\prime\leq_r\Xi^0\) such that
    \[
    \lim_{v\rightarrow L}\Gamma^0 = \mathrm{Arc}_{n_{q}^\prime,i_{q}^\prime}...\mathrm{Arc}_{n_1^\prime,i_1^\prime}\mathrm{Arc}_{2,1}(O).
    \]
    Therefore,
\begin{align*}
        \lim_{v\rightarrow L}\Gamma &= \mathrm{Arc}_{n_m,i_{m}}...\mathrm{Arc}_{n_{q+1},i_q+1} \lim_{\alpha\rightarrow L}\Gamma^0\\
        &=\mathrm{Arc}_{n_m,i_{m}}...\mathrm{Arc}_{n_{q+1},i_{q+1}} \mathrm{Arc}_{n_{q}^\prime,i_{q}^\prime}...\mathrm{Arc}_{n_1^\prime,i_1^\prime}\mathrm{Arc}_{2,1}(O)
    \end{align*}
is an arc-sequence for \(\lim_{v\rightarrow L}\Gamma\). Call it \(\Xi^1\).

We have \(\Xi^\prime\leq_r\Xi^0\), that is, the former has smaller or equal \(n^\prime_j\) and they are identical up to the \(r\)-the move. \(\Xi^0\) is just  \(\Xi\) truncated at the \(q\)-th move, thus \(\Xi^\prime\) is identical to \(\Xi\) up to the \(r\)-th move, and \(n^\prime_j\leq n_j\) afterwards. 

\(\Xi^1\) is defined as \(\Xi^\prime\) up to the \(q\)-th move and then continuing as \(\Xi\). Thus, \(\Xi^1\) is identical to \(\Xi\) before the \(r\)-th move and after the \(q\)-th move, with smaller or equal \(n_j\) in between. This means that by definition \(\Xi^1\leq_r\Xi\).

As \(\Xi^\prime\) is reduced, we have that the sequence \(\Xi^1\) is reduced up to the \(q+1\) move. By Proposition~\ref{prop:sequence reduction} we have \(\hat \Xi\) an equivalent reduced arc-sequence with with \(\hat \Xi \leq_{q+1} \Xi^1\). As \(r\leq q\), we have that  \(\hat \Xi\leq_r\Xi\).

We have thus found \(\hat \Xi\), a reduced arc-sequence for \(\lim_{v\rightarrow L}\Gamma\) with \(\hat \Xi\leq_r\Xi\).  As  \(\Xi\) is a reduced \(q\) sub-BCFW arc-sequence, we have that \(\hat \Xi\) is a reduced \(r\)-sub-BCFW sequence.
\end{proof}

\begin{prop}
\label{prop:arc limits boundary sub BCFW}

    Let \(\Gamma\) be a perfectly oriented BCFW graph with (by Proposition~\ref{prop:BCFW arcs}) a representation as a reduced BCFW arc-sequence \(\Xi\) with vertex-sequence \(\{v_\ell\}_{\ell=1}^{2m-1}\).
    
     Let \(\{v_{\ell_p}\}_{p=1}^q\) be some vertices added by moves numbered \(r_p\) (that is, \(\ell_p \in \{2r_p-1,2r_p-2\}\)), and let \(r = \min \{r_p\}_{p=1}^q \).
 Then
    \(
    \lim_{v_{\ell_p}\rightarrow L_q}... \lim_{v_{\ell_1}\rightarrow L_1}\Gamma 
    \)
    has  a representation as a reduced arc-sequence \(\hat\Xi\) with \(\hat\Xi \leq_r \Xi\), and thus \(\hat\Xi\) is an \(r\)-sub-BCFW sequence.
\end{prop}
\begin{proof}

     By Corollary \(~\ref{order of limits coro}\) limit operations commute for \(OG\) graphs. This means that we can reorder the limits in the expression as we wish. We will reorder the limits such that we take the limits of angles corresponding to the vertices \(\{v_{\ell_q}\}_{p=1}^q\) with \(\ell_p\) in decreasing order. Since \(\ell_p\in\{2r_p-1,2r_p-2\}\) that means that \(r_p\) are in non-increasing order.

     We will prove the claim by induction on \(q\). The case of \(q=1\) reduces to our previous claim.
     For the induction step consider
     \(
    \lim_{v_{\ell_{q+1}}\rightarrow L_{q+1}} \Gamma^\prime \defeq \lim_{v_{\ell_{q+1}}\rightarrow L_{q+1}}\lim_{v_{\ell_{q}}\rightarrow L_{q}}... \lim_{v_{\ell_{1}}\rightarrow L_1}\Gamma,
    \)
    for \(q>1\).
    By the induction hypothesis  \(\Gamma^\prime\) has a representation  as a reduced \(r\)-sub-BCFW sequence \(\Xi^\prime\) and  \(\Xi^\prime \leq_r \Xi\) , where \(r = \min \{r_p\}_{p=1}^q \).
    As we reordered \(r_p\) to be non-increasing, we have that \(r_{q+1}\leq r\). Thus \(r_{q+1} = \min \{r_p\}_{p=1}^{q+1} \). By the previous Lemma we have a representation for \( \lim_{v_{\ell_{q+1}}\rightarrow L_{q+1}} \Gamma^\prime\) as a reduced \(r_{q+1}\) sub-BCFW arc-sequence \(\hat\Xi\) with \(\hat\Xi \leq_{r_{q+1}} \Xi\). This implies it is also \({r_{q+1}}\) sub-BCFW, finishing the proof.
    \end{proof}

    \subsection{Classifying Codimension One Boundaries}
    
In this section we will classify codimension $1$ boundaries of orthitroid cells in \(\OGnon{k}{2k}\). We start by seeing when boundaries of graphs correspond to codimension \(1\) boundary cells.

Consider a BCFW graph \(\Gamma\). By Proposition~\ref{prop:limits are boundaries} all the boundary cells of \(\Omega_\Gamma\) correspond to taking some series of limit operations on \(\Gamma\). By Theorem~\ref{param bij} we know that the dimension of \(\Omega_{\Gamma_0}\) is upper bounded by the number of internal vertices of \(\Gamma_0\), with equality if and only $\Gamma_0$ is reduced.

The limit operation can only reduce the number of internal vertices in a reduced graph. Thus, all codimension \(1\) boundaries of \(\Omega_\Gamma\) are of the form \(\Omega_{\Gamma_0}\) with \(\Gamma_0 = \lim_{v\rightarrow L}\Gamma\) for a \emph{single} internal vertex in some orientation, and \(L\) being one of the two possible limits (be it \(0\) or \(\frac{\pi}{2}\) if the vertex is with trigonometric orientation, and \(0\) or \(\infty\) if the vertex is with hyperbolic orientation). 

It is important to note that performing a limit operation on a non-reduced graph might not reduce the dimension of the corresponding cell. That is precisely because some limit operations on non-reduced graphs can result in a graph in the same equivalence class as the original, as equivalence moves~\ref{eqmove:1} and~\ref{eqmove:2} are equivalent to opening an internal vertex. However, preforming a limit operation on a reduced graph will always reduce the dimension.

What still requires a resolution is which limit operations reduce the dimension by exactly one. Since taking a limit of a single vertex removes that vertex alone from the graph, this question is equivalent to whether the resulting graph is reduced. To summarize, we have
\begin{prop}[\cite{companion}]
    \label{prop:codim 1 boundaries iff reduced}
    Let \(\Gamma\) be a reduced OG graph.  \(\Omega_{\Gamma_0}\) is a codimension \(1\) boundary cell of \(\Omega_\Gamma\) iff \(\Gamma_0 = \lim_{v\rightarrow L}\Gamma\) such that \(\Gamma_0\) is immediately reduced after the limit operation.
\end{prop}

\begin{prop}[\cite{companion}]
\label{prop:closure of orthit}
    Closures of orthitroid cells are stratified by orthitroid cells, and are compact. The union of an orthitroid cell with one of its codimension \(1\) boundary cells form a manifold with a boundary.
\end{prop}

\subsubsection{External Boundaries}

Consider a (reduced) BCFW graph \(\Gamma\), with an internal vertex \(v\) that is adjacent to an external vertex. By Theorem~\ref{thm:BCFW are trees}, \(\Gamma\) is a tree of triangles. As \(v\) is part of a graph that is a tree of triangles, we get a configuration as in Figure~\ref{fig:ext int v}.
\begin{figure}[H]
    \centering
\begin{center}
\begin{tikzpicture}[scale = 0.8, every node/.style={scale=0.8}]]
\draw (0,0) circle (2cm);

\filldraw[lightgray] (-.9,1/2) circle (0.6cm);
\node[scale=2.5] (c) at (-.9,1/2)  {\(\Gamma_1\)};
\draw (-.9,1/2) circle (.6cm);

\draw (-1.38541, 0.147329) --(-1.97538, -0.312869);
\draw (-0.806139, 1.09261) --(-0.618034, 1.90211);
\filldraw[black] (-1.49522, 0.988483) circle (1.75pt);
\filldraw[black] (-1.14384, 1.30383) circle (1.75pt);
\filldraw[black] (-1.69938, 0.531408) circle (1.75pt);

\draw  (-0.475736, 0.0757359)--(0.415823, -1.9563);

\filldraw[lightgray] (.9,1/2) circle (0.6cm);
\node[scale=2.5] (c) at (.9,1/2)  {\(\Gamma_2\)};
\draw (.9,1/2) circle (.6cm);

\draw (1.38541, 0.147329) --(1.97538, -0.312869);
\draw (0.806139, 1.09261) --(0.618034, 1.90211);
\filldraw[black] (1.49522, 0.988483) circle (1.75pt);
\filldraw[black] (1.14384, 1.30383) circle (1.75pt);
\filldraw[black] (1.69938, 0.531408) circle (1.75pt);

\draw  (0.475736, 0.0757359)--(-0.415823, -1.9563);

\draw  (-0.67039, -0.0543277)--(0.67039, -0.0543277);
\draw (0, -0.8)node[anchor=north east]{\(v\)};

\end{tikzpicture}
\end{center}
    \caption{ the BCFW graph \(\Gamma\) with internal vertex \(v\) that is adjacent to an external vertex}
    \label{fig:ext int v}
\end{figure}
As we saw in the Section~\ref{sec:boundary cells}, 
 the two boundaries obtained by opening the vertex \(v\) correspond to graphs as in Figure~\ref{fig:opening ext int v}.
 \begin{figure}[H]
     \centering
              
     \begin{subfigure}{0.4\textwidth}
\begin{center}
\begin{tikzpicture}[scale = 0.8, every node/.style={scale=0.8}]]
\draw (0,0) circle (2cm);

\filldraw[lightgray] (-.9,1/2) circle (0.6cm);
\node[scale=2.5] (c) at (-.9,1/2)  {\(\Gamma_1\)};
\draw (-.9,1/2) circle (.6cm);

\draw (-1.38541, 0.147329) --(-1.97538, -0.312869);
\draw (-0.806139, 1.09261) --(-0.618034, 1.90211);
\filldraw[black] (-1.49522, 0.988483) circle (1.75pt);
\filldraw[black] (-1.14384, 1.30383) circle (1.75pt);
\filldraw[black] (-1.69938, 0.531408) circle (1.75pt);

\draw  (-0.415823, 0.145638)--(-0.415823, -1.9563);

\filldraw[lightgray] (.9,1/2) circle (0.6cm);
\node[scale=2.5] (c) at (.9,1/2)  {\(\Gamma_2\)};
\draw (.9,1/2) circle (.6cm);

\draw (1.38541, 0.147329) --(1.97538, -0.312869);
\draw (0.806139, 1.09261) --(0.618034, 1.90211);
\filldraw[black] (1.49522, 0.988483) circle (1.75pt);
\filldraw[black] (1.14384, 1.30383) circle (1.75pt);
\filldraw[black] (1.69938, 0.531408) circle (1.75pt);

\draw  (0.415823, 0.145638)--(0.415823, -1.9563);

\draw  (-0.67039, -0.0543277)--(0.67039, -0.0543277);

 \end{tikzpicture}
\end{center}
\caption{the boundary graph \(\Gamma_L\)}
     \label{fig:opening ext int v gl}
     \end{subfigure}
     \begin{subfigure}{0.4\textwidth}
\begin{center}
\begin{tikzpicture}[scale = 0.8, every node/.style={scale=0.8}]]

\draw (0+6,0) circle (2cm);

\filldraw[lightgray] (-.9+6,1/2) circle (0.6cm);
\node[scale=2.5] (c) at (-.9+6,1/2)  {\(\Gamma_1\)};
\draw (-.9+6,1/2) circle (.6cm);

\draw (-1.38541+6, 0.147329) --(-1.97538+6, -0.312869);
\draw (-0.806139+6, 1.09261) --(-0.618034+6, 1.90211);
\filldraw[black] (-1.49522+6, 0.988483) circle (1.75pt);
\filldraw[black] (-1.14384+6, 1.30383) circle (1.75pt);
\filldraw[black] (-1.69938+6, 0.531408) circle (1.75pt);

\filldraw[lightgray] (.9+6,1/2) circle (0.6cm);
\node[scale=2.5] (c) at (.9+6,1/2)  {\(\Gamma_2\)};
\draw (.9+6,1/2) circle (.6cm);

\draw (1.38541+6, 0.147329) --(1.97538+6, -0.312869);
\draw (0.806139+6, 1.09261) --(0.618034+6, 1.90211);
\filldraw[black] (1.49522+6, 0.988483) circle (1.75pt);
\filldraw[black] (1.14384+6, 1.30383) circle (1.75pt);
\filldraw[black] (1.69938+6, 0.531408) circle (1.75pt);

\draw  (-0.67039+6, -0.0543277)--(0.67039+6, -0.0543277);

\draw [black] plot [smooth, tension=1.2] coordinates {  (-0.415823+6, 0.145638)(6,-0.5) (0.415823+6, 0.145638)};
\draw [black] plot [smooth, tension=1.2] coordinates {  (-0.618034+6, -1.90211)(6,-1.4) (0.618034+6, -1.90211)};

\end{tikzpicture}
\end{center}
\caption{the boundary graph \(\Gamma_R\)}
     \label{fig:opening ext int v gr}
     \end{subfigure}
     \caption{the two ways of opening an internal vertex that is adjacent to an external vertex}
     \label{fig:opening ext int v}
 \end{figure}
Recall that by Proposition~\ref{prop:reduceability} a graph is reduced iff every arc does not cross itself and no pair of arcs crosses more than once. Since \(\Gamma\) is reduced, so is \(\Gamma_L\) as if two arcs cross in \(\Gamma_L\) they also cross in \(\Gamma\). On the other hand, \(\Gamma_R\) is clearly not. By Proposition~\ref{prop:codim 1 boundaries iff reduced}, this implies that \(\Gamma_L\) alone corresponds to a codimension \(1\) boundary. It is also clear that there is a single BCFW cells such that opening a single vertex would result in the graph \(\Gamma_L\). Thus,
\begin{obs}
\label{co-dime 1 external obs}
    \(\Omega_{\Gamma_L}\) is the only codimension \(1\) boundary of \(\Omega_\Gamma\) and \(\Omega_{\Gamma_R}\) is not.  \(\Omega_\Gamma\) is the only BCFW cell with \(\Omega_{\Gamma_L}\) as a boundary.
\end{obs}
\subsubsection{Internal Boundaries}
 Consider a (reduced) BCFW graph \(\Gamma_+\), with an internal vertex \(v\) that is not adjacent to any external vertex. By  Theorem~\ref{thm:BCFW are trees}, \(\Gamma_+\) is a tree of triangles. It is easy to see we must have the configuration as in Figure~\ref{fig:int int v}.
\begin{figure}[H]
    \centering
 \begin{center}
            \begin{tikzpicture}[scale = 0.7, every node/.style={scale=0.7}]]

            \draw (0,0) circle (3cm);

            \filldraw[lightgray] (-1.25,1.25) circle (0.75cm);
            \draw (-1.25,1.25) circle (0.75cm);
            \node[scale=3] (c) at (-1.25,1.25)  {\(\Gamma_1\)};

            \filldraw[lightgray] (1.25,1.25) circle (0.75cm);
            \draw (1.25,1.25) circle (0.75cm);
            \node[scale=3] (c) at (1.25,1.25)  {\(\Gamma_2\)};

            \filldraw[lightgray] (1.25,-1.25) circle (0.75cm);
            \draw (1.25,-1.25) circle (0.75cm);
            \node[scale=3] (c) at (1.25,-1.25)  {\(\Gamma_3\)};

            \filldraw[lightgray] (-1.25,-1.25) circle (0.75cm);
            \draw (-1.25,-1.25) circle (0.75cm);
            \node[scale=3] (c) at (-1.25,-1.25)  {\(\Gamma_4\)};

            \draw (0.875, 0.600481)node[anchor = north]{\(a_2\)} -- (-0.875, -0.600481)node[anchor =south]{\(a_4\)};
            \draw (0.875, -0.600481)node[anchor = south]{\(a_3\)} -- (-0.875, 0.600481)node[anchor =north]{\(a_1\)};

            \draw ( 0.600481,0.875)node[anchor = east]{\(b_2\)} -- (0.600481,-0.875 )node[anchor = east]{\(b_3\)};
            \draw (-0.600481,0.875)node[anchor = west]{\(b_1\)} -- ( -0.600481,-0.875)node[anchor = west]{\(b_4\)};

            \draw (0,0)node[anchor = north]{\(v\)};

            \draw(1.89952, 0.875)--(2.94236, 0.585271);
            \draw(1.89952, -0.875)--(2.94236, -0.585271);
            \draw(-1.89952, 0.875)--(-2.94236, 0.585271);
            \draw(-1.89952, -0.875)--(-2.94236, -0.585271);
            \draw(0.875, 1.89952)--(0.585271, 2.94236);
            \draw(0.875, -1.89952)--(0.585271, -2.94236);
            \draw(-0.875, 1.89952)--(-0.585271, 2.94236);
            \draw(-0.875, -1.89952)--(-0.585271, -2.94236);

            \filldraw[black] (1.94454, 1.94454) circle (2pt);
            \filldraw[black] (2.21171, 1.54866) circle (2pt);
            \filldraw[black] (1.54866, 2.21171) circle (2pt);

            \filldraw[black] (1.94454, -1.94454) circle (2pt);
            \filldraw[black] (2.21171, -1.54866) circle (2pt);
            \filldraw[black] (1.54866, -2.21171) circle (2pt);

            \filldraw[black] (-1.94454, -1.94454) circle (2pt);
            \filldraw[black] (-2.21171, -1.54866) circle (2pt);
            \filldraw[black] (-1.54866, -2.21171) circle (2pt);

            \filldraw[black] (-1.94454, 1.94454) circle (2pt);
            \filldraw[black] (-2.21171, 1.54866) circle (2pt);
            \filldraw[black] (-1.54866, 2.21171) circle (2pt);

        \end{tikzpicture}
        \end{center}
    \caption{the BCFW graph \(\Gamma_+\) with internal vertex \(v\)}
    \label{fig:int int v}
\end{figure}

Recall that by Proposition~\ref{prop:reduceability} a graph is reduced iff every arc does not cross itself and no pair of arcs crosses more than once. we thus get that the sugraphs \(\Gamma_{1,2,3,4}\) are reduced as well. 
There are two ways of taking limits of the angle associated to \(v\) which correspond to the two different ways of opening the vertex \(v\) seen in Figure~\ref{fig:opening int int v} by Defintion~\ref{def opening the vertex}.

\begin{figure}[H]
    \centering
    \begin{subfigure}[t]{0.4\textwidth}
        \begin{center}
            \begin{tikzpicture}[scale = 0.8, every node/.style={scale=0.8}]]

            \draw (0,0) circle (3*0.8cm);

            \filldraw[lightgray] (-1.25*0.8,1.25*0.8) circle (0.75*0.8cm);
            \draw (-1.25*0.8,1.25*0.8) circle (0.75*0.8cm);
            \node[scale=3*0.8] (c) at (-1.25*0.8,1.25*0.8)  {\(\Gamma_1\)};

            \filldraw[lightgray] (1.25*0.8,1.25*0.8) circle (0.75*0.8cm);
            \draw (1.25*0.8,1.25*0.8) circle (0.75*0.8cm);
            \node[scale=3*0.8] (c) at (1.25*0.8,1.25*0.8)  {\(\Gamma_2\)};

            \filldraw[lightgray] (1.25*0.8,-1.25*0.8) circle (0.75*0.8cm);
            \draw (1.25*0.8,-1.25*0.8) circle (0.75*0.8cm);
            \node[scale=3*0.8] (c) at (1.25*0.8,-1.25*0.8)  {\(\Gamma_3\)};

            \filldraw[lightgray] (-1.25*0.8,-1.25*0.8) circle (0.75*0.8cm);
            \draw (-1.25*0.8,-1.25*0.8) circle (0.75*0.8cm);
            \node[scale=3*0.8] (c) at (-1.25*0.8,-1.25*0.8)  {\(\Gamma_4\)};

            \draw (0.875*0.8, -0.600481*0.8)node[anchor = south]{\(a_3\)} -- (-0.875*0.8, -0.600481*0.8)node[anchor =south]{\(a_4\)};
            \draw  (0.875*0.8, 0.600481*0.8)node[anchor = north]{\(a_2\)}-- (-0.875*0.8, 0.600481*0.8)node[anchor =north]{\(a_1\)};

            \draw ( 0.600481*0.8,0.875*0.8)node[anchor = east]{\(b_2\)} -- (0.600481*0.8,-0.875 *0.8)node[anchor = east]{\(b_3\)};
            \draw (-0.600481*0.8,0.875*0.8)node[anchor = west]{\(b_1\)} -- ( -0.600481*0.8,-0.875*0.8)node[anchor = west]{\(b_4\)};

            \draw(1.89952*0.8, 0.875*0.8)--(2.94236*0.8, 0.585271*0.8);
            \draw(1.89952*0.8, -0.875*0.8)--(2.94236*0.8, -0.585271*0.8);
            \draw(-1.89952*0.8, 0.875*0.8)--(-2.94236*0.8, 0.585271*0.8);
            \draw(-1.89952*0.8, -0.875*0.8)--(-2.94236*0.8, -0.585271*0.8);
            \draw(0.875*0.8, 1.89952*0.8)--(0.585271*0.8, 2.94236*0.8);
            \draw(0.875*0.8, -1.89952*0.8)--(0.585271*0.8, -2.94236*0.8);
            \draw(-0.875*0.8, 1.89952*0.8)--(-0.585271*0.8, 2.94236*0.8);
            \draw(-0.875*0.8, -1.89952*0.8)--(-0.585271*0.8, -2.94236*0.8);

            \filldraw[black] (1.94454*0.8, 1.94454*0.8) circle (2*0.8pt);
            \filldraw[black] (2.21171*0.8, 1.54866*0.8) circle (2*0.8pt);
            \filldraw[black] (1.54866*0.8, 2.21171*0.8) circle (2*0.8pt);

            \filldraw[black] (1.94454*0.8, -1.94454*0.8) circle (2*0.8pt);
            \filldraw[black] (2.21171*0.8, -1.54866*0.8) circle (2*0.8pt);
            \filldraw[black] (1.54866*0.8, -2.21171*0.8) circle (2*0.8pt);

            \filldraw[black] (-1.94454*0.8, -1.94454*0.8) circle (2*0.8pt);
            \filldraw[black] (-2.21171*0.8, -1.54866*0.8) circle (2*0.8pt);
            \filldraw[black] (-1.54866*0.8, -2.21171*0.8) circle (2*0.8pt);

            \filldraw[black] (-1.94454*0.8, 1.94454*0.8) circle (2*0.8pt);
            \filldraw[black] (-2.21171*0.8, 1.54866*0.8) circle (2*0.8pt);
            \filldraw[black] (-1.54866*0.8, 2.21171*0.8) circle (2*0.8pt);

         \end{tikzpicture}
        \end{center}
         \caption{the boundary graph \(\Gamma_0\)}
    \label{fig:opening int int v g0}        
    \end{subfigure}
    \begin{subfigure}[t]{0.4\textwidth}
            \begin{center}
            \begin{tikzpicture}[scale = 0.8, every node/.style={scale=0.8}]]

             \draw (0,0) circle (3*0.8cm);

            \filldraw[lightgray] (-1.25*0.8,1.25*0.8) circle (0.75*0.8cm);
            \draw (-1.25*0.8,1.25*0.8) circle (0.75*0.8cm);
            \node[scale=3*0.8] (c) at (-1.25*0.8,1.25*0.8)  {\(\Gamma_1\)};

            \filldraw[lightgray] (1.25*0.8,1.25*0.8) circle (0.75*0.8cm);
            \draw (1.25*0.8,1.25*0.8) circle (0.75*0.8cm);
            \node[scale=3*0.8] (c) at (1.25*0.8,1.25*0.8)  {\(\Gamma_2\)};

            \filldraw[lightgray] (1.25*0.8,-1.25*0.8) circle (0.75*0.8cm);
            \draw (1.25*0.8,-1.25*0.8) circle (0.75*0.8cm);
            \node[scale=3*0.8] (c) at (1.25*0.8,-1.25*0.8)  {\(\Gamma_3\)};

            \filldraw[lightgray] (-1.25*0.8,-1.25*0.8) circle (0.75*0.8cm);
            \draw (-1.25*0.8,-1.25*0.8) circle (0.75*0.8cm);
            \node[scale=3*0.8] (c) at (-1.25*0.8,-1.25*0.8)  {\(\Gamma_4\)};

            \draw (0.875*0.8, -0.600481*0.8)node[anchor = south]{\(a_3\)};
            \draw (-0.875*0.8, -0.600481*0.8)node[anchor =south]{\(a_4\)};
            \draw  (0.875*0.8, 0.600481*0.8)node[anchor = north]{\(a_2\)};
            \draw (-0.875*0.8, 0.600481*0.8)node[anchor =north]{\(a_1\)};

            \draw ( 0.600481*0.8,0.875*0.8)node[anchor = east]{\(b_2\)} -- (0.600481*0.8,-0.875*0.8 )node[anchor = east]{\(b_3\)};
            \draw (-0.600481*0.8,0.875*0.8)node[anchor = west]{\(b_1\)} -- ( -0.600481*0.8,-0.875*0.8)node[anchor = west]{\(b_4\)};

            \draw [black] plot [smooth, tension=1.5] coordinates { (0.875*0.8, -0.600481*0.8) (0.3*0.8,0) (0.875*0.8, 0.600481*0.8)};
            \draw [black] plot [smooth, tension=1.5] coordinates { (-0.875*0.8, -0.600481*0.8) (-0.3*0.8,0) (-0.875*0.8, 0.600481*0.8)};

            \draw(1.89952*0.8, 0.875*0.8)--(2.94236*0.8, 0.585271*0.8);
            \draw(1.89952*0.8, -0.875*0.8)--(2.94236*0.8, -0.585271*0.8);
            \draw(-1.89952*0.8, 0.875*0.8)--(-2.94236*0.8, 0.585271*0.8);
            \draw(-1.89952*0.8, -0.875*0.8)--(-2.94236*0.8, -0.585271*0.8);
            \draw(0.875*0.8, 1.89952*0.8)--(0.585271*0.8, 2.94236*0.8);
            \draw(0.875*0.8, -1.89952*0.8)--(0.585271*0.8, -2.94236*0.8);
            \draw(-0.875*0.8, 1.89952*0.8)--(-0.585271*0.8, 2.94236*0.8);
            \draw(-0.875*0.8, -1.89952*0.8)--(-0.585271*0.8, -2.94236*0.8);

            \filldraw[black] (1.94454*0.8, 1.94454*0.8) circle (2*0.8pt);
            \filldraw[black] (2.21171*0.8, 1.54866*0.8) circle (2*0.8pt);
            \filldraw[black] (1.54866*0.8, 2.21171*0.8) circle (2*0.8pt);

            \filldraw[black] (1.94454*0.8, -1.94454*0.8) circle (2*0.8pt);
            \filldraw[black] (2.21171*0.8, -1.54866*0.8) circle (2*0.8pt);
            \filldraw[black] (1.54866*0.8, -2.21171*0.8) circle (2*0.8pt);

            \filldraw[black] (-1.94454*0.8, -1.94454*0.8) circle (2*0.8pt);
            \filldraw[black] (-2.21171*0.8, -1.54866*0.8) circle (2*0.8pt);
            \filldraw[black] (-1.54866*0.8, -2.21171*0.8) circle (2*0.8pt);

            \filldraw[black] (-1.94454*0.8, 1.94454*0.8) circle (2*0.8pt);
            \filldraw[black] (-2.21171*0.8, 1.54866*0.8) circle (2*0.8pt);
            \filldraw[black] (-1.54866*0.8, 2.21171*0.8) circle (2*0.8pt);

        \end{tikzpicture}
        \end{center}
        \caption{the boundary graph \(\Gamma_R\)}
    \label{fig:opening int int v gr}        
    \end{subfigure}
    \caption{the two was of opening the vertex \(v\) in Figure~\ref{fig:int int v}}
    \label{fig:opening int int v}
\end{figure}

\begin{prop}
    \label{prop:one internal boundary is codim 1}
    \(\Gamma_R\) is not reduced and \(\Gamma_0\) is. In particular, \(\Omega_{\Gamma_R}\) is not a codimensions one boundary of \(\Omega_{\Gamma_+}\), while \(\Omega_{\Gamma_0}\) is.
\end{prop}
\begin{proof}
    It is straightforward to see that \(\Gamma_R\) is not reduced, since the straight paths going from \(b_2\) to \(b_3\) and from \(a_2\) to \(a_3\) cross twice. 
    
    We now argue that \(\Gamma_0\) is reduced.     
    By Proposition~\ref{prop:reduceability} a graph is reduced iff every arc does not cross itself, and no pair of arcs crosses more than once. Thus it is easy to see that if \(\Gamma_+\) is reduced the subgraphs \(\Gamma_{1},\ldots,\Gamma_{4}\) are also reduced.

    First as the straight paths in the middle do not cross themselves and go from one subgraph to another, it is clear that no arc in \(\Gamma_0\) crosses itself. Consider now a pair of arcs in \(\Gamma_0\), \(\tau_\ell\) and \(\tau_r\). If one is contained in one of the subgraphs \(\Gamma_1,\ldots,\Gamma_4\), then as they are all reduced, they cross each other at most once. If none is contained in a subgraph, then they correspond to the continuation of two of the straight paths in the middle.

    Without loss of generality, \(\tau_\ell\) is the continuation of the path from \(a_1\) to \(a_2\). Now, if  \(\tau_r\) is the continuation of the path from \(a_3\) to \(a_4\), they clearly do not cross at all. If however \(\tau_r\) is the continuation of the path from \(b_2\) to \(b_3\), they do cross at least once in the middle, and might cross again inside \(\Gamma_2\). We argue they cannot cross inside \(\Gamma_2\). 

    Indeed, if they cross in \(\Gamma_2\), this means the arcs \(\tau^\prime_{a_2}\) and \(\tau^\prime_{b_2}\) of \(\Gamma_2\) cross inside \(\Gamma_2\). However, this means in \(\Gamma_+\), before opening the vertex \(v\), the straight paths continuing the paths form \(a_2\) to \(a_4\) and form \(b_2\) to \(b_3\) cross each other twice as well. This is impossible as the original graph was reduced.

    A similar argument works for when \(\tau_r\) is the continuation of the path from \(b_1\) to \(b_4\). We thus conclude that no pair of arcs more then once in \(\Gamma_0\), hence it is reduced.

    The 'In particular' part follows from Proposition~\ref{prop:codim 1 boundaries iff reduced}
\end{proof}

Consider again \(\Gamma_+\) (Figure~\ref{fig:int int v}) and \(\Gamma_0\) (Figure~\ref{fig:opening int int v g0}), and notice what the opening of the vertex did to the tree of triangles. To triangles contracted into a square. It is easy to see that precisely two ToT graphs result in \(\Gamma_0\) after opening of a single vertex (see Figure~\ref{fig:two sides}).
\begin{figure}[H]
    \centering
    \begin{subfigure}{0.4\textwidth}
\begin{center}
            \begin{tikzpicture}[scale = 0.7, every node/.style={scale=0.7}]]

            \draw (0,0) circle (3*0.8cm);

            \filldraw[lightgray] (-1.25*0.8,1.25*0.8) circle (0.75*0.8cm);
            \draw (-1.25*0.8,1.25*0.8) circle (0.75*0.8cm);
            \node[scale=3*0.8] (c) at (-1.25*0.8,1.25*0.8)  {\(\Gamma_1\)};

            \filldraw[lightgray] (1.25*0.8,1.25*0.8) circle (0.75*0.8cm);
            \draw (1.25*0.8,1.25*0.8) circle (0.75*0.8cm);
            \node[scale=3*0.8] (c) at (1.25*0.8,1.25*0.8)  {\(\Gamma_2\)};

            \filldraw[lightgray] (1.25*0.8,-1.25*0.8) circle (0.75*0.8cm);
            \draw (1.25*0.8,-1.25*0.8) circle (0.75*0.8cm);
            \node[scale=3*0.8] (c) at (1.25*0.8,-1.25*0.8)  {\(\Gamma_3\)};

            \filldraw[lightgray] (-1.25*0.8,-1.25*0.8) circle (0.75*0.8cm);
            \draw (-1.25*0.8,-1.25*0.8) circle (0.75*0.8cm);
            \node[scale=3*0.8] (c) at (-1.25*0.8,-1.25*0.8)  {\(\Gamma_4\)};

       );

            \draw (0.875*0.8, 0.600481*0.8)node[anchor = north]{\(a_2\)} -- (-0.875*0.8, -0.600481*0.8)node[anchor =south]{\(a_4\)};
            \draw (0.875*0.8, -0.600481*0.8)node[anchor = south]{\(a_3\)} -- (-0.875*0.8, 0.600481*0.8)node[anchor =north]{\(a_1\)};

            \draw ( 0.600481*0.8,0.875*0.8)node[anchor = east]{\(b_2\)} -- (0.600481*0.8,-0.875*0.8 )node[anchor = east]{\(b_3\)};
            \draw (-0.600481*0.8,0.875*0.8)node[anchor = west]{\(b_1\)} -- ( -0.600481*0.8,-0.875*0.8)node[anchor = west]{\(b_4\)};

            \draw(1.89952*0.8, 0.875*0.8)--(2.94236*0.8, 0.585271*0.8);
            \draw(1.89952*0.8, -0.875*0.8)--(2.94236*0.8, -0.585271*0.8);
            \draw(-1.89952*0.8, 0.875*0.8)--(-2.94236*0.8, 0.585271*0.8);
            \draw(-1.89952*0.8, -0.875*0.8)--(-2.94236*0.8, -0.585271*0.8);
            \draw(0.875*0.8, 1.89952*0.8)--(0.585271*0.8, 2.94236*0.8);
            \draw(0.875*0.8, -1.89952*0.8)--(0.585271*0.8, -2.94236*0.8);
            \draw(-0.875*0.8, 1.89952*0.8)--(-0.585271*0.8, 2.94236*0.8);
            \draw(-0.875*0.8, -1.89952*0.8)--(-0.585271*0.8, -2.94236*0.8);

            \filldraw[black] (1.94454*0.8, 1.94454*0.8) circle (2*0.8pt);
            \filldraw[black] (2.21171*0.8, 1.54866*0.8) circle (2*0.8pt);
            \filldraw[black] (1.54866*0.8, 2.21171*0.8) circle (2*0.8pt);

            \filldraw[black] (1.94454*0.8, -1.94454*0.8) circle (2*0.8pt);
            \filldraw[black] (2.21171*0.8, -1.54866*0.8) circle (2*0.8pt);
            \filldraw[black] (1.54866*0.8, -2.21171*0.8) circle (2*0.8pt);

            \filldraw[black] (-1.94454*0.8, -1.94454*0.8) circle (2*0.8pt);
            \filldraw[black] (-2.21171*0.8, -1.54866*0.8) circle (2*0.8pt);
            \filldraw[black] (-1.54866*0.8, -2.21171*0.8) circle (2*0.8pt);

            \filldraw[black] (-1.94454*0.8, 1.94454*0.8) circle (2*0.8pt);
            \filldraw[black] (-2.21171*0.8, 1.54866*0.8) circle (2*0.8pt);
            \filldraw[black] (-1.54866*0.8, 2.21171*0.8) circle (2*0.8pt);

                    \end{tikzpicture}
        \end{center}
                    \caption{\(\Gamma_+\) again}
    \label{fig:two sides g+}
    \end{subfigure}
    \begin{subfigure}{0.4\textwidth}
          \begin{center}
            \begin{tikzpicture}[scale = 0.7, every node/.style={scale=0.7}]]

             \draw (0+6,0) circle (3*0.8cm);

            \filldraw[lightgray] (-1.25*0.8+6,1.25*0.8) circle (0.75*0.8cm);
            \draw (-1.25*0.8+6,1.25*0.8) circle (0.75*0.8cm);
            \node[scale=3*0.8] (c) at (-1.25*0.8+6,1.25*0.8)  {\(\Gamma_1\)};

            \filldraw[lightgray] (1.25*0.8+6,1.25*0.8) circle (0.75*0.8cm);
            \draw (1.25*0.8+6,1.25*0.8) circle (0.75*0.8cm);
            \node[scale=3*0.8] (c) at (1.25*0.8+6,1.25*0.8)  {\(\Gamma_2\)};

            \filldraw[lightgray] (1.25*0.8+6,-1.25*0.8) circle (0.75*0.8cm);
            \draw (1.25*0.8+6,-1.25*0.8) circle (0.75*0.8cm);
            \node[scale=3*0.8] (c) at (1.25*0.8+6,-1.25*0.8)  {\(\Gamma_3\)};

            \filldraw[lightgray] (-1.25*0.8+6,-1.25*0.8) circle (0.75*0.8cm);
            \draw (-1.25*0.8+6,-1.25*0.8) circle (0.75*0.8cm);
            \node[scale=3*0.8] (c) at (-1.25*0.8+6,-1.25*0.8)  {\(\Gamma_4\)};

            \draw (0.875*0.8+6, 0.600481*0.8)node[anchor = north]{\(a_2\)} --(-0.875*0.8+6, 0.600481*0.8)node[anchor =north]{\(a_1\)} ;
            \draw (0.875*0.8+6, -0.600481*0.8)node[anchor = south]{\(a_3\)} -- (-0.875*0.8+6, -0.600481*0.8)node[anchor =south]{\(a_4\)};

            \draw  ( -0.600481*0.8+6,-0.875*0.8)node[anchor = west]{\(b_4\)}-- ( 0.600481*0.8+6,0.875*0.8)node[anchor = east]{\(b_2\)};
            \draw (-0.600481*0.8+6,0.875*0.8)node[anchor = west]{\(b_1\)} --(0.600481*0.8+6,-0.875*0.8 )node[anchor = east]{\(b_3\)} ;

            \draw(1.89952*0.8+6, 0.875*0.8)--(2.94236*0.8+6, 0.585271*0.8);
            \draw(1.89952*0.8+6, -0.875*0.8)--(2.94236*0.8+6, -0.585271*0.8);
            \draw(-1.89952*0.8+6, 0.875*0.8)--(-2.94236*0.8+6, 0.585271*0.8);
            \draw(-1.89952*0.8+6, -0.875*0.8)--(-2.94236*0.8+6, -0.585271*0.8);
            \draw(0.875*0.8+6, 1.89952*0.8)--(0.585271*0.8+6, 2.94236*0.8);
            \draw(0.875*0.8+6, -1.89952*0.8)--(0.585271*0.8+6, -2.94236*0.8);
            \draw(-0.875*0.8+6, 1.89952*0.8)--(-0.585271*0.8+6, 2.94236*0.8);
            \draw(-0.875*0.8+6, -1.89952*0.8)--(-0.585271*0.8+6, -2.94236*0.8);

            \filldraw[black] (1.94454*0.8+6, 1.94454*0.8) circle (2*0.8pt);
            \filldraw[black] (2.21171*0.8+6, 1.54866*0.8) circle (2*0.8pt);
            \filldraw[black] (1.54866*0.8+6, 2.21171*0.8) circle (2*0.8pt);

            \filldraw[black] (1.94454*0.8+6, -1.94454*0.8) circle (2*0.8pt);
            \filldraw[black] (2.21171*0.8+6, -1.54866*0.8) circle (2*0.8pt);
            \filldraw[black] (1.54866*0.8+6, -2.21171*0.8) circle (2*0.8pt);

            \filldraw[black] (-1.94454*0.8+6, -1.94454*0.8) circle (2*0.8pt);
            \filldraw[black] (-2.21171*0.8+6, -1.54866*0.8) circle (2*0.8pt);
            \filldraw[black] (-1.54866*0.8+6, -2.21171*0.8) circle (2*0.8pt);

            \filldraw[black] (-1.94454*0.8+6, 1.94454*0.8) circle (2*0.8pt);
            \filldraw[black] (-2.21171*0.8+6, 1.54866*0.8) circle (2*0.8pt);
            \filldraw[black] (-1.54866*0.8+6, 2.21171*0.8) circle (2*0.8pt);

        \end{tikzpicture}
        \end{center}
            \caption{\(\Gamma_-\), the other side of \(\Gamma_0\)}
    \label{fig:two sides g-}
            \end{subfigure}
    \caption{two bordering BCFW cells}
    \label{fig:two sides}
\end{figure}
Clearly both are ToT graphs, and since their triangle leaves are the same we have that \(\Gamma_-\) has a triangle between the \(1\) and \(2k\) external vertices as well. Thus both are BCFW graphs by Theorem~\ref{thm:BCFW are trees}. As they are clearly different, we have by Proposition~\ref{prop:uniquness of trees} that they are not equivalent. Thus those two graphs correspond to two different BCFW cells, \(\Omega_{\Gamma_+}\) and \(\Omega_{\Gamma^\prime}\) sharing the common codimension \(1\) boundary cell \(\Omega_{\Gamma_0}\).

\begin{coro}
    \label{the two sides coro}
    Orthitroid cells of the form \(\Omega_{\Gamma_0}\) are codimension \(1\) boundaries of exactly two BCFW cells of the form of  \(\Omega_{\Gamma_+}\) and \(\Omega_{\Gamma_-}\), and it is the only codimension \(1\) boundary cell they share.
\end{coro}

\begin{coro}
\label{coro: codim 1 boudnary classification}
    Let \(\Gamma\) be a BCFW graph with oriented internal vertices \(v_i\). The codimension \(1\) boundary cells of  \(\Omega_\Gamma\) are of the form \(\Omega_{\lim_{v_i\rightarrow L}\Gamma}\), for \(L\in\{0,\frac{\pi}{2}\}\) if \(v_i\) is with trigonometric orientation, or \(L\in\{0,\infty\}\) if \(v_i\) is with hyperbolic orientation. For each internal vertex \(v_i\), one of those limits would result in a codimension \(1\) boundary, and the other in a boundary of a  codimension higher than one. In fact, we have a bijection between codimension \(1\) boundary cells \(\Omega_{\lim_{v_i\rightarrow L}\Gamma}\) and internal boundary vertices \(v_i\).

    The codimension \(1\) boundary cells \(\Omega_\Gamma\) are divided into two types which we term \emph{external} and \emph{internal} boundaries.
    \begin{itemize}
        \item \textbf{External Boundaries:} Those correspond to \(v_i\) that are adjacent to an external vertex. Those boundary cells are not boundaries of any other BCFW cell. They correspond to opening of a vertex as seen in Figure~\ref{fig:opening ext int v gl} or its mirror image. 
        
        \item \textbf{Internal Boundaries:} Those correspond to \(v_i\) that are not adjacent to an external vertex. They correspond to opening of a vertex as seen in Figure 6. Those boundary cells are boundaries of precisely two BCFW cells, corresponding to BCFW graphs \(\Gamma_+\), and \(\Gamma_-\) as seen in Figure~\ref{fig:two sides}.  
    \end{itemize}
\end{coro}
\begin{proof}
 Proposition~\ref{prop:codim 1 boundaries iff reduced} shows that codimension \(1\) boundary cells of  \(\Omega_\Gamma\) have the form \(\Omega_{\lim_{\alpha_i\rightarrow L}\Gamma}\). By Proposition~\ref{prop:one internal boundary is codim 1} and Observation~\ref{co-dime 1 external obs}, there is a bijection between internal vertices and codimension \(1\) boundary strata. The rest of the corollary follows from those claims together with Corollary~\ref{the two sides coro}.
 \end{proof}
We call them external and internal boundaries because the external ones map to boundaries of the amplituhedron and the internal ones map to the interior.
\begin{dfn}
\label{def:graph vert bound}
    For \(\Gamma\) a BCFW, and \(v\) an internal vertex, let \(\partial_v\Gamma\) denote the codimension \(1\) boundary graph corresponding the the vertex \(v\).
\end{dfn}
It is useful to note that the codimension \(1\) boundaries are well-behaved under the \(\mathrm {Arc}\) move:
\begin{obs}
\label{obs:arc of codim 1 is codim 1}
      If \(\widehat  \Gamma\) and \(\Gamma=\mathrm {Arc}_{4,r} \widehat \Gamma\) are reduced graphs, \(\widehat v\) is an internal vertex of \(\widehat \Gamma\), and \(v =\mathrm {Arc}_{4,r}(\widehat v)\) is the corresponding vertex in \(\Gamma\), then \(\partial_v \Gamma = \mathrm {Arc}_{4,r} (\partial_{\widehat v} \widehat \Gamma) \).
\end{obs}
\begin{proof}
    Write \(\widehat\Gamma_0 = \partial_{\widehat v}\widehat \Gamma\) and \(\Gamma_0 = \partial_v \Gamma\).
     By Definitions~\ref{def:graph vert bound} and~\ref{def:boundary graph}, we have that \(\widehat\Gamma_0= \lim_{\widehat v ^{\widehat\omega}\rightarrow L}\widehat\Gamma\), for \(\widehat v ^{\widehat\omega}\) an oriented vertex. Write \(\mathrm {Arc}_{4,r} (\widehat v ^{\widehat\omega}) = v^\omega \) for the corresponding vertex in \(\Gamma\). By Corollary~\ref{limit comute for graphs coro} we have that 
    \[
    \mathrm {Arc}_{4,r} \widehat \Gamma_0 =\mathrm {Arc}_{4,r}  \lim_{\widehat v^{\widehat \omega}\rightarrow L} \widehat \Gamma
    =  \lim_{v^{ \omega}\rightarrow L} \mathrm {Arc}_{4,r} \widehat \Gamma
    = \lim_{ v^{ \omega}\rightarrow L}  \Gamma.
    \]
      We now need to show that \(\lim_{ v^{ \omega}\rightarrow L}  \Gamma = \Gamma_0\). By Corollary~\ref{coro: codim 1 boudnary classification} we have that there is exactly one limit operation on the vertex \(v^\omega\) that would result in a codimension \(1\) boundary. Thus, if \(\lim_{ v^{ \omega}\rightarrow L}  \Gamma\) is a codimension \(1\) boundary of \(\Gamma\), it must be equal to \(\Gamma_0\).  Meaning it is enough to show \(\lim_{ v^{ \omega}\rightarrow L}  \Gamma\) is a codimension \(1\) boundary of \(\Gamma\).

     By Corollary~\ref{coro: codim 1 boudnary classification} we have that there is exactly one limit operation on the either vertices \( \widehat v^{\widehat \omega}\) and \(v^{ \omega}\) that would result in a codimension \(1\) boundary of \(\widehat \Gamma\) and \(\Gamma\) respectively. Let \(\widehat \Gamma_1\) be the other boundary of \(\widehat \Gamma\), that is, \(\widehat \Gamma_1 =\lim_{\widehat v^{\widehat \omega}\rightarrow \widehat R} \widehat \Gamma\) is a boundary graph of \(\widehat \Gamma\) of codimension \(p>1\), meaning that \(\mathrm{dim} (\Omega_{\widehat \Gamma}) - \mathrm{dim} (\Omega_{\widehat \Gamma_1}) = p\). Recall that the number of internal vertices of a reduced graph is the dimension of the corresponding cell. By Corollary~\ref{limit comute for graphs coro} we have that 
    \[
    \mathrm {Arc}_{4,r} \widehat \Gamma_1 =\mathrm {Arc}_{4,r}  \lim_{\widehat v^{\widehat \omega}\rightarrow \widehat R} \widehat \Gamma
    =  \lim_{v^{ \omega}\rightarrow \widehat R} \mathrm {Arc}_{4,r} \widehat \Gamma
    = \lim_{ v^{ \omega}\rightarrow \widehat R}  \Gamma.
    \]
     As reduction cannot increase the number of internal vertices, and that by Definition~\ref{def:arc}, we have that the the number of internal vertices in \(\mathrm {Arc}_{4,r} \widehat \Gamma_1\) is at most two more then of \( \widehat \Gamma_1\), we have that 
    \(
    \mathrm{dim} (\Omega_{\lim_{ v^{ \omega}\rightarrow \widehat R}  \Gamma})\leq\mathrm{dim} (\Omega_{\widehat \Gamma_1 })+2.
    \)
   As \(\Gamma\) is reduced, it has exactly two more internal vertices than \(\widehat \Gamma\). Meaning \(\mathrm{dim} (\Omega_{\Gamma })=\mathrm{dim} (\Omega_{\widehat \Gamma })+2\), and thus \(\mathrm{dim} (\Omega_{\Gamma }) - \mathrm{dim} (\Omega_{\lim_{ v^{ \omega}\rightarrow \widehat R} \Gamma})\geq p >1\). This means that \(\lim_{ v^{ \omega}\rightarrow \widehat R} \Gamma\) is not a boundary of codimension \(1\). Therefore, \(\lim_{ v^{ \omega}\rightarrow \widehat L} \Gamma\) must be of codimension \(1\), and \(\lim_{ v^{ \omega}\rightarrow \widehat L} \Gamma = \Gamma_0\).
\end{proof}
\subsection{Miscellanea}
In this section we will set up some additional simple notions about BCFW cells and their boundaries. These notions will be used primarily in Section~\ref{separation section}.
\begin{dfn}
\label{def:boundary triplet}
        Given \(\Gamma_+\) and \(\Gamma_-\) be BCFW graphs as in Figure~\ref{fig:two sides} with \(2k\) external vertices,  and \(\Gamma_0\) their common codimension \(1\) boundary as in Figure 6.  The triplet \(\Delta = (\Gamma_\epsilon)_{\epsilon\in\{+,0,-\}}\) will be called a \emph{boundary triplet} of \(\OGnon{k}{2k}\). We will call a an external arc that is a common to all three graphs, a \emph{external arc} of \(\Delta\).
        
        Define \(\Omega_\Delta \defeq \bigsqcup_{\epsilon\in\{+,0,-\}}\Omega_{\Gamma_\epsilon}\), and for \(G\) a series of moves, write \(G \Delta \defeq (G \Gamma)_{\epsilon\in\{+,0,-\}}\).

\end{dfn}
\begin{dfn}
\label{def:boundary vertex}
    Let \(\Delta = (\Gamma_+,\Gamma_0,\Gamma_-)\) be a boundary triplet in \(\OGnon{k}{2k}\) and \(\mathcal V(\Gamma_{\epsilon}),~\epsilon\in\{+,0,-\}\) be their respective internal vertices. The graph \(\Gamma_0\) was obtained from \(\Gamma_{+,-}\) by opening a vertex. Let that vertex be referred to as the \emph{boundary vertex} \(v_\pm\in\mathcal V(\Gamma_\pm)\) respectively.

\end{dfn}

\begin{obs}
\label{obs:promomting triplets}
    If \(\tau_\ell\) is a external common \(n\)-arc of a boundary triplet of \(k\)-OG graphs \(\Delta\), then \(\Delta = \mathrm{Arc}_{n,\ell}\Delta_0\), where \(\Delta_0\) is a boundary triplets of \((k-1)\)-OG graphs.
\end{obs}
\begin{obs}
\label{obs:boundary triplet back prom}
    Every boundary triplet of \(k\)-OG graphs for \(k>4\) is of the form of \(\Delta = \mathrm{Arc}_{4,2l}\Delta_0\), where \(\Delta_0\) is a boundary triplets of \((k-1)\)-OG graphs. If \(v_\pm\) are the boundary vertices of \(\Delta_0\), then the boundary vertices of \(\Delta \) are \(\mathrm{Arc}_{4,2l}(v_\pm)\) 
\end{obs}
\begin{proof}
     Let \((\Gamma_{\epsilon }^{k+1})_{\epsilon \in\{0,\pm\}}\) be a boundary triplet of \(\OGnon{k+1}{2k+2}\). By definition, they must be of the form of the graphs in Figures~\ref{fig:two sides g+},~\ref{fig:opening int int v g0}, and~\ref{fig:two sides g-} respectively.  As these graph have \(2k+2>8\) external vertices, it must hold that a least one of \(\Gamma_1\), \(\Gamma_2\), \(\Gamma_3\), or \(\Gamma_4\) contains an external arc of the graph \(\Gamma_0^{k+1}\), or else \(\Gamma_0^{k+1}\) is the spider graph (Figure~\ref{fig:spider}) but then \(k+1 = 4\). This external arc must also be an external arc of graphs \(\Gamma_{+,-}^{k+1}\) which are BCFW graphs. By Proposition~\ref{prop:BCFW arcs} it must be an external \(4\)-arc starting at an even index. Applying Observation~\ref{obs:promomting triplets}, we get that \(\Delta = \mathrm{Arc}_{4,2l}\Delta_0\), where \(\Delta_0\) is a boundary triplets of \((k-1)\)-OG graphs.
\end{proof}

\begin{figure}[H]
    \centering
   \begin{center}
\begin{tikzpicture}[scale = 0.7]
\draw (0,0) circle (2cm);

\draw(-1.84776, 0.765367)node[anchor=east]{\(6\)}--(1.84776, -0.765367)node[anchor=west]{\(2\)};

\draw({-0.765367, 1.84776})node[anchor=south]{\(5\)}--(0.765367, -1.84776)node[anchor=north]{\(1\)};

\draw(-0.765367, -1.84776)node[anchor=north]{\(8\)}--(1.84776, 0.765367)node[anchor=west]{\(3\)};

\draw(-1.84776, -0.765367)node[anchor=east]{\(7\)}--(0.765367, 1.84776)node[anchor=south]{\(4\)};

\draw (0+5,0) circle (2cm);

\draw({-0.765367+5, 1.84776})node[anchor=south]{\(5\)}--(1.84776+5, -0.765367)node[anchor=west]{\(2\)};

\draw(-1.84776+5, 0.765367)node[anchor=east]{\(6\)}--(0.765367+5, -1.84776)node[anchor=north]{\(1\)};

\draw(-0.765367+5, -1.84776)node[anchor=north]{\(8\)}--(0.765367+5, 1.84776)node[anchor=south]{\(4\)};

\draw(-1.84776+5, -0.765367)node[anchor=east]{\(7\)}--(1.84776+5, 0.765367)node[anchor=west]{\(3\)};

\end{tikzpicture}
\end{center}
    \caption{\(\Gamma_+\)  and \(\Gamma_-\), the two BCFW cells for \(k=4\)}
    \label{fig:BCFWk4}
\end{figure}
 \begin{figure}[H]
     \centering
   \begin{center}
\begin{tikzpicture}[scale = 0.7]
\draw (0,0) circle (2cm);

\draw({-0.765367, 1.84776})node[anchor=south]{\(5\)}--(1.84776, -0.765367)node[anchor=west]{\(2\)};

\draw(-1.84776, 0.765367)node[anchor=east]{\(6\)}--(0.765367, -1.84776)node[anchor=north]{\(1\)};

\draw(-0.765367, -1.84776)node[anchor=north]{\(8\)}--(1.84776, 0.765367)node[anchor=west]{\(3\)};

\draw(-1.84776, -0.765367)node[anchor=east]{\(7\)}--(0.765367, 1.84776)node[anchor=south]{\(4\)};

\end{tikzpicture}
\end{center}
     \caption{the only internal boundary for \(k=4\), also known as the 'spider' graph}
     \label{fig:spider}
 \end{figure}
 Notice that when we take the trigonometric orientation where the edges around each triangle are oriented clockwise for \(\Gamma_+\) and \(\Gamma_-\), the inherited orientation for \(\Gamma_0\) is the same (see Figure~\ref{fig:opening orient}). 

\begin{figure}[H]
    \centering
\begin{center}
\scalebox{0.7}{
            \begin{tikzpicture}

             \draw(-1.25, 0.75)--(1.25, -0.75)
              node[pos=0](a1){}node[pos=0.01](a2){} node[pos=0.44](b1){}node[pos=0.45](b2){} node[pos=0.55](c1){}node[pos=0.56](c2){} node[pos=0.99](d1){}node[pos=1](d2){};
            \draw [very thick, -latex, shorten <= 5] (a2)--(a1);
            \draw [very thick, -latex, shorten <= 5] (b1)--(b2);
            \draw [very thick, -latex, shorten <= 5] (c2)--(c1);
            \draw [very thick, -latex, shorten <= 5] (d1)--(d2);
               
            \draw (-1.25, -0.75)-- (1.25, 0.75)            
            node[pos=0.19](a1){}node[pos=0.2](a2){} node[pos=0.25](b1){}node[pos=0.26](b2){} node[pos=0.74](c1){}node[pos=0.75](c2){} node[pos=0.8](d1){}node[pos=0.81](d2){};
            \draw [very thick, -latex, shorten <= 5] (a1)--(a2);
            \draw [very thick, -latex, shorten <= 5] (b2)--(b1);
            \draw [very thick, -latex, shorten <= 5] (c1)--(c2);
            \draw [very thick, -latex, shorten <= 5] (d2)--(d1);
            \draw (0.75,1.25)--( 0.75,-1.25)
             node[pos=0](a1){}node[pos=0.01](a2){} node[pos=0.59](b1){}node[pos=0.6](b2){} node[pos=0.8](c1){}node[pos=0.81](c2){};
            \draw [very thick, -latex, shorten <= 5] (a2)--(a1);
            \draw [very thick, -latex, shorten <= 5] (b1)--(b2);
            \draw [very thick, -latex, shorten <= 5] (c2)--(c1);
            \draw ( -0.75,-1.25)--(-0.75,1.25)
             node[pos=0](a1){}node[pos=0.01](a2){} node[pos=0.59](b1){}node[pos=0.6](b2){} node[pos=0.8](c1){}node[pos=0.81](c2){};
            \draw [very thick, -latex, shorten <= 5] (a2)--(a1);
            \draw [very thick, -latex, shorten <= 5] (b1)--(b2);
            \draw [very thick, -latex, shorten <= 5] (c2)--(c1);

            \draw (0,0) node[anchor=north]{\(v\)};

            \draw (1.25+5, -0.75) -- (-1.25+5, -0.75)
             node[pos=0](a1){}node[pos=0.01](a2){} node[pos=0.59](b1){}node[pos=0.6](b2){} node[pos=0.8](c1){}node[pos=0.81](c2){};
            \draw [very thick, -latex, shorten <= 5] (a2)--(a1);
            \draw [very thick, -latex, shorten <= 5] (b1)--(b2);
            \draw [very thick, -latex, shorten <= 5] (c2)--(c1);
            \draw  (-1.25+5, 0.75)-- (1.25+5, 0.75)
             node[pos=0](a1){}node[pos=0.01](a2){} node[pos=0.59](b1){}node[pos=0.6](b2){} node[pos=0.8](c1){}node[pos=0.81](c2){};
            \draw [very thick, -latex, shorten <= 5] (a2)--(a1);
            \draw [very thick, -latex, shorten <= 5] (b1)--(b2);
            \draw [very thick, -latex, shorten <= 5] (c2)--(c1);

            \draw ( 0.75+5,-1.25) --(-0.75+5,1.25)
            node[pos=0.19](a1){}node[pos=0.2](a2){} node[pos=0.25](b1){}node[pos=0.26](b2){} node[pos=0.74](c1){}node[pos=0.75](c2){} node[pos=0.8](d1){}node[pos=0.81](d2){};
            \draw [very thick, -latex, shorten <= 5] (a1)--(a2);
            \draw [very thick, -latex, shorten <= 5] (b2)--(b1);
            \draw [very thick, -latex, shorten <= 5] (c1)--(c2);
            \draw [very thick, -latex, shorten <= 5] (d2)--(d1);
             
            \draw (0.75+5,1.25) --  ( -0.75+5,-1.25)
            node[pos=0](a1){}node[pos=0.01](a2){} node[pos=0.44](b1){}node[pos=0.45](b2){} node[pos=0.55](c1){}node[pos=0.56](c2){} node[pos=0.99](d1){}node[pos=1](d2){};
            \draw [very thick, -latex, shorten <= 5] (a2)--(a1);
            \draw [very thick, -latex, shorten <= 5] (b1)--(b2);
            \draw [very thick, -latex, shorten <= 5] (c2)--(c1);
            \draw [very thick, -latex, shorten <= 5] (d1)--(d2);

        \draw (5,0) node[anchor=east]{\(v\)};

            \draw (1.25+2.5, -0.5-5)--(-1.25+2.5, -0.5-5)
             node[pos=0](a1){}node[pos=0.01](a2){} node[pos=0.59](b1){}node[pos=0.6](b2){} node[pos=0.8](c1){}node[pos=0.81](c2){};
            \draw [very thick, -latex, shorten <= 5] (a2)--(a1);
            \draw [very thick, -latex, shorten <= 5] (b1)--(b2);
            \draw [very thick, -latex, shorten <= 5] (c2)--(c1);
            \draw  (-1.25+2.5, 0.5-5)--(1.25+2.5, 0.5-5)
             node[pos=0](a1){}node[pos=0.01](a2){} node[pos=0.59](b1){}node[pos=0.6](b2){} node[pos=0.8](c1){}node[pos=0.81](c2){};
            \draw [very thick, -latex, shorten <= 5] (a2)--(a1);
            \draw [very thick, -latex, shorten <= 5] (b1)--(b2);
            \draw [very thick, -latex, shorten <= 5] (c2)--(c1);

            \draw (0.5+2.5,1.25-5)--( 0.5+2.5,-1.25-5)
             node[pos=0](a1){}node[pos=0.01](a2){} node[pos=0.59](b1){}node[pos=0.6](b2){} node[pos=0.8](c1){}node[pos=0.81](c2){};
            \draw [very thick, -latex, shorten <= 5] (a2)--(a1);
            \draw [very thick, -latex, shorten <= 5] (b1)--(b2);
            \draw [very thick, -latex, shorten <= 5] (c2)--(c1);
            \draw ( -0.5+2.5,-1.25-5)--(-0.5+2.5,1.25-5) 
             node[pos=0](a1){}node[pos=0.01](a2){} node[pos=0.59](b1){}node[pos=0.6](b2){} node[pos=0.8](c1){}node[pos=0.81](c2){};
            \draw [very thick, -latex, shorten <= 5] (a2)--(a1);
            \draw [very thick, -latex, shorten <= 5] (b1)--(b2);
            \draw [very thick, -latex, shorten <= 5] (c2)--(c1);

             \draw [very thick, -latex, shorten <= 0] (2.5*0.3,-5*0.3)--(2.5*0.7,-5*0.7);
             \draw [very thick, -latex, shorten <= 0] (-2.5*0.3+5,-5*0.3)--(-2.5*0.7+5,-5*0.7);

        \draw (2.5,-2.5) node{\(v\rightarrow 0\)};
       
        \end{tikzpicture}
        }
        \end{center}
    \caption{the cells on the side on an internal boundary with orientations}
    \label{fig:opening orient}
\end{figure}
This allows us assign a canonical orientation to BCFW graphs and their boundaries.

\begin{dfn}
    Given BCFW graph, the \emph{canonical orientation} is defined orientation where the edges around each triangle, that is, a vertex in its triangle graph, are oriented clockwise. 
\end{dfn}
\begin{prop}
    The canonical orientation of a BCFW graph is a trigonometric orientation.
\end{prop}
\begin{proof}
    Let \(\Gamma\) be a BCFW graph with \(2k\) external vertices. By Proposition~\ref{thm:BCFW are trees} there exist a disk graph \(G\) with \(k\) external vertices such that \(M(G)=\Gamma\). By Proposition~\ref{prop:uniquness of trees}, we have a unique such \(G\) that is a three-regular tree with a vertex between the \(1\) and \(2k\)  half edges. The canonical orientation is precisely the orientation corresponding to \(G\) according to Proposition~\ref{prop:trigonometric from disk graph}, and thus it is trigonometric.
    \end{proof}

\begin{obs}
\label{obs:codim 1 boundary canonical 0}
    Let \(\Gamma^\omega\) be a BCFW graph with its canonical orientation, and let \(v\) be an internal vertex. Then \(\lim_{v^\omega\rightarrow 0} \Gamma\) is a codimension \(1\) boundary of \(\Gamma\), and the other boundary associated to that \(v\), \(\lim_{v^\omega\rightarrow \frac{\pi}{2}} \Gamma\) is not a codimension \(1\) boundary of \(\Gamma\) .
\end{obs}
\begin{proof}
    This is immediate from Corollary~\ref{coro: codim 1 boudnary classification}.
\end{proof}
\begin{dfn}
Let \(\Gamma\) be a be a BCFW graph and \(\Gamma_0\) the result of some limit operation on \(\Gamma\) such that \(\Omega_{\Gamma_0}\) is a codimension \(1\) boundary of \(\Omega_{\Gamma}\). Define the \emph{canonical orientation} on \(\Gamma_0\) to be the one inherited form the canonical orientation on \(\Gamma\).
\end{dfn}
This is well defined as an external boundary is a boundary of only one BCFW graph, and for an internal boundary the two inherited orientations agree as we have seen in Figure~\ref{fig:opening orient}.

\begin{dfn}
    Let $\Gamma$ be a BCFW graph or a codimension \(1\) boundary of BCFW graph, with internal vertices \(\mathcal V(\Gamma)\). The \emph{canonical parametrization} 
    \(
    \varphi_{\Gamma}:(0,\frac{\pi}{2})^{\mathcal V(\Gamma)}\rightarrow\Omega_{\Gamma}
    \)
    is the parameterization associated with the canonical orientation, where we index the angles by the corresponding internal vertex.
\end{dfn}
\begin{obs}
    \label{obs: canonical para codim one}
    Set \(\Lambda \in \mathrm{Mat}_{2k\times(k+2)}^>\) and let \(\Gamma\) be a \(k\)-BCFW graph. Then there exists an open neighborhood \(U\supset[0, \frac{\pi}{2})^{\mathcal V(\Gamma)}\) such that the canonical parameterization \(\varphi_\Gamma\) can be extended to a diffeomorphism \(\varphi_\Gamma:U\rightarrow \OGnon{k}{2k}\), such that for any \(v\in \mathcal V(\Gamma)\), we have that \( \evalat{\varphi_\Gamma}{v\mapsto0} = \varphi_{\partial _v \Gamma}\), and \(\widetilde \Lambda\) is well-defined on \(\varphi_\Gamma(U)\). 
\end{obs}
\begin{proof}
    By Propositions~\ref{prop:limits commute with para} and~\ref{prop:codim 1 boundaries iff reduced} we can extend \(\varphi_\Gamma\) to \(V_1 \supset[0,\frac{\pi}{2})^{\mathcal V(\Gamma)} \) an open neighborhood such that for any \(v\in \mathcal V(\Gamma)\), we have that \( \evalat{\varphi_\Gamma}{v\mapsto0} = \varphi_{\partial_v \Gamma}\). 
    By Proposition~\ref{prop:extends to a nbhd} , \(\widetilde \Lambda\) can be extended to \(V_2\supset\OGnon{k}{2k}\). Set \(U = \varphi_\Gamma^{-1}(\varphi_\Gamma(V_1)\cap V_2)\) and we have that \(\evalat{\varphi_\Gamma}{U}\) is a diffeomorphism and \(\widetilde \Lambda\) is well-defined on \(\varphi_\Gamma(U)\).
\end{proof}

\section{Injectivity}
\label{sec:inj}

In the following section, our goal is to prove that the BCFW cells are mapped injectively under the amplituhedron map. See Theorem~\ref{BCFW inj} for the precise and stronger statement. Our strategy is to construct an explicit solution to the inverse problem in terms of the natural coordinates on the target space, namely the twistor variables. Recall Definition~\ref{mat vec abs subst def} for the definition of \(\mathcal{F}_{[2k]}^{k \times 2k}\).

\begin{dfn}
    \label{def:twistor sol}
    Let \(\Gamma\) be an OG graph, and \(M \in \mathcal F_{[2k]}^{k\times 2k}\). We will call \(M\) a \emph{twistor-solution of \(\Gamma\) (or \(\Omega_\Gamma\))}, and write \(M \in \mathcal F(\Gamma)\) (or \(M = \mathcal F(\Omega_\Gamma)\)),
    if for every \([C,\Lambda,Y]\in \mathcal U_k^\geq\) with \( C \in \Omega_{\Gamma}\), we have that  \(C = M(\Lambda,Y)\) as elements of \(\Gr{k}{2k}\). If such \(M\) exists, \(\Gamma\) (or \(\Omega_\Gamma\)) is said to be \emph{twistor-solvable}. Just like the matrix representing \(C\), twistor-solutions are defined up to row operations. 
\end{dfn}

\begin{rmk}
    \label{rmk:uniqueness of solution}
    The twistor variables on an orthitroid cell often satisfy some identities, thus the twistor-solution is not unique. It is clear however that multiple twistor-solutions for the same graph are unique up to those identities as they must be equal on the cell corresponding to that graph. Additionally, we give an algorithm to generate twistor-solutions algebraically using the arc-sequence representation of a graph. Such a representation defines a unique twistor-solution.
\end{rmk}

A twistor-solution is, in effect, an inverse of the amplituhedron map on a given cell. In other words, if a cell is twistor-solvable, it is necessarily mapped injectively under the amplituhedron map. We begin by constructing twistor-solutions for some basic cases, and then develop tools that allow us to build twistor-solutions for more complex graphs by induction on the moves described in Section~\ref{sec:prom}. Central among these tools is Proposition~\ref{prop:sol induction prop}, which enables the inversion of the amplituhedron map on BCFW cells and their boundaries, using the characterization from Section~\ref{sec:BCFW Cells and their Boundaries}. The goal of this section is to establish the following result, whose proof is provided in Section~\ref{sec:solving by arcs}:

\begin{thm}
\label{BCFW inj}
    All BCFW graphs and their boundaries are twistor-solvable, and thus the BCFW cells and their boundary cells map injectively by the amplituhedron map.
\end{thm}
\begin{rmk}
\label{rmk:square roots}
Inverting the amplituhedron map for the Orthogonal Momentum amplituhedron requires the use of square roots, whereas in the Momentum Twistor (standard) amplituhedron, inverting the amplituhedron map on BCFW cells involves only rational functions in twistors. In inverting the amplituhedron map on BCFW cells, \cite{amptriag,even2023cluster} found the solution to the inverse problem was expressed as matrices with entries in polynomial rings of twistors. In contrast, the orthogonal momentum amplituhedron relies on more general functions, specifically those found in towers of quadratic extensions of the corresponding polynomial ring.

    We are, however, faced with the uncomfortable fact that the expressions for twistor-solutions often contain square roots. As is well-attested in literature, those are not defined on the whole real line. Alarmed readers can rest assured we will prove the expressions appearing in the square roots are positive on the images of the relevant orthitroid cells in Section~\ref{sec:positively rad}.
\end{rmk}

\begin{rmk}
\label{rmk inj strategy}
    Our strategy for finding twistor-solutions for graphs would be as follows: First find an external arc in \(\Gamma\). Suppose it has support \(I\) of length \(n\), This external arc corresponds to a unique vector in \(C\) by Proposition~\ref{prop:uniqueness of supp}. This means \(C\in \Omega_\Gamma\) has a unique vector \(\mathbf v\) with support \(I\). By its support we have \(n-1\) degrees of freedom for the vector \(\mathbf v\). We also have that \(\mathbf v\cdot\lambda = 0\), which further restricts to \(n-3\) degrees of freedom by Proposition~\ref{prop:twistor lambda}, which is sufficient if \(n \leq 3\). However, we have one additional closed constraint, that \(\mathbf v\,\eta\, \mathbf v^\intercal = 0\), which further restricts to \(n-4\) degrees of freedom. Meaning we do not have enough constraints to find our vector if \(n>4\), we can find exactly one for \(n<4\), but for \(n=4\) we have just enough using the orthogonality constraint.

As one of the constraint is a quadratic equation, it gives us two possible solutions for the case of \(n=4\). We claim that only one of these solutions will correspond to a positive \(C\). Furthermore, by directly solving the case of the top cell of \(\OGnon{3}{6}\) (see Lemma~\ref{triag}), we can identify the correct choice directly for the case of \(k=3\). We will then use promotion to find the correct choice of associated vector for any  external \(4\)-arc for any graph. 

Now, having found a single vector in \(C\in\Omega_\Gamma\), we will use it to reduce the question to a one of smaller \(k\) by removing the corresponding arc from the graph and finding its twistor-solution. This would allow us to use promotion to inductively generate a twistor-solution for any graph that can be built by recursively adding external arcs with support of length \(n\leq 4\) (see Proposition~\ref{prop:sol induction prop}). Recall we have seen in Section~\ref{sec:BCFW Cells and their Boundaries} that any graph that corresponds to a BCFW cell or their boundaries is indeed such a graph (recall Propositions~\ref{prop:BCFW arcs} and~\ref{prop:arc limits boundary sub BCFW}). This will allow us to invert the amplituhedron map on BCFW cells and their boundaries, proving injectivity.
\end{rmk}

\subsection{Basic Twistor-Solutions}
\label{sec:basic twist solution}

By Lemma~\ref{lem:twist in C}, given a point \(Y\) in the amplituhedron we can easily find two rows from its preimage in the positive Orthogonal Grassmannian. In particular, when $k=2$, we can easily invert the amplituhedron map:

\begin{prop}
    Let \(\Gamma\) be a \((k=2)\)-OG graph. If there exists a perfect orientation where \(i,j\in[4]\) are sources, then  
    \(
    \mathcal F (\Gamma) =\scalebox{0.8}{\( 
    \begin{pmatrix}
        \langle{i\,1}\rangle&-\langle{i\,2}\rangle&\langle{i\,3}\rangle&-\langle{i\,4}\rangle\\
        \langle{j\,1}\rangle&-\langle{j\,2}\rangle&\langle{j\,3}\rangle&-\langle{j\,4}\rangle\\
    \end{pmatrix}
    \)}\).
\end{prop}
\begin{proof}
\label{prop ext arc}
    Let \(M\) be the above matrix. For \([C,\Lambda,Y]\in \mathcal U_2^\geq\) we have that \(M(Y,\Lambda) = (\lambda\eta)_{\{i,j\}}\) by Proposition~\ref{prop:twistor lambda}. By Lemma~\ref{lem:twist in C}, \(\mathrm{dim} (\lambda\eta) = 2 \) and \(\lambda\eta\subset C\), therefore \(\lambda\eta = C\) for \(k=2\). It now is enough to show that the \(i,j\) rows are linearly independent. Since \(\Delta_{\{i,j\}}(M) = \pm \langle i \,j\rangle^2\), it is enough to show that \(\Delta_{\{i,j\}}(C) \neq 0\). Since we have an orientation in which  \(i,j\) are sources, the corresponding representation of \(C\) has \(C^{\{i,j\}} = \mathrm{Id}_{2\times 2}\), thus \(\Delta_{\{i,j\}}(C) \neq 0\).
\end{proof}

In the case where \(k = 3 \), the top cell of the positive Orthogonal Grassmannian is also a BCFW cell. 
Albeit more involved, we can still find the twistor-solution for the that cell. Using the fact that \(\mathrm{dim} (\lambda\eta) = 2 \) and \(\lambda\eta\subset C\), together with the additional orthogonality constraint \(C\eta C^\intercal =0\),  we can find exactly two preimages in the Orthogonal Grassmannian:
\begin{dfn}
\label{def:triag}
Let \(\Delta_\pm \in \mathcal{F}_{[6]}^{3\times 6}\) be defined as
\[
\Delta_\pm :=\scalebox{0.8}{\(
\begin{pmatrix}
\langle3\,5\rangle&\pm \langle4\,6\rangle&-\langle1\,5\rangle&\mp \langle2\,6\rangle&\langle1\,3\rangle&\pm \langle2\,4\rangle\\
\langle1\,2\rangle&0&-\langle2\,3\rangle&\langle2\,4\rangle&-\langle2\,5\rangle&\langle2\,6\rangle\\
0&-\langle1\,2\rangle&\langle1\,3\rangle&-\langle1\,4\rangle&\langle1\,5\rangle&-\langle1\,6\rangle\\
\end{pmatrix}
\)}\]
\end{dfn}
Using a direct calculation of the sign of their minors, we can show only one of them lies in the positive part of the Orthogonal Grassmannian. See Appendix~\ref{sec:triag} for the complete proof.
\begin{lem}
\label{triag}
Let \(\Gamma\), be the OG graph representing the top cell of  \(\OGnon{3}{6}\). We have that the twistor-solution \(\mathcal F(\Gamma) =\Delta_+ \).
\end{lem}

\subsection{Promoting Twistor-Solutions}
Having found some basic examples for twistor-solutions, we will now develop the tools needed to construct twistor-solutions inductively using the moves we defined in Sections~\ref{sec:prom}.
\begin{dfn}
\label{twistor sol arc def}
    For \(\tau_\ell\) an external arc of \(\Gamma\) an OG graph with support \(I\), we say that \(\mathbf{v} \in\mathcal F_{I}^{2k}\) is the \emph{twistor-solution} for \(\tau_\ell\) and write
    \(
    \mathbf v = \mathcal F(\tau_\ell,\Gamma),
    \)
    if for any \([C,\Lambda,Y]\in\mathcal U ^\geq_k\) with \(C\) in the interior of \(\Omega_\Gamma\), we have that the vector associated to \(\tau_\ell\) in \(C\) equals \(\mathbf v (\Lambda,Y)\). As the support of the associated vector is \(I\), the support of \(\mathbf v\) is \(I\). Notice that we also require the index-support of the twistor-solution of the  arc to be in \(I\). An arc is said to be \emph{twistor-solvable} if it has a twistor-solution. 
\end{dfn}
Note that just like the associated vector to an arc, the twistor-solution is defined up to scaling, and that Remark~\ref{rmk:uniqueness of solution} applies.

\begin{dfn}
\label{def:twist sol ang}
    Let \(\Gamma^\omega\) be a perfectly oriented \(k\)-OG graph. Label the internal vertices \(\{v_i^\omega\}_{i=1}^n\). Set \(v^\omega\in\{v_i^\omega\}_{i=1}^n\), and recall Definition~\ref{def:C as a func on vert}.

    Suppose that there exists \(a\in \mathcal F_{[2k]}\)  such that for any \([C,\Lambda,Y]\in\mathcal U^\geq_k\) and \(C\in \Omega_\Gamma\), we have that 
    \(
    C(v^\omega,\Gamma)   = a(\Lambda,Y).
    \)
    When such a solution exists we say that the oriented vertex \(v^\omega\) and the corresponding vertex \(v\) are \emph{twistor-solvable} and write \(a = \mathcal F(v^\omega,\Gamma)\) and call it a \emph{twistor-solution} for the oriented \(v\).
\end{dfn}
Note that Remark~\ref{rmk:uniqueness of solution} applies. 
We will now show that twistor-solutions behave nicely under promotion. This will allow us to inductively reduce the problem of finding twistor-solutions.
\begin{prop}
    \label{prop:sol commute for graphs}
    For  a twistor-solvable OG graph  \(\Gamma\), we have that
    
\begin{enumerate}
  \item \(\mathcal F(\mathrm{Cyc}(\Gamma))=\mathrm{Cyc}(\mathcal F(\Gamma)).\)
  \item \(\mathcal F( \mathrm{Inc}_i(\Gamma)) = \mathrm{Inc}_i(\mathcal F(\Gamma))\).
\end{enumerate}
    Furthermore, suppose \(\Gamma\) is reduced and \(\tau_i\) and \(\tau_{i+1}\) (considered mod \(2k\)) are two non-crossing arcs. Suppose there exist a twistor-solution \(a = \mathcal F(v^\omega,\mathrm{Rot}_{i,i+1}(v^\omega) (\Gamma))\). Since the above arcs are non-crossing in \(\Gamma\), we have that \(\mathrm{Rot}_{i,i+1} (\Gamma)\) is reduced by Corollary~\ref{reduced move coro}. We have that
\begin{enumerate}[resume]
  \item \(\mathcal F( \mathrm{Rot}_{i,i+1}(v^\omega)(\Gamma)) = \mathrm{Rot}_{i,i+1}(a)(\mathcal F(\Gamma)).\)
\end{enumerate}
\end{prop}
\begin{proof}
    For the first item, take \([C,\Lambda,Y]\in \mathcal U^{\geq}_k\) with \(C \in \Omega_{\mathrm{Cyc}(\Gamma)}\). Thus \([C,\Lambda,Y] = \mathrm{Cyc}[C_0,\Lambda_0,Y_0]\) with \([C_0,\Lambda_0,Y_0]\in  \mathcal U^{\geq}_k\) and \(C_0 \in \Omega_{\Gamma}\) by Proposition~\ref{prop:prom rules}. Thus \(C_0 = \mathcal F(\Gamma)(\Lambda_0,Y_0)\), and  \(C = \mathrm{Cyc}(\mathcal F(\Gamma)(\Lambda_0,Y_0)) = \mathcal F(\Gamma)(\Lambda,Y)\) by Proposition~\ref{prop:com diag}.
    
    For the other claims the proof is similar, except that for the last item we have to make sure that \(\mathrm{Rot}_{i,i+1}(a)(\mathcal F(\Gamma))\in F_{[2k]}^{k\times 2k}\), as \(a\) is an expression in twistors. Luckily, this is immediate as \(a\in \mathcal F_{[2k]}\) as a twistor-solution for an angle.
\end{proof}

\begin{prop}
    \label{prop:sol commute for arcs}
    Let \(\Gamma\) be an OG graph, and \(\tau_\ell\) be a twistor-solvable external arc with support \(I\). We have that
    \begin{enumerate}
        \item \(\mathcal F(\mathrm{Cyc}(\tau_\ell),\mathrm{Cyc}(\Gamma))=\mathrm{Cyc}(\mathcal F(\tau_\ell,\Gamma)).\)
    \end{enumerate}
    Let \(i\) be such that \(\mathrm{Inc}_i(\tau_\ell)\) is an external arc of \(\mathrm{Inc}_i(\Gamma)\), that is, \(i\notin I\). Then we have that
    \begin{enumerate}[resume]
        \item \(\mathcal F( \mathrm{Inc}_i(\tau_\ell),\mathrm{Inc}_i(\Gamma)) = \mathrm{Inc}_i(\mathcal F(\tau_\ell,\Gamma)).\)
    \end{enumerate}
    Furthermore, suppose \(\Gamma\) is reduced and \(\tau_i\) and \(\tau_{i+1}\) (considered mod \(2k\)) are two non-crossing arcs.  Let \(a = \mathcal F(v^\omega,\mathrm{Rot}_{i,i+1}(v^\omega) (\Gamma))\).  Since the above arcs are non-crossing in \(\Gamma\), we have that \(\mathrm{Rot}_{i,i+1} (\Gamma)\) is reduced by Corollary~\ref{reduced move coro}. Then we have that
    \begin{enumerate}[resume]
        \item \(\mathcal F( \mathrm{Rot}_{i,i+1}(\tau_\ell),\mathrm{Rot}_{i,i+1}(\Gamma)) = \mathrm{Rot}_{i,i+1}(a)(\mathcal F(\tau_\ell,\Gamma)).\)
    \end{enumerate}
\end{prop}
\begin{proof}
    Let the \(R\) be some move as above. Take \([C,\Lambda,Y]\in \mathcal U^{\geq}_k\) with \(C \in \Omega_{R(\Gamma)}\). Thus \([C,\Lambda,Y] = R[C_0,\Lambda_0,Y_0]\) with \(R [C_0,\Lambda_0,Y_0]\in  \mathcal U^{\geq}_k\) and \(C_0 \in \Omega_{\Gamma}\) by Proposition~\ref{prop:prom rules}. As \(\tau_\ell\) is a twistor-solvable external arc of \(\Gamma\), we have that \(\mathcal F(\tau,\Gamma)(\Lambda_0,Y_0)\) is its associated vector in \(C_0\), meaning its support is contained in \(I\).

    We have that \(C = R(C_0)\) thus \(R(\mathcal F(\tau,\Gamma)(\Lambda_0,Y_0))\in C\). Clearly \(R(\mathcal F(\tau,\Gamma)(\Lambda_0,Y_0))\) has support contained in \(R(I)\). By Corollary~\ref{arc supp prom} we have that \(R(\tau_\ell)\) is an external arc of \(R(\Gamma)\) with support \(R(I)\). Thus by Proposition~\ref{prop:uniqueness of supp}, \(R(\mathcal F(\tau,\Gamma)(\Lambda_0,Y_0))\) is its associated vector in \(C\). By Proposition~\ref{prop:com diag}, we have that \(R(\mathcal F(\tau,\Gamma)(\Lambda_0,Y_0)) = R(\mathcal F(\tau,\Gamma))(\Lambda,Y)\). 
    
    If the move was a rotation we required that the angle has \(a\in\mathcal F_{\mathrm{Rot}_{i,i+1}(I)}\) in Definition~\ref{twistor sol arc def}, so we have  that \(R(\mathcal F(\tau,\Gamma))\) has index-support in \(R(I)\).  Thus we have that \(R(\mathcal F(\tau,\Gamma))\) is a twistor-solution for \(R(\tau_\ell)\) in \(R(\Gamma)\).
\end{proof}
\begin{prop}
    \label{prop:sol commute for angles}
    Let \(\Gamma\) be a perfectly orientated OG graph, and \(v^\omega\) a twistor-solvable oriented vertex. We have that
    \begin{enumerate}
        \item \(\mathcal F(\mathrm{Cyc}(v^\omega),\mathrm{Cyc}(\Gamma))=\mathrm{Cyc}(\mathcal F(v^\omega,\Gamma)).\)
        \item For any \(i\in [2k]\), \(\mathcal F( \mathrm{Inc}_i(v^\omega),\mathrm{Inc}_i(\Gamma) )= \mathrm{Inc}_i(\mathcal F(v^\omega,\Gamma)).\)
    \end{enumerate}
    Furthermore, suppose that \(\mathrm{Rot}_{i,i+1} (u^\psi) (\Gamma)\) is reduced and we have that the oriented vertex \(u^\psi\) is twistor-solvable and \(b = \mathcal F(u^\psi, \mathrm{Rot}_{i,i+1} (u^\psi) (\Gamma))\) is its twistor-solution. Then
    \begin{enumerate}[resume]
        \item \( \mathcal F( \mathrm{Rot}_{i,i+1}(v^\omega),\mathrm{Rot}_{i,i+1}  (\Gamma)) = \mathrm{Rot}_{i,i+1} (b) (\mathcal F(v^\omega,\Gamma)).\)
    \end{enumerate}
\end{prop}
\begin{proof}
    Let the \(R\) be some move as above. Take \([C,\Lambda,Y]\in \mathcal U^{\geq}_k\) with \(C \in \Omega_{R(\Gamma)}\). Thus \([C,\Lambda,Y] = R[C_0,\Lambda_0,Y_0]\) with \(R [C_0,\Lambda_0,Y_0]\in  \mathcal U^{\geq}_k\) and \(C_0 \in \Omega_{\Gamma}\) by Proposition~\ref{prop:prom rules}. By Remark~\ref{rmk:oriented vert bij} we have that a choice of \(C_0 \in \Omega_\Gamma\) is equivalent to a choice of angle to each oriented vertex in \(\Gamma\).  As \(v^\omega\) is twistor-solvable, we have that \(\mathcal F(v^\omega,\Gamma)(\Lambda_0,Y_0) = C_0(v^\omega,\Gamma)\), the angle corresponding to \(v^\omega\) in \(C_0\).

    We have that \(C = R(C_0)\). \(R(\Gamma)\) is reduced by assumption or by Proposition~\ref{prop:reduceability}.  By Remark~\ref{rmk:oriented vert bij}, we have that a choice of \(C \in \Omega_{R(\Gamma)}\) is equivalent to a choice of angle to each oriented vertex in \(R(\Gamma)\). By Corollary~\ref{coro:moves dont change angles}, \(C(R(v^\omega),R(\Gamma))=C_0(v^\omega,\Gamma)\). Thus \(C(R(v^\omega),R(\Gamma)) = \mathcal F(v^\omega,\Gamma)(\Lambda_0,Y_0)\). By Proposition~\ref{prop:com diag}, we have that \(\mathcal F(v^\omega,\Gamma)(\Lambda_0,Y_0) =R (\mathcal F(v^\omega,\Gamma))(\Lambda,Y)\). We conclude that \(C(R(v^\omega),R(\Gamma)) = R (\mathcal F(v^\omega,\Gamma))(\Lambda,Y)\), finishing the proof.
\end{proof}
Let us consider special cases where the twistor-solutions are, in some sense, invariant under promotions.
\begin{prop}
    \label{prop:sol invariance for arcs}
    Let \(\Gamma\) be an OG graph and \(\tau_\ell\) be a twistor-solvable external arc of with support \(I\).
    For \(I<i,j\), we have that 
    \begin{enumerate}
        \item \(\mathcal F(\mathrm{Inc}_i(\tau_\ell),\mathrm{Inc}_i(\Gamma))^I =\mathcal F(\tau_\ell,\Gamma)^I.\)
        \item \( \mathcal F( \mathrm{Inc}_i(\tau_\ell), \mathrm{Inc}_i(\Gamma)) =\mathcal F(\mathrm{Inc}_j(\tau_\ell),\mathrm{Inc}_j(\Gamma))).\)
    \end{enumerate}
    Let \(R= \mathrm{Rot}_{i,i+1}\) with \(i,i+1\notin I\).
    We have that
    \begin{enumerate}[resume]
        \item \(\mathcal F( R(\tau_\ell),R(\Gamma)) = \mathcal F(\tau_\ell,\Gamma).\)
    \end{enumerate}
\end{prop}
\begin{proof}
    For the items \(1\) and \(2\), by the previous claim we have \(\mathcal F( \mathrm{Inc}_i(\tau_\ell), \mathrm{Inc}_i(\Gamma)) =\mathrm{Inc}_i(\mathcal F(\tau_\ell,\Gamma)) \). Notice that \(\mathrm{Inc_i}\) does not have any effect on abstract twistors with indices smaller then \(i\). Since the support and index-support of \(\mathcal F(\tau_\ell,\Gamma)\) is contained in the support of \(\tau_\ell\), \(I\), we have that \(\mathrm{Inc}_i\) just adds two padding zeros to the end of the vector \(\mathcal F(\tau_\ell,\Gamma)\), proving the claims.
    
   For the item \(3\), the proof is the same as for the Proposition~\ref{prop:sol commute for arcs}, except that now since we have that \(\mathcal F(\tau_\ell,\Gamma)\) has both support and index-support contained in \(I\) by Definition~\ref{twistor sol arc def}, we have that \(R(\mathcal F(\tau_\ell,\Gamma))=\mathcal F(\tau_\ell,\Gamma)\).
\end{proof}

\subsection{Solving by Arcs}
\label{sec:solving by arcs}

We can now reduce the problem of finding a twistor-solution for OG graphs, to that of finding the twistor-solution to external \(n\)-arcs. Recall that by Remark~\ref{rmk inj strategy} we are interested in external arcs with support length \(n\leq 4\). This will allow us to prove Theorem~\ref{BCFW inj}, which is the main goal of this section. Our main tool will be Proposition~\ref{prop:sol induction prop}, whose proof can be found in Section~\ref{sec:proof prop:sol induction prop}.

\begin{prop}
\label{prop:sol induction prop}
     Let \(\Gamma\) be a reduced twistor-solvable \(k\)-OG graph such that \(\{l=i_1, i_2,...,i_{n-2}\}\) does not contain any arcs, and let \(n\leq 4\), then 
     \(
     \mathcal F( \mathrm{Arc}_{n,\ell,k}(v^\omega)(\Gamma)) =  \mathrm{Arc}_{n,\ell,k}( \mathcal F(\Gamma)).
     \)
\end{prop}
      \begin{coro}
      \label{sub BCFW are twistor-solvable coro}
        Let \(\Gamma\) be a reduced OG graph with a reduced arc-sequence \(\Xi\) which is sub-BCFW. Then \(\Gamma\) is a twistor-solvable graph.
    \end{coro}
    \begin{proof}
        Immediate from Definition~\ref{def:sequences} with Proposition~\ref{prop:sol induction prop}.
    \end{proof}
This is, in fact, all we need to prove Theorem~\ref{BCFW inj}:
\begin{proof}[Proof of Theorem~\ref{BCFW inj}]
    Recall that by Definition~\ref{def:twistor sol} we call \(M\in \mathcal{F}^{k\times 2k}\) a \emph{twistor-solution of \(\Gamma\) (or \(\Omega_\Gamma\))},
    if for every \([C,\Lambda,Y]\in \mathcal U_k^\geq\) with \( C \in \Omega_{\Gamma}\), we have that  \(C = M(\Lambda,Y)\). This means that finding a twistor-solution to a \(k\)-OG graph \(\Gamma\) is equivalent to finding the unique preimage of a point \(Y \in \widetilde\Lambda(\Omega_\Gamma)\subset\mathcal O_k\) in the cell \(\Omega_\Gamma\subset\OGnon{k}{2k}\). Thus graphs being twistor-solvable means their Orthitroid cells map injectively by the amplituhedron map.  
    
    BCFW graphs are twistor-solvable by Proposition~\ref{prop:BCFW arcs} together with Proposition~\ref{prop:sol induction prop}. We conclude that BCFW cells map map injectively by the amplituhedron map.

    By Corollary~\ref{sub BCFW are twistor-solvable coro} we have that graphs with sub-BCFW sequences are twistor-solvable. By Proposition~\ref{prop:limits are boundaries} we have that boundaries cells of BCFW cells correspond to boundary graphs of BCFW graphs. By Proposition~\ref{prop:arc limits boundary sub BCFW} all such graphs have sub-BCFW arc-sequences, and thus are twistor-solvable. We conclude that boundary cells of BCFW cells map injectively by the amplituhedron map 
    \end{proof}
    See Example~\ref{exm:twistor sol} in Appendix~\ref{apx:calc exm} for an example for an explicit calculation of a twistor-solution for a \(4\)-OG graph.
\subsubsection{Proof of Proposition~\ref{prop:sol induction prop}}\label{sec:proof prop:sol induction prop}
\begin{dfn}
    If for all OG graphs \(\Gamma\) with external arc \(\tau_\ell\) with support of size \(n\), \(\tau_\ell\) is twistor solvable, we say that \emph{external \(n\)-arcs are twistor-solvable}.
\end{dfn}
\begin{prop}
\label{prop:sol vnlk prop}
    \label{n<=4 twist sol prop}
    For \(n=2,3,4\), external \(n\)-arcs are twistor-solvable. Consider a \(k\)-OG graph \(\Gamma\) with corresponding permutation \(\tau\). Let \(\tau_\ell\) be an external \(n\)-arc labeled such that \(\tau(\ell) =\ell+n \) considered mod \(2k\). We have that its twistor-solution is \(\mathcal F(\tau_\ell,\Gamma) = \mathbf v_{n,\ell,k}\)  as defined in Definition~\ref{def:vec angle exp}.
\end{prop}
Recall Remark~\ref{rmk:square roots}.
\begin{proof}{proof of Proposition~\ref{prop:sol vnlk prop}} We will prove this separately for each \(n\):

    \underline{\textbf{For \(n=2\):}}  \(\Gamma\) having an external arc with support \(\{\ell,\ell+1\}\) means we have \(\Gamma = \mathrm{Inc}_l(\Gamma_0)\). This means that \(C\in\Omega_{\Gamma}\) contains a vector as in the corollary, and by Proposition~\ref{prop:uniqueness of supp}, it must be the one associated to the arc. Thus, it is the twistor-solution by definition.
    
    \underline{\textbf{For \(n=3\):}}
    Let \([C,\Lambda,Y]\in\mathcal U^\geq_k\) with \(C\in\Omega_\Gamma\), such that \(\tau_\ell\) is an external arc of support \(I=\{l=i_1,i_2,..,i_n = \tau(\ell)\}\) (with \(i_j\) being consecutive mod \(2k\)). Let \(\mathbf u\in C\) be the associated vector to the arc \(\tau_\ell\), which has support contained in \(I\) by Proposition~\ref{prop:uniqueness of supp}. By Lemma~\ref{lem:twist in C} we have that for any \(j\in[2k]\), we have that
    \(
    \sum_{i=1}^{2k}\mathbf u^i\twist{Y}{j}{i} = 0.
    \)
    Since \(\mathbf u^i\) is zero for \(i\notin I\), we have 
    \(
    \sum_{i\in I}\mathbf u^\mu\twist{Y}{j}{i} = 0.
    \)
    Since \(\mathrm{dim} \lambda = 2\) from Proposition~\ref{prop:twistor lambda}, we can now solve the equations to get that
    \[
    \mathbf{u}^i = \begin{cases}
        \twist{Y}{i_2}{i_3}_{(2k)}&i=i_1\\
        -\varepsilon_{i}\twist{Y}{i_1}{i_3}_{(2k)}&i=i_2\\
        \varepsilon_{i}\twist{Y}{i_1}{i_2}_{(2k)}&i=i_3\\
        0 &\mathrm{otherwise},\\
    \end{cases}
    \]
    where \(\varepsilon_{i}\defeq 1\) if \(\ell\leq i\) and \(-1\) otherwise, and \(i_j = \ell+j-1\).

    \underline{\textbf{For \(n=4\):}} As this proof is more involved, we will prove this in separately in Section~\ref{sec:case n=4}.
\end{proof}

\begin{coro}
    \label{n<=4 twist solv coro}
    All external \(n\)-arcs with \( n\leq4\) are twistor-solvable.
\end{coro}

 Although it is immaterial to the case of BCFW cells and their boundaries, external arcs with support of length five and above are not generally twistor-solvable. As stated in Remark \ref{rmk inj strategy}, we are mainly interested in external \(n\)-arcs for \(2\leq n\leq 4\).
 
Having found twistor-solutions for arcs, we can now use this information to find twistor-solutions for oriented vertices.

\begin{prop}
    \label{prop:sol arc to angle}
    Given \(\tau_\ell\) a twistor-solvable external arc of \(\Gamma\) a reduced OG graph with support \(I=\{l_1,\ell_2,..,\ell_n\}\) with \(\ell=\ell_1\) (with \(\ell_j\) being consecutive mod \(2k+2\)), let \(\omega\) be a \(\tau_\ell\)-proper orientation for \(\Gamma\) (recall Definition~\ref{prop orientation def}). Oriented vertices on \(\tau_\ell\) be \(\{v^\omega_i\}_{i=1}^{n -2}\) enumerated along the arc in the direction from \(\ell\) to \(\tau(\ell)\). Then \(v^\omega_i\) are twistor-solvable with the index-support of the solution contained in \(I\), and the twistor-solution depends only on \(n,\ell,\) and \(i\) (specifically, it does not depend on \(\Gamma\) or \(k\)). The twistor solution is \(\mathcal F(v^\omega_i,\Gamma) =\mathcal \alpha_{n,\ell,k,i}\) as defined in Definition~\ref{def:vec angle exp}.
\end{prop}
\begin{proof}
    Let \([C,\Lambda,Y]\in\mathcal U^\geq_k\) with \(C\in\Omega_\Gamma\). Write \(C(v^\omega_i,\Gamma) = \alpha_i\). By the proper orientation, the vector associated to \(\tau_\ell\) is exactly 
    \[
    \mathbf u^{i} = 
    \begin{cases}
        1 &i=\ell\\
        \epsilon_i\cosh(\alpha_i)\prod_{j=1}^{i-1} \sinh(\alpha_j)&i\in I\setminus\tau_\ell\\
        \epsilon_i\prod_{j=1}^{i-1} \sinh(\alpha_j) & i = \tau(\ell)\\
        0 &\mathrm{otherwise}.
    \end{cases}
    \]
    where \(\epsilon_i = (-1)^{k+1}\) if \(i<\ell\) and \(1\) if \(\ell<i\). This is so since \(\tau_\ell\) is an external arc thus all of the paths from \(\ell\) go over either no sources or \(k-1\) of them, thus this is just the result of the parametrization we get from the orientation.

    By Proposition~\ref{prop:uniqueness of supp} we have that \(\mathbf u\) is the associated vector to \(\tau_\ell\), thus \(\mathbf u = \mathcal F(\tau_\ell,\Gamma)(\Lambda,Y) = \mathbf v_{n,\ell,k}(\Lambda,Y)\) by Proposition~\ref{prop:sol vnlk prop} (up to scaling). Thus we have that 
    \begin{align*}
    \alpha_1 &= \mathrm{arccosh}\left(\frac{\epsilon_{\ell_2}\mathbf v_{n,\ell,k}(\Lambda,Y)^{\ell_2}}{\mathbf v_{n,\ell,k}(\Lambda,Y)^{\ell_1}}\right)\\
    \alpha_2 &= \mathrm{arccosh}\left(\frac{\epsilon_{\ell_3}\mathbf v_{n,\ell,k}(\Lambda,Y)^{\ell_3}}{\mathbf v_{n,\ell,k}(\Lambda,Y)^{\ell_1}\sinh{(\alpha_1)}}\right) 
    \end{align*}

    etc. We can thus get the twistor-solution to all the angles by induction: 
    \[
   \alpha_{n,\ell,k,i}= \mathrm{arccosh}\left(\frac{\epsilon_{\ell_{i+1}}\mathbf v_{n,\ell,k}^{\ell_{i+1}}}{\mathbf v_{n,\ell,k}^{\ell_1}\prod_{j=1}^{i-1}\mathcal \sinh(\alpha_{n,\ell,k,j})}\right)
    .\]

    Since the index-support of \(\mathcal F(\tau_\ell,\Gamma)\) is in \(I\) by definition, we have that the index-support of \(\mathcal F(v^\omega_i,\Gamma)\) is also contained in \(I\). We built the twistor-solutions from \(\mathbf v_{n,\ell,k}\), thus they do not depend on \(\Gamma\). 

    Now, notice that by definition of \(\mathbf v_{n,\ell,k}\) and by definition of the \(\mathrm{Cyc}\) move, for \(r\in[n]\),

    \[
     \varepsilon_{\ell_r}\mathbf v_{n,\ell,k}^{\ell_r}=\varepsilon_{\ell_r}(\mathrm{Cyc}^{\ell-1}(\mathbf v_{n,1,k}))^{\ell_r}=\mathrm{Cyc}^{\ell-1} (\mathbf v_{n,1,n}^{r}).
    \]
    Thus
     \[
     \alpha_{n,\ell,k,i}= \mathrm{arccosh}\left(\frac{\mathrm{Cyc}^{\ell-1} (\mathbf v_{n,1,n}^{i+1})}{\mathbf \mathrm{Cyc}^{\ell-1} (\mathbf v_{n,1,n}^{1})\prod_{j=1}^{i-1}\mathcal \sinh(\alpha_{n,\ell,k,j})}\right)
    ,\]
    which does not depend on \(k\), and matches Definition~\ref{def:vec angle exp}.
\end{proof}
\begin{rmk}
\label{rmk:no hyp func}
    Notice that no actual hyperbolic functions would ultimately appear in our twistor-solutions. As we only use angles for arguments \(\sinh\) and \(\cosh\) in the \(\mathrm{Rot}\) move, the only expressions we will have appearing are:

     \[
     \cosh(\alpha_{n,\ell,k,i})= \left(\frac{\mathrm{Cyc}^{l-1} (\mathbf v_{n,1,n}^{i+1})}{\mathbf \mathrm{Cyc}^{l-1} (\mathbf v_{n,1,n}^{1})\prod_{j=1}^{i-1}\mathcal \sinh(\alpha_{n,\ell,k,j})}\right)
    ,\]
    and
    \(
    \sinh(\alpha_{n,\ell,k,i})= \sqrt{ \cosh(\alpha_{n,\ell,k,i})^2-1}
    \)
    similarly defined recursively, which will be algebraic.
\end{rmk}

    We would now like to consider OG graphs built recursively by adding on twistor-solvable arcs. Recall the \(\mathrm{Arc}\) move defined in Definition~\ref{def:arc}.
\begin{obs}
\label{obs:reduce arc move}
    For \(\Gamma^\prime\) a reduced \(k\)-OG graph, \(l\in[2k]\) and \(n>1\). Suppose \(I^\prime =\{l=i_1, i_2,...,i_{n-2}\}\) (\(i_j\) consecutive mod \(2k\)) does not contain any arcs of \(\Gamma^\prime\). We have that \(\Gamma =\mathrm{Arc}_{n,\ell,k}(\Gamma^\prime)\) is a reduced \((k+1)\)-OG graph with \(\tau_\ell\) an external arc with support \(I = \{i_1, i_2,...,i_{n}\}\) with \(\ell=i_1\) (\(i_j\) again consecutive mod \(2k\)).
\end{obs}
\begin{proof}
    The \(\mathrm{Arc}\) move is made of one \(\mathrm{Inc}\) move and a few \(\mathrm{Rot}\) moves, thus it increases \(k\) by 1.

    We prove the rest by induction on \(n\). 
    
    If \(n=2\) we have that \(\mathrm{Arc}_{2,\ell} = \mathrm{Inc_{\ell}}\) for which the claim is trivial by Corollaries~\ref{arc supp prom} and~\ref{reduced move coro}.

    For the induction step, write
    \[
    \Gamma = \mathrm{Arc}_{n,\ell,k}(v^\omega)(\Gamma^\prime)\defeq \mathrm{Rot}_{i_{n-1},i_n}(v^\omega_{n-2})\mathrm{Arc}_{n-1,\ell,k}(v^\omega)(\Gamma^\prime) =\mathrm{Rot}_{i_{n-1},i_n}(\hat \Gamma).
    \]
    We have by the induction hypothesis that \(\hat \tau_\ell\) is an external arc of \(\hat \Gamma \) a reduced graph with support \(\hat I = \{\ell=i_1, i_2,...,i_{n-1}\}\). By corollaries~\ref{reduced move coro} and~\ref{arc supp prom}, we have that \(\Gamma\) is a reduced graph with external arc \(\mathrm{Rot}_{i_{n-1},i_n} (\hat \tau_\ell) = \tau_\ell\) with support \(\mathrm{Rot}_{i_{n-1},i_n} (\hat I )= I\).
\end{proof}
We are now ready to prove Proposition~\ref{prop:sol induction prop}:
\begin{proof}[Proof of Proposition~\ref{prop:sol induction prop}]

     By Observation~\ref{obs:reduce arc move} we get that \(\tau_\ell\) is an external arc with support \(I = \{\ell=i_1, i_2,...,i_{n}\}\) (\(i_j\) consecutive mod \(2k\)), with \(n\leq 4\). Thus, \(\tau_\ell\) is a twistor-solvable arc.
     
     Set a \(\tau_\ell\)-proper orientation for \(\Gamma\). The angles associated to the internal vertices on \(\tau_\ell\) are \(\{\alpha_i\}_{i=1}^{n -2}\) counting from \(\ell\) to \(\tau(\ell)\).  By Proposition~\ref{prop:sol arc to angle}, we have that the \(\mathcal F (v^\omega_i,\Gamma) = \alpha_{n,\ell,k,i}\). By repeated application of Proposition~\ref{prop:sol commute for graphs}, we have that 
     \begin{align*}
         \mathcal F( \mathrm{Arc}_{n,\ell,k}&(v^\omega)(\Gamma)) = \\
          &= \mathcal F (\mathrm{Rot}_{i_{n-1},i_n}(v^\omega_{n -2})\cdot...\cdot\mathrm{Rot}_{i_3,i_4}(v^\omega_{2})\mathrm{Rot}_{i_2,i_3}(v^\omega_{1})\mathrm{Inc}_{\ell}(\Gamma))\\
         &= \mathrm{Rot}_{i_{n-1},i_n}(\alpha_{n,\ell,k,n -2})\cdot...\cdot\mathrm{Rot}_{i_3,i_4}(\alpha_{n,\ell,k,2})\mathrm{Rot}_{i_2,i_3}(\alpha_{n,\ell,k,1})\mathrm{Inc}_{l}(\mathcal F (\Gamma))\\
         &= \mathrm{Arc}_{n,\ell,k}( \mathcal F(\Gamma)).
     \end{align*}
\end{proof}
\subsubsection{Solving External \texorpdfstring{\(4\)}{4}-arcs }
\label{sec:case n=4}

We will now prove Proposition~\ref{prop:sol vnlk prop} for the case of \(n=4\). First, let us reduce the problem to the case of \(\OGnon{3}{6}\). Consider the following graph:
\begin{figure}[H]
    \centering
                \begin{center}
\begin{tikzpicture}[scale = 0.5]
\draw (0,0) circle (2cm);

\draw(1., -1.73205)node[anchor=north]{\(1\)}--(-1., 1.73205)node[anchor=south]{\(4\)};

\draw(2., 0.)node[anchor=west]{\(2\)}--(-1., -1.73205)node[anchor=north]{\(6\)};

\draw(1., 1.73205)node[anchor=south]{\(3\)}--(-2., 0.)node[anchor=east]{\(5\)};

\end{tikzpicture}
\end{center}
    \caption{the simplest OG graph with an external \(4\)-arc}
    \label{fig:simlpest 4arc}
\end{figure}
\begin{prop}
\label{prop:arc sol depends on supp}
    For a reduced \(k\)-OG graph \(\Gamma\) with corresponding permutation \(\tau\), and an external \(4\)-arc \(\tau_\ell \) of \(\Gamma\) with support \(I = \{\ell,\ell+1,\ell+2,\ell+3\}\) considered mod \(2k\), the twistor solution of \(\tau_\ell\) is \(
    \mathcal{F}(\tau_\ell,\Gamma) = \mathrm{Cyc}^{\ell-1}_{k}\mathrm{Inc}_5^{k-3}\mathcal F(\{1,4\},\Gamma_4),
    \) where \(\Gamma_4\) is as seen in Figure~\ref{fig:simlpest 4arc}.
    
\end{prop}
\begin{proof}
    
    By Proposition~\ref{prop:sol commute for arcs} we can use the \(\mathrm{Cyc}\) move to reduce the problem to the case of \(I=\{1,2,3,4\}\) (which implies \(l=1\) and \(\tau_\ell  = \{1,4\}\)).

    We can reduce the number of vertices in the graph \(\Gamma\) using the \(\mathrm{Rot}^{-1}\) move according to Proposition~\ref{prop:sol invariance for arcs} until all the internal vertices are on the arc  \(\{1,4\}\), without changing the twistor-solution:
    If there are \(i,i+1\notin I\) such that the arcs \(\tau_i\) and \(\tau_{i+1}\) are crossing, we can use the equivalence move number 3 to replace the graph with an equivalent graph such that the crossing vertex is adjacent to an external vertex. We can now eliminate it using a \(\mathrm{Rot}^{-1}_{i,i+1}\) move without changing the twistor-solution according to Proposition~\ref{prop:sol invariance for arcs}. We will continue doing so until there are no more such arcs, meaning the only pairs of crossing arcs are ones where one of the arcs is contained in \(I\). The only arc contained in \(I\) is \(\{1,4\}\) as it is external, thus we can assume all the internal vertices are on the arc \(\{1,4\}\) in \(\Gamma\).

    Since the arc \(\{1,4\}\) is external, there are no arcs contained \(\{1,2,3,4\}\). Meaning that for any \(\{r,\tau(r)\}\) an arc that does not cross \(\{1,4\}\), we have that \(4<r,\tau(r)\). Since all the internal vertices are on the arc \(\{1,4\}\), we must have that \(\{r,\tau(r)\}\) is external with support of size \(2\). This means we have an arc \(\{r,r+1\}\) with \(4<r\), and \(\Gamma = \mathrm{Inc}_r (\Gamma_0)\) for some graph \(\Gamma_0\). We can now apply \(\mathrm{Inc}_r^{-1}\) according to Propositions~\ref{prop:sol commute for arcs}  and~\ref{prop:sol invariance for arcs} to reduce \(k\) by \(1\). As \(\{1,4\}\) is external, there are exactly \(2\) arcs that cross it. That means that there are exactly \(k - 3\) arcs that do not cross \(\{1,4\}\). By induction, we can now assume that all arcs cross \(\{1,4\}\) in \(\Gamma\) and \(k=3\).

    So now we can assume \(\ell=1\) and  \(k=3\). That is, that \(\{1,4\}\) is an external arc in \(\Gamma\), all the internal vertices are on that arc, and there are no arcs that do not cross it. This means we must have exactly \(2\) other arcs going straight across  \(\{1,4\}\), This means the graphs is precisely \(\Gamma_4\). In this case, the claim is trivial, finishing the proof.
\end{proof}

Let us now deal with the case of \(\OGnon{3}{6}\). We will do so by considering a different point of view on the problem of finding a twistor-solution for the top cell of \(\OGnon{3}{6}\), discussed in Section~\ref{sec:basic twist solution} and Lemma~\ref{triag}.
Consider the top cell of \(\OGnon{3}{6}\) with the following hyperbolic orientation and choice of angles:

\begin{center}
\scalebox{0.8}{
\begin{tikzpicture}[scale = 1]
\draw (0,0) circle (2cm);

\draw({-0.517638, 1.93185})node[anchor=south]{\(4\)}--(1.41421, -1.41421)node[anchor=north]{\(1\)};

\draw [very thick, -latex, shorten <=55] (1.41421, -1.41421)--(0.351694, 0.426123);

\draw(-1.93185, -0.517638)node[anchor=east]{\(5\)}--(1.93185, -0.517638)node[anchor=west]{\(2\)};

\draw [very thick, -latex, shorten <=55] (-1.93185, -0.517638)--(0.193185, -0.517638);

\draw(-1.41421, -1.41421)node[anchor=north]{\(6\)}--(0.517638, 1.93185)node[anchor=south]{\(3\)};

\draw [very thick, -latex, shorten <=55] (-1.41421, -1.41421)--(-0.351695, 0.426122);

\node[anchor=north] (c) at (0.896575-0.1, -0.517638)  {\(\alpha\)};
\node[anchor=east] (c) at (0, 1.03528-0.05)  {\(\beta\)};
\node[anchor=north] (c) at (-0.896575+0.1, -0.517638)  {\(\gamma\)};

\end{tikzpicture}
}
\end{center}

This defines this parameterization (the last two rows do not concern us at this moment):
\[
C(\alpha,\beta,\gamma) = 
\begin{pmatrix}
\mathrm{sinh}(\alpha)\mathrm{sinh}(\beta) &\mathrm{cosh}(\alpha)\mathrm{cosh}(\beta)&\mathrm{cosh}(\beta)&1&0&0\\
C_2^1&C_2^2&C_2^3&0&1&0\\
C_3^1&C_3^2&C_3^3&0&0&1\\
\end{pmatrix}
\]

The associated vector to \(\tau_1\) is
\(
\mathbf{v} = (\mathrm{sinh}(\alpha)\mathrm{sinh}(\beta) ,\,\mathrm{cosh}(\alpha)\mathrm{sinh}(\beta),\,\mathrm{cosh}(\beta),\,1,\,0,\,0).\)

By Lemma~\ref{lem:twist in C} we have that   \(\mathbf v \lambda ^\intercal=0\). Thus the four entries of \(\mathbf{v}\) are defined by the following equations:

\begin{align}
\label{defing eqns for v}
\begin{split}
    &\mathbf v_{1} \twist{Y}{2}{1}+\mathbf v_{3}\twist{Y}{2}{3}+\mathbf  v_{4}\twist{Y}{2}{4} = 0\\
    &\mathbf v_{1} \twist{Y}{3}{1}+\mathbf  v_{2}\twist{Y}{3}{2}+\mathbf  v_{4}\twist{Y}{3}{4} = 0\\
    &\mathbf v_{1}^2-\mathbf  v_{2}^2+\mathbf v_3 ^2 - \mathbf v_{4}^2 = 0\\
\end{split}
\end{align}
This set of equations has two solutions:
 \begin{align*}
     &\mathbf{v}_{\pm}^i =\\
      & \begin{cases}
        S_{\{2,3,4\}}(\Lambda,Y)&i=1\\
       \twist{Y}{ 1}{4}\twist{Y}{ 2}{4} - \twist{Y}{ 1}{3}\twist{Y}{ 2}{3}\pm\twist{Y}{ 3}{4}\sqrt{S_{\{1,2,3,4\}}(\Lambda,Y)}&i=2\\
       \twist{Y}{ 1}{2}\twist{Y}{ 2}{3} - \twist{Y}{ 1}{4}\twist{Y}{ 3}{4}\mp\twist{Y}{ 2}{4}\sqrt{S_{\{1,2,3,4\}}(\Lambda,Y)}&i=3\\
       \twist{Y}{ 1}{3}\twist{Y}{ 3}{4} - \twist{Y}{ 1}{2}\twist{Y}{ 2}{4}\pm\twist{Y}{ 2}{3}\sqrt{S_{\{1,2,3,4\}}(\Lambda,Y)}&i=4\\
        0 &\mathrm{otherwise},\\
    \end{cases}         
     \end{align*}
Notice that by Propositions~\ref{prop:S orth} and~\ref{prop:Consecutive Twistors}, we have \(\sqrt{S_{\{1,2,3,4\}}(\Lambda,Y)} = \twist{Y}{5}{6} \). It is easy to see \(\mathbf{v}_+\) corresponds to \(\Delta_+\). Since we know \(C\) has a non-zero vector satisfying equation (~\ref{defing eqns for v}) that has support \(\{1,\,2,\,3,\,4\}\), by Proposition~\ref{prop:uniqueness of supp}, \(\mathbf v_+\) must be non-zero. Thus, it is the correct solution.

We are now ready to prove the final case of Proposition~\ref{prop:sol vnlk prop}:
\begin{proof}[Proof of the \(n=4\) case of Proposition~\ref{prop:sol vnlk prop}]
    Let \(\Gamma_0 = \mathrm{Arc}_{4,2}\mathrm{Arc}_{3,2}\mathrm{Arc}_{2,1}(O)\), that is the top cell of  \(OG_{\geq 0 } \left(3,\, 6\right)\). By Lemma~\ref{triag} we have that \(\mathcal F(\Gamma_0) =\Delta_+\). By Proposition~\ref{prop:uniqueness of supp} we have that \(\Delta_+\) contains a unique vector with support \(I=\{1,2,3,4\}\). By the above discussion, it corresponds to the vector \(\mathbf v_+\) defined above. Notice that \(\mathbf v_+ = \mathbf{v}_{4,1,3}(\Lambda,Y)\) from Definition~\ref{def:vec angle exp} for \([C,\Lambda,Y]\in\mathcal U^\geq_3\) and \(C\in\Omega_{\Gamma_0}\). Thus, we can deduce that \(\mathcal{F}(\{1,4\},\Gamma_0)=\mathbf v_{4,1,3}\).
    By Proposition~\ref{prop:arc sol depends on supp} we can now deduce that 
        \(
    \mathbf v_{4,1,3} = \mathcal F(\{1,4\},\Gamma_4),
    \)
    with \(\Gamma_4\) as seen in Figure~\ref{fig:simlpest 4arc}. It is easy to see that for \(\mathbf v_{4,\ell,k} \) as defined in Definition~\ref{def:vec angle exp} we have that
    \[
    \mathbf v_{4,\ell,k} =\mathrm{Cyc}^{\ell-1}_{k}\mathrm{Inc}_5^{k-3}\mathbf v_{4,1,3} = \mathrm{Cyc}^{\ell-1}_{k}\mathrm{Inc}_5^{k-3}\mathcal F(\{1,4\},\Gamma_4).
    \]
    Thus by Proposition~\ref{prop:arc sol depends on supp}, we have that for all \(k\)-OG graph \(\Gamma\) with corresponding permutation \(\tau\), and an external \(4\)-arc \(\tau_\ell\) of \(\Gamma\), \(\tau_\ell\) is twistor-solvable and \(\mathcal F(\tau_\ell,\Gamma) = \mathbf v_{4,\ell,k}\). Finishing the proof. 
\end{proof}

   \section{Local Separation}
   \label{separation section}
   The goal of this section is to prove the following theorem:
   \begin{thm}
    \label{thm:local sep}
        Let \(\Lambda\in \mathrm{Mat}^>_{2k\times(k+2)}\) and let \(\Delta = (\Gamma_+,\Gamma_0,\Gamma_-)\) be a boundary triplet of \(k\)-OG graphs. Then there is an open neighborhood \(U\supset \Omega_{\Gamma_0}\) in \(\OG{k}{2k}\), such that \(U\cap \Omega_\Delta\) is mapped injectively by the amplituhedron map \(\widetilde \Lambda\).
    \end{thm}
This theorem will be pivotal for the proof of the BCFW tiling of the amplituhedron in Section~\ref{sec:tiling}. In order to prove Theorem~\ref{thm:local sep} in Section~\ref{sec:proving local separation} we will use a topological argument based on comparing the derivatives of different functions that define the boundary between cells (See Proposition~\ref{prop:sep 4-native} for details of those functions). In order to apply the argument, we will first need to show the two following claims:
\begin{thm}
\label{thm:ampli diffeo}
    Let \(\Lambda\in \mathrm{Mat}_{2k\times(k+2)}^>\), \(\Gamma\) and let be a \(k\)-BCFW graph or a boundary of a BCFW graph. Then \(\evalat{\widetilde \Lambda}{\Omega_\Gamma}\) is a diffeomorphism onto its image.
\end{thm}
\begin{thm}
\label{thm:BCFW and codim 1 bounds are smooth}
    Let \(\Lambda\in \mathrm{Mat}_{2k\times(k+2)}^>\), and let \(\Gamma\) be a \(k\)-BCFW graph or a codimension \(1\) boundary of BCFW graph. Then \(\widetilde \Lambda (\Omega_\Gamma)\subset Y_\Lambda^k\) is a smooth submanifold.
\end{thm}
Recall Definition~\ref{def:Ylk}, the definition of \(Y_\Lambda^k\).
\subsection{Positively-Radical Vertices}
\label{sec:positively rad}
The goal of this section is to prove Theorem~\ref{thm:ampli diffeo} and to lay the groundwork for the proof of Theorem~\ref{thm:BCFW and codim 1 bounds are smooth}. To that end, we take a closer look at the specific functions defining the twistor-solutions introduced in Section~\ref{sec:inj}. These functions are algebraic expressions in abstract twistors, built from rational functions and square roots. Such expressions are well-behaved when the square roots are taken over positive quantities. Our aim is therefore twofold: first, to understand which expressions may appear under radicals in the twistor-solutions, and second, to demonstrate that these expressions are indeed positive on the image of the orthitroid cell to which they correspond. This analysis also justifies Remark~\ref{rmk:square roots}.

There is a natural analogy here with classical Euclidean constructions. Numbers expressible in closed form using only integers, addition, subtraction, multiplication, division, and square roots are known as constructible numbers, familiar from straightedge and compass constructions. The functions appearing in the twistor-solutions can be seen as analogous—constructed from similar operations, though using twistor variables in place of integers.
\begin{dfn}
    \label{def:native vertex}
    Let \(\Gamma\) be an OG graph with \(\tau_\ell\) an external \(n\)-arc. Then the an internal vertex \(v\) on \(\tau_\ell\) is called a \emph{native vertex} to \(\tau_\ell\) or an \emph{\(n\)-native vertex}.
    
    Let \(\Xi\) be an arc-sequence of \(m\geq0\) consecutive \(\mathrm{Arc}_{n_i,\ell_i}\) moves with \(n_i\leq4\) such that \(\Gamma = \Xi(\Gamma^\prime)\) is reduced after each move. Suppose that \(v^\prime\) is a vertex on an external \(n\)-arc \(\tau^\prime_{\ell^\prime}\) in \(\Gamma^\prime\), and that \(v = \Xi(v^\prime)\) and \(\tau_\ell = \Xi(\tau^\prime_{\ell^\prime})\). Then we would call \(v\) an \emph{\(n\)-native vertex} or \emph{native vertex} to the arc \(\tau_\ell\) in \(\Gamma\) as well (see Figure~\ref{fig:4native}). \(\tau_\ell\) would be called \(n\)-descendant. \(v^\prime\) and \(\tau^\prime_{\ell^\prime}\) would be called \emph{ancestral-vertex} and \emph{ancestral-arc} of \(v\) and \(\tau_\ell\) respectively. We will call \(\Xi\) an \text{ancestry-sequence} of \(v\) to \(\tau^\prime_{\ell^\prime}\).

    \(\omega\) will be called an \emph{\(\tau_\ell\)-proper orientation} or  \emph{\(\tau^\prime_{\ell^\prime}\)-proper orientation} for \(v\) if \(v^\omega = \Xi(v^{\prime \omega^\prime})\) and \(\omega^\prime\) is a \(\tau^\prime_{\ell^\prime}\)-proper orientation for \(v^\prime\). Notice that a proper orientation is always hyperbolic.   
\end{dfn}
Note a vertex can be native to more then one arc.
\begin{figure}[H]
    \centering
                \begin{center}
                \scalebox{0.8}{
\begin{tikzpicture}

\draw (0,-2.4)node[anchor=north]{\(\Gamma^\prime\)};
\draw (6,-2.4)node[anchor=north]{\(\Gamma\)};

\draw (0,0) circle (2cm);
\draw(1.41421, 1.41421)node[anchor=south]{\(3\)}--({-0.517638, -1.93185})node[anchor=north]{\(6\)}node[pos=0.43](a4){} node[pos=.42](b4){} node[pos=0.08](a5){} node[pos=.07](b5){} node[pos=.82](a6){} node[pos=.81](b6){};
\fill[blue] (0.448286*2, 0.258819*2) circle (2pt);
\draw (0.448286*2+0.1, 0.258819*2)node[anchor=north]{\(v^\prime\)};
\draw[blue](-1.93185, 0.517638)--(1.93185, 0.517638)node[pos=0.43](a1){} node[pos=.42](b1){} node[pos=0.08](a2){} node[pos=.07](b2){} node[pos=.82](a3){} node[pos=.81](b3){};
\draw(-1.93185, 0.517638)node[anchor=east]{\(5\)};
\draw(1.93185, 0.517638)node[anchor=west]{\(2\)} ;
\draw(-1.41421, 1.41421)node[anchor=south]{\(4\)}--(0.517638, -1.93185)node[anchor=north]{\(1\)}node[pos=0.43](a7){} node[pos=.42](b7){} node[pos=0.08](a8){} node[pos=.07](b8){} node[pos=.82](a9){} node[pos=.81](b9){};

\draw[blue] [very thick, -latex, shorten <=5] (a1)--(b1);
\draw[blue] [very thick, -latex, shorten <=5] (a2)--(b2);
\draw[blue] [very thick, -latex, shorten <=5] (a3)--(b3);

\draw[very thick, -latex, shorten <=5] (a4)--(b4);
\draw[very thick, -latex, shorten <=5] (a5)--(b5);
\draw[very thick, -latex, shorten <=5] (a6)--(b6);

\draw[very thick, -latex, shorten <=5] (a7)--(b7);
\draw[very thick, -latex, shorten <=5] (a8)--(b8);
\draw[very thick, -latex, shorten <=5] (a9)--(b9);

\draw (0+6,0) circle (2cm);
\draw (0+6+0.15,0)node[anchor=north]{\(v\)};

\draw({-0.765367+6, 1.84776})node[anchor=south]{\(5\)}--(1.84776+6, -0.765367)node[anchor=west]{\(2\)}node[pos=0.43](a1){} node[pos=.42](b1){} node[pos=0.15](a2){} node[pos=.14](b2){} node[pos=.7](a3){} node[pos=.69](b3){};
\draw [very thick, -latex, shorten <=5] (a1)--(b1);
\draw [very thick, -latex, shorten <=5] (a2)--(b2);
\draw [very thick, -latex, shorten <=5] (a3)--(b3);

\draw(-1.84776+6, 0.765367)node[anchor=east]{\(6\)}--(0.765367+6, -1.84776)node[anchor=north]{\(1\)}node[pos=0.43](a1){} node[pos=.42](b1){} node[pos=0.15](a2){} node[pos=.14](b2){} node[pos=.7](a3){} node[pos=.69](b3){};
\draw [very thick, -latex, shorten <=5] (a1)--(b1);
\draw [very thick, -latex, shorten <=5] (a2)--(b2);
\draw [very thick, -latex, shorten <=5] (a3)--(b3);

\draw(0.765367+6, 1.84776)node[anchor=south]{\(4\)}--(-0.765367+6, -1.84776)node[anchor=north]{\(8\)}node[pos=0.32](a1){} node[pos=.31](b1){} node[pos=0.9](a2){} node[pos=.08](b2){} node[pos=.81](a3){} node[pos=.80](b3){}node[pos=0.55](a4){} node[pos=.54](b4){};
\draw [very thick, -latex, shorten <=5] (a1)--(b1);
\draw [very thick, -latex, shorten <=80] (a2)--(b2);
\draw [very thick, -latex, shorten <=5] (a3)--(b3);
\draw [very thick, -latex, shorten <=5] (a4)--(b4);

\draw[blue](-1.84776+6, -0.765367)--(1.84776+6, 0.765367)node[pos=0.32](a1){} node[pos=.31](b1){} node[pos=0.9](a2){} node[pos=.08](b2){} node[pos=.81](a3){} node[pos=.80](b3){}node[pos=0.55](a4){} node[pos=.54](b4){};
\draw[blue] [very thick, -latex, shorten <=5] (a1)--(b1);
\draw[blue] [very thick, -latex, shorten <=80] (a2)--(b2);
\draw[blue] [very thick, -latex, shorten <=5] (a3)--(b3);
\draw[blue] [very thick, -latex, shorten <=5] (a4)--(b4);

\draw(-1.84776+6, -0.765367)node[anchor=east]{\(7\)};
\draw(1.84776+6, 0.765367)node[anchor=west]{\(3\)};
\fill[blue] (6, 0) circle (2pt);

\draw [very thick, -latex, shorten <=0] (2.3,0)--(3.7,0);

\draw(3,0)node[anchor=south]{\(\mathrm{Arc}_{4,2}\) };
\end{tikzpicture}
}
\end{center}
    \caption{a \(4\)-native vertex.}
    \label{fig:4native}
\end{figure}
\begin{exm}
Consider Figure~\ref{fig:4native}. We have \(\Gamma = \mathrm{Arc}_{4,2} \Gamma^\prime\). On \(\Gamma^\prime\), we have \(a^\prime = \{2,5\}\) being an external \(4\)-arc, with a \(a^\prime \)-proper orientation, and a highlighted internal vertex \(v^\prime\) on \(\{2,5\}\). In the graph \(\Gamma\), we have the the corresponding arc \(a = \mathrm{Arc}_{4,2} (a^\prime) = \{3,7\} \), with the inherited \(a\)-proper orientation. The corresponding vertex \(v =\mathrm{Arc}_{4,2} (v^\prime )\) is highlighted. 

The arc \(a\) is not an external \(4\)-arc, but it is \(4\)-descendant. We have that \(v\) is \(4\)-native to the arc \(a\). \(a^\prime\) and \(v^\prime\) are ancestral-arc and ancestral-vertex to \(a\) and \(v\). The ancestry-sequence is \(\Xi = \mathrm{Arc}_{4,2}\).
\end{exm}
Although a bit daunting, the notion of an \(n\)-native vertex is not new. This is the same structure of BCFW graphs and their boundaries expressed in Proposition~\ref{prop:BCFW arcs} and Section~\ref{sec:arc seq}, but now from the point of view of a single vertex. Since by Proposition~\ref{n<=4 twist sol prop} we have a very good control of vertices on external \(n\)-arcs with \(n\leq 4\), we would use this structure to reduces various claims to the case of a vertex being on such an arc.
\begin{obs}
\label{obs:BCFW 4 native}
    Let \(\Gamma\) be a BCFW graph or a boundary of a BCFW graph. Then every internal vertex in \(\Gamma\) is \(n\)-native to some arc with \(n\leq 4\), and every arc in \(\Gamma\) is \(n\)-descendant with \(n\leq 4\). If \(\Gamma\) is a BCFW graph, the all internal vertices are \(4\)-native to some arc, and every arc in \(\Gamma\) is \(4\)-descendant. 
\end{obs}
\begin{proof}
    This is immediate form Propositions~\ref{prop:BCFW arcs} and~\ref{prop:arc limits boundary sub BCFW}.
\end{proof}
As stated before, we aim to study the expressions which can appear inside radicals in the twistor-solutions. Definition~\ref{def:radicand} aims to capture precisely those expressions.
\begin{dfn}
\label{def:radicand}
    Consider an OG graph \(\Gamma\) with an external \(4\)-arc \(\tau_\ell\), that is \(\Gamma = \mathrm{Arc}_{4,\ell} \Gamma_0\) for some graph \(\Gamma_0\). Let the support of \(\tau_\ell\) be \(I \subset [2k]\). We will call the Mandelstam variable \(S_I\) the \emph{the arc-radicand} of \(\tau_\ell\).

    Let \(v\) be a vertex that is \(4\)-native of \(\tau_\ell\) a \(4\)-descendant arc. Let \(\Xi\) be an ancestry-sequence of \(v\) to \(\tau_\ell\). We will call \(\Xi(S_I)\)  an \emph{arc-radicand} of \(\tau_\ell\) and a \emph{vertex-radicand} of \(v\).
\end{dfn}
\begin{rmk}
    A choice of an arc-sequence representation for a BCFW graph or a boundary of a BCFW graph gives each vertex and arc a unique ancestral-arc and radicand. 
\end{rmk}
\begin{rmk}
\label{rmk:what is a radicand}
    For a vertex \(v\) on an external \(4\)-arc of length four \(\tau_\ell\), the expression \(S_I\) that is arc-radicand of \(\tau_\ell\) and a vertex-radial of \(v\), is precisely the expression appearing in the radical of the of twistor-solution \(a = \mathcal F (v^\omega, \Gamma)\), with \(\omega\) a \(\tau_\ell\)-proper orientation by Proposition~\ref{prop:sol arc to angle} and Corollary~\ref{n<=4 twist solv coro}.

    Let \(\Xi\) be some ancestry-sequence of \(\Xi(v)\). Then by Proposition~\ref{prop:sol commute for angles} we have that \\\(\mathcal F(\Xi(v^\omega),\Xi(\Gamma)) = \Xi(a)\). Thus radicand \(S_i\) in \(a\) becomes \(\Xi(S_i)\) in the twistor-solution for \(\Xi(v)\) under a \(\tau_\ell\)-proper orientation. Notice \(\Xi(S_i)\) is how we defined a vertex-radicand of \(\Xi(v)\).
\end{rmk}
\begin{proof}
    This is immediate by Propositions~\ref{prop:arc limits boundary sub BCFW} and~\ref{prop:BCFW arcs}.
\end{proof}
Our strategy for proving Theorem~\ref{thm:ampli diffeo} is showing the twistor-solutions for vertices in BCFW cells and their boundaries are smooth. Observation~\ref{obs:pos rad is smooth} will be our chief strategy for showing functions are smooth on the amplituhedron:
\begin{dfn}
    We say \(f\in\mathcal F_{[2k]}\) is \text{positive} or \emph{zero} (on an OG graph \(\Gamma\)) if for any \([C,\Lambda,Y]\in\mathcal U^\geq_k\) (with \(C\in\Omega_\Gamma\)) we have that \(f(\Lambda,Y)>0\) or \(f(\Lambda, Y)= 0\) respectively.
\end{dfn}
Recall that we are concerned with the functions appearing in the twistor-solutions studied in Section~\ref{sec:inj}. Those include rational expressions and square roots. As stated before, such functions are well-behaved when those square roots, or radicals, are positive. 
\begin{dfn}
    We say \(a\in\mathcal F\) is \emph{positively-radical}  (on an OG graph \(\Gamma\)) if \(a\) has a representation as an algebraic expression in abstract twistors such that all of the radicals are positive (on \(\Gamma\)).
\end{dfn}
\begin{obs}
\label{obs:pos rad is smooth}
    If \(f\) is positively-radical in a \(k\)-OG graph \(\Gamma\), then for every \(\Lambda\in \mathrm{Mat}^>_{2k \times (k+2)}\), the map \(f(\Lambda,\bullet):\widetilde \Lambda(\Omega_\Gamma)\subset\mathcal O_k(\Lambda)\rightarrow \mathbb R\) can be extended to a smooth map in an open neighborhood of \(\widetilde \Lambda(\Omega_\Gamma)\).
\end{obs}
It is important to note that the existence of a positively-radical twistor-solution for an oriented vertex does not depend on the orientation selected for the vertex. It is thus a property inherent to the vertex itself. 
\begin{obs}
\label{obs:we are too rad for orientations}
    Let \(v^\omega\) be an hyperbolically oriented internal vertex in an OG graph \(\Gamma\), with twistor-solution \(\mathrm{arccosh}(f) = \mathcal F(v^\omega, \Gamma)\). Let \(v^\psi\) be a different hyperbolic orientation for the same vertex. Then \(v^\psi\) is twistor-solvable, and for \( \mathrm{arccosh}(g) = F(v^\psi, \Gamma)\) we have that \(f\) is positively-radical on \(\Gamma\) iff \(g\) is.
\end{obs}
\begin{proof}
    First let us discuss the issue of the uniqueness of expressions in abstract twistors. Recall Remark~\ref{rmk:uniqueness of solution}. We have that all of the identities defining the ring of twistors on a specific orthitroid cell are polynomial. Thus if one representation of the solution is positively-radical, all other representations are as well.
    Now consider different oriented vertices. The claim is immediate from the fact that either \(\mathrm{arccosh}(f)\) is a twistor-solution for \(v^\psi\) or \ \(\mathrm{arccosh}\left(\sqrt{\frac{f^2}{f^2-1}}\right)\) is by Proposition~\ref{prop:angles repara}, and as by Definitions~\ref{def:param} and~\ref{def:twist sol ang}, we must have that all twistor-solutions for any oriented vertex is positive on the graph.
\end{proof}
\begin{dfn}
\label{def:pos rad vert}
    For an internal vertex \(v\) in an OG graph \(\Gamma\), we say \(v\) is \emph{positively-radical} if there exists an hyperbolic orientation \(\omega\) for \(v\), such that \(v^\omega\) has a twistor-solution \(\mathrm{arccosh}(f) = \mathcal F(v^\omega, \Gamma)\), with \(f\) being positively-radical on \(\Gamma\).
\end{dfn}
\begin{rmk}
    Be not dismayed by the appearance of \(\mathrm{arccosh}\). As discussed in Remark~\ref{rmk:no hyp func}, we are purely concerned with algebraic functions. The \(\mathrm{arccosh}\) is simply here to turn the expression of the twistor-solution from Proposition~\ref{prop:sol arc to angle}, which is of the form \(\alpha = \mathrm{arccosh}(f)\) where \(f\) is algebraic in twistors, back to an algebraic expression. By the definition of the \(\mathrm{Rot}_{i,i+1}(\alpha)\) move in Definition~\ref{abs substit def}, we clearly have that \(\mathrm{Rot}_{i,i+1}(\mathrm{arccosh}(f))\) is algebraic in twistors, and as has all \(\sinh(\alpha)\) must be positive by Definition~\ref{def:param}, we get that \(\mathrm{Rot}_{i,i+1}(\mathrm{arccosh}(f))\) is positively-radical if \(f\) is.
\end{rmk} 

As we hope is now becoming clear, our plan for proving Theorem~\ref{thm:ampli diffeo} is to use the fact that all vertices on BCFW graphs and their boundaries are \(n\)-native for \(n\leq 4\) to show all their internal vertices are positively-radical, thus admitting a smooth twistor-solution. We will then use those solutions to show we can invert the amplituhedron map using smooth functions. Consider the following proposition, which we will prove in Section~\ref{sec:proof of prop:the natives are radicalized}:
\begin{prop}
\label{prop:the natives are radicalized}
    Let \(\Gamma\) be an OG graph with an \(n\)-native vertex \(v\) with \(n\leq4\). Then \(v\) is positively-radical.
\end{prop}
\begin{coro}
    \label{coro:BCFW are rad}
    All internal vertices in BCFW graphs and their boundary graphs are positively-radical.
\end{coro}
\begin{proof}
    This is immediate from Proposition~\ref{prop:the natives are radicalized} and Observation~\ref{obs:BCFW 4 native}.
\end{proof}
We are now ready to prove Theorem~\ref{thm:ampli diffeo}:
\begin{proof}[Proof of Theorem~\ref{thm:ampli diffeo}] 
   By the parameterization from Theorem~\ref{param bij}, we have a diffeomorphism \(\varphi:U\rightarrow \Omega_\Gamma\), where \(U\) is an open subset of \(\mathbb R^n\) and \(n\) is the number of internal vertices in \(\Gamma\). By Corollary~\ref{coro:BCFW are rad}, all internal vertices on \(\Gamma\) are positively-radical. By Observations~\ref{obs:pos rad is smooth} and~\ref{obs:we are too rad for orientations}, this means they are twistor-solvable regardless of their hyperbolic orientation with a smooth twistor-solution. As by Proposition~\ref{prop:hyperbolic choice per arc} \(\Gamma\) admits a hyperbolic orientation, we have a smooth map \(f:\widetilde \Lambda(\Omega_\Gamma)\rightarrow \mathbb R^n\), such that \( f\circ \evalat{\widetilde \Lambda}{\Omega_\Gamma} \circ\varphi = \mathrm{id}\). This means \(\evalat{\widetilde \Lambda}{\Omega_\Gamma}\) has a smooth inverse. As it is clearly smooth, we can conclude \(\evalat{\widetilde \Lambda}{\Omega_\Gamma}\) is a diffeomorphism onto its image.
\end{proof}
\subsubsection{Proof of Proposition~\ref{prop:the natives are radicalized}}
\label{sec:proof of prop:the natives are radicalized}
We will now present the proof for Proposition~\ref{prop:the natives are radicalized}. Our approach will be similar to what we presented in Section~\ref{sec:inj}: We will use Lemma~\ref{lem:vanishing of i i+1 in  k=3} as a base case for \(\OGnon{3}{6}\), continue to extend the claim to vertices on external arcs in Proposition~\ref{prop:external radicand doen't vanish}, and then use Lemma~\ref{lem:prom radical pos} to finally extend to all relevant vertices to prove Proposition~\ref{prop:the natives are radicalized}. 
\begin{lem}[\cite{companion}]
\label{lem:vanishing of i i+1 in  k=3}
     Let \(\Gamma\) be a \(3\)-OG graph, and \(\tau\) be the permutation corresponding to \(\Gamma\). Consider \(C\in \Omega_\Gamma \subset\OGnon{3}{6}\), and \([C,\Lambda,Y]\in\mathcal{U}_{3}^\geq\). Then \(\twist{Y}{i}{i+1} = 0\) iff \(\tau(i) = i+1\) (where \(i+1\) is considered mod \(6\)). 

     Specifically, \(\twist{Y}{i}{i+1} = 0\) means \(C \in \overline{\Omega}_{\Gamma_0}\), the closure of the orthitroid cell corresponding to the graph \(\Gamma_0\) as seen in Figure~\ref{fig: i i+1 =0}, which is codimension \(2\) in \(\OGnon{3}{6}\).

\begin{figure}[H]
\centering

\begin{center}
\begin{tikzpicture}[scale = 0.6]

\draw (0,0) circle (2cm);
\draw (1.2, -1.74)node[anchor=north]{\(i+1\)};
\draw(-1, -1.74)node[anchor=north]{\(i\)};
\draw [black] plot [smooth, tension=1.2] coordinates {  (-1., -1.73205)(0,-1.3) (1., -1.73205)};
\draw(2,0)--(-1., 1.73205);

\draw(-2,0)--(1., 1.73205);
\end{tikzpicture}
\end{center}
    
     \caption{\(\Gamma_0\), the graph corresponding to the cell where \(\twist{Y}{i}{i+1} = 0\) in \(\OGnon{3}{6}\)}
     \label{fig: i i+1 =0}
 \end{figure}
\end{lem}
Consider a BCFW graph \(\Gamma\), and an internal vertex \(v\) that is adjacent to an external vertex.

\begin{prop}
\label{prop:external radicand doen't vanish}
     Let \(\Gamma\) be a \(k\)-OG graph, and \(\tau_\ell\) an external \(4\)-arc of \(\Gamma\), and \(v\) a vertex on \(\tau_\ell\). Let \(S_I\) be the arc-radicand of \(\tau_\ell\) and of \(v\). Then \(S_I\) is positive on \(\Gamma\).
\end{prop}
\begin{proof}
    The proof resembles the proof of Proposition~\ref{prop:arc sol depends on supp}.  Let \(C\in\Omega_\Gamma\), and \([C,\Lambda,Y]\in \mathcal U_k^\geq\). Recall that \(\Gamma\) must be of the form illustrated in Figure~\ref{fig:ext arc} with \(I = \{i_1,i_2,i_3,i_4\}\). Using the \(\mathrm{Cyc}\) move and Proposition~\ref{s cyc}, it is enough to show for \(I = \{1,2,3,4\}\), that is \(\tau_\ell = \{1,4\}\). 

   We can reduce the number of vertices in the graph \(\Gamma\) using the \(\mathrm{Rot}^{-1}\) move according to Proposition~\ref{s rot} until all the internal vertices are on the arc  \(\{1,4\}\), without changing the value of \(S_I(\Lambda,Y)\):
    If there are \(i,i+1\notin I\) such that the arcs \(\tau_i\) and \(\tau_{i+1}\) are crossing, we can use the equivalence move number 3 to replace the graph with an equivalent graph such that the crossing vertex is adjacent to an external vertex. We can now eliminate it using a \(\mathrm{Rot}^{-1}_{i,i+1}\) move without changing the arc-radicand according to Proposition~\ref{s rot}. We will continue doing so until there are no more such arcs, meaning the only pairs of crossing arcs are ones where one of the arcs is contained in \(I\). The only arc contained in \(I\) is \(\{1,4\}\) as it is external, thus we can assume all the internal vertices are on the arc \(\{1,4\}\) in \(\Gamma\).

    Since  \(\{1,4\}\) is external, there are no arcs contained \(I\). Meaning that for any \(\{r,\tau(r)\}\) an arc in \(\Gamma\) that does not cross \(\{1,4\}\), we have that \(4<r,\tau(r)\). Since all the internal vertices are on \(\{1,4\}\), we must have that \(\{r,\tau(r)\}\) is external with support of size \(2\). This means we have an arc \(\{r,r+1\}\) with \(4<r\), and \(\Gamma = \mathrm{Inc}_r (\Gamma_0)\) for some graph \(\Gamma_0\). We can now apply \(\mathrm{Inc}_r^{-1}\) according to Proposition~\ref{s inc} to reduce \(k\) by \(1\). As \(\{1,4\}\) is external, there are exactly \(2\) arcs that cross it. That means that there are exactly \(k - 3\) arcs that do not cross \(\{1,4\}\). By induction, we can now assume that all arcs cross \(\{1,4\}\) in \(\Gamma\) and \(k=3\).

    So now we can assume \(\tau_\ell = \{1,4\}\) and  \(k=3\). Now, by Proposition~\ref{prop:S orth} we have that \(S_{\{1,2,3,4\}} = \langle 5\,6\rangle ^2 \) on \([C,\Lambda,Y]\). That means that \(S_{\{1,2,3,4\}}(\Lambda, Y)\) is non-negative, and is zero iff \(\twist{Y}{5}{6} = 0\). By Lemma~\ref{lem:vanishing of i i+1 in  k=3}, since we know \(\Gamma\) has an external \(4\)-arc, we must have that \(\twist{Y}{5}{6}\neq 0\). Thus \(S_{\{1,2,3,4\}}(\Lambda, Y)>0\), finishing the proof.
\end{proof}
\begin{coro}
\label{coro:ext is pos rad}
    Let \(\Gamma\) be an OG graph with an external \(n\)-arc \(\tau_\ell\) with \(n\leq 4\), and \(v\) a vertex on \(\tau_\ell\). Then \(v\) is positively-radical on \(\Gamma\).
\end{coro}
\begin{proof}
    This is precisely as described in Remark~\ref{rmk:what is a radicand}. By proposition~\ref{prop:sol vnlk prop}, \(\tau_\ell\) is twistor-solvable. Set \(v^\omega\) a \(\tau_\ell\)-proper orientation for \(v\), and we have that \(v^\omega\) is twistor-solvable by Proposition~\ref{prop:sol arc to angle}. The twistor-solution \(\mathcal F(v^\omega,\Gamma)\) has no radicals if \(n<4\), and if \(n=4\) the only radical is the the square-root of a vertex-radicand of \(v\), which its positive on \(\Gamma\) by Proposition~\ref{prop:external radicand doen't vanish}.
\end{proof}
We will now want to show the property of being positively-radical behaves nicely under promotion with the \(\mathrm{Arc}\) move, which dominates the structure of BCFW cells and their boundaries as seen in Section~\ref{sec:BCFW Cells and their Boundaries}. This will allow us to extend the claim to all \(n\)-native vertices with \(n\leq 4\).
\begin{lem}
\label{lem:prom radical pos}
    Let  \(f\in\mathcal F\) be an expression in abstract twistors. If \(f\) is positively-radical on a reduced \(k\)-OG graph \(\Gamma\),  \(\mathrm{Arc}_{n,\ell,k}\Gamma\) is reduced, and \(n\leq4\), then \(\mathrm{Arc}_{n,\ell,k}(a)\) is positively-radical on \(\mathrm{Arc}_{n,\ell,k}\Gamma\). 
\end{lem}
\begin{proof}
     Recall that by Definition~\ref{def:arc},
     \[
    \mathrm{Arc}_{n,\ell,k}(a)\\
    \defeq \mathrm{Rot}_{\ell+n-2,\ell+n-1}(\alpha_{n,\ell,k,n -2})\cdot...\cdot\mathrm{Rot}_{\ell+2,\ell+3}(\alpha_{n,\ell,k,2})\mathrm{Rot}_{\ell+1,\ell+2}(\alpha_{n,\ell,k,1})\mathrm{Inc}_{l}(a),
    \]
    Write \(\mathrm{Arc}_{n,\ell,k}(a) = \widehat a\), and let \( \widehat \rho\) be a radicand in \( \widehat a\). We have two possibilities:
    \begin{enumerate}
        \item \( \widehat\rho =\mathrm {Arc}_{n,\ell,k}(\rho)\), where \( \rho\) is a radicand in \(a\).
        \item \(\widehat\rho\) is added by \(\mathrm{Rot}_{\ell+i,\ell+i+1}(\alpha_{n,\ell,k,i})\) for some \(i\).
    \end{enumerate}
    
    First consider possibility \(1\), that is, \(\widehat  \rho =\mathrm {Arc}_{n,\ell,k}(\rho)\), where \(\rho\) is a radicand in \( a\):
    By assumption, we have that all of the radicals in \(\widehat a\) are positive on \(\Gamma\). Specifically, we have that \(\rho\) is positive on \(\Gamma\).  Now, by Proposition~\ref{prop:com diag} we have that \(\mathrm {Arc}_{n,\ell,k}( \rho) = \widehat  \rho\) is positive on \(\mathrm {Arc}_{n,\ell,k} \Gamma\).

    Now consider possibility \(2\): Write \(\alpha_{n,\ell,k,i} = \mathrm{arccosh}(b)\). By Definition~\ref{abs substit def}, the radicals added by \(\mathrm{Rot}_{\ell+i,\ell+i+1}(\alpha_{n,\ell,k,i} )\) 
    are either \(\sqrt{\widehat\rho}=\sinh(\alpha_{n,\ell,k,i}) = \sqrt{b^2-1}\)
    itself, or a radical from the expression for \(\alpha_{n,\ell,k,i}\). By Proposition~\ref{prop:sol arc to angle}, we have that \(\alpha_{n,\ell,k,i}\) is a twistor-solution of a vertex \(u\) in \(\widehat \Gamma\) on the external \(n\)-arc \(\tau_\ell\),
    the one added by the \(\mathrm{Arc}\) move. By Definition~\ref{def:param}, all angles are positive, thus \(\sqrt{b^2-1}\) is positive on \(\mathrm{Arc}_{n,\ell,k}\Gamma\). For \(\sqrt{\widehat\rho}\neq\sinh(\alpha_{n,\ell,k,i})\), we have that \(\widehat\rho\) must be a radicand in the expression \(b\). Notice that by Proposition~\ref{prop:sol arc to angle}, we have that only radicand in \(\cosh(\alpha_{n,\ell,k,i})\) is the arc-radicand of \(\tau_\ell\). Thus  \(\widehat \rho\) is positive on \(\mathrm{Arc}_{n,\ell,k}\Gamma\) by Proposition~\ref{prop:external radicand doen't vanish}.
\end{proof}
\begin{coro}
    \label{coro:prom of rad vertex is rad}
    Let  \(v\) be a positively-radical vertex in a reduced OG graph \(\Gamma\). If \(\mathrm{Arc}_{n,\ell,k}\Gamma\) is reduced and \(n\leq4\), then \(\mathrm{Arc}_{n,\ell,k}(v)\) is a positively-radical vertex in \(\mathrm{Arc}_{n,\ell,k}\Gamma\).    
\end{coro}
\begin{proof}
    \(v\) is positively-radical. Which means there exists an hyperbolic orientation \(v^{\omega}\) such that \(\mathrm{arccosh}(a) = \mathcal F(v^{\omega},\Gamma)\), and \(a\) is positively-radical on \(\Gamma\).

     By Corollary~\ref{coro:ext is pos rad} we have that the vertices added by the move to \(\mathrm{Arc}_{n,\ell,k}\Gamma\) are positively-radical. By Observation~\ref{obs:we are too rad for orientations}, they are twistor-solvable with whichever hyperbolic orientation. Thus, by Proposition~\ref{prop:sol commute for angles}, we have that
    \[
    \mathrm{Arc}_{n,\ell,k}(\mathrm{arccosh}(a))= \mathrm{arccosh}(\mathrm{Arc}_{n,\ell,k}(a))= \mathcal F(\mathrm{Arc}_{n,\ell,k}(v^{ \omega})),\mathrm{Arc}_{n,\ell,k}(\Gamma)).
    \]
    By Lemma~\ref{lem:prom radical pos} we have that \(\mathrm{Arc}_{n,\ell,k}(a)\) is positively-radical on \(\mathrm{Arc}_{n,\ell,k}(\Gamma)\). As \(\mathrm{Arc}_{n,\ell,k}(v^{ \omega}))\) is clearly with hyperbolic orientation, we have that \(\mathrm{Arc}_{n,\ell,k}(v)\) is positively-radical in \(\mathrm{Arc}_{n,\ell,k}\Gamma\).
\end{proof}
\begin{proof}[Proof of~\ref{prop:the natives are radicalized}]
     By Definition~\ref{def:native vertex}, we have \(\Xi\), an ancestry-sequence such that \(\Gamma = \Xi(\Gamma^\prime)\) and \(v = \Xi(v^{\prime}) \), with \(v^\prime\) a vertex on the arc \(\tau^\prime_{\ell^\prime}\) an external \(n\)-arc of \(\Gamma^\prime\). By Corollary~\ref{coro:ext is pos rad}, we have that \(v^{\prime}\) is positively-radical. Which means there exists an hyperbolic orientation \(v^{\prime \omega^\prime}\) such that \(\mathrm{arccosh}(a^\prime) = \mathcal F(v^{\prime \omega^\prime},\Gamma^\prime)\), and \(a^\prime\) is positively-radical on \(\Gamma^\prime\). Suppose that \(\Xi\) is composed of \(m\) \(\mathrm{Arc}\) moves. We will prove this by induction on \(m\). For the base case, the case where \(m=0\), the claim is trivial.
     
    For the step, let \(\Xi = \mathrm {Arc}_{\widehat n,r} \widehat \Xi\). Write  \(\widehat \Gamma = \widehat \Xi(\Gamma^\prime)\), and  \(\widehat v = \widehat \Xi(v^{\prime})\). As \(\Xi\) is an ancestry-sequence, we have that \(\widehat n \leq 4\). \(\widehat v = \widehat \Xi(v^{\prime})\) is positively-radical by the induction hypothesis.
    As \(\mathrm{Arc}_{\widehat n,r}\widehat \Gamma = \Gamma \), and \(\mathrm{Arc}_{\widehat n,r}(\widehat v) = v\), we have that \(v\) is a positively-radical vertex by Corollary~\ref{coro:prom of rad vertex is rad}.
\end{proof}

\subsection{Vertex-Separators}
The goal of this section is to prove Theorem~\ref{thm:BCFW and codim 1 bounds are smooth}. The challenging part will be to prove that the images of the codimension \(1\) boundaries of orthitroid cells are smooth submanifolds. We will do so by finding functions termed \emph{vertex-separators} for each internal vertex in a BCFW graph, show they vanish on the corresponding codimension \(1\) boundary, and show their differential does not vanish on that boundary. In fact, we will show the differential is positive in the direction of the cell (See Proposition~\ref{prop:sep 4-native} for details), which will prove pivotal to the proof of~\ref{thm:local sep} in Section~\ref{sec:proving local separation}. In order to do that, we will first need to better understand the relationship of derivatives to the notions developed so far. 
\begin{dfn}
    \label{def:vertex sep}
    Let \(\Gamma\) be an OG graph, and let \(v_i\) be the \(i\)-th internal vertex in an external \(4\)-arc \(\tau_\ell = \{\ell,\tau(\ell)\}\) with \(l<\tau(\ell)\) (that means \(i\) can be either \(1\) or \(2\), \(i=1\) will denote the one closest to the external vertex \(\ell\)). We will call \(S_{\{\ell+1,\ell+2,\ell+3\}}\) a \emph{vertex-separator} of \(v_1\), and \(S_{\{\ell,\ell+1,\ell+2\}}\) a \emph{vertex-separator} of \(v_2\).

    Let \(v\) be an internal vertex in an OG graph \(\Gamma\) such that \(\Gamma = G(\Gamma^\prime)\) and \(v = G(v^\prime)\), where \(G\) is a series of moves. If \(S^\prime\) is a vertex-separator of \(v^\prime\) then, \(G(S^\prime)\) will be a \emph{vertex-separator} of \(v\). Notice any \(4\)-native vertex has a vertex separator.   
  If a Mandelstam variable \(S_I\) for \(I\subset [2k]\) is a vertex-separator of a vertex \(v\), we will consider the complimentary Mandelstam \(S_{[2k]\setminus I}\) to be a vertex-separator for \(v\) as well, as they are equal by Proposition~\ref{prop:S orth}.

\end{dfn}
We will now consider the derivative of the  vertex-separators when the angle associated to the vertex approaches zero. We will therefore define a notion of the derivative of a function with respect to an oriented vertex, and we show the sign does not depend on its orientation.
\begin{dfn}
    \label{def:vertex derivative}
    Let \(v^\omega\) be an oriented vertex in a perfectly orientated OG graph \(\Gamma^\omega\), and write \(\Gamma_ 0 =\lim_{v^\omega\rightarrow L}\Gamma\) for finite \(L\). Let \(C(\alpha_1,...,\alpha_n)\) be the corresponding parameterization of \(\Omega_\Gamma\), such that \(\alpha_1\) is the angle associated to \(v^\omega\). Let \([C(\alpha_1,...,\alpha_n),\Lambda, Y(\alpha_1,...,\alpha_n)]\in\mathcal U^\geq_k\). For \(f\in \mathcal F\), define \(\evalat*{\frac{\partial}{\partial v^\omega}f}{\Gamma_0}\) by 
    \[\evalat*{\frac{\partial}{\partial v^\omega}f}{\Gamma_0}(\Lambda, Y(L,...,\alpha_n)) = \evalat*{\frac{\partial}{\partial t}f(\Lambda, Y(t,\alpha_2,...,\alpha_n))}{t = L},\]
    which is well-defined since the amplituhedron map is defined for a an open neighborhood around \(t=L\) by Proposition~\ref{prop:extends to a nbhd}.
    
    For a different orientation \(v^\psi\), by Proposition~\ref{prop:angles repara} we have that \(C(\alpha_1,...,\alpha_n)(v^\psi,\Gamma) = g(\alpha_1)\) for some function \(g\). Write 
    \(
   \evalat*{\frac{\partial}{\partial v^\omega}f}{\Gamma_0} =\evalat*{\frac{\mathrm d}{\mathrm dt} g(t)}{t=L}\evalat*{\frac{\partial}{\partial v^\psi}f}{\Gamma_0}.
    \)
\end{dfn}
\begin{obs}
\label{obs:vertex derivative orientation}
    Let \(v^\omega\) be an oriented vertex in an OG graph \(\Gamma\), and let \(f\in \mathcal F\), then the sign of \(\evalat{\frac{\partial}{\partial v^\omega}f}{\partial_v \Gamma}\) does not depend on \(\omega\).
    \end{obs}
    The property of the derivative with respect to a vertex being positive is thus well-defined. As will be instrumental in the future, it also behaves nicely under promotion. 
    \begin{dfn}
    \label{def:vert der sans ori}
        Let \(v^\omega\) be an oriented vertex in an OG graph \(\Gamma\), and let \(f\in \mathcal F\). If \(\evalat{\frac{\partial}{\partial v^\omega}f}{\partial_v \Gamma}\) is positive on \(\partial_v \Gamma\) for any perfect orientation \(\omega\), then we would say \emph{\(\evalat{\frac{\partial}{\partial v}f}{\partial_v \Gamma}\) is positive.}
    \end{dfn}
    \begin{obs}
\label{obs:triv diff prom}
    Consider \(f\in\mathcal F\), a move \(G\) that is an \(\mathrm{Inc}\), \(\mathrm{Cyc}\), or \(\mathrm{Rot}\) move with a fixed angle, a \(k\)-OG graph \(\Gamma\), and suppose \(\Gamma^\prime = G(\Gamma)\) is a \(k^\prime\)-OG graph. Let \(C(t)\in \OGnon{k}{2k}\), \([C(t),\Lambda,Y(t)]\in \mathcal{U}_k^\geq\) for \(t\geq 0\) with \(C(0)\in \Omega_\Gamma\). Then for \([C^\prime(t),\Lambda^\prime,Y^\prime(t)] = G [C(t),\Lambda,Y(t)]\in \mathcal{U}_{k^\prime}^\geq\), we have that \(C^\prime(0)\in \Omega_{\Gamma^\prime}\), and
    \[
     \evalat*{ \frac{\partial}{\partial t}f(\Lambda,Y(t))}{t=0} =\evalat*{ \frac{\partial}{\partial t}(Gf)(\Lambda^\prime,Y^\prime(t))}{t=0}.
    \]
\end{obs}
\begin{proof}
    This is immediate from the fact that \(f(\Lambda,Y(t)) =(Gf)(\Lambda^\prime,Y^\prime(t))\) by Proposition~\ref{prop:com diag}.
\end{proof}
    \begin{obs}
    \label{obs:ver der prom}
Consider \(f\in\mathcal F\), a \(\mathrm{Rot}\), \(\mathrm{Cyc}\), or  \(\mathrm{Inc}\) move \(G\), a \(k\)-OG graph \(\Gamma\), and suppose \(\Gamma^\prime = G(\Gamma)\) is a \(k^\prime\)-OG graph. Let \(v^\prime\) be an internal vertex in \(\Gamma^\prime\), and write \(v = G(v^\prime)\) and \(f^\prime = G(f)\). Then \(\evalat{\frac{\partial}{\partial v}f}{\partial_v \Gamma}\) is positive iff \(\evalat{\frac{\partial}{\partial v^\prime}f^\prime}{\partial_{v^\prime} \Gamma^\prime}\) is positive.

\end{obs}
\begin{proof}
    This is immediate from Observation~\ref{obs:triv diff prom}, and Corollary~\ref{coro:moves dont change angles}. 
\end{proof}
Consider the following Proposition, which we will prove in Section~\ref{sec:Proof of Proposition prop:sep 4-native}:
 \begin{prop}
\label{prop:sep 4-native}
    For \(\Gamma\) an OG graph, a \(4\)-native vertex \(v\) in an OG graph \(\Gamma\) with vertex-separator \(S\), and \(\Gamma_0 = \partial_v\Gamma\). Then \(S\) is both positively-radical and zero on \(\Gamma_0\), and \(\evalat{\frac{\partial}{\partial v}S}{\Gamma_0}\) is positive.
\end{prop}
This is, in fact, all we need in order to prove Theorem~\ref{thm:BCFW and codim 1 bounds are smooth}:
\begin{proof}[Proof of Theorem~\ref{thm:BCFW and codim 1 bounds are smooth}]

    Recall that by Proposition~\ref{prop:mom cons}, we have that \(\widetilde \Lambda (\Omega_\Gamma)\subset \mathbf Y_\Lambda^k\).
    First consider \(\Gamma\), a \(k\)-BCFW graph. By Proposition~\ref{prop:zerolocus submfld}, we have that \(\mathrm{dim}(\mathbf Y^k_\Lambda) = 2k-3\), which is the number of internal vertices of a \(k\)-BCFW graph by Theorem~\ref{thm:BCFW are trees}. This means that \(\mathrm{dim}(\Omega_\Gamma) = 2k-3\). By Theorem~\ref{thm:ampli diffeo}, \(\evalat{\widetilde\Lambda}{\Omega_\Gamma}\) is a diffeomorphism onto its it image. We can conclude that \(\widetilde\Lambda(\Omega_\Gamma)\subset \mathbf Y_\Lambda^k \) is a smooth submanifold.
    
Now for \(\Gamma_0=\partial_v\Gamma = \lim_{v^\omega\rightarrow L} \Gamma\), a codimension \(1\) boundary of a \(k\)-BCFW graph. By definition~\ref{def:boundary graph}, this means \(\mathrm{dim}(\Omega_{\Gamma_0})=2k-4\). By Theorem~\ref{thm:ampli diffeo}, \(\evalat{\widetilde\Lambda}{\Omega_{\Gamma_0}}\) is a diffeomorphism onto its it image, thus \(\mathrm{dim}(\widetilde\Lambda(\Omega_{\Gamma_0}))=2k-4\).

By Proposition~\ref{prop:sep 4-native}, we have that \(S\), a vertex-separator of \(v\), is positively-radical and zero on \(\Gamma_0\), and \(\evalat*{\frac{\partial}{\partial v}S}{\Gamma_0}\) is positive. By Observation~\ref{obs:pos rad is smooth} this means \(S(\Lambda,\bullet)\) is smooth in \(U\), an open neighborhood containing \(\widetilde\Lambda(\Omega_{\Gamma_0})\). Furthermore, since \(\evalat{\frac{\partial}{\partial v}S}{\Gamma_0}\) is positive, we know \(\evalat{\frac{\partial}{\partial v^\omega}S}{\Gamma_0}\) is positive for any perfect orientation \(\omega\) on \(v\) by Definition~\ref{def:vert der sans ori}. By Definition~\ref{def:vertex derivative}, \(\evalat{\frac{\partial}{\partial v^\omega}S}{\Gamma_0}=\mathrm{d}(S(\Lambda,\bullet))(\mathbf v )\) on \(\widetilde\Lambda(\Omega_{\Gamma_0})\) for some direction \(\mathbf v\). This means the differential of \(S(\Lambda, \bullet)\) does not vanish in some open neighborhood of \(\widetilde\Lambda(\Omega_{\Gamma_0})\). By restricting \(U\), we can now assume the differential of \(S(\Lambda,\bullet)\) does not vanish in \(U\). We conclude \(S(\Lambda,\bullet)=0\) defines a smooth codimension \(1\) submanifold of \(U\subset\mathbf Y_\Lambda^k\).

As by Proposition~\ref{prop:sep 4-native}, we have that \(S\) is zero on \(\Gamma_0\), we have that \(\widetilde\Lambda(\Omega_{\Gamma_0})\) is a contained in the zero locus of \(S(\Lambda,\bullet)\) in \(U\). As they both have the same dimension, by the same argument as for the previous case, we get that \(\widetilde\Lambda(\Omega_{\Gamma_0})\) is a smooth submanifold of the zero locus of \(S(\Lambda,\bullet)\). We can thus conclude that \(\widetilde \Lambda (\Omega_{\Gamma_0})\subset \mathbf Y_\Lambda^k\) is a smooth submanifold.

\end{proof}
\subsubsection{Proof of Proposition~\ref{prop:sep 4-native}}
\label{sec:Proof of Proposition prop:sep 4-native}
We will now present the proof for Proposition~\ref{prop:sep 4-native}. Our approach will be similar to what we presented in Section~\ref{sec:proof of prop:the natives are radicalized}: We will use Proposition~\ref{prop:sep k=3} as a base case for \(\OGnon{3}{6}\), continue to extend the claim to vertices on external arcs in Proposition~\ref{prop:sep external}, and then use promotion to extend to general \(4\)-native vertices and prove Proposition~\ref{prop:sep 4-native}. 
\begin{prop}
\label{prop:sep k=3}
    For \(\Gamma\) as in Figure~\ref{fig:top cell}, corresponding to the the top cell of \(\OGnon{3}{6}\), an internal vertex \(v\)  in \(\Gamma\) with vertex-separator \(S\), and \(\Gamma_0 = \partial_v \Gamma\) the codimension \(1\) boundary resulting from opening \(v\). Then \(S\) is both positively-radical and zero on \(\Gamma_0\), and \(\evalat{\frac{\partial}{\partial v}S}{\Gamma_0}\) is positive.
\end{prop}

The outline of the proof is as follows: By rotational symmetry and the \(\mathrm{Cyc}\) move, it is enough to prove for one internal vertex. To show \(\evalat{\frac{\partial}{\partial v}S}{\Gamma_0}\) is positive, it is enough to prove that \(\evalat{\frac{\partial}{\partial v^\omega}S}{\Gamma_0}\) is positive on \(\Gamma_0\) for some \(\omega\) by Observation~\ref{obs:vertex derivative orientation}. Recall that by Corollary~\ref{coro: codim 1 boudnary classification} we have exactly one boundary graph of codimension \(1\) for each vertex in a BCFW graph, and that the top cell of \(\OGnon{3}{6}\) is BCFW. Thus, it is enough to verify that the derivative is indeed positive for one parametrization of the cell set by one orientation. We can show this by a simple direct calculation. The detailed proof of Proposition~\ref{prop:sep k=3} together with this calculation can be found in Appendix~\ref{apx:sep k=3}.

    \begin{prop}
\label{prop:sep external}
    For \(\Gamma\) an OG graph, an internal vertex \(v\) on an external \(4\)-arc \(\tau_\ell\) in \(\Gamma\) with vertex-separator \(S\), and \(\Gamma_0 = \partial_v\Gamma\). Then \(S\) is both positively-radical and zero on \(\Gamma_0\), and \(\evalat{\frac{\partial}{\partial v}S}{\Gamma_0}\) is positive.
\end{prop} 
    \begin{proof} To show \(\evalat{{}\frac{\partial}{\partial v}S}{\Gamma_0}\) is positive, it is enough to prove that \(\evalat{\frac{\partial}{\partial v^\omega}S}{\Gamma_0}\) is positive for some \(\omega\) by Observation~\ref{obs:vertex derivative orientation}. Let us choose \(\omega\) to be the \(\tau_\ell\)-proper orientation. 
    The proof is similar to that of Proposition~\ref{prop:arc sol depends on supp}.  Let \(C\in\Omega_\Gamma\), and \([C,\Lambda,Y]\in \mathcal U_k^\geq\). Recall that \(\Gamma\) must be of the form seen in Figure~\ref{fig:ext arc} with the support of \(\tau_\ell\) being \(I = \{i_1,i_2,i_3,i_4\}\). Using the \(\mathrm{Cyc}\) move and Proposition~\ref{s cyc}, it is enough to show for \(I = \{1,2,3,4\}\) meaning \(\tau_\ell = \{1,4\}\). 

   We can reduce the number of vertices in the graph \(\Gamma\) using the \(\mathrm{Rot}^{-1}\) move according to Proposition~\ref{s rot} until all the internal vertices are on the arc  \(\{1,4\}\), without changing the value of \(S(\Lambda,Y)\) or \(\evalat*{\frac{\partial}{\partial v^\omega}S}{\Gamma_0}(\Lambda,Y)\):
    If there are \(i,i+1\notin I\) such that the arcs \(\tau_i\) and \(\tau_{i+1}\) are crossing, we can use the equivalence move number 3 to replace the graph with an equivalent graph such that the crossing vertex is adjacent to an external vertex. We can now eliminate it using a \(\mathrm{Rot}^{-1}_{i,i+1}\) move without changing the value of the vertex-separator according to Proposition~\ref{s rot} or its derivative according to Observation~\ref{obs:triv diff prom}. We will continue doing so until there are no more such arcs, meaning the only pairs of crossing arcs are ones where one of the arcs is contained in \(I\). The only arc contained in \(I\) is \(\{1,4\}\) as it is external, thus we can assume all the internal vertices are on the arc \(\{1,4\}\) in \(\Gamma\).

    Since  \(\{1,4\}\) is external, there are no arcs contained in \(I\). Meaning that for any \(\{r,\tau(r)\}\), an arc which does not cross \(\{1,4\}\), we have that \(4<r,\tau(r)\). Since all the internal vertices are on \(\{1,4\}\), we must have that \(\{r,\tau(r)\}\) is external with support of size \(2\). This means we have an arc \(\{r,r+1\}\) with \(4<r\), and \(\Gamma = \mathrm{Inc}_r (\Gamma_0)\) for some graph \(\Gamma_0\). We can now apply \(\mathrm{Inc}_r^{-1}\) according to Proposition~\ref{s inc} to reduce \(k\) by \(1\) without changing the value of the vertex-separator or its derivative according to Observation~\ref{obs:triv diff prom}. As \(\{1,4\}\) is external, there are exactly \(2\) arcs that cross it. That means that there are exactly \(k - 3\) arcs that do not cross \(\{1,4\}\). By induction, we can now assume that all arcs cross \(\{1,4\}\) in \(\Gamma\) and that \(k=3\). The claim reduces to Proposition~\ref{prop:sep k=3}. As \(S\in\mathcal F\) does not change throughout the process, it is trivially positively radical as well.
\end{proof}
 \begin{proof}[Proof of Proposition~\ref{prop:sep 4-native}]
    Let \(\Xi\) be an ancestry-sequence of \(v\), and write \(\Xi(\Gamma^\prime) = \Gamma\), \(\Xi(\Gamma_0^\prime) = \Gamma_0\), and \(\Xi(v^\prime) = v\). By Definition~\ref{def:vertex sep}, there exist \(\Xi\) such that \(S = \Xi(S^\prime)\), with \(S^\prime\) being a vertex-separator of \(v^\prime\). By Observation~\ref{obs:arc of codim 1 is codim 1}, we have that \(\partial_{v^\prime} \Gamma^\prime = \Gamma^\prime_0\). By Proposition~\ref{prop:sep external}, \(S^\prime\) is both positively-radical and zero on \(\Gamma_0^\prime\), and \(\evalat{\frac{\partial}{\partial v^\prime}S^\prime}{\Gamma_0^\prime}\) is positive. By Proposition~\ref{prop:com diag}, \(S^\prime\) being positive on \(\Gamma_0^\prime \) means \(S\) is positive on \(\Gamma_0\). By Lemma~\ref{lem:prom radical pos}, \(S^\prime\) being positively-radical on \(\Gamma_0^\prime \) means \(S\) is positively-radical on \(\Gamma_0\).  By Observation~\ref{obs:ver der prom}, we have that \(\evalat{\frac{\partial}{\partial v^\prime}S^\prime}{\Gamma_0^\prime}\) is positive means \(\evalat{\frac{\partial}{\partial v}S}{\Gamma_0}\) is positive. 
\end{proof}
 \subsection{Proving Local Separation}
 \label{sec:proving local separation}
 Having proved the necessary prerequisites, in this section we will finally prove Theorem~\ref{thm:local sep}. Our strategy for doing so will be  applying the following lemma:
 \begin{lem}
    \label{lem:local sep form connectedness}
    Let \(X\) be a smooth manifold, and let \(X_+,X_0,X_-\) be three disjoint smooth submanifolds of dimensions \(n,n-1,n\) respectively, such that \(X_\epsilon \cup X_0\subset X\) is a smooth submanifold with a boundary for \(\epsilon=\pm\). Let \(Y\) be a smooth manifold of dimension \(n\) and \(f:X\rightarrow Y\) a smooth map. Write \(Y_\epsilon\defeq f(X_\epsilon)\) for \(\epsilon\in\{+,0,-\}\). Let \(\mathbf v_\pm\) be continuous sections of the tangent bundle at \(X_0\) pointing in the direction of \(X_\pm\), and write \(\mathbf u_\pm \defeq \mathrm{d}f( \mathbf v_\pm)\). Let \(S_\pm\in C^1(Y)\).
    
    Suppose the following assumptions are satisfied:
    \begin{enumerate}
        \item \(f\) is injective on \(X_\epsilon\) for \(\epsilon\in\{+,0,-\}\).
        \item \(\forall y\in Y_0\), \(S_\epsilon(y)=0\) and \(\mathrm d S_\epsilon(\mathbf u_\epsilon (y))>0\) for \(\epsilon=\pm\).
        \item For each \(y\in Y_0\), \(\frac{\mathrm d S_+(\mathbf v)}{\mathrm d S_-(\mathbf v)}<0\) for all \(\mathbf v\in T_{y}Y\setminus T_y Y_0\).
    \end{enumerate}
    Then there exists open \(U\supset X_0\) such that \(f\) is injective on \(U\cap(X_+\cup X_0\cup X_-)\).
    
\end{lem}
\begin{proof}
    We have that \( X_+ \cup X_0\subset X\) and \(X_- \cup X_0\subset X\) are both smooth submanifolds with a boundary, and \(X_0\subset X\) is itself a smooth submanifold. \(\mathbf v_+\) is a section of the tangent bundle at \(X_0\) pointing towards \(X_+\). By assumption \(2\), we have that \(\mathrm{d}(S_+ \circ f)(\mathbf v_+)>0\) and \(\evalat{S_+}{X_0}=0\). As \(X_0\) is codimension \(1\) in \(X_+\cup X_0\), this implies there exists an open neighborhood \(V\subset X\) of \(X_0\) such that \((S_+ \circ f)(V\cap X_+) >0 \). Now, \(\mathbf v_-\) is a section of the tangent bundle at \(X_0\) pointing towards \(X_-\). We have by assumptions \(2\) that  \(\mathrm{d}(S_- \circ f)(\mathbf v_-)>0\). Therefore, assumption \(3\) implies \(\mathrm{d}(S_+ \circ f)(\mathbf v_-)<0\). Just as before, we can conclude there exists an open neighborhood \(W\subset X\) of \(X_0\) such that \((S_+ \circ f)(W\cap X_-) <0 \).

 Setting \(U \defeq V\cap W\), we have that \(U\) is an open neighborhood of \(X_0\) such that \(S_+\circ f\) is positive on \(U\cap X_+\) and negative on \(U\cap X_-\). Recall that by assumption \(2\), we also have that \(S_+\circ f\) is zero on \(X_0\). This implies \(f(U\cap X_+),f(U\cap X_0),f(U\cap X_-)\) are disjoint. By assumption \(1\), \(f\) is injective on each of those sets. We conclude \(f\) is injective on \(U\cap(X_+\cup X_0\cup X_-)\).
\end{proof}

    Theorems~\ref{thm:ampli diffeo} and~\ref{thm:BCFW and codim 1 bounds are smooth} bring us relatively close to showing the case of Theorem~\ref{thm:local sep} satisfies the assumptions of Lemma~\ref{lem:local sep form connectedness}. Our main obstacles now would be showing it satisfies assumptions \(3\) (\(\forall y\in Y_0\), \(S_\epsilon(y)=0\) and \(\mathrm d S_\epsilon(\mathbf u_\epsilon (y))>0\) for \(\epsilon=\pm\)) and \(4\) (for each \(y\in Y_0\), \(\frac{\mathrm d S_+}{\mathrm d S_-}\) is constant on \(T_{y}Y\setminus T_y Y_0\), and \(\frac{\mathrm d S_+}{\mathrm d S_-}<0\)). We will now address assumption \(3\):
\begin{obs}
\label{obs:triplet sections}
    Let \(\Lambda\in \mathrm{Mat}_{2k\times(k+2)}^>\), and let \(\Delta = (\Gamma_+,\Gamma_0,\Gamma_-)\) be a boundary triplet with boundary vertices \(v_\epsilon\) in \(\Gamma_\epsilon\) for \(\epsilon=\pm\) respectively. For \(\epsilon=\pm\), let \(S_\epsilon\) be vertex-separators of \(v_\epsilon\). Then there exist \(\mathbf v_\epsilon\) continous sections of the tangent bundle at \(\Omega_{\Gamma_0}\) pointing towards the cells \(\Omega_{\Gamma_\pm}\) respectively, such that:
    \begin{enumerate}
        \item \(S_\epsilon\) is positively-radical and zero on \(\Gamma_0\).
        \item \(\mathrm{d}(S_\epsilon\circ \widetilde \Lambda)(\mathbf v_\epsilon(x))>0 \) for any \(x\in \widetilde\Lambda(\Omega_{\Gamma_0})\) .
    \end{enumerate}
\end{obs}
\begin{proof}
    Take \(\omega\) to be the canonical orientation on the graphs. We have that \(\Gamma_0 = \lim_{v_+^\omega\rightarrow 0 }\Gamma_+=\lim_{v_-^\omega\rightarrow 0 }\Gamma_-\) by Observation~\ref{obs:codim 1 boundary canonical 0}. Let \(\varphi_\epsilon\) be the canonical parameterizations of \(\Gamma_\epsilon\) for \(\epsilon \in\{+,0,-\}\), and let \(\alpha_\pm\) be the angles associated to \(v_\pm\) respectively. We have by Observation~\ref{obs: canonical para codim one} that \(\evalat{\varphi_\pm}{v_\pm\mapsto 0} = \varphi_0\), thus \(\alpha_\pm\) define \(\mathbf v_\pm\) continuous sections of the tangent bundle at  \(\Omega_{\Gamma_0}\) that either vanish or point towards the cells  \(\Omega_{\Gamma_\pm}\) respectively.

    By Observation~\ref{obs:BCFW 4 native}, we have both \(v_\pm\) are \(4\)-native vertices. Thus by Proposition~\ref{prop:sep 4-native}, we have that there exist \(S_\pm\) vertex-separators for \(v_\pm\) that are both positively-radical and zero on \(\Gamma_0\), and \(\evalat{\frac{\partial}{\partial v\pm }S_\pm}{\Gamma_0}\) are positive. By Observation~\ref{obs:vertex derivative orientation}, we have that \(\evalat{\frac{\partial}{\partial v_\pm^\omega }S_\pm}{\Gamma_0}\) are positive. By Definition~\ref{def:vertex derivative}, we have that \(\evalat{\frac{\partial}{\partial v_\epsilon^\omega }S_\epsilon}{\Gamma_0}\) is precisely \(\mathrm{d}(S_\epsilon\circ \widetilde\Lambda)(\mathbf v_\epsilon)\) for \(\epsilon =\pm\). Thus \(\mathrm{d}(S_\epsilon\circ \widetilde\Lambda)(\mathbf v_\epsilon(x))>0\) for any for any \(x\in \widetilde\Lambda(\Omega_{\Gamma_0})\), and we must have that \(\mathbf v_\epsilon(x)\) does not vanish.
\end{proof}
We all that is left now is showing assumption \(4\) (for each \(y\in Y_0\), \(\frac{\mathrm d S_+}{\mathrm d S_-}\) is constant on \(T_{y}Y\setminus T_y Y_0\), and \(\frac{\mathrm d S_+}{\mathrm d S_-}<0\)). Recall Definition~\ref{def:Ylk} for the definition of \(Y^k_\Lambda\).  Let us introduce the following notation:
\begin{dfn}
\label{def:Ykg}
    Define 
    \[\mathbf Y^k \defeq \bigcup_{\Lambda\in \mathrm{Mat}^>_{2k\times(k+2)}}  \{\Lambda\}\times \mathbf Y^k_\Lambda,\] and for \(\Gamma\) a \(k\)-OG graph, define  \[\mathbf Y^k_\Gamma \defeq \bigcup_{\Lambda\in \mathrm{Mat}^>_{2k\times(k+2)}}  \{\Lambda\}\times \widetilde \Lambda(\Omega_{\Gamma}).\] 
\end{dfn}
\begin{obs}
    \label{obs:Y connected smooth submanifold}
    \(\mathbf Y^k\) is a smooth manifold, and for \(\Gamma\) a \(k\)-OG graph that is a BCFW graph or a codimension \(1\) boundary of a BCFW graph, we have that \(Y_\Gamma^k\subset Y^k\) is a connected smooth submanifold.
\end{obs}
\begin{proof}
    \(\mathrm{Mat}_{2k\times(k+2)}^>\) is clearly a connected smooth manifold. By Proposition~\ref{prop:zerolocus submfld}, \(\mathbf Y_\Lambda^k\) are smooth manifolds. As \(\mathbf Y_\Lambda^k\) clearly change smoothly with \(\Lambda\),  \(\mathbf Y^k\) is a smooth manifold as well.

    By Theorem~\ref{param bij}, \(\Omega_{\Gamma}\) are connected, and thus so are \( \{\Lambda\}\times \widetilde \Lambda(\Omega_{\Gamma})\). By Theorem~\ref{thm:ampli diffeo}, they are smooth manifolds. By Theorem~\ref{thm:BCFW and codim 1 bounds are smooth}, \(\widetilde \Lambda(\Omega_{\Gamma})\subset \mathbf Y^k_\Lambda\) are smooth submanifolds, thus \(Y_\Gamma^k\subset Y^k\) are smooth connected submanifolds.  
\end{proof}
This construction imitates the universal amplituhedron defined in \cite{universal}. Our strategy would now be to apply Lemma~\ref{lem:sign is const} on \(\mathbf Y_{\Gamma_0}^k\subset \mathbf Y^k\) to show it is enough to show the assumption is satisfied for one \(\Lambda\in \mathrm{Mat}^>_{2k\times(k+2)}\) for each \(k\). We would then use the Lemma and compute a single case for \(k=4\) to prove the assumption holds for all cases of \(k=4\) and \(\Lambda\in \mathrm{Mat}^>_{8\times6}\) in Proposition~\ref{prop:k=4 sep sign}. We then continue and use induction on \(k\) to deduce the assumption holds for all \(\Lambda\in \mathrm{Mat}^>_{2k\times(k+2)}\) with \(k\geq4\) in Proposition~\ref{prop:sep sign}. Let us begin:

\begin{lem}
\label{lem:sign is const}
     Let \(Y\) be a smooth manifold of dimension \(n\), \(S_\epsilon\in C^1(Y)\) for \(\epsilon =\pm\), and \(Y_0\subset Y\) a smooth connected submanifold of \(Y\) of codimension \(1\). Write \(\evalat{TY}{Y_0}\defeq \iota^*TY,\), where $\iota:Y_0\hookrightarrow Y$ is the embedding. Suppose the following assumptions hold:
    \begin{enumerate}
        \item \(\evalat{S_\epsilon}{Y_0}=0\) for \(\epsilon=\pm\).
        \item There exist  \(\mathbf u_\epsilon\in \Gamma^0 (\evalat{TY}{Y_0})\) continuous sections of \(\evalat{TY}{Y_0}\), such that \(\forall y\in Y_0\),\\ \(\mathrm d S_\epsilon(\mathbf u_\epsilon (y))>0\) for \(\epsilon=\pm\).
    \end{enumerate}
    Then we have that 
        \begin{enumerate}
        \item For each \(y\in Y_0\), \(\frac{\mathrm d S_+}{\mathrm d S_-}\) is constant on \(T_{y}Y\setminus T_y Y_0\). 
        \item \(\frac{\mathrm d S_+}{\mathrm d S_-}\) has a constant sign on \((\evalat{TY}{Y_0})\setminus T Y_0\).
    \end{enumerate}
\end{lem}
\begin{proof}
      For the first conclusion, let us fix some \(y\in Y_0\). We claim that \(\frac{\mathrm d S_+(\mathbf v)}{\mathrm d S_-(\mathbf v)}\) is constant on \( T_{y}Y\setminus T_{y}Y_0\) for \(y\in Y_0\). We have that \(Y_0\) is a smooth connected submanifolds of \(Y\) of codimension \(1\). Thus, we have that \(T_{y}Y = V\oplus T_{y}Y_0\), where \(V \defeq \mathrm{Span} \{\mathbf v\}\) for some \(0\neq \mathbf v\in T_y Y\setminus T_y Y_0\). Take  \(\mathbf w\in T_{y}Y\setminus T_{y}Y_0\). We have that \(\mathbf w =a \mathbf  v + b \mathbf u\), with \(\mathbf u\in T_{x_0}Y_0\) and \(a\neq 0\). Since by assumption \(1\), \(\evalat{S_\pm}{Y_0}=0\), we have that \(\mathrm d S_\pm(\mathbf u) = 0\). Therefore \(\frac{\mathrm d S_+(\mathbf v)}{\mathrm d S_-(\mathbf v)}=\frac{\mathrm d S_+(\mathbf w)}{\mathrm d S_-(\mathbf w)}\) by linearity. Meaning that \(\frac{\mathrm d S_+(\mathbf w)}{\mathrm d S_-(\mathbf w)}\) is constant \(\forall \mathbf w\in T_{y}Y\setminus T_{y}Y_0\), proving the first part of the claim. 

    For the second conclusion, write \(\sigma:Y_0\to \mathbb R\) by \(\sigma(y) \defeq \frac{\mathrm d S_+(\mathbf u_+(y))}{\mathrm d S_-(\mathbf u_+(y))}=\frac{\mathrm d S_+(\mathbf u_-(y))}{\mathrm d S_-(\mathbf u_-(y))}\), which is continuous as \(S_\pm\in C^1(Y)\) and \(\mathbf u_\pm(y)\) are continuous. By assumption \(2\), \(\mathrm d S_-(\mathbf u_-(y)) \) does not vanish, and thus \(\sigma\) is well-defined. We claim that the sign of \(\sigma\) is constant on \(Y_0\). If it is not, since \(Y_0\) is connected, \(\sigma\) must vanish for some \(x\in Y_0\). This means \(\mathrm d S_+(\mathbf u_+(y)) = 0\), which is impossible by assumption \(2\). We conclude \(\sigma(y)\) has constant sign on \(Y_0\), and thus that  \(\frac{\mathrm d S_+}{\mathrm d S_-}\) has a constant sign on \(\left(\evalat{TY}{Y_0}\right)\setminus T Y_0\), proving the second part of the claim.
\end{proof}
Now for the base case of the induction:
\begin{prop}
\label{prop:k=4 sep sign}
     Let \(\Delta = (\Gamma_+,\Gamma_0,\Gamma_-)\) be the only boundary triplet of \((k=4)\)-OG graphs, as seen in Figures~\ref{fig:BCFWk4} and~\ref{fig:spider}. Let \(v_\pm\) be the boundary vertex in \(\Gamma_\pm\) respectively. Write  \(\mathbf Y^4_\epsilon \defeq\mathbf Y^4_{\Gamma_\epsilon}\) for \(\epsilon\in\{+,0,-\}\). Then there exist \(S_\pm\) vertex-separators of \(v_\pm\) respectively, such that \(\frac{\mathrm d S_+}{\mathrm d S_-}<0\) on \(T_{y}\mathbf Y^4\setminus T_y \mathbf Y^4_0\) for all \(y\in Y_0\).
\end{prop}
We prove this using a direct calculation. See Appendix~\ref{apx:k=4 sep sign} for the calculation together with the proof of Proposition~\ref{prop:k=4 sep sign}. Now, we would use the base case and show the induction step to prove the claim for general boundary triplets with \(k\geq4\):
\begin{prop}
\label{prop:sep sign}
     Let \(\Delta = (\Gamma_+,\Gamma_0,\Gamma_-)\) be a boundary triplet of \(k\)-OG graphs. Let \(v_\pm\) be the boundary vertex in \(\Gamma_\pm\) respectively. Write  \(\mathbf Y^k_\epsilon \defeq\mathbf Y^k_{\Gamma_\epsilon}\) for \(\epsilon\in\{+,0,-\}\). Then there exist \(S_\pm\) vertex-separators of \(v_\pm\) respectively, such that \(\frac{\mathrm d S_+}{\mathrm d S_-}<0\) on \(T_{y}\mathbf Y^k\setminus T_y \mathbf Y^k_0\) for all \(y\in \mathbf Y^k_0\).
\end{prop}
\begin{proof}
    We will prove this by induction on \(k\), and use Proposition~\ref{prop:k=4 sep sign} as a base case, as there are only boundary triplets for \(k\geq4\).  

    For the induction step, assume the claim is true for all \(k\geq4\). And let \(\Delta^{k+1} = (\Gamma_+^{k+1},\Gamma_0^{k+1},\Gamma_-^{k+1})\) be a boundary triplet of \((k+1)\)-OG graphs. Let \(v_\pm^{k+1}\) be the boundary vertex in \(\Gamma_\pm^{k+1}\) respectively. Recall that by~\ref{obs:Y connected smooth submanifold}, \(\mathbf Y_0^{k+1}\subset\mathbf  Y^{k+1}\) is a smooth connected submanifold for any \(k\ge3\). We want to find \(S_\pm^{k+1}\) vertex-separators of \(v_\pm^{k+1}\) respectively, such that \(\frac{\mathrm d S_+^{k+1}}{\mathrm d S_-^{k+1}}<0\) on \(T_{y}\mathbf Y^{k+1}\setminus T_y \mathbf Y_0^{k+1}\) for all \(y\in \mathbf Y_0^{k+1}\). By Observations~\ref{obs:triplet sections} and~\ref{obs:Y connected smooth submanifold}, Lemma~\ref{lem:sign is const} applies. Thus, it is enough to show for one point \(y\in \mathbf Y^{k+1}\) and one \(\mathbf u \in T_{y}\mathbf Y^{k+1}\setminus T_y \mathbf Y_0^{k+1}\).

    Fix \(\Lambda ^{k+1}\in\mathrm{Mat}_{(2k+2)\times(k+3)}^>\). \(\Delta^{k+1}\) is a boundary triplet of \((k+1)\)-OG graphs. By Observation~\ref{obs:boundary triplet back prom}, we have \(\Delta^{k+1} =\mathrm{Arc}_{4,\ell}(\Delta^k)\), for \(\Delta^k = (\Gamma^k_+,\Gamma_0^k,\Gamma_-^k)\) a boundary triplet of \(k\)-OG graphs with \(v^k_\pm\) the boundary vertex in \(\Gamma^k_\pm\) respectively, and with \(v^{k+1}_\pm=\mathrm{Arc}_{4,\ell}(v^{k}_\pm)\) respectively. 

     By Observation~\ref{obs:BCFW 4 native}, \(v^{k}_\pm\) are \(4\)-native vertices, thus there exist vertex-separators \(S^k_{\pm}\) of \(v^{k}_\pm\) respectively, such that \(\frac{\mathrm d S_+}{\mathrm d S_-}<0\) on \(T_{x}\mathbf Y^k\setminus T_x \mathbf Y_0^k\) for all \(x\in \mathbf Y_0^k\) by the induction hypothesis.

    Let \(\varphi^k(\alpha_1,...,\alpha_n)\) be the canonical parameterization of \(\Gamma_+^{k}\), with \(\alpha_1\) being the angle associated to \(v_+^{k}\). By Observation~\ref{obs: canonical para codim one}, we have that \(\varphi^k(0,\alpha_2,...,\alpha_n)\) is the canonical parameterization for \(\Gamma_0^k\). Write \(\Lambda^k = \mathrm{Inc}_l^{-1} \Lambda^{k+1}\). By Definition~\ref{def:vertex sep}, we have that \(S^{k+1}_\pm = \mathrm{Arc}_{4,\ell} S_\pm^k \)  are vertex-separators for \(v^{k+1}_\pm\) respectively. By Observation~\ref{obs:BCFW 4 native}, \(v^{k+1}_\pm\) are \(4\)-native vertices, thus  we have that \(\evalat{\frac{\partial}{\partial v_\epsilon^{k+1}}S^{k+1}_\epsilon}{\Gamma_0^{k+1}}\) are positive for \(\epsilon =\pm\) by Proposition~\ref{prop:sep 4-native}.

    By Corollary~\ref{coro:graph cell com diag coro}, we have that \(\varphi^{k+1}(\alpha_1,...,\alpha_{n+2})= \mathrm{Arc}_{4,\ell}(\alpha_{n+1},\alpha_{n+2})\varphi^k(\alpha_1,...,\alpha_n)\) is a parameterization for \(\Gamma^{k+1}_+\). Since by Observation~\ref{obs: canonical para codim one}  \(\varphi^k(0,\alpha_2,...,\alpha_n)\) is the canonical parameterization for \(\Gamma_0^k\), we have that \(\varphi^{k+1}(0,\alpha_2,...,\alpha_{n+2})\) is a parameterization for \(\Gamma^{k+1}_0\). Set some positive angles \(\alpha_2,...,\alpha_{n+2}\). Write 
    \(x^k = (\Lambda^k,\widetilde\Lambda^k(\varphi^k(0,\alpha_2,...,\alpha_n))),  \) and 
    \(x^{k+1} = (\Lambda^{k+1},\widetilde\Lambda^{k+1}(\varphi^{k+1}(0,\alpha_2,...,\alpha_n,\alpha_{n+1},\alpha_{n+2}))). \)
    
    Let \(\mathbf v^\epsilon_+\in T_{x^\epsilon} \mathbf Y^\epsilon \setminus T_{x^\epsilon}\mathbf Y^\epsilon\) for \(\epsilon = k,k+1\) be the directions defined increasing \(\alpha_1\). We will show they do not vanish. Since both \(\evalat{\frac{\partial}{\partial v_+^{\epsilon}}S^{\epsilon}_+}{\Gamma_0^{\epsilon}}\) are positive, we know both \(\evalat{\frac{\partial}{\partial (v_+^{\epsilon})^\omega}S^{\epsilon}_+}{\Gamma_0^{\epsilon}}\) are positive for any perfect orientation \(\omega\) on \(v_+^{\epsilon}\) by Definition~\ref{def:vert der sans ori}. By Definition~\ref{def:vertex derivative}, \(\evalat{\frac{\partial}{\partial (v_+^{\epsilon})^\omega}S^{\epsilon}_+}{\Gamma_0^{\epsilon}}=\mathrm{d}S_+^{\epsilon}(\mathbf v _+^{\epsilon})\) for some orientation \(\omega\). Thus \(\mathrm{d}S_+^{\epsilon}(\mathbf v ^{\epsilon}_+)\) are positive for both \(\epsilon =k,k+1\), and \(\mathbf v ^{\epsilon}_+\) do not vanish.\
    
    By Observation~\ref{obs:triv diff prom}, we have that \(\mathrm{d}S_+^{k+1}(\mathbf v ^{k+1})=\frac{\partial}{\partial\alpha_1}S_+^{k+1} (x^{k+1})=\frac{\partial}{\partial\alpha_1}S_+^{k} (x^{k})=\mathrm{d}S_+^{k}(\mathbf v ^{k}),\)
    and in the same way we have that \(\mathrm{d}S_-^{k+1}(\mathbf v ^{k+1}) = \mathrm{d}S_-^{k}(\mathbf v ^k)\). Meaning that \(\frac{\mathrm d S_+^{k+1}(\mathbf v ^{k+1})}{\mathrm d S_-^{k+1}(\mathbf v ^{k+1})}= \frac{\mathrm d S_+^{k}(\mathbf v ^{k})}{\mathrm d S_-^{k}(\mathbf v ^{k})}<0\). By Lemma~\ref{lem:sign is const}, we have that \(\frac{\mathrm d S_+^{k+1}}{\mathrm d S_-^{k+1}}<0\) on \(T_{y}\mathbf Y^{k+1}\setminus T_y \mathbf Y_0^{k+1}\) for all \(y\in \mathbf Y_0^{k+1}\).
\end{proof}
We can now deduce the claim holds for each specific \(\Lambda\in\mathrm{Mat}_{2k\times(k+2)}^>\), which we need in order to prove Theorem~\ref{thm:local sep}.
\begin{coro}
\label{coro:sep sign}
     Let \(\Delta = (\Gamma_+,\Gamma_0,\Gamma_-)\) be a boundary triplet of \(k\)-OG graphs, and \(\Lambda\in\mathrm{Mat}_{2k\times(k+2)}^>\). Let \(v_\pm\) be the boundary vertex in \(\Gamma_\pm\) respectively. Write \(M \defeq  \widetilde\Lambda(\Omega_{\Gamma_0})\). Then there exist \(S_\pm\) vertex-separators of \(v_\pm\) respectively, such that \(\frac{\mathrm d S_+}{\mathrm d S_-}<0\) on \(T_{y} \mathbf Y^k_\Lambda \setminus T_y M \) for all \(y\in  M\).
\end{coro}
\begin{proof}
    That is immediate from Proposition~\ref{prop:sep sign} and Definition~\ref{def:Ylk} and~\ref{def:Ykg} by holding \(\Lambda\) constant.
\end{proof}
    \begin{proof}[Proof of Theorem~\ref{thm:local sep}]
        We will show Lemma~\ref{lem:local sep form connectedness} applies.

        We have that \(f=\widetilde \Lambda\) is a smooth map that is defined in an open neighborhood \(X \supset \OGnon{k}{2k}\) in  \(\OG{k}{2k}\) by Proposition~\ref{prop:extends to a nbhd}. By Theorem~\ref{thm:BCFW and codim 1 bounds are smooth}, we have that \(X_+ = \Omega_{\Gamma_+},X_0 = \Omega_{\Gamma_0},X_- = \Omega_{\Gamma_-}\) are smooth submanifolds, with \(X_0\) in the boundary of both \(X_\pm\). By Theorem~\ref{param bij}, \(X_+,X_0,X_-\) are connected and of dimensions \(2k-3,2k-4,2k-3\) respectively. By Proposition~\ref{prop:limits commute with para}, both \(X_\pm \cup X_0\) are smooth manifolds with boundary. \(Y = \widetilde\Lambda(X)\) is of dimension \(n=2k-3\) by Theorem~\ref{dimension thm}.

        By Theorem~\ref{thm:ampli diffeo}, \(\evalat{f}{X_\epsilon}\) is a diffeomorphism onto its image for \(\epsilon\in\{+,0,-\}\), thus injective, hence Assumption \(1\) of  Lemma~\ref{lem:local sep form connectedness} holds.

        By Corollary~\ref{coro:sep sign}, there exist \(S_\pm\) vertex-separators of \(v_\pm\) respectively, such that \(\frac{\mathrm d S_+}{\mathrm d S_-}<0\) on \(T_{y}Y\setminus T_y Y_0\) for all \(y\in Y_0\), implying assumption \(3\) of Lemma~\ref{lem:local sep form connectedness}.
        
        By Observation~\ref{obs:triplet sections}, we have that there exist \(\mathbf v_\epsilon\in \Gamma^0 (\evalat{TY}{Y_0})\), continuous sections of the tangent bundle at \(\Omega_{\Gamma_0}\) pointing towards the cells \(\Omega_{\Gamma_\pm}\) respectively, such that:
    \begin{enumerate}
        \item \(S_\epsilon\) is positively-radical and zero on \(\Gamma_0\).
        \item \(\mathrm{d}(S_\epsilon\circ \widetilde \Lambda)(\mathbf v_\epsilon(x))>0 \) for any \(x\in \widetilde\Lambda(\Omega_{\Gamma_0})\) .
    \end{enumerate}
        By Observation~\ref{obs:pos rad is smooth}, \(S_\epsilon\) being positively radical on \(\Gamma_0\) means they are smooth in a neighborhood containing \(\widetilde \Lambda(\Omega_{\Gamma_0})\). Meaning that  \(S_\epsilon\) satisfy assumption \(2\) of  Lemma~\ref{lem:local sep form connectedness}. Therefore Lemma~\ref{lem:local sep form connectedness} applies. We conclude that there is an open neighborhood \(U\supset \Omega_{\Gamma_0}\) in \(\OG{k}{2k}\), such that \(U\cap \Omega_\Delta\) is mapped injectively by the amplituhedron map \(\widetilde \Lambda\).
        \end{proof}

\section{Strong Positivity and Non-Negative 
Mandelstam Variables}\label{sec:immanant_pos}

In this section, we find a set of matrices \(\Lambda\in \mathrm{Mat}^>_{k\times2k}\) such that all Mandelstam variables \(S_I\), for cyclically consecutive subsets $I\subseteq[2k]$, are non-negative on \(\mathcal{O}_k(\Lambda)\) (see Theorem \ref{nonneg mandelstam thm}). We call these matrices \emph{strongly positive matrices} and denote them \(\mathcal L^>_k\) (see Definition~\ref{def:strongly positive} for the precise definition). Our approach is inspired by a similar result proved by Galashin \cite{origami} for the momentum amplituhedron, defined by Damgaard, Ferro, Lukowski, and Parisi in \cite{momAmpli}. 

\begin{thm}
\label{nonneg mandelstam thm}
    For \(Y\in \mathcal O_k(\widetilde\Lambda)\) with \(\Lambda\in \mathcal L^>_k\), and \(I\subset[2k]\) with \(1<\left| I\right|<2k-1\) such that \(I\) is cyclically consecutive, we have that \(S_I(\Lambda, Y)\geq0\).
\end{thm}
The full proof of Theorem~\ref{nonneg mandelstam thm} can be found in Section~\ref{sec:Temperley–Lieb Immanants}.
This theorem brings us close to the main result of this section, which is the following:
\begin{thm}
\label{thm external boundaries}
    For \(\Lambda\in\mathcal{L}_k^>\) and \(\Gamma_0\) an external boundary of BCFW cell \(\Gamma\) in \(\OGnon{k}{2k}\), we have that \(\widetilde \Lambda(\Omega_{\Gamma_0})\) is contained in the boundary of the amplituhedron \(\mathcal{O}_k(\Lambda)\).
\end{thm}
The proof of Theorem~\ref{thm external boundaries} can be found in Section~\ref{sec:ext boundaries Pos}.

As well as proving the above claims, we have two more goals for this section. In Section~\ref{sec:strongly pos mat}, we will show the set of strongly positive matrices behaves nicely under the local moves we used throughout the previous sections. In Section~\ref{sec:finding strongly pos mat}, we would show the set of strongly positive matrices is not empty.

For the rest of this section, in the spirit of Proposition~\ref{prop:mom cons}, keep in mind that we can view \(\mathcal{O}_k (\Lambda) \) as the image of \(\OGnon{k}{2k}\) under the map \(C \mapsto C^\perp\cap \Lambda^\intercal \in \mathrm{Gr}_2 (\Lambda^\intercal)\), where \(\mathrm{Gr}_2 (\Lambda^\intercal)\) is the space of two dimension linear subspaces of \(\Lambda^\intercal\). 

\begin{lem}[\cite{origami}, Equation~2.16]
    \label{alt ortho}
    For \(I\in\binom{[n]}{k}\) and \(C\in \Grnon{k}{n}\), we have
    \(
    \Delta_{[n]\setminus I}(C^\perp) = \Delta_I(C\eta).
    \)
\end{lem}
\begin{dfn}\label{def:mom_amp}
Let \(M^{\perp > }_{n\times k}\) be the space of \(n \times k\) real full rank \(k\) matrices which are orthogonal to a positive \(n \times (n-k)\) matrix.
For \((\Lambda,\widetilde\Lambda)\in M^{\perp >}_{n\times (n-k+2)}\times M^{ > }_{n\times (k+2)}\), the \emph{momentum amplituhedron} \(\mathcal{M}_{k,n}(\Lambda,\widetilde\Lambda)\) is defined as the image of \(\Grnon{k}{n}\) under the map\\
\(
\Phi_{\Lambda,\widetilde\Lambda}:\Grnon{k}{n}\rightarrow \Gr{2}{n}\times\Gr{2}{n}
\) defined by
\(C\mapsto ({\underline{\lambda}},\widetilde{{\underline{\lambda}}}):=(C\cap\Lambda^\intercal,C^\perp\cap\widetilde\Lambda^\intercal).
\)
\({{\underline{\lambda}}}\) and \(\widetilde{{\underline{\lambda}}}\) are always two dimensional, and thus \(({\underline{\lambda}},\widetilde{\underline{\lambda}})\in \Gr{2}{n}\times\Gr{2}{n}\). 

The Mandelstam variables on \(\Gr{2}{n}\times\Gr{2}{n}\) are defined as
\(
S_I({\underline{\lambda}},\widetilde{\underline{\lambda}}) = \sum_{\{i,j\}\subset I} \langle i \, j \rangle_{\underline{\lambda}} \left[i\,j\right]_{\widetilde{\underline{\lambda}}}
\)
for \(I\subset[n]\), where \(\langle i \, j \rangle_{\underline{\lambda}} \defeq \Delta_{\{i,j\}}({\underline{\lambda}})\) and \(\left[i\,j\right]_{\widetilde{\underline{\lambda}}} \defeq \Delta_{\{i,j\}}(\widetilde{\underline{\lambda}})\).    
\end{dfn}
Note that the definition of the ABJM amplituhedron is closely analogous: by setting \(n=2k\), \(\eta \Lambda = \widetilde{\Lambda}\), and changing the domain of the map in Definition~\ref{def:mom_amp} to \(\OGnon{k}{2k}\) instead of \(\Grnon{k}{2k}\), the resulting space is naturally isomorphic to the ABJM amplituhedron. Moreover, the \(\lambda\) defined in Definition~\ref{def lambda} coincides with \({\underline{\lambda}} \eta = \widetilde{\underline{\lambda}}\) of the momentum amplituhedron, and the Mandelstam variables restrict to those discussed in Definition~\ref{def mand}.

\subsection{Temperley–Lieb Immanants}
\label{sec:Temperley–Lieb Immanants}
The goal of this section is to prove Theorem~\ref{nonneg mandelstam thm}. To do so, we need to study more closely the Mandlestam variables and consider the Temperley–Lieb immanants defined by Lam \cite{Lam} following Rhoades and Skandera \cite{TLImma}. See also the treatment in \cite{origami}.
\begin{dfn}\label{def:temp_lieb}
Let \(\tau\) be an involution on \([n]\), that is, a permutation such that \(\tau^2 = \text{id}\). Let \(S(\tau) \defeq \{\ell\in [n]  : \tau(\ell) \neq \ell\}\) and let \(T \subset [n]\setminus S(\tau)\) such that \(2 \left|T\right|+\left|S(\tau)\right| = 2k \). We say that \((\tau,T)\) is a \emph{\((k,n)\)-partial non-crossing pairing} if there are no indices \(1\leq a < b < c < d \leq n\) such that \(\tau(a) = c\) and \(\tau(b)=d\). Let \(\mathcal{T}_{k,n}\) be the set of \((k,n)\)-partial non-crossing pairings. 

For \(A,B\in\binom{[n]}{k}\), we say \((\tau,T)\in \mathcal T_{k,n}\) is \emph{compatible} with \((A,B)\) if \(T = A\cap B\), \(S(\tau)=(A\setminus B)\cup (B\setminus A)\), and \(\tau(A\setminus B) = B\setminus A\). Write \(\mathcal T_{k,n}(A,B)\) for the set of \((\tau,T)\in \mathcal T_{k,n}\) compatible with \((A,B)\)

The \emph{Temperley--Lieb immanants} of \(C\in \Gr{k}{n}\) are the set of functions \(\Delta_{\tau,T}(C)\) for all\\ \((\tau,T)\in \mathcal T_{k,n}\), uniquely defined by the equations
\(
\Delta_A(C)\Delta_B (C)  = \sum_{(\tau,T)\in \mathcal T_{k,n}(A,B)} \Delta_{\tau,T}(C)
\) for all \( A,B\in \binom{[n]}{k}\). They are non-negative for \(C\in\Grnon{k}{n}\).    
\end{dfn}
Write \(I(i,j)\) for the elements \(\{i+1,i+2,...,j\}\) taken mod \(n\). A pair \(\{\ell,\tau(\ell)\}\) with \(l\in[n]\) and \(l\neq \tau(\ell)\) is called an \emph{arc} of \(\tau\). For \(\{\ell,\tau(\ell)\}\) an arc is an \emph{I-special arc} if \(\left|\{\ell,\tau(\ell)\}\cap I \right| = 1\). Let an \((i,j,\tau,T)\)-marking be a function \(\mu: S(\tau) \rightarrow \{\ell,R,J\}\) that satisfies:
\begin{itemize}
    \item \(\tau\) has exactly two \(I(i,j)\)-special arcs \(\{\ell, \tau (\ell)\}\) and \(\{r, \tau (r)\}\) with \(\mu (\ell) = \mu (r) = J\) and \(\mu (\tau(\ell)) = \mu (\tau(r)) = R\). We will refer to those as \emph{\(J\)-arcs}.
    \item For all other arcs of \(\tau\), we have one endpoint being sent to \(L\) and one to \(R\).
\end{itemize}
Let \(\mathcal{M}_{(i,j,\tau,T)}\) be the set of all \((i,j,\tau,T)\)-markings. Let \(d^{i,j}_\mu\) be the number of \(I(i,j)\)-special arcs of \(\tau\) between \(\{\ell, \tau (\ell)\}\) and \(\{r, \tau (r)\}\). Let \(R_\mu \defeq T\cup \mu^{-1}(R)\), \(L_\mu \defeq T\cup \mu^{-1}(L)\), and \(J_\mu \defeq \mu^{-1}(J)\).

\begin{thm}[\cite{origami}, Theorem~6.2]
\label{s immanant thm}
    For \(C\in\Grnon{k}{n}\), \((\Lambda,\widetilde\Lambda)\in M^{\perp > }_{n\times (n-k+2)}\times M^{ > }_{n\times (k+2)}\), and \(\Phi_{\Lambda,\widetilde\Lambda} (C) = ({\underline{\lambda}},\widetilde{\underline{\lambda}})\), we have
    \(
    S_{I(i,j)}({\underline{\lambda}},\widetilde{\underline{\lambda}}) = \sum_{(\tau,T)\in \mathcal T_{k,n}} c_{\tau,T}^{i,j}(\Lambda,\widetilde\Lambda)\Delta_{\tau,T}(C)
    ,\)
    where
    \begin{equation}
        \label{cijtT}
    c_{\tau,T}^{i,j}(\Lambda,\widetilde\Lambda)=\sum_{\mu\in\mathcal{M}_{(i,j,\tau,T)}}(-1)^{d^{i,j}_\mu} \Delta_{L_\mu}(\Lambda^{\perp \intercal})\Delta_{R_\mu\cup J_\mu}(\widetilde\Lambda^\intercal).
    \end{equation}
    
\end{thm}
Thus, if every \( c_{\tau,T}^{i,j}(\Lambda,\widetilde\Lambda)\geq 0\) then all the Mandelstam variables are  non-negative on the amplituhedron. 
\begin{dfn}
\label{def:strongly positive}
    Let \(\mathcal{T}_{k,n}^{i,j}\) be the set of \((\tau,T)\in \mathcal{T}_{k,n}\) such that \(\tau\) has at least two \(I(i,j)\)-special arcs. \((\Lambda,\widetilde\Lambda)\in M^{\perp > }_{n\times (n-k+2)}\times M^{ > }_{n\times (k+2)}\) would be called \emph{strongly positive} if \( c_{\tau,T}^{i,j}(\Lambda,\widetilde\Lambda)> 0\) for any \(i+2\leq j\leq i+n-2\) and \((\tau,T)\in \mathcal{T}_{k,n}^{i,j}\) (we will call those \emph{non-trivial} \(c_{\tau,T}^{i,j}\)). Let \(\mathcal{L}_k^>\), which we will call \emph{strongly positive matrices}, be the space of \(\widetilde\Lambda\) such that \((\eta \widetilde\Lambda, \widetilde\Lambda)\in M^{\perp > }_{2k\times (k+2)}\times M^{ > }_{2k\times (k+2)}\) are immanant positive. 
\end{dfn}
From here until the end of the section we shall restrict to the case \(\eta\Lambda = \widetilde\Lambda\) and \(n = 2k\), and we will be interested in finding conditions on $\Lambda$ making the pair $(\Lambda,\widetilde\Lambda)$ immanant positive.

\begin{proof}[Proof of Theorem~\ref{nonneg mandelstam thm}]
    Set \(C\in \OGnon{k}{2k}\) such that \(C\widetilde\Lambda = Y\). It holds that \(C\eta = C^\perp\), thus
    \[
    ({\underline{\lambda}},\widetilde{{\underline{\lambda}}}):=(C\cap\Lambda^\intercal,C^\perp\cap\widetilde\Lambda^\intercal) =
    ((C\eta)^\perp\cap(\eta\widetilde\Lambda)^\intercal,C^\perp\cap\widetilde\Lambda^\intercal),
    \]
    and therefore
    \(\lambda = \underline{\widetilde\lambda}=\underline{\lambda}\eta\). This means that 
    \(
    (-1)^{i-j+1}\langle i \, j \rangle_{\underline{\lambda}} = \left[i\,j\right]_{\widetilde{\underline{\lambda}}} = \Delta_{\{i,j\}}(\lambda)=\twist{Y}{i}{j}_{\widetilde\Lambda},
    \)
    by Proposition~\ref{prop:twist pluckers}. Therefore,
   \[
    S_I(\widetilde\Lambda, Y)=
    \sum_{\{i,\,j\} \subset I} (-1)^{i-j+1} \twist{Y}{i}{j}_{\widetilde\Lambda}^2
    =\sum_{\{i,j\}\in I} \langle i \, j \rangle_{\underline{\lambda}} \left[i\,j\right]_{\widetilde{\underline{\lambda}}}
    =S_I({\underline{\lambda}},\widetilde{\underline{\lambda}})
\geq
    0
    \] 
\end{proof}

\subsection{Strongly Positive Matrices}
\label{sec:strongly pos mat}
We will now show that \(\widetilde\Lambda\in \mathcal L^>_k\) is preserved by \(\mathrm{Rot}\), \(\mathrm{Inc}\), and \(\mathrm{Cyc}\) moves.
    \begin{prop}
    \label{imm pos closed cyc prop}
        For \(\widetilde\Lambda\in\mathcal{L}^>_k\), we have that \(\mathrm{Cyc}(\widetilde\Lambda)\in\mathcal{L}^>_k\).
    \end{prop}
    \begin{proof}
    Write \(\mathrm{Cyc}(\widetilde{\Lambda}) = \widetilde\Lambda^\prime\). We have for  \((\tau,T)\in \mathcal T_{k,2k}^{i,j}\), by Lemma~\ref{alt ortho}

        \begin{align*}
             c_{\tau,T}^{i,j}(\eta \widetilde\Lambda^\prime,\widetilde\Lambda^\prime)\,&=\sum_{\mu\in\mathcal{M}_{(i,j,\tau,T)}}(-1)^{d^{i,j}_\mu} \Delta_{L_\mu}((\eta\widetilde\Lambda^\prime)^{\perp \intercal})\Delta_{R_\mu\cup J_\mu}(\widetilde\Lambda^{\prime \intercal})\\
             &= \sum_{\mu\in\mathcal{M}_{(i,j,\tau,T)}}(-1)^{d^{i,j}_\mu} \Delta_{L_\mu^c}(\widetilde\Lambda^{\prime \intercal})\Delta_{R_\mu\cup J_\mu}(\widetilde\Lambda^{\prime \intercal})\\
              &= \sum_{\mu\in\mathcal{M}_{(i,j,\tau,T)}}(-1)^{d^{i,j}_\mu} \Delta_{L_\mu^c+1}(\widetilde\Lambda^{ \intercal})\Delta_{R_\mu\cup J_\mu+1}(\widetilde\Lambda^{ \intercal})
        \end{align*}
        by Corollary~\ref{cyc plu} where similarly we define adding 1 to an index mod \(2k\).

        Define \(\tau^\prime(\ell+1) = \tau(\ell)+1\), \(T^\prime = T+1\), \(i^\prime = i+1\), \(j^\prime = j+1\). We have  \((\tau^\prime,T^\prime)\in\mathcal T^{i,^\prime,j^\prime}_{k,2k}\) 
        as \((\tau,T)\in \mathcal T_{k,2k}^{i,j}\). For \(\mu\in\mathcal{M}_{(i,j,\tau,T)}\) define \(\mu^\prime (\ell+1) = \mu(\ell)\) and we have 
        \(\mu^\prime\in\mathcal{M}_{(i^\prime,j^\prime,\tau^\prime,T^\prime)}\). 
        Since change all of the arcs together by cycling the indices, we have that \(d^{i^\prime,j^\prime}_{\mu^\prime} = d^{i,j}_\mu\). For a set of indices \(A\subset[2k]\) write \(A+1\) for the set resulting form adding one to each index mod \(2k\). Since we just moved all the indices by 1, we have \(L_\mu^c+1 = L_{\mu^\prime}\) and \(R_\mu\cup J_\mu+1 = R_{\mu^\prime}\cup J_{\mu^\prime}\). Therefore
        \begin{align*}
             c_{\tau,T}^{i,j}(\eta \widetilde\Lambda^\prime,\widetilde\Lambda^\prime)\,&= \sum_{\mu\in\mathcal{M}_{(i,j,\tau,T)}}(-1)^{d^{i,j}_\mu} \Delta_{L_\mu^c+1}(\widetilde\Lambda^{ \intercal})\Delta_{R_\mu\cup J_\mu+1}(\widetilde\Lambda^{ \intercal})\\
             &= \sum_{\mu\in\mathcal{M}_{(i,j,\tau,T)}}(-1)^{d^{i^\prime,j^\prime}_{\mu^\prime}} \Delta_{L_{\mu^\prime}^c}(\widetilde\Lambda^{ \intercal})\Delta_{R_{\mu^\prime}\cup J_{\mu^\prime}}(\widetilde\Lambda^{ \intercal})\\
             &= \sum_{\mu^\prime\in\mathcal{M}_{(i^\prime,j^\prime,\tau^\prime,T^\prime)}}(-1)^{d^{i^\prime,j^\prime}_{\mu^\prime}} \Delta_{L_{\mu^\prime}^c}(\widetilde\Lambda^{ \intercal})\Delta_{R_{\mu^\prime}\cup J_{\mu^\prime}}(\widetilde\Lambda^{ \intercal})\\
             &=\,c_{\tau^\prime,T^\prime}^{i^\prime,j^\prime}(\eta \widetilde\Lambda,\widetilde\Lambda).
        \end{align*}
        Thus \(\widetilde\Lambda^\prime\) is strongly positive if \(\widetilde\Lambda\) is.   
    \end{proof}
    \begin{prop}\label{prop:imm_under_inv_inj}
        For \(\widetilde\Lambda\in\mathcal{L}^>_{k+1}\), we have that \(\mathrm{Inc}^{-1}_{s}(\widetilde\Lambda)\in\mathcal{L}^>_{k}\).
    \end{prop}
    \begin{proof}
        Enough to show for \(s=2k-1\). Write \(\mathrm{Inc}^{-1}_{s}(\widetilde\Lambda) =\widetilde\Lambda^\prime\). We have by Lemma~\ref{alt ortho}
 \begin{align*}
             c_{\tau,T}^{i,j}(\eta \widetilde\Lambda^\prime,\widetilde\Lambda^\prime)=&\sum_{\mu\in\mathcal{M}_{(i,j,\tau,T)}}(-1)^{d^{i,j}_\mu} \Delta_{L_\mu}((\eta\widetilde\Lambda^\prime)^{\perp \intercal})\Delta_{R_\mu\cup J_\mu}(\widetilde\Lambda^{\prime \intercal})\\
             =& \sum_{\mu\in\mathcal{M}_{(i,j,\tau,T)}}(-1)^{d^{i,j}_\mu} \Delta_{L_\mu^c}(\widetilde\Lambda^{\prime \intercal})\Delta_{R_\mu\cup J_\mu}(\widetilde\Lambda^{\prime \intercal}).
        \end{align*}
        Thus by Proposition~\ref{prop:inc plu},
        \begin{align*}
             c_{\tau,T}^{i,j}(\eta \widetilde\Lambda^\prime,\widetilde\Lambda^\prime)
             =& \sum_{\mu\in\mathcal{M}_{(i,j,\tau,T)}}(-1)^{d^{i,j}_\mu} \Delta_{L_\mu^c\cup \{s\}}(\widetilde\Lambda^{\intercal})\Delta_{R_\mu\cup J_\mu\cup\{s\}}(\widetilde\Lambda^{\intercal})\\
             &+\sum_{\mu\in\mathcal{M}_{(i,j,\tau,T)}}(-1)^{d^{i,j}_\mu} \Delta_{L_\mu^c\cup \{s+1\}}(\widetilde\Lambda^{\intercal})\Delta_{R_\mu\cup J_\mu\cup\{s\}}(\widetilde\Lambda^{\intercal})\\
             &+\sum_{\mu\in\mathcal{M}_{(i,j,\tau,T)}}(-1)^{d^{i,j}_\mu} \Delta_{L_\mu^c\cup \{s\}}(\widetilde\Lambda^{\intercal})\Delta_{R_\mu\cup J_\mu\cup\{s+1\}}(\widetilde\Lambda^{\intercal})\\
             &+\sum_{\mu\in\mathcal{M}_{(i,j,\tau,T)}}(-1)^{d^{i,j}_\mu} \Delta_{L_\mu^c\cup \{s+1\}}(\widetilde\Lambda^{\intercal})\Delta_{R_\mu\cup J_\mu\cup\{s+1\}}(\widetilde\Lambda^{\intercal}).
        \end{align*}
        First consider 
        \(
        \sum_{\mu\in\mathcal{M}_{(i,j,\tau,T)}}(-1)^{d^{i,j}_\mu} \Delta_{L_\mu^c\cup \{s+1\}}(\widetilde\Lambda^{\intercal})\Delta_{R_\mu\cup J_\mu\cup\{s\}}(\widetilde\Lambda^{\intercal})
        \).\\ Define \(\tau_1(s) = s\), \(\tau_1(s+1) = s+1\) and  \(\tau_1(\ell) = \tau(\ell)\) otherwise. Then \(S(\tau_1) =S(\tau)\). Set \(T_1 = T\cup \{s\}\). We have that \(\mathcal{M}_{(i,j,\tau,T)} = \mathcal{M}_{(i,j,\tau_1,T_1)} \). Clearly we have \((\tau_1,T_1)\in\mathcal T_{k+1,2k+2}^{i,j}\). Notice that since we added \(s\) to \(T\), we added \(s\) to both \(L_\mu\) and \(R_\mu\cup J_\mu\), thus
        \begin{align*}
             &\sum_{\mu\in\mathcal{M}_{(i,j,\tau,T)}}(-1)^{d^{i,j}_\mu} \Delta_{L_\mu^c\cup \{s+1\}}(\widetilde\Lambda^{\intercal})\Delta_{R_\mu\cup J_\mu\cup\{s\}}(\widetilde\Lambda^{\intercal})\\
             =&\sum_{\mu\in\mathcal{M}_{(i,j,\tau_1,T_1)}}(-1)^{d^{i,j}_\mu} \Delta_{L_\mu^c}(\widetilde\Lambda^{\intercal})\Delta_{R_\mu\cup J_\mu}(\widetilde\Lambda^{\intercal})\\
             =&\,c_{\tau_1,T_1}^{i,j}(\eta \widetilde\Lambda,\widetilde\Lambda).
        \end{align*}
        Similarly, by  defining \((t_2,T_2)\) with \(s+1\) instead of \(s\), we get 
         \begin{align*}
             &\sum_{\mu\in\mathcal{M}_{(i,j,\tau,T)}}(-1)^{d^{i,j}_\mu} \Delta_{L_\mu^c\cup \{s\}}(\widetilde\Lambda^{\intercal})\Delta_{R_\mu\cup J_\mu\cup\{s+1\}}(\widetilde\Lambda^{\intercal})\\
             =&\sum_{\mu\in\mathcal{M}_{(i,j,\tau_2,T_2)}}(-1)^{d^{i,j}_\mu} \Delta_{L_\mu^c}(\widetilde\Lambda^{\intercal})\Delta_{R_\mu\cup J_\mu}(\widetilde\Lambda^{\intercal})\\
             =&\,c_{\tau_2,T_2}^{i,j}(\eta \widetilde\Lambda,\widetilde\Lambda).
        \end{align*}
Finally, define \(\tau_3(s) =s+1\), \(\tau_3(s+1) = s\), and \(\tau_3(\ell) = \tau(\ell)\) otherwise, and \(T_3 = T\). We still have that \((\tau_3,T_3)\in \mathcal T_{k+1,2k+2}^{i,j}\) as the new arc crosses none of the old arcs. As \(i,j\leq2k-2\), \(\{s,s+1\}\) is not a special arc. Thus for any \(\mu\in\mathcal{M}_{i,j,\tau_3,T_3}\) we either have \(\mu(s) = L\) and \(\mu(s+1) = R\), or \(\mu(s) = R\) and \(\mu(s+1) = L\). Thus,

        \begin{align*}
        c_{\tau_3,T_3}^{i,j}(\eta \widetilde\Lambda,\widetilde\Lambda)=&
        \sum_{\mu\in\mathcal{M}_{(i,j,\tau_3,T_3)}}(-1)^{d^{i,j}_\mu} \Delta_{L_\mu^c}(\widetilde\Lambda^{\intercal})\Delta_{R_\mu\cup J_\mu}(\widetilde\Lambda^{\intercal})\\
        =&\sum_{\mu\in\mathcal{M}_{(i,j,\tau_3,T_3)},\,\,\mu(s)=L}(-1)^{d^{i,j}_\mu} \Delta_{L_\mu^c}(\widetilde\Lambda^{\intercal})\Delta_{R_\mu\cup J_\mu}(\widetilde\Lambda^{\intercal})\\
        &+\sum_{\mu\in\mathcal{M}_{(i,j,\tau_3,T_3)},\,\,\mu(s)=R}(-1)^{d^{i,j}_\mu} \Delta_{L_\mu^c}(\widetilde\Lambda^{\intercal})\Delta_{R_\mu\cup J_\mu}(\widetilde\Lambda^{\intercal}) \\
        =&\sum_{\mu\in\mathcal{M}_{(i,j,\tau,T)}}(-1)^{d^{i,j}_\mu} \Delta_{L_\mu^c\cup\{s+1\}}(\widetilde\Lambda^{\intercal})\Delta_{R_\mu\cup J_\mu\cup\{s+1\}}(\widetilde\Lambda^{\intercal})\\
        &+\sum_{\mu\in\mathcal{M}_{(i,j,\tau,T)}}(-1)^{d^{i,j}_\mu} \Delta_{L_\mu^c\cup\{s\}}(\widetilde\Lambda^{\intercal})\Delta_{R_\mu\cup J_\mu\cup\{s\}}(\widetilde\Lambda^{\intercal}).
        \end{align*}
        Combining the above we get
        \(
        c_{\tau,T}^{i,j}(\eta \widetilde\Lambda^\prime,\widetilde\Lambda^\prime) = c_{\tau_1,T_1}^{i,j}(\eta \widetilde\Lambda,\widetilde\Lambda)+c_{\tau_2,T_2}^{i,j}(\eta \widetilde\Lambda,\widetilde\Lambda)+c_{\tau_3,T_3}^{i,j}(\eta \widetilde\Lambda,\widetilde\Lambda).
        \)
        Therefore, if \(\widetilde\Lambda\) is strongly positive then so is \(\widetilde\Lambda^\prime\).
    \end{proof}

Recall for \([C,\widetilde\Lambda,Y]\in \mathcal{U}_k\), we defined
    \(
    \mathrm{Rot}_{i,i+1}^{-1} (\alpha) [C,\widetilde\Lambda,Y] = [\mathrm{Rot}_{i,i+1}^{-1}(\alpha) C,\mathrm{Rot}_{i,i+1}(\alpha)^{-1}\widetilde\Lambda,Y],
    \)
that is \( \mathrm{Rot}_{i,i+1}^{-1}(\widetilde\Lambda)= R_{i,i+1}(\alpha)\widetilde\Lambda\), with \(R_{i,i+1}(\alpha)\) being the hyperbolic rotation as defined in Definition~\ref{def:rot}. In Section~\ref{sec:finding strongly pos mat} we will also show that
\begin{prop}
\label{imm pos rot prop}
For \(\widetilde\Lambda\in\mathcal{L}^>_k\) and \(\alpha>0\), we have that \(\mathrm{Rot}_{i,i+1}(\alpha)^{-1}(\widetilde\Lambda)\in\mathcal{L}^>_{k}\).
\end{prop}
Combining Propositions~\ref{imm pos closed cyc prop},~\ref{prop:imm_under_inv_inj},~\ref{imm pos rot prop}, we obtain
\begin{coro}
\label{coro:imm pos L closed}
    \(\mathcal{L}^>_{k}\) is closed under the action of the inverses of the \( \mathrm{Rot}\), \( \mathrm{Inc}\), and \( \mathrm{Cyc}\). As it is defined as a combination of the above, it is also closed under the inverse of the \( \mathrm{Arc}\) move.
\end{coro}

\subsection{Finding Strongly Positive Matrices}
\label{sec:finding strongly pos mat}
We now show that \(\mathcal L_k^>\) is non-empty, and in fact contains the very nice large subset we define below, inspired by a related construction of \cite{origami}.
\begin{dfn}
    Given \(A\in \mathrm{Mat}_{(k+m)\times n} \), with \(k+m<n\), define
    \[
    \pi_{k-m,k+m}: \mathrm{Mat}_{n\times n}\rightarrow \mathrm{Mat}_{(k-m)\times n}\times \mathrm{Mat}_{(k+m)\times n},
    \]
    by 
    \(
    \pi_{k-m,k+m}(A) = (A_{\{k-m\}},A_{\{k+m\}}).
    \)
    That is, \(A\) restricted to the \(\{k\pm m\}\) rows respectively.
\end{dfn}

    \begin{dfn}
    \label{regular moves for tot pos def}
        For \(s\in[n-1]\) let \(x_s(t)\) and \(y_s(t)\) be the matrices obtained by taking the \(2k \times 2k\) identity matrix and replacing the block on the \(\{s, s+1\}\) rows and columns by 
        \(
        \begin{pmatrix}
            1&t\\
            0&1
        \end{pmatrix}
        \)
        and
        \(
        \begin{pmatrix}
            1&0\\
            t&1
        \end{pmatrix}
        \)
        respectively. For \(s\in[n]\) let \(h_s(t)\) be the matrix obtained by taking the \(2k\times 2k\) identity matrix and scaling the \((s,s)\) entry by \(t\). 

        Let \(\mathcal G\) be the semigroup generated by \(x_s(t)\), \(y_s(t)\) and \(h_s(t)\) for \(t>0\).
    \end{dfn}

    \begin{lem}[\cite{origami}, Section~6.2]
    \label{pos poly lemma}
        Write \((\Lambda^{ \perp \intercal }_0,\widetilde\Lambda^{\intercal}_0) = \pi_{k-2,k+2}(M)\), and \((\Lambda^{\perp \intercal },\widetilde\Lambda^\intercal) = \pi_{k-2,k+2}(M \, g(t))\), where \(t>0\) and \(g(t)\in\{x_s(t),y_s(t),h_s(t)\}\). Then for \((\tau,T)\in T_{k,2k}\) the expression
        \(
            c^{i,j}_{\tau,T}(\Lambda, \widetilde\Lambda) 
        \)
        is a \(\mathbb{Z}_{\geq 0}[t]\)-linear combination of non-trivial \(c^{i,j}_{\tau^\prime,T^\prime}(\Lambda_0,  \widetilde\Lambda_0)\), with the coefficient of \(c^{i,j}_{\tau,T}(\Lambda_0,  \widetilde\Lambda_0)\) being non-zero.
    \end{lem}

    \begin{dfn}[\cite{origami}, Section~6.2]
        Let \(\mathbf{Fl_{>0}}(k-2,k+2)\defeq \{\pi_{k-m,k+m}(M)|M\in \mathrm{Mat}_{n\times n}^{\gg0}(\mathbb{R})\}\), where \( \mathrm{Mat}_{n\times n}^{\gg0}(\mathbb{R})\) be the \(n\times n\) \emph{totally} positive matrices, that is, matrices with all minors of all sizes being positive.
    \end{dfn}

        \begin{thm}[\cite{origami}, Theorem~6.5]
    \label{immanant pos totaly pos thm}
    For \(2\leq k\leq n-2\), and \((\Lambda^{\perp \intercal },\widetilde\Lambda^\intercal)\in \mathbf{Fl_{>0}}(k-2,k+2)\), we have that \((\Lambda,\widetilde \Lambda)\) are strongly positive. 
    \end{thm}

    Observe that acting with the generators \( \{x_s(t), y_s(t), h_s(t) \}\) for \( t \in \mathbb{R}_{>0} \) on a totally positive matrix by multiplication from the left results in a totally positive matrix. Totally positive matrices have been extensively studied in the past. Fomin and Zelevinsky \cite{FZ} give a very powerful characterization of the space of totally positive matrices, including the following:

    \begin{thm}[\cite{FZ}, Theorem~5]
           \label{totally pos para thm}
        There exists a series of \(g_i\in \{x_s\, y_s, h_s \}\) such that any matrix \( M \in \mathrm{Mat}_{n\times n }^{\gg0}(\mathbb{R}) \) can be represented as
        \(
        M = g_1(t_1) g_2(t_2)...g_l(t_l)
        \)
         for some \(t_i>0\).
    \end{thm}

\begin{obs}
    For \(g\in \{x_s\, y_s, h_s \},~t>0\) and every \(n\times n\) matrix \(A\), the \(I, J\in\binom{[n]}{k}\) minor \(\Delta_{I,J}(Ag(t))\) of \(Ag(t)\), can be written as a $\mathbb{Z}_{\geq0}[t]$-linear combination of the minors of $A,$ such that the coefficient of \(\Delta_{I,J}(A) \) is non-zero. In particular, if \(t>0\) and \(\Delta_{I,J}(A)>0\), then  \(\Delta_{I,J}(Ag(t))>0\).
\end{obs}
 
\begin{coro}
\label{more positive coro}
Let \(M = f_1(t_1) f_2(t_2)...f_l(t_r)\) a series of \(f_i\in \{x_s\, y_s, h_s \}\) with \(t_i>0\). If there exist a series \(i_1<i_2<...<i_l\) such that \(f_{i_j} = g_j\) for \(j\in[l]\) then \(M\) is totally positive.
    
\end{coro}

    Let us consider a different point of view. For \(M\) an \(n\times n\) matrix define
    \(
    c_{\tau,T}^{i,j}(M) \defeq c_{\tau,T}^{i,j}(\Lambda,\widetilde \Lambda),
    \)
    where \((\Lambda^{\perp \intercal },\widetilde\Lambda^\intercal)=\pi_{k-m,k+m}(M)\).

 Recall that we are interested in the case of \(n=2k\), \(m=2\), and \(\eta \Lambda = \widetilde \Lambda\) for the orthogonal momentum amplituhedron. When we define those using \(\pi_{k-m,k+m}\), we have that \(\Lambda^{\perp\intercal}\subset\widetilde\Lambda^\intercal \) as spaces. Thus \(\eta \Lambda = \widetilde \Lambda\) is equivalent to \(\Lambda^{\perp\intercal} \, \eta \,\Lambda^{\perp}=0\).

Let \(e_i\) be the standard basis vectors of \(\mathbb{R}^{2k}\). Define the \(2k\times 2k\)\ matrix \(\Lambda_{0,k}\) by \((\Lambda_{0,k}^\intercal)_i =e_{2i-1}+e_{2i}\) for \(1\leq i \leq k-2\), and \((\Lambda_{0,k}^\intercal)_i = e_i\) for \(k-2<i\leq 2k\), and the rest of the rows chosen arbitrarily to get a matrix of rank \(2k\).

\[
\Lambda_{0,k}^\intercal =\scalebox{0.7}{\(
\begin{pmatrix}
    1&1&&&&&&&&&\\
    &&1&1&&&&&&&\\
    &&&&\ddots&&&&&&\\
    &&&&&1&1&&&&\\
    &&&&&&&1&&&\\
    &&&&&&&&1&&\\
    &&&&&&&&&1&\\
    &&&&&&&&&&1\\
    *&*&*&*&...&*&*&*&*&*&*\\
    &\vdots&&&&&&&&\vdots&\\
    *&*&*&*&...&*&*&*&*&*&*\\
\end{pmatrix}
\)}\]

In Section~\ref{the base case section}, we will prove the following lemma:
\begin{lem}
\label{pos on L0 lemma}
    Fix $k>2.$ For \((\Lambda_{k},\widetilde\Lambda_{k})\) such that \((\Lambda_{k}^{\perp \intercal },\widetilde\Lambda_{k}^\intercal) = \pi_{k-2,k+2}(\Lambda_{0,k}^\intercal)\), we have
    \\\({
     c_{\tau,T}^{i,j}(\Lambda_{k},\widetilde\Lambda_{k})\geq0
     }\)
\end{lem}

Theorem~\ref{immanant pos totaly pos thm} tells us that non-trivial \(c_{\tau,T}^{i,j}(M)\) are  strictly positive for  \(M\in  \mathrm{Mat}_{2k\times2k}^{\gg0}(\mathbb{R})\). Since $\mathrm{Mat}_{2k\times2k}^{\gg0}(\mathbb{R})$ is open in \(\mathrm{Mat}^*_{2k\times2k}(\mathbb{R})\), the zero-locus of each polynomial $c_{\tau,T}^{i,j}$ is sparse. 
Since \(\Lambda_{0,k}^\intercal\) is invertible, 
    \(\Lambda_{0,k}^\intercal \cdot  \mathrm{Mat}_{2k\times 2k}^{\gg0}(\mathbb{R})\subset \mathrm{Mat}^*_{2k\times2k}(\mathbb{R})\) is also open, and hence non-trivial \(c_{\tau,T}^{i,j}\) are not identically zero on \(\Lambda_{0,k}^\intercal \cdot   \mathrm{Mat}_{2k\times 2k}^{\gg0}(\mathbb{R})\).
    \begin{coro}
        \label{L0M pos coro}
         Non-trivial \(c_{\tau,T}^{i,j}\) are positive on \(\Lambda_{0,k}^\intercal \cdot  \mathrm{Mat}_{2k\times 2k}^{\gg0}(\mathbb{R})\).
    \end{coro}
    \begin{proof}
        Take \(L \in  \Lambda_{0,k}^\intercal   \cdot \mathrm{Mat}_{2k\times 2k}^{\gg0}(\mathbb{R})\). By Theorem~\ref{totally pos para thm} we have that \(L = L(t) =  \Lambda_{0,k}^\intercal M(t)\) for 
         \(
        M(t) = g_1(t_1) g_2(t_2)...g_l(t_l)
        \)
         for \(t\in\mathbb R_{+}^l\). By Lemma~\ref{pos poly lemma} we have that \(c_{\tau,T}^{i,j}(L(t))\) is a  \(\mathbb{Z}_{\geq 0}[t]\)-linear combination of non-trivial \(c^{i,j}_{\tau^\prime,T^\prime}( \Lambda_{0,k}^\intercal)\). By Lemma~\ref{pos on L0 lemma}, these are non-negative, therefore \(c_{\tau,T}^{i,j}(L(t))\in \mathbb R_{\geq 0 }[t]\). Since, as explained above, \(c_{\tau,T}^{i,j}\) are not identically zero on \(\Lambda_{0,k}^\intercal  \cdot  \mathrm{Mat}_{2k\times 2k}^{\gg0}(\mathbb{R})\), 
         they must all be positive for \(t\in\mathbb R_{+}^l\).
    \end{proof}
     \begin{dfn}
        For \(i\in[n-1]\) let \(r_i(t)\) be the matrices obtained by taking the \(2k \times 2k\) identity matrix and replacing the block on the \(\{i, i+1\}\) rows and columns by 
        \(
        \begin{pmatrix}
            \cosh(t)&\sinh(t)\\
            \sinh(t)&\cosh(t)
        \end{pmatrix}.
        \)
        \\
        Let \(\mathcal R\) be the semigroup generated by \(r_i(t)\) for \(t>0\). Define also $\mathcal {R}^{\gg0}$ to be the intersection $\mathrm{Mat}_{2k\times 2k}^{\gg0}(\mathbb R)\cap \mathcal R$.
    \end{dfn}
    Notice for with \(\xi(t)\defeq \sinh(t) \cosh(t)\) we have 
    \begin{align*}
        r_i(t) 
        &=
        y_i\left(\tanh(t)\right)x_i\left(\xi(t)\right)h_i\left(\mathrm{sech}(t)\right)h_{i+1}\left({\cosh(t)}\right)\\
        &= 
        x_i\left(\tanh(t)\right)y_i\left(\xi(t)\right)h_{i+1}\left(\mathrm{sech}(t)\right)h_{i}\left({\cosh(t)}\right).
    \end{align*}

    \(\mathcal R\) is a sub-semigroup of \(\mathcal G\) defined in Definition~\ref{regular moves for tot pos def}. Note that \(r_i(t) \,\eta\, r_i(t)^\intercal = \eta\), thus acting with those by multiplication from the left on \(\Lambda^{\perp\intercal}\) will preserve the property that \(\Lambda^{\perp\intercal} \, \eta \,\Lambda^{\perp}=0\). 

    \begin{coro}
        Strongly positive \( \mathcal L_k^>\) matrices are closed under action by multiplication from the left by \(r_i(t)\) with \(t>0\). 
    \end{coro}
    \begin{proof}
        This is immediate from Lemma~\ref{pos poly lemma}, together with the fact that \(r_i(t)\) preserves the property that \(\eta \Lambda = \widetilde \Lambda\): For  \(\widetilde\Lambda = r_i(t)\widetilde\Lambda_0,\) and \(\Lambda^{\perp } = r_i(t) \Lambda^{\perp }_0 \). As for any \(r_i(t)\in\mathcal R\) we have \(r_i(t)^\intercal\,\eta\, r_i(t) = \eta\), we conclude
        \(
         (\Lambda^\perp)\cdot (\eta \widetilde \Lambda) =\Lambda^{\perp\intercal } \eta\widetilde \Lambda = \Lambda_0^{\perp\intercal }r_i(t)^\intercal \eta r_i(t) \widetilde \Lambda _0 =  \Lambda_0^{\perp\intercal }\eta \widetilde \Lambda _0= \\=(\Lambda^\perp_0)\cdot (\eta \widetilde \Lambda_0)=0.
        \)
        Thus \(\Lambda = \eta \widetilde\Lambda\), and \(\eta \Lambda = \widetilde \Lambda\).
    \end{proof}

    We turn to prove Proposition~\ref{imm pos rot prop}:
    
    \begin{proof}[Proof of Proposition~\ref{imm pos rot prop}]
        We need to show that for \(\widetilde\Lambda\in\mathcal{L}^>_k\), we have that \(\mathrm{Rot}_{i,i+1}^{-1}(\widetilde\Lambda)\in\mathcal{L}^>_{k}\).

        If \(i<2k\) we have that \(\mathrm{Rot}_{i,i+1}^{-1}(t)(\widetilde\Lambda) = r_i(t)\widetilde \Lambda\). If \(i=2k\) we can use the \(\mathrm{Cyc}\) move and Proposition~\ref{imm pos closed cyc prop} to reduce the problem to the previous case as \(\mathrm{Rot}_{1,2}^{-1}(t)\mathrm{Cyc}(\widetilde \Lambda) = \mathrm{Cyc}(\mathrm{Rot}_{2k,1}^{-1}(t)\widetilde \Lambda)\). 
    \end{proof}

    By Theorem~\ref{totally pos para thm} we have that \(M\in  \mathrm{Mat}_{2k\times 2k}^{\gg0}(\mathbb{R})\)  is a product 
    \(
    M = g_1(t_1) g_2(t_2)...g_\ell(t_\ell)
    \)
    with \(g_i\in \{x_s\, y_s, h_s \}\) and \(t_i>0\). Let us write \(M = M(t)\) and construct a new matrix \(\hat{M}(t)\) in the following way:
    \begin{dfn}
        For \(i\in[2k-1]\) write:
        \begin{align*}
            &\hat x_i(t) = r_i(t) 
        =
        y_i\left(\tanh(t)\right)x_i\left(\xi(t)\right)h_i\left(\mathrm{sech}(t)\right)h_{i+1}\left({\cosh(t)}\right),\\
            &\hat y_i(t) = r_{i}(t) = 
        x_i\left(\tanh(t)\right)y_i\left(\xi(t)\right)h_{i+1}\left(\mathrm{sech}(t)\right)h_{i}\left({\cosh(t)}\right),\\
            &\hat h_i(t)=r_{i}(t) = 
        x_i\left(\tanh(t)\right)y_i\left(\xi(t)\right)h_{i+1}\left(\mathrm{sech}(t)\right)h_{i}\left({\cosh(t)}\right),\\
        &\hat h_{2k}(t)=r_{2k-1}(t)=
        y_{2k}\left(\tanh(t)\right)x_{2k}\left(\xi(t)\right)h_{2k}\left(\mathrm{sech}(t)\right)h_{2k+1}\left({\cosh(t)}\right).
        \end{align*}

        For \(M(t) = g_1(t_1) g_2(t_2)...g_\ell(t_\ell)\)  write
        \(
        \hat M= \hat g_1(t_1) \hat g_2(t_2)...\hat g_\ell(t_\ell).
        \)
    \end{dfn}
    By Corollary~\ref{more positive coro} and theorem~\ref{totally pos para thm} we have that \(\hat{M}\) is totally positive for any \(t\in\mathbb R_+^l\). 
     \begin{dfn}
               Define 
        \(
        \mathbf{OF}_{>0}(k-2,k+2)\defeq \{\pi_{k-2,k+2}(\Lambda_{0,k}^\intercal M)|M\in\mathcal {R}^{\gg0}
        \}.
        \)
    \end{dfn}
    \begin{thm}
        For \(2\leq k\leq n-2\), \(\mathbf{OF}_{>0}(k-2,k+2)\) is non empty, and for \((\Lambda^{\perp \intercal },\widetilde\Lambda^\intercal)\in \mathbf{O F}_{>0}(k-2,k+2)\), every \((\Lambda,\widetilde \Lambda)\) is both strongly positive and satisfies \(\eta \Lambda = \widetilde \Lambda\).
    \end{thm}
    \begin{proof}
        For the first part, it is enough to show that \(\mathcal {R}^{\gg0}=\mathrm{GL}_{2k}^{\gg0}(\mathbb R)\cap \mathcal R\) is not empty. For every \(M\in \mathrm{GL}_{2k}^{\gg0}(\mathbb R)\), clearly \(\hat M\in \mathrm{GL}_{2k}^{\gg0}(\mathbb R)\cap \mathcal R\), hence $\mathcal {R}^{\gg0}\neq\emptyset.$

        For the second part, consider \((\Lambda^{\perp \intercal },\widetilde\Lambda^\intercal)\in \mathbf{O F}_{>0}(k-2,k+2)\).\\ Then \((\Lambda^{\perp \intercal },\widetilde\Lambda^\intercal) = \pi_{k-2,k+2}(\Lambda_{0,k}^\intercal M)\) with \(M \in\mathrm{GL}_{2k}^{\gg0}(\mathbb R) \). Thus \((\Lambda,\widetilde\Lambda)\) are strongly positive by Corollary~\ref{L0M pos coro}.

        Finally, consider \((\Lambda^{\perp \intercal }_0,\widetilde\Lambda^\intercal_0) = \pi_{k-2,k+2}(\Lambda_{0,k}^\intercal)\). It is easy to check that \(\eta \Lambda_0 = \widetilde \Lambda_0\). Now, for \(R\in \mathcal R\), \((\Lambda^{\perp \intercal },\widetilde\Lambda^\intercal) = \pi_{k-2,k+2}(\Lambda_{0,k}^\intercal R) = (\Lambda^{\perp \intercal }_0 R^\intercal,\widetilde\Lambda^\intercal_0R^\intercal).\) Indeed, \(\widetilde\Lambda = R\widetilde\Lambda_0,\) and \(\Lambda^{\perp } = R \Lambda^{\perp }_0 \). Since every \(R\in\mathcal R\)satisfies \(R^\intercal\eta R = \eta\), we have
        \(
         (\Lambda^\perp)\cdot (\eta \widetilde \Lambda) =\Lambda^{\perp\intercal } \eta\widetilde \Lambda = \Lambda_0^{\perp\intercal }R^\intercal \eta R \widetilde \Lambda _0 =  \Lambda_0^{\perp\intercal }\eta \widetilde \Lambda _0 =(\Lambda^\perp_0)\cdot (\eta \widetilde \Lambda_0)=0.
        \)
        Thus \(\Lambda = \eta \widetilde\Lambda\), and \(\eta \Lambda = \widetilde \Lambda\).
    \end{proof}

\begin{coro}
    \label{imma pos L is not empty}
    \(\mathcal{L}_k^{>}\) is not empty.
\end{coro}

\begin{proof}
    For \((\Lambda^{\perp \intercal },\widetilde\Lambda^\intercal)\in \mathbf{O F}_{>0}(k-2,k+2)\), we have that \(\widetilde\Lambda\in\mathcal{L}_k^{>} \).
\end{proof}

\subsection{Proof of Lemma~\ref{pos on L0 lemma}}
\label{the base case section}
We will now prove Lemma~\ref{pos on L0 lemma}. The proof will be rather technical. 

Recall that
\[
\Lambda_{0,k}^\intercal =\scalebox{0.7}{\(
\begin{pmatrix}
    1&1&&&&&&&&&\\
    &&1&1&&&&&&&\\
    &&&&\ddots&&&&&&\\
    &&&&&1&1&&&&\\
    &&&&&&&1&&&\\
    &&&&&&&&1&&\\
    &&&&&&&&&1&\\
    &&&&&&&&&&1\\
    *&*&*&*&...&*&*&*&*&*&*\\
    &\vdots&&&&&&&&\vdots&\\
    *&*&*&*&...&*&*&*&*&*&*\\
\end{pmatrix}.
\)}\]
\begin{proof}[Proof of Lemma~\ref{pos on L0 lemma}]
    The Pl\"ucker coordinates of \(\Lambda_{k}^{\intercal \perp},\widetilde\Lambda_{k}^\intercal\) are all either 1 or 0. We will analize for which indices \(i,j\) and \((i,j,\tau,T)\)-markings \(\mu\) the terms \({\Delta_{L_\mu}(\Lambda_{k}^{\intercal \perp})\Delta_{R_\mu\cup J_\mu}(\widetilde\Lambda_{k}^\intercal)}\) are  $1,$ and to find the corresponding signs \((-1)^{d_\mu^{i,j}}\). We will act inductively.

    Write \(I = I(i,j)\), \(H = \{2k-3,2k-2,2k-1,2k\}\). We must have \(H\cap L_\mu = \emptyset\) otherwise \(\Delta_{L_\mu}(\Lambda_{k}^{\intercal \perp}) = 0\), as \((\Lambda_{k}^{\intercal \perp})^H = 0\), and therefore 
    \(
    H\cap T = \emptyset
    .\) We must also have \(H\subset R_\mu \cup J_\mu\), otherwise \(\Delta_{R_\mu \cup J_\mu} (\widetilde\Lambda_{k}) = 0\) as \[
    (\widetilde\Lambda_{k})^{H^c}_{\{k-1,k,k+1,k+2\}} = 0
    ,\quad\text{and}\quad 
    {(\widetilde\Lambda_{k})^{H}_{\{k-1,k,k+1,k+2\}} = \mathrm{Id}_{4\times4}}
    .\]
    This means 
    \(
    H\subset\mu^{-1}(R)\cup\mu^{-1}(J)\subset S(\tau)
    .\)

    Similarly, for each \(1\leq q\leq k-2\) exactly one index of the pair \(\{\ell,r\}\defeq\{2q-1,2q\}\) is in \(L_\mu\), and exactly one is in \(R_\mu\cup J_\mu\). 
    
    Since these are the only \((i,j,\tau,T)\)-markings with \({\Delta_{L_\mu}(\Lambda_{k}^{\intercal \perp})\Delta_{R_\mu\cup J_\mu}(\widetilde\Lambda_{k}^\intercal)}\) are non-zero, we are going to define \(\widetilde{\mathcal M}_{(i,j,\tau,T)}\) as follows: \begin{dfn}
        \(\widetilde{\mathcal M}_{(i,j,\tau,T)}\) are precisely the markings \(\mu\in\mathcal M{(i,j,\tau,T)}\) such that:
    \begin{itemize}
        \item For each \(1\leq q\leq k-m\) that exactly one index of the pair \(\{\ell,r\}\defeq\{2q-1,2q\}\) is in \(L_\mu\) and exactly one is in \(R_\mu\cup J_\mu\).
        \item \(H\subset\mu^{-1}(R)\cup\mu^{-1}(J)\), and thus \(H\subset R_\mu\cup J_\mu\) and \(H \cap L_\mu = \emptyset\).
    \end{itemize}
    \end{dfn}
     Indeed, it is easy to see that for these \((i,j,\tau,T)\)-markings \({\Delta_{L_\mu}(\Lambda_{k}^{\intercal \perp})\Delta_{R_\mu\cup J_\mu}(\widetilde\Lambda_{k}^\intercal)}\) are 1. We can thus write 
     \begin{prop}
     \label{prop:cMtilde}
          \[
     c_{\tau,T}^{i,j}(\Lambda_k,\widetilde\Lambda_k)=\sum_{\mu\in\widetilde{\mathcal{M}}_{(i,j,\tau,T)}}(-1)^{d^{i,j}_\mu}.
     \]
     \end{prop}

     Consider now \(\mu \in \widetilde{\mathcal{M}}_{(i,j,\tau,T)}\). For each \(1\leq q\leq k-2\) exactly one index of the pair \(\{\ell,r\}\defeq\{2q-1,2q\}\) is in \(L_\mu\), and exactly one is in \(R_\mu\cup J_\mu\).
     
    \underline{\textbf{If it is the same index}}: Then  without loss of generality we have
    \(l \in (R_\mu\cup J_\mu) \cap L_\mu = T,\) and  \({r \notin R_\mu\cup J_\mu \cup L_\mu = S(\tau) \cup T}.\)
    
    \underline{\textbf{If it is not the same index}}: Then without loss of generality we have 
    \[
    \ell\in L_\mu,\,\,\,\,\, r\in R_\mu\cup J_\mu,
    \]
    \[
    r\notin L_\mu,\,\,\,\,\, \ell\notin R_\mu\cup J_\mu,
    \]
    and thus it must hold that
    \(
    \ell,r \in  R_\mu\cup J_\mu \cup L_\mu\setminus T = S(\tau)
    .\)
    Both of these properties depend only on \((\tau,T)\) and not on \(\mu\).
    \begin{prop}
    \label{prop:cond on  tT}
    For every \(i,j\) and \((\tau,T)\in\mathcal T_{k,2k}\), if \(c_{\tau,T}^{i,j}(\Lambda_k,\widetilde\Lambda_k) \neq 0\), then:
    \begin{enumerate}
        \item For each \(1\leq q\leq k-2\), the pair \(\{2q-1,2q\}\) are either both in \(S(\tau)\) or both are not in \(S(\tau)\) and exactly one is in \(T\).
        \item \(H \subset S(\tau)\).
    \end{enumerate}
    \begin{proof}
         If this is not the case, \(\mu \in \widetilde{\mathcal{M}}_{(i,j,\tau,T)}=\emptyset\) and thus by Proposition~\ref{prop:cMtilde}, we have
         \(c_{\tau,T}^{i,j}(\Lambda_k,\widetilde\Lambda_k) = 0.\)
    \end{proof}
       
    \end{prop}
    We will thus restrict our attention only to such cases as described in Proposition~\ref{prop:cond on  tT}.
    \begin{dfn}
        Let \(\widetilde{\mathcal T}_{k,2k}\) be the set of \((\tau,T)\in\mathcal T_{k,2k}\) that satisfy the conditions in the previous Proposition~\ref{prop:cond on  tT}
    \end{dfn}
    \begin{prop}
    \label{prop:T empty}
    For \(i,j\) and \((\tau, T)\in \widetilde{\mathcal T}_{k,2k}\) such that \(T\neq\emptyset\), there exist \(i^\prime,j^\prime\) and \((\tau^\prime, T^\prime) \in \widetilde{\mathcal T}_{k-1,2k-2}\) such that
        \(
        c_{\tau,T}^{i,j}(\Lambda_k,\widetilde\Lambda_k)=c_{\tau^\prime,T^\prime}^{i^\prime,j^\prime}(\Lambda_{k-1},\widetilde\Lambda_{k-1})
        .\)
    \end{prop}
    This claim will allow us to disregard such cases by induction on \(k\).
    \begin{proof}
    We know \(H\subset S(\tau)\). Take \(1\leq q\leq k-2\) and a pair \(\{\ell,r\}\defeq\{2q-1,2q\}\) such that \(\ell\in T\).
        Thus \(\ell,r\notin S(\tau)\) and are not contained in any arc.
    We would like to define a new \((\tau^\prime ,T^\prime)\in \mathcal{T}_{k-1,2k-2}\) and \(i^\prime,j^\prime\), such that
    \(
    c_{\tau,T}^{i,j}(\Lambda_{k},\widetilde\Lambda_{k})= 
    c_{\tau^\prime,T^\prime}^{i^\prime,j^\prime}(\Lambda_{k-1},\widetilde\Lambda_{k-1})
    ,\)
    by deleting the \(\{2q-1,2q\}\) indices; that is, define \((i^\prime,j^\prime,\tau^\prime,T^\prime)\) by the following procedure:

    First define the bijection \(\nu_{\{\ell,r\}}:[2k]\setminus\{\ell,r\}\rightarrow[2k-2]\) by \(\nu_{\{\ell,r\}}(a) = a-1\) if \(a>2q\) and \(\nu_{\{\ell,r\}}(a)=a\) otherwise. Now
    \begin{itemize}
        \item \(i^\prime =\nu_{\{\ell,r\}}(i)\), and \(j^\prime =\nu_{\{\ell,r\}}(j)\).
        \item \(\tau^\prime:[2k-2]\rightarrow[2k-2]\) is defined by \(\tau^\prime(\nu_{\{\ell,r\}}((a)) = \nu_{\{\ell,r\}}(\tau(a))\), for every \(a \in [2k]\setminus \{\ell,r\}\).
        \item \(T^\prime = \nu_{\{\ell,r\}}(T\setminus\{l\})\) (which are stable points of \(\tau^\prime\) by the above definition).
    \end{itemize}
    Because  \(S(\tau^\prime) = \nu_{\{\ell,r\}}(S(\tau))\), we have \(2\left|T^\prime\right| +\left|S(\tau^\prime)\right| = 2k-2.\)  As \(\nu_{\{\ell,r\}}\) preserves orderings of indices \(\tau^\prime\) is also non-crossing, iff \(\tau\) is non-crossing. Thus \((\tau^\prime,T^\prime)\in \mathcal{T}_{k-1,2k-2}\). As we removed a pair of indices \(\{2q-1,2q\}\) and \(\nu_{\{\ell,r\}}\subset S(\tau^\prime)\), we have \((\tau^\prime,T^\prime)\in \widetilde{\mathcal{T}}_{k-1,2k-2}\). Since \(\nu_{\{\ell,r\}}\) preserves orderings of indices \(\tau^\prime\) the image of an arc of \(\tau\) is a special arc of \(\tau^\prime\) if and only if it were a special arc of \(\tau\). We can similarly define
     for \(\mu\in{\mathcal{M}_{(i,j,\tau,T)}}\), a new \(\mu^\prime\in{\mathcal{M}_{(i^\prime,j^\prime,\tau^\prime,T^\prime)}}\) by
    \begin{itemize}
        \item \(\mu^\prime(\nu_{\{\ell,r\}}^{-1}(a)) = \mu(a)\) for any \(a\in S(\tau)\).
    \end{itemize}
    Observe that that this operation is invertible: given  \(\hat \mu \in \widetilde{\mathcal{M}} _{(i^\prime, j^\prime, \tau^\prime, T^\prime)}\) we can find a unique \(\tilde\mu \in \widetilde{\mathcal{M}} _{(i, j, \tau, T)}\) such that \(\tilde\mu^\prime = \hat \mu\) by applying the bijection \(\nu_{\{\ell,r\}}\).
    Furthermore, as \(\{\ell,r\} = \{2q-1,2q\}\subset S(\tau)\) for some \(1\leq q \leq k-2\), iff  \(\{\ell,r\} = \{2q^\prime-1,2q^\prime\} \subset S(\tau^\prime)\) for some \(1\leq q^\prime \leq k-3\). Thus \(\mu^\prime \in \widetilde{\mathcal{M}} _{(i^\prime, j^\prime, \tau^\prime, T^\prime)}\) iff  \(\mu \in \widetilde{\mathcal{M}} _{(i , j , \tau , T )}\).
    We thus obtain a bijection \(\mu \in \widetilde{\mathcal{M}} _{(i , j , \tau , T )}\rightarrow \mu^\prime \in \widetilde{\mathcal{M}} _{(i^\prime, j^\prime, \tau^\prime, T^\prime)}  \) by \(\mu\mapsto \mu^\prime\).

    As we didn't remove any arc, and did not change the ordering of the indices, 
    \(d^{i,j}_\mu \) and \(d^{i^\prime,j^\prime}_{\mu^\prime} \) must be equal.
    Thus,
    \begin{align*}
        c_{\tau,T}^{i,j}(\Lambda_{k},\widetilde\Lambda_{k})\,&=\sum_{\mu\in\widetilde{\mathcal{M}}_{(i,j,\tau,T)}}(-1)^{d^{i,j}_\mu}
        =\sum_{\mu\in\widetilde{\mathcal{M}}_{(i,j,\tau,T)}}(-1)^{d^{i^\prime,j^\prime}_{\mu^\prime}} \\
        &=\sum_{\mu^\prime\in\widetilde{\mathcal{M}}_{(i^\prime,j^\prime,\tau^\prime,T^\prime)}}(-1)^{d^{i^\prime,j^\prime}_{\mu^\prime}}
        =c_{\tau^\prime,T^\prime}^{i^\prime,j^\prime}(\Lambda_{k-1},\widetilde\Lambda_{k-1})
    \end{align*}
    \end{proof}
    
    Since we have \(H\subset S(\tau)\) meaning \(H\cap T =\emptyset\), for each \(\ell\in T\) there must be a \(q\) with \(1\leq q\leq k-2\) such that \(\ell\in\{2q-1,2q\}\). We can thus continue removing such pairs, by induction, until we reach the case \(T = \emptyset\). It is therefore enough to show \(
    c_{\tau,T}^{i,j}(\Lambda_{k},\widetilde\Lambda_{k})\geq 0 
    \) for cases where \(T = \emptyset\), and assume they are positive otherwise by induction. 
    Since \(2\left|T\right| + |S(\tau)| = 2k\), we must have \(S(\tau) = [2k]\), and that \(\tau\) is an involution with no fixed points. We will now consider only these cases. 

    Recall that, by definition, for every \(\mu \in \widetilde{\mathcal M} _{(i,j,\tau,T)}\), for each \(1\leq q\leq k-2\), exactly one index from the pair \(\{\ell,r\}\defeq\{2q-1,2q\}\) is in \(L_\mu\), and exactly one is in \(R_\mu\cup J_\mu\). They are disjoint, since we assume \(T=\emptyset\). Without loss of generality \(
        \ell\in L_\mu,\ell\notin R_\mu\cup J_\mu\) and \(
        r\notin L_\mu, r\in R_\mu\cup J_\mu.\)
    
\subsubsection{Circle Graphs}
Recall \(H\defeq \{2k-3,2k-2,2k-1,2k\}\).

    \begin{dfn}

    For \(\tau\) as defined in Definition~\ref{def:temp_lieb}, and \(I\subset[2k]\), define the \emph{circle graph} \(\Gamma(\tau) \defeq \Gamma \) as the following graph embedded in a disk: 
    \begin{itemize}
        \item The vertices \(\mathcal V_\Gamma\) are the indices \([2k]\), arranged along the boundary of the disk in a counter-clockwise order.
        \item For every \(q\in[k]\), \(\{2q-1,2q\}\) is an edge. These edges will be called \emph{\(O\)-edges} and will be drawn on the boundary of the disk. If \(\{2q-1,2q\}\subset H\),  These edges will be termed \emph{\(H\)-edges}, and will be drawn dashed.
        \item For every \(\ell\in[2k]\), \(\{\ell,\tau(\ell)\}\) is an edge, contained in the interior of the disk. These are called the \emph{\(\tau\)-edges}. If  \(\{\ell,\tau(\ell)\}\) is an I-special arc, we will call the corresponding edge a \emph{special edge}.
    \end{itemize}

    Since \(\tau\) is non-crossing, we may assume that the edges of \(\Gamma\) are non-crossing as well.
    Since \(\tau\) is a total pairing, as \(T = \emptyset\), the graph is \(2\)-regular, that is, it is a disjoint set of cycles. A cycle that contains a special edge would be called a \emph{special cycle}, and a cycle that contains an \(H\)-edge would be called an \emph{\(H\)-cycle}. An \(n\)-special cycle is a cycle with exactly \(n\) special edges. Given \(\mu \in \widetilde{ \mathcal M} _{(i,j,\tau,T)}\), an edge that corresponds to a \(J\)-arc would be called a \emph{\(J\)-edge}, and a cycle that contains an \(J\)-edge is a \emph{\(J\)-cycle}. 
    \end{dfn}
    \begin{figure}[H]
    \centering
 
      \begin{center}
\begin{tikzpicture}[scale=0.8, every node/.style={scale=0.8}]

\draw[line width=0.30mm,blue] (0.469303,-2.96307) arc (-81.:-63.:3);

\draw[line width=0.30mm,blue] (2.12132,-2.12132) arc (-45.:-27.:3);

\draw[line width=0.30mm,blue] (2.96307,-0.469303) arc (-9.:9.:3);

\draw[line width=0.30mm,blue] (2.67302,1.36197) arc (27.:45.:3);

\draw[line width=0.30mm,blue] (1.36197,2.67302) arc (63.:81.:3);

\draw[line width=0.30mm,blue] (-0.469303,2.96307) arc (99.:117.:3);

\draw[line width=0.30mm,blue] (-2.12132,2.12132) arc (135.:153.:3);

\draw[line width=0.30mm,blue] (-2.96307,0.469303) arc (171.:189.:3);

\draw[line width=0.30mm,blue, dashed] (-2.67302,-1.36197) arc (207.:225.:3);

\draw[line width=0.30mm,blue, dashed] (-1.36197,-2.67302) arc (243.:261.:3);

\draw[black,dashed,line width=0.30mm] (-2.42705, -1.76336) arc (126.:108.:18.9413);

\draw[black,line width=0.25mm] (0.469303,-2.96307) arc (-171.:135-360.:5.88783);

\draw[black,line width=0.25mm] (1.36197,-2.67302) arc (-153.:45.-360:0.475153);

\draw[black,line width=0.25mm] (2.67302,-1.36197) arc (-117.:81.-360:0.475153);

\draw[black,line width=0.25mm] (2.96307,0.469303) arc (-81.:117.-360:0.475153);

\draw[black,line width=0.70mm] (1.36197,2.67302) arc (-27+180.:351-180:18.9413);

\draw[black,line width=0.25mm] (0.469303,2.96307) arc (-9.:333.-360:18.9413);

\draw[black,line width=0.70mm] (-0.469303,2.96307) arc (9.:315.-360:5.88783);

\draw[black,line width=0.25mm] (-1.36197,2.67302) arc (27.:225.-360:0.475153);

\draw[black,line width=0.25mm] (-2.67302,1.36197) arc (63.:297.-360:1.52858);

\draw[black,line width=0.25mm] (-2.96307,0.469303) arc (81.:279.-360:0.475153);

\filldraw[black] (0.469303,-2.96307) circle (0.07cm);

\node[anchor=north](c) at (0.469303,-2.96307){1};
\node[anchor=east, purple](c) at (0.469303,-2.96307){\(L\)};

\filldraw[black] (1.36197,-2.67302) circle (0.07cm);

\node[anchor=north](c) at (1.36197,-2.67302){2};
\node[anchor=west,purple](c) at (1.36197,-2.67302){\(R\)};

\filldraw[black] (2.12132,-2.12132) circle (0.07cm);

\node[anchor=north](c) at (2.12132,-2.12132){3};
\node[anchor=west,purple](c) at (2.12132,-2.12132){\(L\)};

\filldraw[black] (2.67302,-1.36197) circle (0.07cm);

\node[anchor=west](c) at (2.67302,-1.36197){4};
\node[anchor=south,purple](c) at (2.67302,-1.36197){\(R\)};

\filldraw[black] (2.96307,-0.469303) circle (0.07cm);

\node[anchor=west](c) at (2.96307,-0.469303){5};
\node[anchor=north,purple](c) at (2.96307,-0.469303){\(L\)};

\filldraw[black] (2.96307,0.469303) circle (0.07cm);

\node[anchor=west](c) at (2.96307,0.469303){6};
\node[anchor=south,purple](c) at (2.96307,0.469303){\(R\)};

\filldraw[black] (2.67302,1.36197) circle (0.07cm);

\node[anchor=west](c) at (2.67302,1.36197){7};
\node[anchor=south,purple](c) at (2.7,1.45){\(L\)};

\filldraw[black] (2.12132,2.12132) circle (0.07cm);

\node[anchor=south](c) at (2.12132,2.12132){8};
\node[anchor=west,purple](c) at (2.13,2.13){\(R\)};

\filldraw[black] (1.36197,2.67302) circle (0.07cm);

\node[anchor=south](c) at (1.36197,2.67302){9};
\node[anchor=west,purple](c) at (1.36197,2.67302){\(J\)};

\filldraw[black] (0.469303,2.96307) circle (0.07cm);

\node[anchor=south](c) at (0.469303,2.96307){10};
\node[anchor=east,purple](c) at (0.469303,2.96307){\(L\)};

\filldraw[black] (-0.469303,2.96307) circle (0.07cm);

\node[anchor=south](c) at (-0.469303,2.96307){11};
\node[anchor=west,purple](c) at (-0.469303,2.96307){\(R\)};

\filldraw[black] (-1.36197,2.67302) circle (0.07cm);

\node[anchor=south](c) at (-1.36197,2.67302){12};
\node[anchor=east,purple](c) at (-1.36197,2.67302){\(L\)};

\filldraw[black] (-2.12132,2.12132) circle (0.07cm);

\node[anchor=south](c) at (-2.12132,2.12132){13};
\node[anchor=east,purple](c) at (-2.145,2.13){\(R\)};

\filldraw[black] (-2.67302,1.36197) circle (0.07cm);

\node[anchor=east](c) at (-2.67302,1.36197){14};
\node[anchor=north,purple](c) at (-2.67302,1.36197){\(L\)};

\filldraw[black] (-2.96307,0.469303) circle (0.07cm);

\node[anchor=east](c) at (-2.96307,0.469303){15};
\node[anchor=south,purple](c) at (-2.96307,0.469303){\(L\)};

\filldraw[black] (-2.96307,-0.469303) circle (0.07cm);

\node[anchor=east](c) at (-2.96307,-0.469303){16};
\node[anchor=north,purple](c) at (-2.96307,-0.469303){\(R\)};

\filldraw[black] (-2.67302,-1.36197) circle (0.07cm);

\node[anchor=east](c) at (-2.67302,-1.36197){17};
\node[anchor=south,purple](c) at (-2.67302,-1.36197){\(R\)};

\filldraw[black] (-2.12132,-2.12132) circle (0.07cm);

\node[anchor=north](c) at (-2.12132,-2.12132){18};
\node[anchor=east,purple](c) at (-2.12132,-2.12132){\(J\)};

\filldraw[black] (-1.36197,-2.67302) circle (0.07cm);

\node[anchor=north](c) at (-1.36197,-2.67302){19};
\node[anchor=east,purple](c) at (-1.36197,-2.6){\(R\)};

\filldraw[black] (-0.469303,-2.96307) circle (0.07cm);

\node[anchor=north](c) at (-0.469303,-2.96307){20};
\node[anchor=west,purple](c) at (-0.469303,-2.96307){\(R\)};

\end{tikzpicture}

\end{center}

    \caption{An example of a circle graph with a marking}
    \label{fig:circ graph}
\end{figure}

    \paragraph{Example:} in Figure~\ref{fig:circ graph} we see an example of the circle graph \(\Gamma(\tau)\) for \(k=10\), \(i=6\), \(j=17\), 
    \[
    \tau = (1,8) (2,3) (4,5) (6,7) (9,20) (10,19) (11,18) (12,13) (14,17) (15,16)
    \] in cycle notation, \(T = \emptyset\), together with a marking \(\mu\in\widetilde{\mathcal M}_{(i,j,\tau,T)}\).
    
    The \(\tau\)-edges are colored black, with the \(J\)-edges being thicker. The \(O\)-edges are colored blue with the \(H\)-edges dashed blue. The \(i\)-\(j\) line is dashed in black. The marking \(\mu\) is colored purple.

    The special edges are \(\{1,8\},\{6,7\},\{9,20\},\{10,19\},\{11,18\}\). The indices \(\{1,2,3,4,5,6,7\}\), form a 2-special cycle,  \(\{9,10,19,20\}\) a 2-special \(H\)-cycle, and the indices \(
    \{11,12,13,14,17,18\}\) form a 1-special \(H\) cycle. The cycle on \(\{15,16\}\) is neither an \(H\)-cycle nor a special cycle. There is one special edge strictly between the two \(J\)-edges \(\{9,20\}\) and \(\{11,18\}\), namely the edge \(\{10,19\}\). Thus \(d^{i,j}_\mu = 1\).

     As all the edges are on the boundary of the disk, we must have that the regions bounded by the cycles are disjoint, that is, we cannot have a circle embedded in the region bounded by another circle.
    Since the vertices are on the boundary of the disk, the edges are non crossing, and the graph is two-regular, the \(i\)-\(j\) line can cross each cycle either twice or not at all. Furthermore, by definition, the special edges are the ones that cross \(i\)-\(j\) line. Since the \(i\)-\(j\) line passes through the interior of the disk, we have that it only crosses \(\tau\)-edges, while it may or may not start and end on \(O\)-edges.
    This means that the first and last cycles that the \(i\)-\(j\) line intersects can either be 1-special or 2-special, while any other special cycle is 2-special.

    \begin{prop}
    \label{prop:special cycles}
        For \((\tau,T=\emptyset)\in\widetilde{\mathcal T}_{k,2k}\), \(\Gamma(\tau)\) is a disjoint set of cycles with the areas bounded by the cycles being disjoint. The first and last cycles the \(i\)-\(j\) line passes through can either be 1-special or 2-special, while any other special cycle is 2-special.
    \end{prop}
   
While \(\mu \in \widetilde{ \mathcal M} _{(i,j,\tau,T)}\) imposes conditions on the labeling of vertices connected by either a \(\tau\)-edge or an \(O\)-edge, the labeling of vertices in different cycles is entirely independent.
    \begin{prop}
    \label{prop:regular cycles}
           Let \(C\subset[2k]\) is cycle of \(\Gamma\), and let \(V_C\) be its vertices. Suppose that \(C\) is neither a special cycle nor an \(H\)-cycle. Then there exists \((\tau^\prime, T^\prime)\in \widetilde{\mathcal{T}}_{k,2k}\), such that \(T^\prime\neq \emptyset\), and 
    \(
    c^{i,j}_{\tau,T}(\Lambda_{k},\widetilde\Lambda_{k})=2 c^{i,j}_{\tau^\prime,T^\prime}(\Lambda_{k},\widetilde\Lambda_{k})
    .\)
    \end{prop}
    Thus we can assume \(c^{i,j}_{\tau,T}(\Lambda_{k},\widetilde\Lambda_{k})\geq 0\) by induction via Proposition~\ref{prop:T empty}.
    \begin{proof}
    Define \(\tau^\prime\) and \(T^\prime\) as follows: First define \(\tau^\prime(\ell) = \tau(\ell)\) for \(\ell\notin V_C\), and \(\tau(\ell) = \ell\) otherwise. Now, define \(T^\prime = T\cup\{2\ell : \{2\ell-1,2\ell\}\subset V_C\}\).

\(\tau^\prime\) is non-crossing if \(\tau\) is, and \(2\left|T^\prime\right| + \left|S(\tau^\prime)\right| = 2\left|T\right| + \left|S(\tau)\right| = 2k\), thus \((\tau^\prime,T^\prime)\in {\mathcal T}_{k,2k}\). \(H\subset S(\tau^\prime)\), and by construction for each \(1\leq q\leq k-2\), the pair \(\{2q-1,2q\}\) is either contained in \(S(\tau^\prime)\) or avoids \(S(\tau^\prime)\), and exactly one of $2q-1,2q$ is in \(T^\prime\). Thus \((\tau^\prime,T^\prime)\in \widetilde {\mathcal T}_{k,2k}\).
    For \(\mu \in \widetilde{\mathcal M}_{(i,j,\tau,T)}\), define \(\mu^\prime = \evalat{\mu}S(\tau^\prime)\). We have essentially removed from $\tau$ the arcs in the cycle \(V_C\) to get \(\tau^\prime\). Since \(V_C\) is not special we did not remove any special arcs from \(\tau\). This means \(\mu^\prime \in \widetilde{\mathcal M}_{(i,j,\tau^\prime,T^\prime)}\), and that  \(d^{i,j}_{\mu^\prime}=d^{i,j}_{\mu}\).

    We claim that the map \(\widetilde{\mathcal M}_{(i,j,\tau,T)}\rightarrow\widetilde{\mathcal M}_{(i,j,\tau^\prime,T^\prime)}\) defined by \(\mu\mapsto \mu^\prime\) is two-to-one:
    Take \(\mu^\prime\in \widetilde{\mathcal M}_{(i,j,\tau^\prime,T^\prime)} \). We will show there exist only two \(\hat\mu_{1,2}\in \widetilde{\mathcal M}_{(i,j,\tau,T)} \) such that  \(\hat\mu_{1,2}^\prime = \mu^\prime\).

    For every index \(\ell\notin V_C\), we have \(\mu^\prime(\ell)=\mu(\ell)\), thus we must have \(\hat\mu_{1,2}(\ell) = \mu(\ell)\). As for the indices in \(V_C\), notice we have that 
    \(
    V_C \cap S(\tau^\prime) = \emptyset
    ,\)
    and
    \(
    T^\prime = \{q : q\in V_C \text{ is even}\} 
    .\)
    Notice that for every edge \(\{r,\ell\}\subset V_C\), \(\{r,\ell\}\) is either a non-special \(\tau\)-edge, or a non-\(H\) \(O\)-edge. This means we must have one edge labeled \(R\), and one labeled \(L\), by \(\hat\mu_{1,2}\in\widetilde{\mathcal M}_{(i,j\tau,T)}\). Meaning that the labeling of vertices in \(V_C\) must be alternating \(R\) and \(L\). Indeed there are exactly two possible options for \(\evalat{\hat\mu_{1,2}}{V_C}\). One where the odd-numbered vertices go to \(R\) and the even-numbered vertices go to \(L\) (let us call it \(\hat\mu_1\)), and one where the even-numbered vertices go to \(R\) and the odd-numbered vertices go to \(L\) (let us call it \(\hat\mu_2\)). 
    
 Thus the map \(\widetilde{\mathcal M}_{(i,j,\tau,T)}\rightarrow\widetilde{\mathcal M}_{(i,j,\tau^\prime,T^\prime)}\) defined by \(\mu\mapsto \mu^\prime\) is two-to-one, and we have

    \begin{align*}
        c_{\tau,T}^{i,j}(\Lambda_{k},\widetilde\Lambda_{k})\,&=\sum_{\mu\in\widetilde{\mathcal{M}}_{(i,j,\tau,T)}}(-1)^{d^{i,j}_\mu} 
        =\sum_{\mu\in\widetilde{\mathcal{M}}_{(i,j,\tau,T)}}(-1)^{d^{i ,j}_{\mu^\prime}} \\
        &=2\sum_{\mu^\prime\in\widetilde{\mathcal{M}}_{(i ,j,\tau^\prime,T^\prime)}}(-1)^{d^{i ,j}_{\mu^\prime}} 
        =2 c_{\tau^\prime,T^\prime}^{i ,j}(\Lambda_{k},\widetilde\Lambda_{k}).
    \end{align*}
    \end{proof}

    As \(T^\prime\neq\emptyset\), we can assume by induction \(c_{\tau^\prime,T^\prime}^{i ,j}(\Lambda_{k},\widetilde\Lambda_{k}) \geq 0\), and therefore so is \(c_{\tau,T}^{i,j}(\Lambda_{k},\widetilde\Lambda_{k})\).

    If \(\Gamma(\tau)\) has a non-special, non-\(H\) cycle (by Proposition~\ref{prop:regular cycles}), or if \(T \neq \emptyset \) (by Proposition~\ref{prop:T empty}),  we can now assume that \(c_{\tau,T}^{i,j}(\Lambda_{k},\widetilde\Lambda_{k})\geq 0\). We will now consider \((\tau,T)\in \widetilde{\mathcal T}_{k,2k}\) with \(T =\emptyset\) such that \(\Gamma(\tau)\) has no non-special, non-\(H\) cycles.

    \begin{dfn}
   Fix \(\mu\in\widetilde{\mathcal M}_{(i,j,\tau,T)}\) and a path \(P = \{\{a_1,a_2\},\{a_2,a_3\},...\{a_{q-1},a_q\}\}\) in \(\Gamma\).
   
   A path \(P\) that does not contain any \(J\) or \(H\)-edge will be called a \emph{regular path}.
    A path \(P\) such that \(a_1, a_q\) both belong to a \(H\)-edge and \(P\) itself contains no \(H\)-edges will be called an \emph{\(H\)-path}. Since they belong to an \(H\)-edge, we must have \(\{a_1,a_q\}\subset R_\mu\cup J_\mu\).
    An \(H\)-path that contains a special edge will be called a \emph{special path}. A special path is \emph{\(n\)-special} if it contains exactly \(n\) special edges.
    \end{dfn}

    \begin{prop}
    \label{prop:HJpaths}
        For \(\mu\in\widetilde{\mathcal M}_{(i,j,\tau,T)}\) regular paths must alternate between non-\(H\) \(O\)-edges and non-\(J\) \(\tau\)-edges, and must also alternate between vertices in \(L_\mu\) and \(R_\mu\cup J\mu\). Furthermore, every \(H\)-path is special and must contain exactly one \(J\)-edge, and every \(J\)-edge is contained in an \(H\)-path.
    \end{prop}
    \begin{proof}
    Suppose \(P\) is a regular \(H\)-path.
    Every vertex in \(\Gamma\) is two-regular and belongs to one \(O\)-edge and one \(\tau\)-edge. Since we have that \(a_1\) belongs to a \(H\)-edge (which must be \(O\)) that does not belong to \(P\), we have that \(\{a_1,a_2\}\) must be a non-\(J\) \(\tau\)-edge, and therefore we must have \(a_2\in L_\mu\). Similarly, \(\{a_2,a_3\}\) must then be a non-\(H\) \(O\)-edge, and therefore we must have \(a_3 \in R_\mu \cup J_\mu\). Continuing in that fashion, we get that for any \(2\leq2\ell\leq q\), we have that \(\{a_{2\ell-1},a_{2\ell}\}\) must be a non-\(J\), \(\tau\)-edge with \(a_{2\ell}\in L_\mu\), and  for any \(3\leq2\ell+1\leq q\) we have that \(\{a_{2\ell},a_{2\ell+1}\}\) must be a non-\(H\), \(O\)-edge with \(a_{2\ell+1}\in R_\mu \cup J_\mu\). 
    
    We know \(a_q \) belongs to an \(H\)-edge, thus we have \(a_q\in R_\mu \cup J_\mu\). That means \(q\) must be odd. However, that means \(\{a_{q-1},a_{q}\}\) must be a non-\(H\), \(O\)-edge (which can not be an \(H\)-edge because it is contained in \(P\)), and thus the other edge \(a_q\) belongs to must be a \(\tau\)-edge (which also can not be a \(H\)-edge, as only \(O\)-edges might be \(H\)). We thus have that \(a_q\), does not belong to any \(H\)-edge, which is contradiction. Which means \(P\) must contain a \(J\)-edge. 

    To reiterate: Regular paths must alternate between non-\(H\), \(O\)-edges and non-\(J\), \(\tau\)-edges, and must also alternate between vertices in \(L_\mu\) and \(R_\mu\cup J\mu\). Since the number of vertices and the number of edges in a path are of differing parity, we have that no such path can end on both the same kind of edge and the same kind of vertex. 
    This means we cannot have a regular \(H\)-path. By the same argument, no regular path can connect two vertices in \(R_\mu\cup J_\mu\) that are contained in a \(J\)-edge, and therefore we cannot have a \(J\)-edge in a non-\(H\) cycle. That means that any cycle must contain the same number of \(H\) and \(J\)-edges arranged alternatively around the cycle, and every \(H\)-path must contain exactly one \(J\)-edge. 
    \end{proof}
    Since we have only two \(J\)-edges, we can conclude
    \begin{prop}
    \label{prop:2 H paths}
    For \(\widetilde{\mathcal{M}}_{(i,j,\tau,T)}\) to be non-empty, and thus for 
    \[
    c_{\tau,T}^{i,j}(\Lambda_{k},\widetilde\Lambda_{k})=\sum_{\mu\in\widetilde{\mathcal{M}}_{(i,j,\tau,T)}}(-1)^{d^{i,j}_\mu} 
    \]
    to be non-zero, we must have exactly two \(H\)-paths.
     \end{prop}
     
    \begin{prop}
    \label{prop:all H cycles}
           Let \(C\subset [n]\) be cycle of \(\Gamma\), and let \(V_C\) be its vertices. Suppose that \(C\) is not an \(H\)-cycle. We claim that there exist \((\tau^\prime, T^\prime)\in \widetilde{\mathcal{T}}_{k,2k}\) such that \(T^\prime\neq \emptyset\), and 
    \(
    c^{i,j}_{\tau,T}(\Lambda_{k},\widetilde\Lambda_{k})=2 c^{i,j}_{\tau^\prime,T^\prime}(\Lambda_{k},\widetilde\Lambda_{k})
    .\)
    \end{prop}
    Thus we can assume \(c^{i,j}_{\tau,T}(\Lambda_{k},\widetilde\Lambda_{k})\geq 0\) by induction via Proposition~\ref{prop:T empty}.
    \begin{proof}
    By Proposition~\ref{prop:regular cycles} it is enough to show for \(C\) a special non-\(H\)-cycle.
    
    By Proposition~\ref{prop:HJpaths}, we know that for \(\mu\in\widetilde{\mathcal{M}}_{(i,j,\tau,T)}\) we cannot have that a non-\(H\) cycle would be a \(J\)-cycle. That means that every special non-\(H\) cycle cannot ever be \(J\). Thus it is labeled the same a non-special non-\(H\) cycle and the argument from Proposition~\ref{prop:regular cycles} applies. The only change is that now we are removing some special edges, so let us make sure  \((-1)^{d^{i,j}_{\mu^\prime}}=(-1)^{d^{i,j}_{\mu}}\): By Proposition~\ref{prop:special cycles} we have that \(C\) is either 2-special or 1-special. 

    If it is 1-special, it must be either the first or the last special-edge on the \(i\)-\(j\) line, thus removing it will not change the number of special edges between the \(J\)-edges and thus \((-1)^{d^{i,j}_{\mu^\prime}}=(-1)^{d^{i,j}_{\mu}}\).

    If it is 2-special, since we have that the areas bounded by the cycles being disjoint, its two special edges must be consecutive on the \(i\)-\(j\) line, thus removing them will not change the parity of the number of special edges between the \(J\)-edges and thus \((-1)^{d^{i,j}_{\mu^\prime}}=(-1)^{d^{i,j}_{\mu}}\).

    Thus we have \((-1)^{d^{i,j}_{\mu^\prime}}=(-1)^{d^{i,j}_{\mu}}\), completing the proof.
    \end{proof}

    We can now assume that if \(\Gamma(\tau)\) has non-\(H\) cycle (by Proposition~\ref{prop:all H cycles}), or if \(T \neq \emptyset \) (by Proposition~\ref{prop:T empty}), we have \(c_{\tau,T}^{i,j}(\Lambda_{k},\widetilde\Lambda_{k})\geq 0\). We will now consider \((\tau,T)\in \widetilde{\mathcal T}_{k,2k}\) with \(T =\emptyset\) such that \(\Gamma(\tau)\) has no non-special, non-\(H\) cycles.
    
    \begin{prop}
        Let \(\mu_1, \mu_2 \in\widetilde{\mathcal{M}}_{(i,j,\tau,T)}\) be labels that have the same \(J\)-edges, and  \(\mu_1(\ell) =  \mu_2(\ell)\) for any vertex \(\ell\) that is contained in a \(J\)-edge, then we must have \(\mu_1 = \mu_2\). 
    \end{prop}

    \begin{proof}
    By Proposition~\ref{prop:all H cycles} we can assume \(\Gamma = \Gamma(\tau)\) has no cycles that are not \(H\), and by Proposition~\ref{prop:HJpaths} all \(H\)-cycles must have regular paths connecting alternating \(J\) and \(H\)-edges.
    Let \(P\) be a regular path such that \(a_1,a_q\in R_\mu\cup J\mu\) with \(a_1\) contained in an \(H\)-edge and  \(a_q\) is contained in a \(J\)-edge. The vertices on the path must alternate between \(R_\mu\cup J\mu\) and \(L_\mu\). Since the vertices \(a_2,...,a_{q-1}\) are not contained in any \(J\)-edge, they must alternate between being labeled \(R\) and \(L\) by \(\mu_{1,2}\) and \(\mu_{1,2}(a_2) = \mu_{1,2}(a_{q-1}) = L\). 
    
    As all vertices that do not belong to an \(H\) or \(J\)-edge are contained in such a path, the labels of \(\mu\in\widetilde{\mathcal{M}}_{(i,j,\tau,T)}\) must are set on any vertices that do not belong to an \(H\) or \(J\)-edge. As vertices on \(H\)-edges that do no belong to a \(J\)-edge must be labeled \(R\), we have that \(\mu\) is set on any vertex that does not belong to a \(J\)-edge.
    \end{proof}

    \begin{prop}
    \label{prop:two special}
        Suppose \(i,j\) and \((\tau,T=\emptyset)\in\widetilde{\mathcal{T}}_{k,2k}\) are such that \(\Gamma\) has a 2-special \(H\)-path.
        Then
        \(
        c^{i,j}_{\tau,T}(\Lambda_{k},\widetilde\Lambda_{k})=0.
        \)
    \end{prop}
    \begin{proof}
        Let \(P\) be a 2-special \(H\)-path that is contained in a cycle  in \(\Gamma = \Gamma(\tau)\). Let those two special edges be \(e_1, e_2\in P\).
        Since every special \(H\)-path must contain exactly one \(J\)-edge (by Proposition~\ref{prop:HJpaths}), we know that for \(\mu \in\widetilde{\mathcal{M}}_{(i,j,\tau,T)}\) either \(e_1\) or \(e_2\) is a \(J\)-edge of \(\mu\). Since \(P\) is contained in a cycle, and \(\Gamma\) has that areas bounded by the cycles are disjoint (Proposition~\ref{prop:special cycles}), we must have that \(e_1\) and \(e_2\) are adjacent on the \(i\)-\(j\) line. Let us write \(e_1 = \{\ell_1,\tau(\ell_1)\}\) and \(e_2 = \{\ell_2,\tau(\ell_2)\}\) for \(\ell_{1,2}<\tau(\ell_{1,2})\), and
        \(
    \widetilde{\mathcal{M}}_{(i,j,\tau,T)}^{E} \defeq \{ \mu\in\widetilde{\mathcal{M}}_{(i,j,\tau,T)} : E \text{ are \(J\)-edges of } \mu \}
    ,\)
    for \(E\) a set of special edges.
    We have a bijection \(\widetilde{\mathcal{M}}_{(i,j,\tau,T)}^{\{e_1\}}\rightarrow\widetilde{\mathcal{M}}_{(i,j,\tau,T)}^{\{e_2\}}\) by \(\mu\mapsto\mu^\prime\) by switching the labeling of the edges \(e_1\) and \(e_2\): 
    \[
    \mu^\prime (\ell) = 
    \begin{cases}
        \mu(\ell_1), &\ell=\ell_2\\
        \mu(\ell_2), &\ell=\ell_1\\
        \mu(\tau(\ell_1)), &\ell=\tau(\ell_2)\\
        \mu(\tau(\ell_2)), &\ell=\tau(\ell_1)\\
        \mu(\ell), &\text{otherwise}.\\
        
    \end{cases}
    \]

    Since \(e_1\) and \(e_2\) are adjacent on the \(i\)-\(j\) line, we have that \((-1)^{d^{i,j}_\mu} = -(-1)^{d^{i,j}_{\mu^\prime}} \), and thus:
    \begin{align*}
        c_{\tau,T}^{i,j}(\Lambda_{k},\widetilde\Lambda_{k})\,=\sum_{\mu\in\widetilde{\mathcal{M}}_{(i,j,\tau,T)}}(-1)^{d^{i,j}_\mu}&=\sum_{\mu\in\widetilde{\mathcal{M}}_{(i,j,\tau,T)}^{\{e_1\}}}(-1)^{d^{i,j}_\mu}+\sum_{\mu\in\widetilde{\mathcal{M}}_{(i,j,\tau,T)}^{\{e_2\}}}(-1)^{d^{i,j}_\mu}\\
        &=\sum_{\mu\in\widetilde{\mathcal{M}}_{(i,j,\tau,T)}^{\{e_1\}}}(-1)^{d^{i,j}_\mu}+\sum_{\mu\in\widetilde{\mathcal{M}}_{(i,j,\tau,T)}^{\{e_1\}}}(-1)^{d^{i,j}_{\mu^\prime}}\\
        &=\sum_{\mu\in\widetilde{\mathcal{M}}_{(i,j,\tau,T)}^{\{e_1\}}}(-1)^{d^{i,j}_\mu}-\sum_{\mu\in\widetilde{\mathcal{M}}_{(i,j,\tau,T)}^{\{e_1\}}}(-1)^{d^{i,j}_{\mu}} = 0.
    \end{align*}
    \end{proof}

    Meaning that for any \(i,j\) and \((\tau,T)\in \widetilde{\mathcal T}_{k,2k}\) such that \(\Gamma(\tau)\) has a 2-special \(H\)-path, we have \( c_{\tau,T}^{i,j}(\Lambda_{k},\widetilde\Lambda_{k}) = 0\).

    \begin{prop}
    \label{prop:2 1 special}
        Suppose \(i,j\) and \((\tau,T=\emptyset)\in\widetilde{\mathcal{T}}_{k,2k}\) are such that \(\Gamma\) has a exactly two 1-special \(H\)-paths, and no other \(H\)-paths.
        Then
        \(
        c^{i,j}_{\tau,T}(\Lambda_{k},\widetilde\Lambda_{k})\geq0.
        \)
    \end{prop}

    \begin{proof}
    We have that \(\Gamma\) has exactly two 1-special \(H\)-paths \(P_1\) and \(P_2\) with special edges \(e_1\) and \(e_2\) respectively. Since for any \(\mu\in\widetilde{\mathcal M}_{(i,j,\tau,T)}\) must have that every \(H\)-path has exactly one \(J\)-edge (by Proposition~\ref{prop:HJpaths}), those two special edges must be \(J\). 

     If \(P_1\) and \(P_2\) are contained in the same cycle \(C\), since \(\Gamma\) has that areas bounded by the cycles are disjoint (Proposition~\ref{prop:special cycles}), we must have that \(e_1\) and \(e_2\) are adjacent on the \(i\)-\(j\) line. This means \(d^{i,j}_\mu\) must be zero, and thus
     \[
    c_{\tau,T}^{i,j}(\Lambda_{k},\widetilde\Lambda_{k})=\sum_{\mu\in\widetilde{\mathcal{M}}_{(i,j,\tau,T)}}(-1)^{d^{i,j}_\mu}>0.
    \]

    Now assume \(P_1\) and \(P_2\) are not contained in the same cycle, but rather in cycles \(C_1\) and \(C_2\) respectively. As \(P_1\) is 1-special, \(C_1\) is 1-special as well. Indeed, if \(C_1\) were 2-special, it means it contains two \(H\)-paths (by Proposition~\ref{prop:HJpaths}) -- but we said \(\Gamma\) has exactly two \(H\)-paths that are not contained in the same cycle. So we must conclude both \(C_1\) and \(C_2\) are 1-special. Since any cycle that crosses the \(i\)-\(j\) line that is not the first or the last must be 2-special (by Proposition~\ref{prop:special cycles}), we get that any special cycle between \(e_1\) and \(e_2\) is 2-special. That means the number of special edges between them must be even, and thus \((-1)^{d^{i,j}_\mu}\) is positive. Therefore 
    \[
    c_{\tau,T}^{i,j}(\Lambda_{k},\widetilde\Lambda_{k})=\sum_{\mu\in\widetilde{\mathcal{M}}_{(i,j,\tau,T)}}(-1)^{d^{i,j}_\mu}>0.
    \]
    \end{proof}

    Now we can finally finish the proof of Lemma~\ref{pos on L0 lemma}: By Proposition~\ref{prop:T empty} we can assume \(T=\emptyset\). We know \(\Gamma\) is a set of disjoint cycles by Proposition~\ref{prop:special cycles}. By Proposition~\ref{prop:all H cycles} we can assume \(\Gamma\) contains only \(H\)-edges. By Proposition~\ref{prop:cMtilde} we have that 
    \[
     c_{\tau,T}^{i,j}(\Lambda_k,\widetilde\Lambda_k)=\sum_{\mu\in\widetilde{\mathcal{M}}_{(i,j,\tau,T)}}(-1)^{d^{i,j}_\mu}.
     \]
     By Proposition~\ref{prop:2 H paths} we can assume \(\Gamma\) has exactly two \(H\)-paths, that can be either 2 or 1-special (by Proposition~\ref{prop:special cycles}). By Proposition~\ref{prop:two special} we can assume \(\Gamma\) has only 1-special paths. By Proposition~\ref{prop:2 1 special} we know that in those cases we have \(c^{i,j}_{\tau,T}(\Lambda_{k},\widetilde\Lambda_{k})\geq0\). Therefore we have
        \(
        c^{i,j}_{\tau,T}(\Lambda_{k},\widetilde\Lambda_{k})\geq0,
        \)
    finally finishing the proof of Lemma~\ref{pos on L0 lemma}.
\end{proof}
\subsection{The Boundary of the Amplituhedron}
\label{sec:ext boundaries Pos}

In this section we will prove Theorem~\ref{thm external boundaries}. In order to do so, we will take a closer look at external boundaries defined in Corollary~\ref{coro: codim 1 boudnary classification} and apply the tools developed in Section~\ref{separation section}.

\begin{prop}
\label{prop:ext bound cont mand sep}
     Let \(v\) be an internal vertex in a \(k\)-BCFW graph \(\Gamma\), such that \(\partial_v \Gamma\) is an external boundary. Then there exist a cyclically consecutive \(I\subset[2k]\) of odd length, such that the Mandelstam variable \(S_I\) is a vertex-separators of \(v\).
     \end{prop}
     \begin{proof}[Proof of Proposition~\ref{prop:ext bound cont mand sep}]
    By Corollary~\ref{coro: codim 1 boudnary classification}, we have that \(\Gamma\) and \(v\) as in Figure~\ref{fig:graph with ext boundary}. 
    \begin{figure}[H]
    \centering
\begin{center}
\begin{tikzpicture}[scale=0.8, every node/.style={scale=0.8}]
\draw (0,0) circle (2cm);

\filldraw[lightgray] (-.9,1/2) circle (0.6cm);
\node[scale=2.5] (c) at (-.9,1/2)  {\(\Gamma_1\)};
\draw (-.9,1/2) circle (.6cm);

\draw (-1.38541, 0.147329) --(-1.97538, -0.312869);
\draw (-0.806139, 1.09261) --(-0.618034, 1.90211)node[anchor=south]{\(j+1\)};
\filldraw[black] (-1.49522, 0.988483) circle (1.75pt);
\filldraw[black] (-1.14384, 1.30383) circle (1.75pt);
\filldraw[black] (-1.69938, 0.531408) circle (1.75pt);

\draw  (-0.475736, 0.0757359)--(0.415823, -1.9563)node[anchor=north]{\(i+1\)};

\filldraw[lightgray] (.9,1/2) circle (0.6cm);
\node[scale=2.5] (c) at (.9,1/2)  {\(\Gamma_2\)};
\draw (.9,1/2) circle (.6cm);

\draw (1.38541, 0.147329) --(1.97538, -0.312869);
\draw (0.806139, 1.09261) --(0.618034, 1.90211)node[anchor=south]{\(j\)};
\filldraw[black] (1.49522, 0.988483) circle (1.75pt);
\filldraw[black] (1.14384, 1.30383) circle (1.75pt);
\filldraw[black] (1.69938, 0.531408) circle (1.75pt);

\draw  (0.475736, 0.0757359)--(-0.415823, -1.9563)node[anchor=north]{\(i\)};

\draw  (-0.67039, -0.0543277)--(0.67039, -0.0543277);
\draw (0, -0.8)node[anchor=north east]{\(v\)};

 \end{tikzpicture}
 \end{center}
    \caption{a graph with an external boundary}
    \label{fig:graph with ext boundary}
\end{figure}
We will prove the stronger statement that \(S_{\{i+1,i+2,...,j\}}\) and \(S_{\{j+1,i+2,...,i\}}\) are both vertex-separators of \(v\), and that both sets are of odd length. By Corollary~\ref{coro:BCFW are rad}, \(v\) is \(4\)-native. We will prove it by induction on the length of the ancestry-sequence of \(v\). By Definition~\ref{def:vertex sep}, as complimentary Mandelstam variables are equivalent by Proposition~\ref{prop:S orth}, it enough to show only one of those Mandelstam variables is a vertex-separator.

    For the base case, consider \(v\) a vertex on an external arc. The claim is now trivially true by Definition~\ref{def:vertex sep}.
    
    For the induction step, suppose \(\Gamma   = \mathrm{Arc}_{4,n}(\Gamma^\prime)\) and \(v = \mathrm{Arc}_{4,n}(v^\prime)\) are such that \(\partial_v \Gamma\) is an external boundary, that is, \(\Gamma\) and \(v\) are as seen in Figure~\ref{fig:graph with ext boundary}. Then new external arc added by the \(\mathrm{Arc}\) move is either in \(\Gamma_1\) or \(\Gamma_2\). Thus we must have that \(\Gamma^\prime\) and \(v^\prime\) are also as seen in Figure~\ref{fig:graph with ext boundary} as well. Let us label the vertices marked in the Figure as \(j^\prime+1,j^\prime,i^\prime,i^\prime+1\) respectively for this case. We clearly have that \(i = \mathrm{Arc}_{4,n}(i^\prime)\) and \(j = \mathrm{Arc}_{4,n}(j^\prime)\). By Corollary~\ref{coro: codim 1 boudnary classification}, we have that \(\partial_{v^\prime}\Gamma^\prime\) is an external boundary, and by Theorem~\ref{thm:BCFW are trees}, we have that \(\Gamma^\prime\) is a BCFW graph.

    Write \(I^\prime = \{i^\prime+1,i^\prime+2,...,j^\prime\}\) and \(J^\prime =\{j^\prime+1,i^\prime+2,...,i^\prime\}\). By the induction hypothesis, we have that that \(S_{I^\prime} \) and \(S_{J^\prime}\) are both vertex-separators of \(v^\prime\) and of odd length. By Definition~\ref{def:vertex sep}, we have that \(\mathrm{Arc}_{4,n}(S_{I^\prime})\) and \(\mathrm{Arc}_{4,n}(S_{J^\prime})\) are both vertex-separators of \(v\). By Theorem~\ref{thm:BCFW are trees}, we have that \(n\) is even. Recall that by Definition~\ref{def:arc}
     \(
    \mathrm{Arc}_{4,n}
    \defeq \mathrm{Rot}_{n+2,n+3}\mathrm{Rot}_{n+1,n+2}\mathrm{Inc}_{n}.
    \)
    By Definition~\ref{abs substit def} we have that \(\mathrm{Inc}_{n}(S_{I^\prime}) = S_{\mathrm{Inc}_{n}(I^\prime)}\) and \(\mathrm{Inc}_{n}(S_{J^\prime}) = S_{\mathrm{Inc}_{n}(J^\prime)}\). We clearly have that either \(\mathrm{Inc}_{n}(I^\prime)\) or \(\mathrm{Inc}_{n}(J^\prime)\) are consecutive. If it is the former, write \(I = \mathrm{Inc}_{n}(I^\prime)\) and \(J = [2k]\setminus I\). If the latter, write \(J = \mathrm{Inc}_{n}(J^\prime)\) and \(I = [2k]\setminus J\). \(I\) and \(J\) are now both consecutive, and since \(I^\prime\) and \(J^\prime\) are both of odd length, so are \(I\) and \(J\). 
    
    Since the arc added by the \(\mathrm{Arc}\) move is either in \(\Gamma_1\) or \(\Gamma_2\), the effect of the \(\mathrm{Arc}\) move on the indices \(i^\prime,i^\prime+1,j^\prime,j^\prime+1\) is the same as the \(\mathrm{Inc}\) move alone. Therefore, we have that \(I=\{i+1,i+2,...,j\}\) and \(J=\{j+1,i+2,...,i\}\). \(S_I\) and \(S_J\) are vertex-separators of \(\mathrm{Inc}_n (v^\prime)\) in the graph \(\mathrm{Inc}(\Gamma^\prime)\) by Definition~\ref{def:vertex sep}. As the arc added by the \(\mathrm{Arc}\) move is either in \(\Gamma_1\) or \(\Gamma_2\), we must have that the support of the arc,  \(\{n,n+1,n+2,n+3\}\) is contained in either \(I\) or \(J\). Thus, by Observation~\ref{obs:s rot}, we have that \(S_I\) and \(S_J\) are invariant under the \(\mathrm{Rot}\) moves, and therefore \(S_I=\mathrm{Arc}_{4,n}(S_{I^\prime})\) and \(S_J=\mathrm{Arc}_{4,n}(S_{J^\prime})\) which are vertex-separators of \(v\).
\end{proof}

\begin{proof}[Proof of Theorem~\ref{thm external boundaries}]
 By Corollary~\ref{coro: codim 1 boudnary classification} we have that \(\Gamma_0 = \partial_v\Gamma\) for some \(v\), an internal vertex of the \(k\)-BCFW graph \(\Gamma\). Meaning that by Proposition~\ref{prop:ext bound cont mand sep} we can find \(S\), a vertex-separator of \(v\), such that \(S=S_I\) with \(I\subset [2k]\) being cyclically consecutive.

 By Observation~\ref{obs:BCFW 4 native}, we have that \(v\) is \(4\)-native. Thus by Proposition~\ref{prop:sep 4-native}, have that \(S\) is positively-radical and zero on \(\Gamma_0\), and \(\evalat{\frac{\partial}{\partial v }S}{\Gamma_0}\) is positive. By Observation~\ref{obs:vertex derivative orientation}, we have that \(\evalat{\frac{\partial}{\partial v^\omega }S}{\Gamma_0}\) is positive for any perfect orientation \(\omega\) on \(v\). By Definition~\ref{def:vertex derivative}, we have that \(\evalat{\frac{\partial}{\partial v^\omega }S}{\Gamma_0}\) is precisely the derivative of \(S(\Lambda,\bullet)\) on \(\widetilde \Lambda(\Omega_{\Gamma_0})\) in the direction defined by increasing \(v^\omega\). Meaning \(S(\Lambda,\bullet)\) is zero with a non-vanishing derivative on \(\widetilde \Lambda(\Omega_{\Gamma_0})\). 
 
 Fix \(x\in\widetilde\Lambda(\Omega_{\Gamma_0})\). By the previous paragraph, \(S(\Lambda,x)\) is zero with a non-vanishing derivative. This means that for any neighborhood \(U\ni x\), there exists a point \(y\in U\) with \(S(\Lambda,y)<0\). By Theorem~\ref{nonneg mandelstam thm}, \(S(\Lambda,y)\geq 0\) for any \(y\in \mathcal O_k(\Lambda)\). We can conclude \(x\) is in the boundary of \(\mathcal O_k(\Lambda)\), and thus \(\widetilde\Lambda(\Omega_{\Gamma_0})\) is contained in the boundary of \(\mathcal O_k(\Lambda)\).
\end{proof}

\section{The BCFW Tiling of the ABJM Amplituhedron}\label{sec:tiling}
In this section, we prove Theorem~\ref{thm:tilings}. The proof proceeds in two steps. First, Proposition~\ref{prop:const_deg} establishes, using results from earlier sections and the topological tools developed here, that the degree of the amplituhedron map is constant over the union of BCFW cells. Then, Proposition~\ref{prop:deg_1} shows that this constant degree is in fact equal to \(1\). We begin by stating these main propositions and then proceed to prove Theorem~\ref{thm:tilings}.
\begin{dfn}
    Define $\Ext_k$ to be the set of orthitroid cells that codimension $1$ external boundaries of cells in $\BCFW_k$.
\end{dfn}
Recall Observation~\ref{co-dime 1 external obs} for a characterization such OG graphs.
\begin{prop}\label{prop:const_deg}
Fix a strongly positive $\Lambda.$ Then 
$\widetilde{\Lambda}(\bigcup_{\Omega\in\BCFW_k}\Omega)$ is dense in the amplituhedron.
Let $S$ be the union of all BCFW cells and their internal boundaries of codimension $1.$ Then $\widetilde{\Lambda}(S)$ is an open dense subset of the amplituhedron, and, moreover, for an open dense subset of this space there is a constant number of preimages in $S.$ In addition, 
\(\partial\mathcal{O}_k(\Lambda)=\bigcup_{\Omega\in\Ext_k}\widetilde\Lambda\left(\overline{\Omega}\right).\)
\end{prop}
We call this number the \emph{degree} of the amplituhedron map on $S$.
\begin{coro}\label{cor:boundaries}
The boundaries of the ABJM amplituhedron, for strongly positive \(\Lambda\), are the closures of images of external boundaries of BCFW cells.    
\end{coro}
\begin{prop}\label{prop:deg_1}
    Fix $k\geq 4.$ Assume that for every  \(\Lambda^{k-1}\in \mathcal L^>_{k-1}\), the \(\widetilde \Lambda^{k-1}\) images of different orthitroid cells represented by graphs of $\BCFW_{k-1}$ do not intersect, then the same holds for  \(\widetilde \Lambda^{k}\) images of cells represented by graphs of $\BCFW_{k}$ for every  \(\Lambda^{k}\in \mathcal L^>_{k}\).
\end{prop}
\begin{proof}[Proof of Theorem~\ref{thm:tilings}]
The case $k=3$ follows from Theorem~\ref{BCFW inj}, since in this case $\BCFW_3$ consists of a single graph. The proof is now a simple induction. Assume the theorem holds for $k-1.$ Then by Proposition~\ref{prop:const_deg} the union of images of BCFW cells in $\BCFW_k$ are dense in the amplituhedron, they are locally separated and of constant degree. This degree is $1,$ by Proposition~\ref{prop:deg_1} and the induction. Thus, images of different BCFW cells are also disjoint. And the induction follows. 
\end{proof}

\subsection{Constant Degree}
In order to prove Proposition~\ref{prop:const_deg}, we need some topological preparations.
The following lemma is a slight adjustment of a very nice argument from \cite{4168411,4169287}.
\begin{lem}\label{lem:connectedness_in_complement_stratification}
Let $M$ be smooth connected manifold of dimension $n$, and $S_1,\ldots,S_N$ smooth manifolds of dimensions at most $n-2.$ Let $f_i:S_i\to M$ be smooth injections, for $i=1,\ldots, N.$ Then $M\setminus \bigsqcup f_i(S_i)$ is connected.
\end{lem}
\begin{proof}
Write $S=\coprod S_i,~f=\coprod f_i.$ Our convention is that manifolds are second countable. Therefore, we can find a countable cover of $S$ by subsets $D_j,j=1,2,\ldots$ that are diffeomorphic to closed disks in $\mathbb R^{n-a_j},$ for $a_j\geq 2.$ We can also assume that each $D_i$ is contained in the interior of another closed disk $D'_i$, satisfying the same properties. 

Every differential connected manifold has a complete metric \cite{complete_metric}. Recall that the space of maps from a compact space $X$ to a complete metric space $Y$, endowed with the compact-open topology, has a complete metric inducing the topology, see for example \cite[Theorem 2.4.1]{Hirsch}. This metric can be taken to be $d_\infty(f,g)=\sup_{x\in X}d(f(x),g(x)),$ where $d(\cdot)$ is the complete metric on $Y.$

Write $\mathcal{P}_{x,y}(N)$ for the space of paths from $x$ to $y$ in $N,$ endowed with the compact-open topology. Then by the previous paragraph $\mathcal{P}_{x,y}(N)$ is completely metrizable.

Note that since each $D_i$ is a closed disk of dimension at most $n-2,$ and hence compact, $M$ is Hausdorff and $\evalat*{f}{D_i}$ is injective, every $f(D_i)$ is a disk of dimension at most $n-2$ \emph{topologically} embedded in $M.$ 

We will now show
\begin{prop}\label{prop:global}
Let $M$ be a manifold, and $x,y\in M$. Let  $D$ be a smooth closed disk of dimension at most $n-2,$ smoothly embedded in the interior of another closed topological disk $D'$ of the same dimension, both are mapped by a smooth injection to $M\setminus\{x,y\}.$ Then $\mathcal{P}_{x,y}(M\setminus D)$ is open and dense in $\mathcal{P}_{x,y}(M).$  
\end{prop}
This would imply the lemma, as 
\(\mathcal{P}_{x,y}(M\setminus f(S)) = \bigcap_i \mathcal{P}_{x,y}(M\setminus f(D_i))\) is the intersection of countably many dense open sets. Since $\mathcal{P}_{x,y}(M)$ is completely metrizable, it is a Baire space, hence this intersection is non empty and dense.

In order to prove Proposition~\ref{prop:global}, we will use a nice argument of Moshe Kohan \cite{4168411}.
\begin{prop}\label{prop:kohan}
Let $U$ be a connected, orientable manifold of dimension $n$ and $C$ a closed subspace homeomorphic to a $n-2$-dimensional manifold. Then $U\setminus C$ is connected.
\end{prop}
\begin{proof}
We first prove Proposition~\ref{prop:kohan}. 
By Poincar\'e-Lefschetz duality as written in \cite[Proposition 7.14, Ch. VIII]{Dold}, we have that
\(\check{H}^i_c(C)\simeq H_{n-i}(U,U\setminus C),\) where the left hand side the the \v Cech cohomology with compact support, and the right hand side is the relative singular homology. 
\(\check{H}^i_c(C)\simeq H^i_c(C),\) where the right hand side is the singular homology, since $C$ is a topological manifold. Thus, 
\(H^{n-1}(C)=0\Rightarrow H_1(U,U\setminus C)=0.\)
Using the long exact sequence for homology we get
\[\ldots\rightarrow H_1(U)\to H_1(U,U\setminus C)\simeq0\rightarrow \tilde{H}_0(U\setminus C)\rightarrow\tilde{H}_0(U)\simeq0,\]
where $\tilde{H}_*$ is the reduced homology, and the last equality is since $U$ is connected. Therefore
\(\tilde{H}_0(U\setminus C)\simeq 0,\) implying $U\setminus C$ is connected.
\end{proof}

Returning to Proposition~\ref{prop:global}, we denote by $d(\cdot)$ a fixed complete metric on $M,$ and by $d_\infty$ for the induced metric on $\mathcal{P}_{x,y}(M).$
We have added the technical requirement on $D',$ since it allows us to apply Morse-Sard's theorem \cite[Chapter 3]{Hirsch}
on the smooth map $\evalat*{f}{\text{int}(D')}:\text{int}(D')\to M$ between the manifolds $\text{int}(D')$ and $M,$ to deduce that every open subset of $M$ contains points from $M\setminus \text{int}(D'),$ hence also from $M\setminus D.$ 

Since $D$ is compact $\mathcal{P}_{x,y}(M\setminus D)$ is open in $\mathcal{P}_{x,y}(M).$ Indeed, if $\gamma\in\mathcal{P}_{x,y}(M\setminus D)$ then
\(\gamma:[0,1]\to M\setminus D,\) with \(\gamma(0)=x,~\gamma(1)=y\). Write
\[\rho=\inf_{z\in[0,1],a\in D}d(\gamma(z),a).\] $\rho>0$ by the compactness of $[0,1],~D.$ Thus, every $\gamma'\in\mathcal{P}_{x,y}(M)$ with $d_\infty(\gamma,\gamma')<\rho/2$ belongs to $\mathcal{P}_{x,y}(M\setminus D).$

As for denseness, let $\gamma:[0,1]\to M$ be an element of $\mathcal{P}_{x,y}(M).$ Then $\gamma^{-1}(D)$ is a closed subset of $[0,1]$ which does not include \(0\) or \(1\).
For $\epsilon>0$ we will find an element $\gamma'\in\mathcal{P}_{x,y}(M\setminus D)$ with $d_\infty(\gamma,\gamma')<\epsilon.$

By using Lebesgue's covering lemma we pick $n$ large enough so that the diameter of $\gamma([\frac{i}{n},\frac{i+1}{n}])$ is smaller than $\epsilon/2$ for $i=0,1,\ldots,n-1.$ 
Denote by $U_i=U_i^\epsilon$ the open set
\[U_i=\{x\in M|\inf_{t\in[\frac{i}{n},\frac{i+1}{n}]}d(x,\gamma(t))<\epsilon\}.\]We may assume $\epsilon$ is small enough so that all these sets are orientable and connected.

By the Morse-Sard argument above, we can find a point $x_i\in U_i\setminus D,$ for $i=1,\ldots,n-1.$ Write $x_0=x,x_n=y.$
We now apply Proposition~\ref{prop:kohan} with $U_i$ as the ambient manifold, and $D\cap U_i$ as the embedded closed subset to find a path \[\gamma'_i:[\frac{i}{n},\frac{i+1}{n}]\to U_i\setminus D,~~~~~\text{with }\gamma'_i(\frac{i}{n})=x_i,~\gamma'_i(\frac{i+1}{n})=x_{i+1}.\]We glue these paths along their endpoints to obtain $\gamma'\in\mathcal{P}_{x,y}(M\setminus D)$ which by construction satisfies
\(d_\infty(\gamma',\gamma)<\epsilon.\)
\end{proof}

The following lemma is a modification of \cite{amptriag}[Proposition 8.5] to our needs.
We recall that a manifold is assumed to be without boundary unless explicitly specified otherwise.
\begin{prop}\label{prop:criterion_for_positive_const_deg}
Let $f:B\to N$ be a smooth submersion between two manifolds (\emph{without boundary}) $B,N,$ where the dimension of $N$ is $n.$ Let $L$ be a connected open subset of $B$ with a compact closure $\overline{L}\subset B.$ Denote $f(\overline{L})$ by $K.$
Let $\{S_a\}_{a\in A}$ be a collection of $n$-dimensional submanifolds  of $B,$ which are contained in $\overline{L}$ and satisfy the following properties:
\begin{enumerate}
\item $\overline{S_a}$ is compact. The (topological) boundary of each $S_a$ has a stratification $\partial S_a = \bigcup_{i=1}^{k_a}S_{a;i}$ where each $S_{a;i}$ is a submanifold of $B$ of dimension at most $n-1.$ We moreover assume that the union of $S_a$ with the dimension $n-1$ boundaries $S_{a;i}$ is a \emph{manifold with boundary}, and that the closure of every $S_{a;i}$ is the union of $S_{a;i}$ with other spaces $S_{a;j}$ of smaller dimensions.
    \item 
    For every $S_{a;i}$ of dimension $n-1$ either $f(S_{a;i})$ is contained in $\partial K,$ or  $S_{a;i}=S_{b;j}$ for exactly one other $b,j.$ We call the former boundaries \emph{external} and the latter \emph{internal}. 
    \item Write $S'_a$ for the manifold with boundary which is the union of $S_a$ and its internal boundaries. Then $S'_a\cap S'_b$ is either empty, or it is the union of the common internal boundaries of $S_a$ and $S_b.$
    \item $f$ is injective on every $S_a$ and every boundary stratum $S_{a;i}.$ In addition, if $S_{a;i}=S_{b;j}$ are internal boundaries, and if $S_{a,b}$ is the (topological) manifold obtained by gluing the manifolds with boundaries $S_a\cup S_{a;i},$ and $S_b\cup S_{b;j}$ along the common $S_{a;i}=S_{b;j}$, then $f$ is locally injective near every $y\in S_{a;i}\hookrightarrow S_{a,b}.$
\end{enumerate}
Then 
\begin{enumerate}
\item\(f(\bigcup_{a\in A}\overline{S_a})=K.\)
\item Write $S$ for the space obtained by gluing the manifolds with boundaries $S'_a,$ along the identified internal boundaries, and $R$ is the union of spaces $S_{a;i}$ of dimensions at most $n-2.$ Then $f(S)$ is open, $f(S)\setminus f(R)$ is open and dense in $f(S),$ and every point of $f(S)\setminus f(R)$ has the same number of preimages in $S.$

\item $\partial K$ is the union of $\overline{f(S_{a;i})}$ taken over the external boundaries $S_{a;i}.$
\end{enumerate}
\end{prop}

\begin{proof}
Items 1,3 guarantee that $S$ is a (topological) manifold. Item 4 guarantees it is a topological cover of its image. Thus, $f(S)$ is open. By Morse-Sard's theorem \cite[Chapter 3]{Hirsch}, for example, $f(S)\setminus f(R)$ is non empty and dense in $f(S)$. By Item 1 it is also open.

Assume towards contradiction that $\overline{f(S)}\neq K$. Since $\overline{S}$ is compact,  also $\overline{f(S)}$ is. 
Since $f$ is a submersion, $f(L)$ is open. Since $\overline{L}$ is compact, the closure of $f(L)$ is $K.$ This implies that $K$ is also the closure of its interior, hence $U:=\text{int}(K)\setminus\overline{f(S)}$ is open, and, by assumption, non empty.
Take $q\in U, p\in f(S)\setminus f(R).$ 
Since $L$ is connected, also $f(L)$ is. By Item 4 and Lemma~\ref{lem:connectedness_in_complement_stratification}, $f(L)\setminus f(R)$ is connected. 

Connect $p,q$ by a path
\(u:[0,1]\to f(L)\setminus f(R),~u(0)=p,~u(1)=q.\) This path does not pass through the boundary $\partial K,$ since $f(L)$ is open. Since $f(S)$ is open, 
\(\{t\in[0,1]~|~u(t)\in f(S)\}\) is open. Since $u(0)\in f(S)$ there must exist a minimal $t$ with $u(t)\notin f(S).$ Thus $u(t)\in f(\overline{S}).$ But $\overline{S}\setminus S$ is contained in the union of $R$ and the external boundaries, whose $f-$images are missed by $u(t)$. This contradiction shows the first statement. 

The second part is proven similarly, assume towards contradiction that there are $p,p'\in f(L)\setminus f(R)$ with different number of preimages $S,$ being $m\neq m'$ respectively. We again connect them by a path in $f(L)\setminus f(R)$, and consider the minimal time $t$ for which $u(t)$ has a number of preimages different than $m,$ the number of preimages of $u(0).$ The same argument shows $u(t)$ must either have a preimage in $R$ or in the external boundaries which map to $\partial K,$ which is again impossible. 

Regarding the last part, as \(\overline{S_{a;i}}\) is compact as a closed subset of \(\overline{S_{a}}\), which is compact by assumption. Therefore, we have that $\overline{f(S_{a;i})}=f(\overline{S_{a;i}})$. Clearly the union of $\overline{f(S_{a;i})}$ taken over external boundaries is contained in $\partial K,$ since each such $f(S_{a;i})\subseteq \partial K,$ and $\partial K$ is closed.
For the other containment, assume the existence of \[y\in \partial K\setminus\bigcup_{S_{a;i}~\text{is external}}\overline{f(S_{a;i})}.\]
Then since the removed set is closed, $y$ has a neighborhood $B$ not intersecting it. Since $y$ is a boundary point, $B$ must contain a point $q\notin K,$ and since $K=\overline{f(S)\setminus f(R)}$ also a point  
$p\in f(S)\setminus f(R).$ As above, we can connect $p,q$ by a path contained in $B\setminus f(R)$, and we will reach the exact same contradiction as above when we test the minimal time for which the path leaves $f(S),$ and observing that at this time the path must either intersect $f(R)$ or the image of an external boundary.
\end{proof}
\begin{proof}[Proof of Proposition~\ref{prop:const_deg}]
This proposition is an immediate consequence of Proposition~\ref{prop:criterion_for_positive_const_deg} once we verify its conditions.

In our setting $B$ is a small neighborhood of $\OGnon{k}{2k}$ in the orthogonal Grassmannian, $N$ is the zero locus of the momentum conservation equations in the target Grassmannian $\Gr_{k,k+2}$ and $f=\widetilde{\Lambda}.$ By Proposition~\ref{prop:extends to a nbhd}, \(f\) extends to \(B\) as a submersion, and by Proposition~\ref{prop:zerolocus submfld}, \(N\) is a smooth manifold. 
The spaces $S_a$ are the BCFW cells $\Omega\in\BCFW_k.$ 

Condition 1 is satisfied by Proposition~\ref{prop:closure of orthit}.
Condition 2 is a consequence of Corollary~\ref{coro: codim 1 boudnary classification} and Theorem~\ref{thm external boundaries}. 
Condition 3 is immediate in our case.
Condition 4 follows from Theorem~\ref{BCFW inj}, and Theorem~\ref{thm:local sep}.
Thus, we may apply Proposition~\ref{prop:criterion_for_positive_const_deg}, finishing the proof. 
\end{proof}

\subsection{The Degree Is \texorpdfstring{\(1\)}{1}}
The goal of this section is to prove Proposition~\ref{prop:deg_1}. We rely on several preparatory lemmas.
\begin{prop}
\label{prop:S345 arc}
    For \(\Lambda\in \mathcal L^>_{k}\) and \(Y = \widetilde \Lambda (C)\) for \(C \in \OGnon{k}{2k}\), we have that \(S_{\{3,4,5\}}(\Lambda,Y) =0\) iff \(C\in\Omega_\Gamma\) with \(\Gamma\) having an external arc with support contained in \(\{3,4,5\}\). Moreover, if ${\Gamma_0}$ is any boundary graph of a BCFW graph, and for some $C\in\Omega_{\Gamma_0}$ $S_{\{3,4,5\}}(\Lambda,Y)$ vanishes on $\widetilde{\Lambda}(C)$, then it vanishes on all $\widetilde{\Lambda}(\Omega_{\Gamma_0}).$
\end{prop}
\begin{obs}\label{obs:3vecsupport}
\(V\in \Gr{k}{n}\) contains a non-zero vector with support contained in \(J\subset [n]\) iff \(\Delta_{I}(V) = 0\) for all \(I\in   \binom{[n]\setminus J}{k}\). 
\end{obs}
\begin{proof}
If $V$ contains such a vector, let $C$ be a matrix representation containing this vector as a row. Then clearly $\Delta_I(C)=0$ for every $I\in\binom{[n]\setminus J}{k}.$ 

Conversely, suppose $\Delta_I(V)=0$ for every $I\in\binom{[n]\setminus J}{k}.$ Take a matrix representation $C$ for $V.$ Its restriction to the columns $[n]\setminus J$ is of rank smaller than $k.$ Thus, there is a non trivial linear combination $v$ of $C'$s rows whose restriction to the columns $[n]\setminus J$ is $0.$ 
Since $\text{rk}(C)=k,~v\neq 0.$ Thus, $v\in V$ is a non zero vector supported on the entries of $J.$\end{proof}

\begin{proof}[Proof of Proposition~\ref{prop:S345 arc}]

    Let $C\in \OGnon{k}{2k}.$ We first show that \(C\in\Omega_\Gamma\), with \(\Gamma\) having an external arc \(\tau_\ell\) with support contained in \(\{3,4,5\}\), iff $C$ has a non-zero vector with support contained \(\{3,4,5\}\), where \(\tau\) is the permutation of \(\Gamma\). The first direction is obvious from the parameterization corresponding to the \(\tau_\ell\)-proper orientation.

    For the second direction, suppose \(\Gamma\) does not have an external arc with support contained in \(J = \{3,4,5\}\). Hence \(J\) contains no arcs at all.
    If \(J\) contains no arcs, then by Proposition~\ref{prop:hyperbolic choice per arc},  we can pick a perfect orientation such that \(I\supset J\) are sinks. This yields a parameterization for \(C\in \Omega_\Gamma\) such that \(C^{[2k]\setminus I} = \mathrm{id}_k\). In particular, \(\Delta_{[2k]\setminus I}(C)\neq 0\), hence, by Observation~\ref{obs:3vecsupport}, \(C\) does not contain a non-zero vector with support in \(J\).

    Thus, \(C\in\Omega_\Gamma\) for \(\Gamma\) having an external arc \(\tau_\ell\) supported on \(J\) iff \(C\in \OGnon{k}{2k}\) has a vector with support in \(J\). 
    This reduces Proposition~\ref{prop:S345 arc} to showing \(S_{J} =0\) iff \(C\in \OGnon{k}{2k}\) has a non-zero vector with support in \(J\). By Observation~\ref{obs:3vecsupport}, it is enough to show that
    \begin{prop}
    \label{prop:subclaim for 10.7} For strongly positive \(\Lambda\),
        \(S_{J} =0\) iff \(\Delta_{I}(C) = 0\) for all \(I\in   \binom{[n]\setminus J}{k}\).
    \end{prop}

    By Theorem~\ref{s immanant thm}, we have  
     \[
    S_{J} = \sum_{(\sigma,T)\in \mathcal T_{k,n}} c_{\sigma,T}^{3,6}(\Lambda \eta,\Lambda)\Delta_{\sigma,T}(C),
    \]
    with non-trivial \(c_{\sigma,T}^{3,6}(\Lambda \eta,\Lambda)\) being positive by the definition of \(\mathcal L^>_{k}\). Recall that \(c_{\sigma,T}^{3,6}\) is non-trivial iff \( \sigma\) has two \(J\)-special arcs, which is equivalent to saying \(|\sigma(J)\cap J|< 2\). As the immanants are non-negative for \(C\in\Grnon{k}{n}\), \(S_{\{3,4,5\}}(\Lambda,Y)\) would be zero iff all immanants \(\Delta_{\sigma,T}(C)\) with \(|\sigma(J)\cap J|< 2\) are zero.
    
    Recall that the immanants are defined by 
    \[
    \Delta_A(C)\Delta_B (C)  = \sum_{(\sigma,T)\in \mathcal T_{k,n}(A,B)} \Delta_{\sigma,T}(C)
    \] for all \(A,B\in \binom{[2k]}{k}\).
    For \(B = [2k]\setminus A\), we get
    \[
    \Delta_A(C)^2  = \sum_{(\sigma,T)\in \mathcal T_{k,n}(A,B)} \Delta_{\sigma,T}(C)
    \]  by Theorem~\ref{thm:comp pluckers}.
     Recall \((\sigma,T)\in \mathcal T_{k,n}(A,B)\) iff \(T = A\cap B\), \(S(\sigma)=(A\setminus B)\cup (B\setminus A)\), and \(\sigma(A\setminus B) = B\setminus A\). Thus \(T=\emptyset\), \(S(\sigma)= [2k]\) and \(\sigma(A) = B\).

     For the first direction of Proposition~\ref{prop:subclaim for 10.7}, assume that \(S_{J} = 0 \) on $\widetilde{\Lambda}(C)$. Note that all immanants \(\Delta_{\sigma,T}(C)\) with \(|\sigma(J)\cap J|< 2\) are zero.  Consider \(I\in \binom{[2k]\setminus J}{k}\). We have
     \[
     \Delta_I(C)^2  = \sum_{(\sigma,T)\in \mathcal T_{k,n}(I,[2k]\setminus I)} \Delta_{\sigma,T}(C).
     \]
    For \((\sigma,T)\in \mathcal T_{k,n}(I,[2k]\setminus I)\) it holds that \(\sigma([2k]\setminus I)\cap ([2k]\setminus I)=\emptyset\). Since \(J\subset [2k]\setminus I\), we have \(|\sigma(J)\cap J|=0\). Thus all summands are zero and hence \(\Delta_I(C) = 0\) for all \(I\in \binom{[2k]\setminus J}{k}\), proving the first direction of Proposition~\ref{prop:subclaim for 10.7}.

    For the second direction, let \(C\) be such that all \(I\in \binom{[2k]\setminus J}{k}\), \(\Delta_I(C) = 0\). 
    Consider \(\Delta_{\hat \sigma, \hat T}(C)\) with non-trivial \(c_{\hat \sigma, \hat T}^{3,6}(\Lambda \eta,\Lambda)\), that is, \(|\hat \sigma(J)\cap J| <2\). We will show \(S_{J} = 0\) by 
 showing all such \(\Delta_{\hat \sigma, \hat T}(C)=0\).
    
    \((\sigma,T)\in \mathcal T_{k,n}(A,B)\) iff \(T = A\cap B\), \(S(\sigma)=(A\setminus B)\cup (B\setminus A)\), and \(\sigma(A\setminus B) = B\setminus A\). Construct \(A\) of size $k$ such that either \(J \subset A\) or \(A\cap J = \emptyset\) as follows: 

    \begin{enumerate}
        \item \underline{\textbf{ If \(\hat \sigma(J)\cap J  \not\subset \hat T\):}} First set \(\hat T\subset A\). Then, for each arc \(\hat{\sigma}(r) = \ell\), if one of \(\{\ell, \sigma(\ell)\}\) is in \(J\), add the other to \(A\); if neither is, add an arbitrary element of the arc to \(A\).
        
        As \(|\hat \sigma(J)\cap J| =0\) we cannot have that both \(r, \hat  \sigma(r)\in J\). As \(\hat  T\) are all fixed points of \(\hat  \sigma\) by definition of \(\mathcal  T_{k,n}\), we must have \(J\cap \hat T = \emptyset\) and thus \(A\cap J = \emptyset\). As \(|S(\hat \sigma)|+2|\hat T| = 2k\) by definition of \(\mathcal T_{k,n}\), indeed \(|A| = k\).

        \item \underline{\textbf{ If \(\hat \sigma(J)\cap J \subset \hat T\):}} First set \(\hat T\subset A\). Then, for each arc \(\hat \sigma(r) = \ell\), if one of \(\{\ell,\sigma(\ell)\}\) is in \(J\), add that one to \(A\); if neither is, add an arbitrary element of the arc to \(A\). 
        
        As \(|\hat \sigma(J)\cap J| =1\) we cannot have that both \(r,\hat \sigma(r)\in J\), unless \(\hat \sigma(r) = \ell\) is not an arc of \(\hat \sigma\). As \(\hat T\) are all stable points of \(\hat \sigma\) by definition of \(\mathcal T_{k,n}\) and we can only have one stable point in \(J\) we must have \(J\cap\hat  T  =\{i\}\) and all other elements of \(J\) are contained in arcs. Thus \(J\subset A\). As \(|S(\hat \sigma)|+2|\hat T| = 2k\) by definition of \(\mathcal T_{k,n}\), we have that \(|A| = k\).
    \end{enumerate}

    We have
    \[    
    0 = \Delta_A(C)\Delta_{\hat \sigma(A)} (C)  = \sum_{(\sigma,T)\in \mathcal T_{k,n}(A,B)} \Delta_{\sigma,T}(C)
    \]
    by Theorem~\ref{thm:comp pluckers} as \(J\subset [2k]\setminus A\) or \(J\subset A\). Since all the summands are non-negative, we conclude that \(\Delta_{\sigma,T}(C) = 0\) for all \((\sigma,T)\in T_{k,n}(A,\hat \sigma(A))\). 
    
    By construction, \(\hat T = A\cap \hat \sigma(A)\) as \(\hat T\) are stable points, \(S(\hat \sigma)=(A\setminus \hat \sigma(A))\cup (\hat \sigma(A)\setminus A)\) as \(A\) contains one element out of every arc, and 
    \(
    \hat \sigma(A\setminus \hat \sigma(A)) = \hat \sigma(A\setminus\hat  T) =\hat \sigma(A)\setminus \hat T =  \hat \sigma(A)\setminus A.
    \)
    Thus \((\hat\sigma,\hat T)\in T_{k,n}(A,\hat \sigma(A))\). We conclude that \(\Delta_{\hat \sigma, \hat T}(C)=0\), and thus \(S_{\{3,4,5\}}(\Lambda,Y) = 0\). Moreover, if ${\Gamma_0}$ is any boundary graph of a BCFW graph, and for some $C\in\Omega_{\Gamma_0}$ $S_{\{3,4,5\}}$ vanishes on $\widetilde{\Lambda}(C)$, then it vanishes on all $\widetilde{\Lambda}(\Omega_{\Gamma_0})$

   Suppose  ${\Gamma_0}$ is any boundary graph of a BCFW graph, and for some $C\in\Omega_{\Gamma_0}$ $S_{\{3,4,5\}}$ vanishes on $\widetilde{\Lambda}(C)$. The proof shows that the Mandelstam vanishes precisely when certain sets of Pl\"ucker coordinates vanish, and that if this happens for some $C\in\Omega_{\Gamma_0},$ it happens for all $\Omega_{\Gamma_0}.$ Thus  $S_{\{3,4,5\}}$ then vanishes on  all $\widetilde{\Lambda}(\Omega_{\Gamma_0})$.
\end{proof}

\begin{lem}\label{lem:corners_dont cover_codim1}
Fix $k,$ and a strongly positive $\Lambda\in \mathrm{Mat}^>_{2k\times(k+2)}.$ Let $\Gamma_0$ be an external boundary of a BCFW cell mapping to the zero locus of $S_{\{3,4,5\}}.$ Let $\widehat{\Gamma}_1,\ldots,\widehat{\Gamma}_N$ be the set of corners of all BCFW cells of $\OGnon{k}{2k}$ of codimension at least $2.$ Then 
\[\widetilde{\Lambda}(\Omega_{\Gamma_0})\setminus\left(\bigcup_{i=1}^N\widetilde{\Lambda}(\Omega_{\Gamma_0})\right)\neq\emptyset.\]
\end{lem}
\begin{proof}
By Proposition~\ref{prop:S345 arc} if a point from a stratum of $ \OGnon{k}{2k}$ maps to the zero locus of $S_{\{3,4,5\}}$ then the all stratum maps there. Write $K =\bigcup_{i=1}^N\Omega_{\Gamma_0}$. Assume towards contradiction every $Y\in\widetilde{\Lambda}(\Omega_{\Gamma_0})$ is covered by $\widetilde{\Lambda}(K).$ Let $U$ be a small open subset of $Y.$ By Theorem~\ref{thm external boundaries} we can take $U$ to be small enough so that \(V=U\cap\{Y|~S_{\{3,4,5\}}(Y,\Lambda)=0\}=U\cap\widetilde{\Lambda}(\Omega_{\Gamma_0}),\) and this intersection is a submanifold of dimension $2k-4$. By the assumption $V$ is covered by the image of $K$ and by Proposition~\ref{prop:S345 arc},
\(W=K\cap\widetilde{\Lambda}^{-1}(U)=K\cap\widetilde{\Lambda}^{-1}(V).\)
$W$ is relatively open in $K,$ as the preimage of an open set under a continuous map, and is the union of manifolds of dimensions at most $2k-5.$ By our assumption $W$ surjects onto $V$ under the smooth map $\widetilde\Lambda.$ But $V$ is of dimension $2k-4,$ and this is impossible by Morse-Sard's theorem \cite[Chapter 3]{Hirsch}.
\end{proof}
\begin{lem}\label{lem:deg_1_promotes_deg_1}Fix $k\geq 3.$  Assume that \(\forall \Lambda^{k}\in \mathrm{Mat}^>_{2k\times (k+2)}\), the \(\widetilde \Lambda^k\) images of different BCFW cells of $\BCFW_{k}$ do not intersect. Then \(\forall \Lambda^{k+1}\in \mathrm{Mat}^>_{(2k+2)\times (k+3)}\), the \(\widetilde \Lambda^{k+1}\) images of the cells \(\{\Omega_{\mathrm{Arc}_{3,3}(\Gamma)}\}_{\Omega_\Gamma \in \BCFW_{k}}\) do not intersect.

\end{lem}

\begin{proof}
   Assume towards contradiction that the conclusion of the lemma does not hold. That is, that  \(\forall \Lambda^{k}\in \mathrm{Mat}^>_{2k\times (k+2)}\) the \(\widetilde\Lambda^{k}\) images of different BCFW cells of $\BCFW_{k}$ do not intersect, but the \(\widetilde\Lambda^{k+1}\) images of the cells \(\{\Omega_{\mathrm{Arc}_{3,3}(\Gamma)}\}_{\Omega_\Gamma \in \BCFW_{k}}\) do intersect for some \(\Lambda^{k+1}\in \mathrm{Mat}^>_{(2k+2)\times (k+3)}\). Fix such \(\Lambda^{k+1}\in \mathrm{Mat}^>_{(2k+2)\times (k+3)}\). There exist
    \(
    Y^{k+1}\in \widetilde\Lambda^{k+1}(\Omega_{\mathrm{Arc}_{3,3}(\Gamma_1)})\cap \widetilde\Lambda^{k+1}(\Omega_{\mathrm{Arc}_{3,3}(\Gamma_2)}),
    \)
    for some \(\Omega_{\Gamma_1}\neq\Omega_{\Gamma_2} \in \BCFW_{k}\). This implies the existence of \(C_1^{k+1} \in\Omega_{\mathrm{Arc}_{3,3}(\Gamma_1)}\) and \(C_2^{k+1} \in \Omega_{\mathrm{Arc}_{3,3}(\Gamma_1)}\) such that\(\widetilde\Lambda^{k+1}(C_i^{k+1}) = Y^{k+1}\) for both \(i=1,2\). Label the new oriented vertices added by \(\mathrm{Arc}_{3,3}\) to \(\Gamma_i\) as \(v^\omega_i\) for \(i=1,2\) respectively. By Corollary~\ref{coro:graph cell com diag coro}, we have that \(C_i^{k+1} =\mathrm{Arc}_{3,3}(\alpha_i)(C^k_i)\) with \(C^k_i\in \Omega_{\Gamma_i}\), \(\alpha_i,>0\) for \(i=1,2\) respectively. 
    
    We argue that \(\alpha_1 =\alpha_2\): Recall that by Proposition~\ref{prop:sol vnlk prop}, external arcs with support of length \(3\) are twistor-solvable, and that the \(\mathrm{Arc}_{3,3}\) adds an external arc of support length \(3\). Thus by Proposition~\ref{prop:sol arc to angle} the vertices \(v^\omega_i\) for \(i=1,2\) are twistor-solvable and their angles \(C^{k+1}_i(v^\omega_i,\Gamma_i) = \alpha_i>0\) can be expressed as a function in twistors. The expressions for those functions depend only on the support of the arc by Proposition~\ref{n<=4 twist sol prop}, which is \(\{3,4,5\}\) in both cases. Thus the twistor-solution for \(\mathcal F (v^\omega_1,\Gamma_1)\ = \mathcal F (v^\omega_2,\Gamma_2) = a\). By definition the twistor-solution for an oriented vertex, we have that \(\alpha_1 = a(\Lambda^{k+1},Y^{k+1}) = \alpha_2\). Write \(\alpha\) for \(\alpha_1=\alpha_2\). 
    
    We have that for $i=1,2$, \(C_i^{k+1} =\mathrm{Arc}_{3,3}(\alpha)(C^k_i)\). Write \(\Lambda^k\defeq \mathrm{Arc}_{3,3}(\alpha)^{-1}(\Lambda^{k+1})\). As \(\mathrm{Arc}_{3,3}(\alpha) = \mathrm{Rot_{2,3}(\alpha)\mathrm{Inc_3}}\) by Definition~\ref{def:arc}, we have that \(\Lambda^k\in\mathrm{Mat}^>_{2k\times(k+2)}\) by Corollary~\ref{coro:inc pos} and Definition~\ref{def:rot of lambda}. Now, write \(Y_i^k \defeq \widetilde\Lambda^k(C^k_i)\) for \(i=1,2\). Since \(Y^{k+1}_1=Y^{k+1}_2\), by Observation~\ref{obs:inv arc inj} we have that \(Y^k\defeq Y^{k}_1=Y^{k}_2\). 
    
    We have thus found \(Y^k\) and \(\Lambda^k\in\mathrm{Mat}^>_{2k\times(k+2)}\) such that 
    \(
    Y^{k} \in \widetilde\Lambda^k(\Omega_{\Gamma_1})\cap\widetilde\Lambda^k(\Omega_{\Gamma_2}),\)
    but this shows the \(\widetilde\Lambda^k\) images of different BCFW cells of $\BCFW_{k}$ intersect, which is a contradiction to the assumptions. We conclude the images of the cells \(\{\Omega_{\mathrm{Arc}_{3,3}(\Gamma)}\}_{\Omega_\Gamma \in \BCFW_{k}}\) do not intersect.
\end{proof}
\begin{proof}[Proof of Proposition~\ref{prop:deg_1}]
By Proposition~\ref{prop:const_deg} the amplituhedron map has constant degree on the union of the BCFW cells, and the BCFW cells and their internal boundaries map to the interior of the amplituhedron. Proving Proposition~\ref{prop:deg_1} is thus equivalent to showing that this degree is $1.$ 

Fix \(\Lambda^{k}\in \mathcal L^>_{k}\), and recall that $\Ext_k$ is the set codimension $1$ external boundaries of cells from $\BCFW_k$. Write \(R_k = \coprod_{\Omega \in \Ext_k} \Omega\). We will first show that all points in an image of $\Omega$ for some \(\Omega\in\Ext_k\), which map to the zero locus of $S_{\{3,4,5\}}$, have unique preimage in $R$.

By Proposition~\ref{prop:S345 arc} the only graphs $\Gamma'$ with  $ \Omega_{\Gamma^\prime}\in \Ext_k$ for which $\widetilde{\Lambda}^k (\Omega_{\Gamma'})$ intersects the zero locus of $S_{\{3,4,5\}}$ are graphs having an external arc with support contained in  $\{3,4,5\}$. By Observation~\ref{co-dime 1 external obs} and Theorem~\ref{thm:BCFW are trees}, since they are  codimension \(1\) boundaries of BCFW cells, they must have an external arc with support \(\{3,4,5\}\).
Lemma~\ref{lem:deg_1_promotes_deg_1} proves that under our assumptions the images of the different strata $\widetilde{\Lambda}(\Omega_{\Gamma'})$ do not intersect, hence preimages of a point from $\widetilde{\Lambda}(\Omega_{\Gamma_0})$ \emph{in}
$\bigcup_{\Omega\in\Ext_k}\widetilde{\Lambda}(\Omega)$ can only come from $\Omega_{\Gamma_0},$ and will be unique by Theorem~\ref{BCFW inj}.

We now argue that the above conclusion implies that the constant degree of the amplituhedron map on BCFW cells is $1.$ 
By Lemma~\ref{lem:corners_dont cover_codim1} again we can find $Y\in\widetilde{\Lambda}(\Omega_{\Gamma_0})$ which is not in the image of a corner of codimension $2$ or more of a BCFW cell.

This $Y$ is the limit of a sequence $Y_1,Y_2,\ldots,$ of points in $\bigcup_{\Omega\in\BCFW_k}\widetilde{\Lambda}({\Omega}).$ Had these points multiple preimages, they must have been preimages from different BCFW cells, but from compactness of the closures $\overline{\Omega_\Gamma}$, this would imply that also $Y$ has multiple preimages in $\coprod_{\Omega\in\BCFW_k}\overline{\Omega}$. Since $Y$ has a unique preimage in $\Ext_k,$ and no preimages in BCFW cells and their internal boundaries, by Proposition~\ref{prop:const_deg}, it must have at least one more preimage coming from a corner of a BCFW cell. Contradicting the choice of $Y.$ This contradiction proves the proposition.
\end{proof}
\printbibliography
\pagebreak
\begin{appendices}
\section{The Positive Orthogonal Grassmannian}
\label{app:OG}

This is an overview of the non-negative orthogonal Grassmannian (sometimes called the positive orthogonal Grassmannian), which was introduced in 2014 by Huang, Wen, and Xie \cite{PosOG} in the study of scattering amplitudes in ABJM theory. We present the basic theory and key notions of the non-negative orthogonal Grassmannian that are used throughout this paper.
\begin{dfn}
    \emph{The Grassmannian} \(\Gr{k}{n}\) is defined as the set of \(k\)-dimensional linear sub-spaces of \(\mathbb{R}^n\). Or equivalently \(\Gr{k}{n} := \,_{\mathrm{GL}_{k}(\mathbb R )}\mkern-.5mu\backslash\mkern-2mu^{\mathrm{Mat}_{k \times n}^*(\mathbb R )}, \)
    where \(\mathrm{Mat}_{k \times n}^*(\mathbb R )\) is the space of real \(k\times n\) matrices of full rank, and \(\mathrm{GL}_{k}(\mathbb R )\) acts by multiplication from the left. We will use \(C\in\Gr{k}{n}\) for both a class and the matrix representing it depending on context. 
\end{dfn}

\begin{dfn}
    \emph{The Pl\"ucker embedding} \(\Delta : \Gr{k}{n} \rightarrow \mathrm{\mathbb{RP}}^{\binom{n}{k}-1}\) is defined by

    \(
        \Delta_I (C) = \mathrm{det}\left( C^{I} \right) \), \(\forall I\in \binom{\left[n\right]}{k},
    \)
    where \(\left[n\right]\) is the set \(\{1,\, \ldots ,\,n\}\), \(\binom{\left[n\right]}{k}\) is the set of length \(k\) subsets of \(\left[n\right]\), and \(C^I\) is the matrix formed by taking the columns of the matrix representing \(C\) corresponding to the indices in \(I\).
\end{dfn}

\begin{prop}
\label{prop:plucker relations}
    The Pl\"ucker coordinates for \(C \in \Gr{k}{n}\) satisfy the Pl\"ucker relations:

    For any two ordered sequences of indices \(i_l,\,j_m \in [n]\), that is,
    \(
    i_1 < i_2<\dots<i_{k-1},\\ j_1<j_2<\dots<j_{k+1},
    \)
    one has
    \[
    \sum_{\ell=1}^{k+1} (-1)^\ell \Delta_{\{i_1,\,...,\,i_{k-1},\,j_\ell\}} (C) \Delta_{\{j_1,\,...,\,\hat j_\ell ,\,...,\,j_{k+1}\}} (C) = 0
    \]
    where \(j_1,\,...,\,\hat j_\ell ,\,...,\,j_{k+1}\) denotes the sequence \(j_1,\,\,...,\,j_{k+1}\) with \(j_\ell\) missing.
\end{prop}

\begin{dfn}[\cite{postnikov}, Definition~3.1]
    \emph{The non-negative Grassmannian} (sometimes referred to as the positive Grassmannian) is defined as
    \(
        \Grnon{k}{n} \defeq \Bigl\{C \in \Gr{k}{n} \Big| \, \Delta_I (C) \geq 0,\,\,\,\forall I\in \binom{\left[n\right]}{k}\Bigr\}.
    \)
\end{dfn}
We now turn to the orthogonal Grassmannian and its positive part, which will play a key role in this work.
\begin{dfn}[\cite{OG_source}, Section~2.2]
    \emph{The orthogonal Grassmannian }is defined as
    \(
        \OG{k}{n} \defeq \Bigl\{C \in \Gr{k}{n} \Big| \, C \, \eta \, C^\intercal=0\Bigl\},
    \)
    where \(\eta\) is the diagonal \(n\times n\) matrix with alternating \(1\) and \(-1\) on the diagonal. 
    
    \emph{The non-negative orthogonal Grassmannian} is defined as \(\OGnon{k}{2k}\defeq \OG{k}{2k}\cap\Grnon{k}{2k}\).
\end{dfn}
\begin{thm}[\cite{OG_source}, Proposition~5.1]
    \label{thm:comp pluckers}
    For \(C\in\Grnon{k}{2k}\) the following are equivalent:
    
    \begin{itemize}
        \item {\(C\in\OGnon{k}{2k}\)}

        \item {\(\forall I\in \binom{\left[n\right]}{k} \,\,\,\Delta_{I} (C) = \Delta_{\bar I} (C) ,\) where \(\bar I \defeq \left[n\right] \setminus I\).}
    \end{itemize}
\end{thm}

The positive orthogonal Grassmannian can be decomposed into cells defined the the vanishing of Pl\"ucker variables. That is, cells defined by a certain subset (that is invariant under complement) of the Pl\"ucker variables being zero. These are called \emph{Orthitroid cells}, and they can be indexed by various combinatorial structures. 
\begin{dfn}
    The Orthitroid cell with all minors being positive will be called the \emph{top cell}. It is equivalent to the interior of \(\OGnon{k}{2k}\) or the positive orthogonal Grassmannian.
\end{dfn}

\subsection{OG Graphs}
\label{app:OGgraphs}
The main tool we use to study the Orthitroid cells are the OG graphs. The Orthitroid cells are indexed by equivalence classes of OG graphs, which also give rise to different parameterizations of the Orthitroid cells. These are the \(\OGnon{k}{2k}\) analogues of plabic graphs, which have been studied in relation to \(\Grnon{k}{n}\) by Postnikov \cite{postnikov}.
\begin{dfn}[\cite{PosOG}, Section~2.2]
    \emph{OG graphs} (sometimes referred to as \emph{medial graphs}) are planar graphs embedded in a disc with \(2 k\, \,\,\, (k\in \mathbb N) \) \emph{external vertices}, that is vertices which lie on the boundary of the disc, numbered counter-clockwise.
    The remaining vertices are \emph{internal} and are $4-$regular. An edge is \emph{external} if it is contained in the boundary of the disk, and otherwise it is \emph{internal}. Every external vertex touches a single internal edge and two external edges. An OG graph would be referred to as a \(k\)-OG graph if it is an OG graph with \(2k\) external vertices. The numbering of the external vertices would always be considered cyclically mod \(2k\), that is, the \(2k+i\)-th external vertex would refer to the \(i\)-th external vertex.

    Two OG graphs are considered equivalent if one can be reached from the other by a series of the following \emph{equivalence moves}:

    \begin{enumerate}
    \begin{minipage}{0.9\textwidth}
        \item \label{eqmove:1}

            \
            
            \begin{center}
            \begin{tikzpicture}[scale=0.7]

            \draw (4, 0.5) --(6, 0.5);
        
            \draw [very thick, -latex] (2.25, 1) -- (3.75, 1);

            \draw [black] plot [smooth, tension=0.8] coordinates { (0,0.5) (1,1) (1.5,1.75) (1,2.25) (0.5,1.75) (1,1) (2,0.5)};
        
        \end{tikzpicture}
        \end{center}
        \end{minipage}\\
\begin{minipage}{0.9\textwidth}
        \item \label{eqmove:2}

            \
            
            \begin{center}
            \begin{tikzpicture}[scale=0.7]

            \draw [very thick, -latex] (4.25, 1) -- (5.75, 1);

            \draw [black] plot [smooth, tension=0.7] coordinates { (0,2) (1,1) (2,0.6) (3,1) (4,2)};

            \draw [black] plot [smooth, tension=0.7] coordinates { (0,0) (1,1) (2,1.4) (3,1) (4,0)};

            \draw (6,2)--(8,0);
            \draw (6,0)--(8,2);

        \end{tikzpicture}
        \end{center}
      \end{minipage} \\
        \begin{minipage}{0.9\textwidth}
        \item \label{eqmove:3}

            \
            
            \begin{center}
            \begin{tikzpicture}[scale=0.7]

            \draw (2.93185, 1.51764+0.25882)--(-0.931852, 1.51764+0.25882);
            \draw (2.41421, 2.41421+0.25882)--(0.482362, -0.931852+0.25882);
            \draw (-0.414214, 2.41421+0.25882)--(1.51764, -0.931852+0.25882);

              \draw [very thick, -latex] (3.25, 1) -- (4.75, 1);

            \draw (8.93185, 0.482362-0.25882)--(5.06815, 0.482362-0.25882);
            \draw (7.51764, 2.93185-0.25882)--(5.58579, -0.414214-0.25882);
            \draw (6.48236, 2.93185-0.25882)--(8.41421, -0.414214-0.25882);
        
        \end{tikzpicture}
        \end{center}
             
        \end{minipage}
    \end{enumerate}

    An OG graph is called \emph{reduced} when it has the minimal number of vertices in its equivalence class. As equivalence move~\ref{eqmove:3} does not change the number of vertices, there can be multiple reduced graphs in any equivalence class.

    We will use OG graphs to refer to a specific graph or its equivalence class interchangeably depending on context. Furthermore, when referring to an OG graph we always assume it is reduced, unless specified otherwise.
\end{dfn}
\begin{obs}
    Every OG graph is equivalent to a reduced OG graph.
\end{obs}
\begin{rmk}
    The OG graph corresponding to the top cell of \(\OGnon{3}{6}\) is as seen in Figure~\ref{fig:top cell}.

    \begin{figure}[H]
    \centering
   \begin{center}
\begin{tikzpicture}[scale=0.6]
\draw (0,0) circle (2cm);

\draw({-0.517638, 1.93185})node[anchor=south]{\(4\)}--(1.41421, -1.41421)node[anchor=north]{\(1\)};

\draw(-1.93185, -0.517638)node[anchor=east]{\(5\)}--(1.93185, -0.517638)node[anchor=west]{\(2\)};

\draw(-1.41421, -1.41421)node[anchor=north]{\(6\)}--(0.517638, 1.93185)node[anchor=south]{\(3\)};

\draw (0+5,0) circle (2cm);

\draw({-0.517638+5, -1.93185})node[anchor=north]{\(6\)}--(1.41421+5, 1.41421)node[anchor=south]{\(3\)};

\draw(-1.93185+5, 0.517638)node[anchor=east]{\(5\)}--(1.93185+5, 0.517638)node[anchor=west]{\(2\)};

\draw(-1.41421+5, 1.41421)node[anchor=south]{\(4\)}--(0.517638+5, -1.93185)node[anchor=north]{\(1\)};
\end{tikzpicture}
\end{center}
    \caption{the equivalent OG graphs corresponding to the top cell of \(\OGnon{3}{6}\)}
    \label{fig:top cell}
\end{figure}
\end{rmk}

\begin{dfn}
    A \emph{perfect orientation} for an OG graph is an orientation on its internal edges such that every internal vertex has two in-going and two out-going edges. Unless stated otherwise, the term \emph{orientation} will refer to a perfect orientation.

    External vertices are termed \emph{sinks} if their internal edge is in-going, and \emph{sources} if their internal edge is out-going.

    For \(\Gamma\) an OG graph with \(\omega\) an orientation, (that is, If \(V\) are the vertices and \(E \in \binom{V}{2}\) are the internal edges of \(\Gamma\), \(\omega:E \rightarrow V \times V\) with \(\omega(\{v_1,v_2\}) = (v_1,v_2)\) or \((v_2,v_1)\) for any \(\{v_1,v_2\}\in E\) ) write \(\Gamma^\omega\) for the oriented OG graph.
\end{dfn}

\begin{dfn}
    A \emph{hyperbolic orientation} on an OG graph is a perfect orientation such that every internal vertex has non-alternating in-going and out-going edges. A \emph{trigonometric orientation} is a perfect orientation such that every internal vertex has alternating in-going and out-going edges.
\begin{figure}[H]
    \centering

    \begin{center}
            \begin{tikzpicture}

            \draw (0,2)--(2,0);
            \draw (2,2)--(0,0);

           \draw [very thick, -latex, shorten <= 20] (0,2)--(0.6, 1.4);
            \draw [very thick, -latex, shorten <= 20] (1, 1)--(1.6,0.4);
            \draw [very thick, -latex, shorten <= 20] (2,2)--(1.4, 1.4);
            \draw [very thick, -latex, shorten <= 20] (1, 1)--(0.4,0.4);

             \draw [very thick, -latex, shorten <= 20] (4+4,2)--(4.6+4, 1.4);
            \draw [very thick, -latex, shorten <= 20] (6+4,0)--(5.4+4, 0.6);
            \draw [very thick, -latex, shorten <= 20] (2+8,2)--(1.4+8, 1.4);
            \draw [very thick, -latex, shorten <= 20] (5+4,1)--(4.4+4,0.4);

            \draw (1,-0.75)node[anchor = south]{Hyperbolic and perfect};

            \draw (4,2)--(6,0);
            \draw (4,0)--(6,2);

            \draw [very thick, -latex, shorten <= 20] (4,2)--(4.6, 1.4);
            \draw [very thick, -latex, shorten <= 20] (5, 1)--(5.6,1.6);
            \draw [very thick, -latex, shorten <= 20] (6,0)--(5.4, 0.6);
            \draw [very thick, -latex, shorten <= 20] (5,1)--(4.4,0.4);

            \draw (5,-0.75)node[anchor = south]{Trigonometric and perfect};

            \draw (4+4,2)--(6+4,0);
            \draw (4+4,0)--(6+4,2);

            \draw (5+4,-0.75)node[anchor = south]{Not perfect};
        
        \end{tikzpicture}
        \end{center}
    \caption{possible orientations on vertices}
    \label{fig:posible orient}
\end{figure}
\end{dfn}

\begin{prop}[\cite{PosOG}, Section~2.3.2]
    \label{prop:trigonometric existence}
    A trigonometric orientation exists for any OG graph.

\end{prop}

\begin{dfn}
    Let \(\Gamma\) be an OG graph with \(V\) its internal vertices and \(E\) its internal edges.
    A \emph{path} from \(\ell\) a vertex to \(r\) a vertex is a selection of consecutive edges 
    \[
    P=\{\{\ell,v_1\},\{v_1,v_2\},\{v_2,v_3\},...,\{v_n,r\}\}\subset E
    \] with \(n\geq0\) such that \(v_i\) are internal vertices.

    Every internal edge is four-regular. Thus for each edge \(e\) adjacent to an internal vertex \(v\) we have three other internal edges adjacent to \(v\), two that share a face with \(e\), and one that does not. We will call the latter edge the \emph{edge opposite} to \(e\) at \(v\). A \emph{straight path} is a path in which every consecutive edges are opposite at their common vertex. A \emph{complete path} is a path that goes from an external vertex to an external vertex.
\end{dfn}

\begin{obs}
    Each internal edge is contained in exactly one complete straight path. Each internal vertex is contained in either one or two complete straight paths. For every internal vertex, each complete straight path contains an even number of its adjacent edges. Each external vertex is contained in exactly one complete straight path. 
\end{obs}
We have thus found a total pairing between the external vertices which are labeled by indices from \([2k]\). Two external vertices are paired iff there exist a straight path between them. This pairing defines a permutation on \([2k]\) that is a product of disjoint 2-cycles and has no fixed points. We will call that permutation the permutation \emph{corresponding} to the OG graph. 
\begin{obs}
Equivalence moves~\ref{eqmove:1} and~\ref{eqmove:2} changed the corresponding permutation to an OG graph. Equivalence move~\ref{eqmove:3} does not.
\end{obs}
\begin{dfn}
    Let \(\Gamma\) be an OG graph, and \(\tau\) the corresponding permutation. For \(\ell\in [2k] \) we call \(\tau_\ell \defeq \{\ell,\tau(\ell)\}\) an \emph{arc} of \(\Gamma\). Each arc corresponds to a straight path on the graph from the \(\ell\) to the \(\tau_{\ell}\) external vertex. Thus complete straight paths are in bijection with arcs.
\end{dfn}
\begin{prop}[\cite{companion}]
\label{prop:hyperbolic choice per arc}
    A hyperbolic orientation is equivalent to a choice of an orientation to the complete straight paths of the graph, or conversely the arcs of the graph.
\end{prop}

We say that two arcs are \emph{crossing} at an internal vertex \(v\), if the corresponding straight paths both contain edges adjacent to \(v\). We say that an arc crosses itself at internal vertex \(v\) if the corresponding path contains all of the four edges adjacent to \(v\).
\begin{prop}[\cite{companion}]
\label{prop:reduceability}
    An OG graph is reduced iff neither of its arcs crosses itself and every pair of arcs crosses at most once.
\end{prop}

\begin{dfn}
\label{def:arc sup}
    Let \(\Gamma\) be an OG graph, and \(\tau\) the corresponding permutation. For \(\ell\in [2k]\) write \(I = \{\ell,\ell+1,...,\tau(\ell)\}\) and \(J = \{\tau(\ell),\tau(\ell)+1,...,\ell\}\) considered mod \(2k\), We say that \(\tau_\ell\) is an \emph{external arc} of \(\Gamma\) if \(I\) contains no other arc of \(\Gamma\) or if \(J\) contains no other arc of \(\Gamma\). If the former occurs, we say that \(I\) is the \emph{support} of \(\tau_\ell\), and if the latter occurs we say that \(J\) is the \emph{support} of \(\tau_\ell\) (if both are true then the choice is arbitrary). If \(I\) is the support of \(\tau_\ell\), we have that for any \(r\in I\setminus\tau_\ell\), we have \(\tau(r)\notin I\), meaning \(\tau_r\) crosses \(\tau_\ell\) (see Figure~\ref{fig:ext arc}). If the support of \(\tau_\ell\) is of length \(n\), we would say \(\tau_\ell\) is an \emph{external \(n\)-arc}.
\end{dfn}

    \begin{figure}[H]
        \centering
    \begin{center}
    \scalebox{0.8}{
\begin{tikzpicture}
\draw (0,0) circle (2cm);
\filldraw[lightgray] (0,1/2) circle (1cm);
\node[scale=3] (c) at (0,1/2)  {\(\Gamma_0\)};
\draw (0,1/2) circle (1cm);
\draw (1.96962/2, 1.3473/2) --(3.64697/2, 1.64306/2);
\draw (-1.96962/2, 1.3473/2) --(-3.64697/2, 1.64306/2);
\filldraw[black] (0,3.5/2) circle (2pt);

\filldraw[black] (1.27565/2, 3.20949/2) circle (2pt);
\filldraw[black] (-1.27565/2, 3.20949/2) circle (2pt);

\draw (0.68404/2, -0.879385/2) --(1.68446/2, -3.62803/2)node[anchor=north]{\(i_3\)};
\draw[very thick, -latex, shorten <= 25] (0.68404/2, -0.879385/2)  --(0.68404/8+3*1.68446/8, -3*3.62803/8 -0.879385/8 );

\draw (-0.68404/2, -0.879385/2)  --(-1.68446/2, -3.62803/2)node[anchor=north]{\(i_2\)};
\draw[very thick, -latex, shorten <= 25] (-0.68404/2, -0.879385/2)   --(-0.68404/8-3*1.68446/8,  -0.879385/8-3*3.62803/8);

\draw (-3.27261/2, -2.3/2) node[anchor=north east]{\(i_1\)}-- (3.27261/2, -2.3/2)node[anchor=north west ]{\(i_4\)} ;
\draw[very thick, -latex, shorten <= 45] (-3.27261/2, -2.3/2)-- (0.2/2, -2.3/2) ;

\end{tikzpicture}
}
\end{center}
        \caption{ an OG graph containing an external \(4\)-arc, \(\{i_1,i_4\}\), with support \(\{i_1,i_2,i_3,i_4\}\), with a \(\{i_1,i_4\}\)-proper orientation.}
        \label{fig:ext arc}
    \end{figure}

\begin{dfn}
\label{prop orientation def}
    Given \(\{\ell,\tau(\ell)\}\) an external arc of \(\Gamma\) reduced OG graph with support \(I\) such that \(\ell<\tau(\ell)\). A \emph{\(\tau_\ell\)-proper orientation} of \(\Gamma\) is a hyperbolic orientation where \(\tau_\ell\) is orientated \(\ell\) to \(\tau(\ell)\), and for every \(i\in I\setminus\tau_\ell\), we have \(\tau_i\) oriented \(\tau(i)\) to \(i\) (see Figure~\ref{fig:ext arc}). Other arcs are oriented arbitrarily.
    
    The existence of such an orientation is possible by Proposition~\ref{prop:hyperbolic choice per arc} since \(\tau_\ell\) is an external arc of a reduced graph, and hence the only arc contained entirely in \(I\) is \(\tau_\ell\).
\end{dfn}
 In the spirit of Definition~\ref{def:arc sup}, in the rare occasion where an external arc has two possible supports, we will take care to have the proper orientation be compatible with the chosen support.

\begin{prop}[\cite{companion}]
\label{prop:uniqueness of supp}
    Suppose \(\tau_\ell\) is an external arc of \(\Gamma\), a reduced OG graph with support \(I\), and \(C\in \Omega_\Gamma\). Then there is a vector in \(C\) with support contained in \(I\), unique up to scaling. This vector will be called the \emph{vector associated} with \(\{\ell,\tau(\ell)\}\) (see Figure~\ref{fig:ext arc} for an example of an external arc of a graph with a proper orientation).
    \end{prop}
\begin{prop} [\cite{arkanihamed2014scattering}, Section~13.2]
    \label{prop:orth to perm bij}
    Orthitroid cells are labeled by permutations that are products of disjoint 2-cycles with no fixed points.
\end{prop}
\begin{prop}[\cite{PosOG}, Section~2.2]
\label{prop:orth to graphs bij}
    Orthitroid cells are also in bijection with equivalent classes of OG graphs.
\end{prop}
\begin{dfn}
    The orthitroid cell labeled by an OG graph \(\Gamma\) will be denoted as \(\Omega_\Gamma\).
\end{dfn}
\begin{prop}[\cite{PosOG}, Section~2.2]
    \label{prop:OG graphs to perm bij}
    Two reduced OG graphs are equivalent iff they correspond to the same permutation. Thus, equivalence classes of OG graphs are in bijection with products of disjoint 2-cycles with no fixed points. 
\end{prop}

\begin{dfn}[\cite{PosOG}, Section~2.4.1]  
\label{def:param}
To each perfectly oriented OG graph \(\Gamma^\omega\) we can assign a parameterization of the orthitroid cell \(\Omega_\Gamma\).

    The parameterization is defined in the following way: 

    We first assign an angle to each internal vertex. If the vertex is has a hyperbolic orientation, its angle should be positive. If the vertex has a trigonometric orientation, the angle should be between 0 and \(\frac{\pi}{2}\). These angles will serve as coordinates for the parameterization.

    Let us define a \emph{decision} at an internal vertex \(v\), as an ordered pair of edges \((a, b)\) of the vertex \(v\), such that \(a\) is in-going into \(v\) and \(b\) is out-going from \(v\), with respect to the given orientation.
    
    For each decision at a vertex we assign a weight in the following way: On a vertex with an angle \(\alpha\) and hyperbolic orientation, a \emph{turn}, that is, a decision where \(a\) and \(b\) are adjacent, is assigned the weight \(\cosh \alpha\). A \emph{straight pass}, that is, a decision where \(a\) and \(b\) are non-adjacent, is assigned the weight \(\sinh \alpha\). For a vertex with a trigonometric orientation we define the weights differently. A \emph{right turn}, that is, a decision where \(b\) is just to the right of \(a\), is assigned weight \(\sin \alpha\). A \emph{left turn}, that is, a decision where \(b\) is just to the left of \(a\), is assigned the weight \(\cos \alpha\).

    Define a \emph{path} from an external edge \(i\) to an external edge \(j\), to be a series of decisions \(\{d_i\}_{i=1}^m \) such that the first edge of \(d_1\) is the only edge on the vertex \(i\), the last edge of \(d_m\) is the only edge on the vertex \(j\), and for every \(\ell<m\) the last edge of \(d_\ell\) is the first edge of \(d_{\ell+1}\).

    For a path on the graph going from the external vertex \(i\) to the external vertex \(j\), with that goes around loops in the graph a total of \(w\) times, assign a weight that is the product of the weight of its decisions times the sign \((-1)^{w}\). 
     For external vertices \(i\) and \(j\), define the \emph{boundary measurement} \(M_{i,\,j}\) as the sum of the weights of possible paths going from \(i\) to \(j\).

     As noted in \cite{PosOG}, notice if there is an oriented cycle from \(i\) to \(j\) then the sum will include an infinitely many alternating geometric terms that will sum to 
     \[
     M_{i,\,j}= \sum_{w=0}^{\infty}(-1)^w f_0f_1^w = \frac{f_0}{1+f_1},
     \]
     for \(f_0\) and \(f_1\) some products of hyperbolic and trigonometric functions.

    Finally, for an OG graph $\Gamma,$ an orientation and choices of angles, define the associated \(k \times 2k \) matrix $C$ by \(C_{i,\,j} = (-1)^{\ell}M_{s_i,\,j}\), with \(\{s_i\}_{i=1}^k\subset[2k]\) being the sources, \(j \in [2k]\) goes over the external vertices, and \(\ell\) is the number of sources between \(s_i\) and \(j\).

    If \(\{v_i\}_{i=1}^n\) is some enumeration of the internal vertices of \(\Gamma\), write the resulting parametrization as \(C = \Gamma^\omega(\alpha) =\Gamma^\omega(\alpha_1,...,\alpha_n) \), where \(\alpha_i\) is the angle assigned to the \(v_i\) vertex.
\end{dfn}
\begin{rmk}
    We can allow the angles to take any real value, but the resulting matrix will not necessarily be positive. It will however still be a well defined element of \(\OGnon{k}{2k}\).
\end{rmk}
\begin{coro}
\label{boundary measurements bounded coro}
    Boundary measurements are always positive, and for trigonometric orientations they are bound by \(1\).
\end{coro}

    \begin{thm}[\cite{PosOG}, Section~2.4.1]
    \label{param bij}
         \(C\in \OGnon{k}{2k}\), is independent of the choice of orientation and of the representative of the equivalence class of OG graphs, and is in the orthitroid cell of the given permutation. That is, each orientation for each graph in the same equivalence class would give a parameterization of the same orthitroid cell, the cell that is labeled by the permutation corresponding to the reduced graph, with the parameters being the choice of angles for each internal vertex. For a given reduced graph and perfect orientation, we get a diffeomorphism from the interior of the orthitroid cell labeled by its permutation to 
         \(\left(0,\frac{\pi}{2}\right) ^{n_1} \times \left(0, \infty \right)^{n_2}  \subset \mathbb  R ^{n_1 +n_2},\)
         where \(n_1\) is the number of vertices with a trigonometric orientation, and \(n_2\) is the number of vertices with a hyperbolic orientation.
    \end{thm}

    Notice the dimension of the cell equals the number of internal vertices of a reduce graph. In particular, a reduced OG graph corresponds to a zero dimensional cell iff it has no internal vertices.

\begin{prop}[\cite{companion}]
\label{prop:angles repara}
    Different orientations of the graph induce different parameterizations of the cell. In parameterizations resulting from hyperbolic orientations, the parameters for each vertex undergo the following change of variables:

    \begin{center}
    \scalebox{0.8}{
            \begin{tikzpicture}

            \draw (0,2)--(2,0);
            \draw (2,2)--(0,0);

            \draw [very thick, -latex, shorten <= 20] (0,2)--(0.6, 1.4);
            \draw [very thick, -latex, shorten <= 20] (1, 1)--(1.6,0.4);
            \draw [very thick, -latex, shorten <= 20] (2,2)--(1.4, 1.4);
            \draw [very thick, -latex, shorten <= 20] (1, 1)--(0.4,0.4);

            \draw (1,0.5)node[anchor = south]{\(\alpha\)};

            \draw (4,2)--(6,0);
            \draw (4,0)--(6,2);

            \draw [very thick, -latex, shorten <= 20] (4,2)--(4.6, 1.4);
            \draw [very thick, -latex, shorten <= 20] (5, 1)--(5.6,1.6);
            \draw [very thick, -latex, shorten <= 20] (5,1)--(5.6, 0.4);
            \draw [very thick, -latex, shorten <= 20] (4,0)--(4.6,0.6);

            \draw (5,0.4)node[anchor = south]{\(\beta\)};

             \draw (0+8,2)--(2+8,0);
            \draw (2+8,2)--(0+8,0);

            \draw [very thick, -latex, shorten <= 20] (0+8,2)--(0.6+8, 1.4);
            \draw [very thick, -latex, shorten <= 20] (2+8, 0)--(1.4+8,0.6);
            \draw [very thick, -latex, shorten <= 20] (1+8,1)--(1.6+8, 1.6);
            \draw [very thick, -latex, shorten <= 20] (1+8, 1)--(0.4+8,0.4);

            \draw (1+8,0.5)node[anchor = south]{\(\gamma\)};

        \end{tikzpicture}
        }
        \end{center}

    \begin{center}
    \scalebox{0.8}{
            \begin{tikzpicture}

            \draw (0,2)--(2,0);
            \draw (2,2)--(0,0);

            \draw [very thick, -latex, shorten <= 20] (1,1)--(0.4, 1.6);
            \draw [very thick, -latex, shorten <= 20] (2, 0)--(1.4,0.6);
            \draw [very thick, -latex, shorten <= 20] (1,1)--(1.6, 1.6);
            \draw [very thick, -latex, shorten <= 20] (0, 0)--(0.6,0.6);

            \draw (1,0.5)node[anchor = south]{\(\alpha\)};

            \draw (4,2)--(6,0);
            \draw (4,0)--(6,2);

            \draw [very thick, -latex, shorten <= 20] (5,1)--(4.4, 1.6);
            \draw [very thick, -latex, shorten <= 20] (6, 2)--(5.4,1.4);
            \draw [very thick, -latex, shorten <= 20] (6,0)--(5.4, 0.6);
            \draw [very thick, -latex, shorten <= 20] (5,1)--(4.4,0.4);

            \draw (5,0.4)node[anchor = south]{\(\beta\)};

            \draw (4+4,2)--(6+4,0);
            \draw (4+4,0)--(6+4,2);

            \draw [very thick, -latex, shorten <= 20] (1+4+4,1)--(0.4+4+4, 1.6);
            \draw [very thick, -latex, shorten <= 20] (1+4+4, 1)--(1.6+4+4,0.4);
            \draw [very thick, -latex, shorten <= 20] (2+4+4,2)--(1.4+4+4, 1.4);
            \draw [very thick, -latex, shorten <= 20] (0+4+4, 0)--(0.6+4+4,0.6);

            \draw (5+4,0.4)node[anchor = south]{\(\delta\)};
        
        \end{tikzpicture}
        }
        \end{center}
with \(\sinh{\alpha} = \frac{1}{\sinh{\beta}}=\frac{1}{\tan{\gamma}}=\tan{\delta}\).

\end{prop}
\begin{dfn}
\label{def:oriented vertex}
    Let \(\Gamma\) be an OG graph, \(V\) its internal vertices. Let the edges incident to some \(v\in V\) be \(E_v =e_1,e_2,e_3,e_4\). Consider  \(\omega\), some partial orientation that is defined on \(E_v\). Let us denote  \((v,\evalat*{\omega}{E_v}) \) as \(v^\omega\) and refer to it as an \emph{oriented vertex}, and \(\evalat*{\omega}{E_v}\) will be referred to as \emph{the orientation} of the vertex \(v\) under \(\omega\), or the orientation of \(v^\omega\). Unless stated otherwise, all orientations are assumed to be perfect, that is, two out-going and two in-going edges for each internal vertex. Denote the set of all possible oriented vertices of \(\Gamma\) as \(\widetilde{\mathcal V}(\Gamma)\). If the orientation is clear from context we might omit it.
\end{dfn}
\begin{rmk}
\label{rmk:oriented vert bij}
    Given an oriented OG graph \(\Gamma^\omega\) with internal vertices \(V\), a choice of \(C\in \Omega_\Gamma\) is equivalent to a choice of an angle for each \(v\in V\) by Definition~\ref{def:param} and Theorem~\ref{param bij}, that is, a map \(\varphi_\omega:V\rightarrow \mathbb R\) satisfying certain inequalities. Hold \(C\) constant, and change the orientation of the graph to \(\omega^\prime \). This again defines a parameterization of \(\Omega\), and thus \(C\) is now equivalent to a different choice of angles \(\varphi_{\omega^\prime}:V\rightarrow \mathbb R\). By Proposition~\ref{prop:angles repara}, for each \(v\in V\), the value of \(\varphi_{\omega^\prime}(v)\) is a function of \(\varphi_\omega(v)\) that depends solely on the change of orientation of the edges incident to the vertex \(v\). If the orientation of those edges does not change between \(\omega\) and \(\omega^\prime\), we have that \(\varphi_{\omega^\prime}(v)=\varphi_\omega(v)\).

    Elements \(C\in \Omega_\Gamma\) are thus in bijection with maps \(\varphi:\widetilde{\mathcal V}(\Gamma)\rightarrow \mathbb R\), that satisfy the following conditions:
    \begin{enumerate}
        \item If \(v^\omega\in\widetilde V\) has an hyperbolic orientation then \(\varphi(v^\omega) \in (0,\infty)\).
        \item If \(v^\omega\in\widetilde V\) has a trigonometric orientation then \(\varphi(v^\omega) \in (0,\frac{\pi}{2})\).
        \item For any \(v\in V\) and \(v^\omega,v^{\omega^\prime}\in\widetilde{\mathcal V}(\Gamma)\), \(\varphi(v^\omega)\) and \(\varphi(v^{\omega^\prime})\) relate as described in Proposition~\ref{prop:angles repara}.
    \end{enumerate} 
\end{rmk}
\begin{dfn}
\label{def:C as a func on vert}
    Let \(C\in\Omega_\Gamma\) and for \(\Gamma\) an OG graph. By Remark~\ref{rmk:oriented vert bij}, \(C\) defines a map \(\varphi:\widetilde{\mathcal V}(\Gamma)\rightarrow \mathbb R\). For \(v^\omega \in \widetilde{\mathcal V}(\Gamma) \) write \(C(v^\omega,\Gamma) = \varphi(v^\omega) \).
\end{dfn}
\subsection{Local Moves}
\label{app:local moves}
We can inductively construct all OG graphs from 
 using the \(\mathrm{Inc}_i\) and \(\mathrm{Rot}_{i,\,i+1}\) moves:

\subsubsection{The \texorpdfstring{\(\mathrm{Inc}\)}{Inc} Move}
\label{inc subsubsection}
\begin{dfn}

The \emph{\(\mathrm{Inc}_i\) move} on graphs adds two new vertices between \(i-1\) and \(i\) and connects them with an edge, with indices considered mod \(2k\).
\begin{figure}[H]
    \centering
\begin{center}
\begin{tikzpicture}[scale=0.6]
\draw (0,0) circle (2cm);
\filldraw[lightgray] (0,1/2) circle (1cm);
\node[scale=3] (c) at (0,1/2)  {\(\Gamma\)};
\draw (0,1/2) circle (1cm);
\draw (-1.96962/2, 1.3473/2) --(-3.64697/2, 1.64306/2);
\draw (1.96962/2, 1.3473/2) --(3.64697/2, 1.64306/2)node[anchor=south west]{\(i\)};
\filldraw[black] (-0,3.5/2) circle (2pt);

\filldraw[black] (-1.27565/2, 3.20949/2) circle (2pt);
\filldraw[black] (1.27565/2, 3.20949/2) circle (2pt);

\draw (-0.382683, -0.42388) --(-1.68446/2, -3.62803/2);
\draw (-1.68446/2-0.3, -3.62803/2)node[anchor=north]{\(i-1\)};

\draw[very thick,-latex] (2.3,0)   --(3.7,0);

\draw (0+6,0) circle (2cm);
\filldraw[lightgray] (0+6,1/2) circle (1cm);
\node[scale=3] (c) at (0+6,1/2)  {\(\Gamma\)};
\draw (0+6,1/2) circle (1cm);
\draw (-1.96962/2+6, 1.3473/2) --(-3.64697/2+6, 1.64306/2);
\draw (1.96962/2+6, 1.3473/2) --(3.64697/2+6, 1.64306/2)node[anchor=south west]{\(i+2\)};
\filldraw[black] (-0+6,3.5/2) circle (2pt);

\filldraw[black] (-1.27565/2+6, 3.20949/2) circle (2pt);
\filldraw[black] (1.27565/2+6, 3.20949/2) circle (2pt);

\draw (-0.382683+6, -0.42388) --(-1.68446/2+6, -3.62803/2);
\draw (-1.68446/2-0.3+6, -3.62803/2)node[anchor=north]{\(i-1\)};
\draw (1.68446/2+0.3+6, -3.62803/2)node[anchor=north]{\(i\)};
\draw (3.27261/2+6, -2.3/2)node[anchor=north west]{\(i+1\)};

\draw (1.68446/2+6, -3.62803/2) arc (180+28.7936 :180+28.7936-157.5:0.53254);

\end{tikzpicture}
\end{center}
  \caption{the Inc move}
    \label{fig:inc move}
\end{figure}

Let \(\Gamma\) be the original graph and \(\Gamma^\prime\) be the one after the move, and \(\tau\) and \(\tau^\prime\) be the corresponding permutations. If \(\ell\in[2k]\) write \(\ell^\prime = \mathrm{Inc}_i (\ell) \) for the corresponding index in \(\Gamma^\prime\). That is, if \(\ell<i\) then \(\ell^\prime = \ell\), and if \(\ell\geq i\) then \(\ell^\prime = \ell+2\). Define the action on sets of indices similarly.

For permutations, define \( \mathrm{Inc}_i(\tau)(\ell^\prime) = \mathrm{Inc}_i(\tau (\ell))\), thus \( \mathrm{Inc}_i(\tau) = \tau^\prime\).

If \(\tau_\ell = \{\ell,r\}\) we have that \(\mathrm{Inc}_i (\tau_\ell) = \tau^\prime_{\ell^\prime} = \{\ell^\prime, r^\prime\} =  \{\mathrm{Inc}_i (\ell),\mathrm{Inc}_i (r)\}\) is the corresponding arc in \(\Gamma^\prime\). 

If we have \(\omega\) an  orientation for \(\Gamma\), we will say that the \emph{inherited orientation} for  \(\Gamma^\prime\) is such that the orientation for the edges contained in the original graph are preserved, and for the new arc \(\{i,i+1\}\) we choose \(i\) to \(i+1\) (mod \([2k]\)). We write \(\mathrm{Inc}_i(\omega)\) for the inherited orientation and \(\mathrm{Inc}_i(\Gamma^\omega)=\mathrm{Inc}_i(\Gamma)^{\mathrm{Inc}_i(\omega)}\). Notice that if \(\omega\) is hyperbolic then the inherited orientation is also hyperbolic by Proposition~\ref{prop:hyperbolic choice per arc}. 

For an internal vertex \(v\) in \(\Gamma\), let \(\mathrm{Inc}_i (v)\) be the corresponding vertex in \(\mathrm{Inc}_i (\Gamma)\). For an oriented vertex \(v^\omega\), define  \(\mathrm{Inc}_i(v^\omega) = \mathrm{Inc}_i(v)^{\mathrm{Inc}_i(\omega)}\).

The effect on matrices is as follows (where \(J\) denotes the set of source vertices):

For \(i\neq 2k\),
\[
\left(
\begin{array}{c|c}
C_{\{1,\,...,\,i-1\}\cap J}^{\{1,\,...,\,i-1\}} & C_{{\{1,\,...,\,i-1\}}\cap J}^{\{i,\,...,\,2k\}}\\ \hline
C_{\{i\,...,\,2k\}\cap J}^{\{1,\,...,\,i-1\}} & C_{\{i,\,...,\,2k\}\cap J}^{\{i\,...,\,2k\}}\\
\end{array}
\right)
\mapsto
\left(
\begin{array}{c|cc|c}
C_{\{1,\,...,\,i-1\}\cap J}^{\{1,\,...,\,i-1\}} &&& -C_{{\{1,\,...,\,i-1\}}\cap J}^{\{i,\,...,\,2k\}}\\ \hline
 & 1 & 1 & \\ \hline
-C_{\{i\,...,\,2k\}\cap J}^{\{1,\,...,\,i-1\}} &&& C_{\{i,\,...,\,2k\}\cap J}^{\{i\,...,\,2k\}}\\
\end{array}
\right)
,\]
and for \(i= 2k\),
\[
\left(
\begin{array}{c|c}
C_{\{1,\,...,\,i\}\cap J}^{\{1,\,...,\,i\}} & C_{{\{1,\,...,\,i\}}\cap J}^{\{i+1,\,...,\,2k\}}\\ \hline
C_{\{i+1\,...,\,2k\}\cap J}^{\{1,\,...,\,i\}} & C_{\{i+1,\,...,\,2k\}\cap J}^{\{i+1\,...,\,2k\}}\\
\end{array}
\right)
\mapsto
\left(
\begin{array}{c|c|c}
1 &&  (-1)^k \\ \hline
   & C &\\ 
\end{array}
\right).
\]

\end{dfn}

\subsubsection{The \texorpdfstring{\(\mathrm{Rot}\)}{Rot} Move}
\begin{dfn}
\label{def:rot}
The \emph{\(\mathrm{Rot}_{i,\,i+1}\) move} braids the edges going to \(i\) and \(i+1\) (considered mod \(2k\)) and adds an additional vertex adjacent to the boundary vertices.
\begin{figure}[H]
    \centering
\begin{center}
\begin{tikzpicture}[scale=0.6]
\draw (0,0) circle (2cm);
\filldraw[lightgray] (0,1/2) circle (1cm);
\node[scale=3] (c) at (0,1/2)  {\(\Gamma\)};
\draw (0,1/2) circle (1cm);
\draw (1.96962/2, 1.3473/2) --(3.64697/2, 1.64306/2);
\draw (-1.96962/2, 1.3473/2) --(-3.64697/2, 1.64306/2);
\filldraw[black] (0,3.5/2) circle (2pt);

\filldraw[black] (1.27565/2, 3.20949/2) circle (2pt);
\filldraw[black] (-1.27565/2, 3.20949/2) circle (2pt);

\draw (0.68404/2, -0.879385/2)--(1.68446/2, -3.62803/2);
\draw (1.68446/2+0.3, -3.62803/2)node[anchor=north]{\(i+1\)};

\draw  (-0.68404/2, -0.879385/2)  --(-1.68446/2, -3.62803/2);
\draw (-1.68446/2-0.3, -3.62803/2)node[anchor=north]{\(i\)};

\draw[very thick, -latex, shorten <= 0] (2.3,0)   --(3.7,0);

\draw (0+6,0) circle (2cm);
\filldraw[lightgray] (0+6,1/2) circle (1cm);
\node[scale=3] (c) at (0+6,1/2)  {\(\Gamma\)};
\draw (0+6,1/2) circle (1cm);
\draw (1.96962/2+6, 1.3473/2) --(3.64697/2+6, 1.64306/2);
\draw (-1.96962/2+6, 1.3473/2) --(-3.64697/2+6, 1.64306/2);
\filldraw[black] (0+6,3.5/2) circle (2pt);

\filldraw[black] (1.27565/2+6, 3.20949/2) circle (2pt);
\filldraw[black] (-1.27565/2+6, 3.20949/2) circle (2pt);

\draw  (-0.68404/2+6, -0.879385/2)--(1.68446/2+6, -3.62803/2);
\draw (1.68446/2+0.3+6, -3.62803/2)node[anchor=north]{\(i+1\)};

\draw (0.68404/2+6, -0.879385/2)  --(-1.68446/2+6, -3.62803/2);
\draw (-1.68446/2-0.3+6, -3.62803/2)node[anchor=north]{\(i\)};
\end{tikzpicture}
\end{center}
  \caption{the Rot move}
    \label{fig:rot move}
\end{figure}
Let \(\Gamma\) be the original graph, and \(\Gamma^\prime\) be the one after the move, and \(\tau\) and \(\tau^\prime\) be the corresponding permutations, and consider the indices mod \(2k\). If \(\ell\in[2k]\) write \(\ell^\prime = \mathrm{Rot}_{i,i+1} (\ell) \) for the corresponding index in \(\Gamma^\prime\). That is, if \(\ell = i\) then \(\ell = i+1\),  if \(\ell=i+1\) then \(\ell= i\), and if \(\ell\neq i,i+1\) then \(\ell^\prime =\ell\). 

For permutations, define \(\mathrm{Rot}_{i,i+1}(\tau)(\ell^\prime) = \mathrm{Rot}_{i,i+1}(\tau (\ell))\), thus \(\mathrm{Rot}_{i,i+1}(\tau) = \tau^\prime\).

For arcs  \(\tau_\ell = \{\ell,r\}\) define \(\mathrm{Rot}_{i,i+1}  (\tau_\ell) = \tau^\prime_{\ell^\prime} = \{\ell^\prime, r^\prime\} =  \{\mathrm{Rot}_{i,i+1}  (\ell),\mathrm{Rot}_{i,i+1}  (r)\}\) is the corresponding arc in \(\Gamma^\prime\). 

For general sets of indices define \(\mathrm{Rot}_{i,i+1} (I) = I\cup\{i,i+1\}\) for \(I\cap \{i,i+1\} \neq\emptyset\), and \(\mathrm{Rot}_{i,i+1} (I)  =I\) otherwise. 

If we have \(\omega\) a perfect orientation for \(\Gamma\), we will say that the \emph{inherited orientation} for \(\Gamma^\prime\) the orientation for the edges contained in the original graph are preserved, and the orientation for the new vertex is hyperbolic. We write \(\mathrm{Rot}_{i,i+1}(\omega)\) for the inherited orientation and \(\mathrm{Rot}_{i,i+1}(\Gamma^\omega)=\mathrm{Rot}_{i,i+1}(\Gamma)^{\mathrm{Rot}_{i,i+1}(\omega)}\). Notice the inherited orientation to a hyperbolic orientation is also hyperbolic by Proposition~\ref{prop:hyperbolic choice per arc}.

For an internal vertex \(v\) in \(\Gamma\), let \(\mathrm{Rot}_{i,i+1} (v)\) be the corresponding vertex in \(\mathrm{Rot}_{i,i+1} (\Gamma)\). For an oriented vertex \(v^\omega\), define  \(\mathrm{Rot}_{i,i+1}(v^\omega) = \mathrm{Rot}_{i,i+1}(v)^{\mathrm{Rot}_{i,i+1}(\omega)}\).

If the graph \(\Gamma\) has a \(\{i,i+1\}\) as an arc, then\(\mathrm{Rot}_{i,\,i+1}\) just gives us a non-reduced graph equivalent to \(\Gamma\). If \(\{i,i+1\}\) is not an arc, we can find an orientation where \(i\) and \(i+1\) are both sinks by Proposition~\ref{prop:hyperbolic choice per arc}). 

It is easy to see that for \(\{i,\,i+1\} \neq \{1,\,2k\}\) (mod \(2k\))
\[
C^i \mapsto  C^i \cosh \alpha + C^{i+1} \sinh \alpha
\]
\[
C^{i+1} \mapsto  C^i \sinh \alpha + C^{i+1} \cosh \alpha,
\]
and for \(\{i,\,i+1\} = \{1,\,2k\}\) (mod \(2k\))
\[
C^1 \mapsto  C^1 \cosh \alpha -(-1)^k C^{2k} \sinh \alpha
\]
\[
C^{2k} \mapsto -(-1)^k C^1 \sinh \alpha + C^{2k} \cosh \alpha.
\]
where \(\alpha\) is the positive angle associated with the new vertex. Which means that:
\[
C \mapsto C\,  R_{i,\,i+1} (\alpha),\,\,\, \forall \{i,\,i+1\} \neq \{1,\,2k\}
\]
\[
C \mapsto C\,  R_{1,\,2k} (-(-1)^k \alpha),\,\,\, \{i,\,i+1\} = \{1,\,2k\}
\]
where \(R_{i,\,i+1} (\alpha)\) is the hyperbolic rotation matrix between the \(i\) and \(j\) basis elements in \(\mathbb{R} ^{2k}\) with angle \(\alpha\).
\[
\left(R_{i,\,j}(\alpha)\right)_{a,\,b} \defeq
\begin{cases}
    1 & j \neq a = b \neq i\\
    \cosh \alpha & a= b = i \lor a= b = j\\
    \sinh \alpha & a =i,\,b=j \lor a =j,\,b=i\\
    0 & \mathrm{otherwise}
\end{cases}
\]
Define \(\mathrm{Rot}_{i,i+1}(\alpha)(C) = C \, R_{i,i+1}(\alpha) \) for \(i<2k\) and  \(\mathrm{Rot}_{2k,1}(\alpha)(C) = C \, R_{1,2k}(-(-1)^k\alpha) \). 
\end{dfn}
\begin{dfn}
\label{vertex label}
   When we write \(\mathrm{Rot}_{i,i+1}(v^\omega)\) acting on an oriented graph, we mean to label the new oriented vertex as \(v^\omega\).
\end{dfn}

\begin{obs}
    If \(C\in\Omega_\Gamma\), then \(C^\prime\defeq \mathrm{Rot}_{i,i+1}(\alpha)(C) \in\Omega_{\mathrm{Rot}_{i,i+1}(v^\omega)(\Gamma)}\) for \(\alpha>0\) and \(C^\prime (v^\omega) = \alpha\).
\end{obs}

\begin{rmk}
    It is easy to see that \(\mathrm{Rot}_{i,i+1}(\alpha)^{-1} =\mathrm{Rot}_{i,i+1}(-\alpha)\).
\end{rmk}

\subsubsection{The \texorpdfstring{\(\mathrm{Cyc}\)}{Cyc} Move}
It is useful to consider another move on OG graphs.
\begin{dfn}
    For a \(C\in \Omega_\Gamma\), where \(\Gamma\) is an OG graph with \(2k\) external vertices, define \(\mathrm{Cyc}_k(\Gamma)\) to be the same graph with the index labels rotated one step clockwise:
\begin{figure}[H]
    \centering
    \begin{center}
\begin{tikzpicture}[scale=0.7]
\draw (0,0) circle (2cm);
\filldraw[lightgray] (0,1/2) circle (1cm);
\node[scale=3] (c) at (0,1/2)  {\(\Gamma\)};
\draw (0,1/2) circle (1cm);
\draw (1.96962/2, 1.3473/2) --(3.64697/2, 1.64306/2);
\draw (-1.96962/2, 1.3473/2) --(-3.64697/2, 1.64306/2);
\filldraw[black] (0,3.5/2) circle (2pt);

\filldraw[black] (1.27565/2, 3.20949/2) circle (2pt);
\filldraw[black] (-1.27565/2, 3.20949/2) circle (2pt);

\draw (0.68404/2, -0.879385/2)--(1.68446/2, -3.62803/2);
\draw (1.68446/2+0.3, -3.62803/2)node[anchor=north]{\(i+2\)};
\draw (0.68404/2, -0.879385/2+0.3)node[anchor=north west]{\(i+1\)};

\draw  (-0.68404/2, -0.879385/2) --(-1.68446/2, -3.62803/2);
\draw (-1.68446/2-0.3, -3.62803/2)node[anchor=north]{\(i+1\)};
\draw  (-0.68404/2-0.1, -0.879385/2+0.3)node[anchor=north east]{\(i\)};

\end{tikzpicture}
\end{center}
\caption{the Cyc move}
    \label{fig:cyc move}
\end{figure}

We will omit \(k\) when it is clear from context.

    Let \(\Gamma\) be the original graph and \(\Gamma^\prime\) be the one after the move, and \(\tau\) and \( \tau^\prime\) be the corresponding permutations. If \(\ell\in[2k]\) write \(\ell^\prime = \mathrm{Cyc} (\ell) \) for the corresponding index in \(\Gamma^\prime\). That is,  \(\ell^\prime =\ell + 1\). Define the action on sets of indices similarly. 

    For permutations, define \( \mathrm{Cyc} (\tau)(\ell^\prime) = \mathrm{Cyc} (\tau (\ell))\), thus \( \mathrm{Cyc} (\tau) = \tau^\prime\).
    
    If \(\tau_\ell = \{\ell,r\}\) we have that \(\mathrm{Cyc} (\tau_\ell) = \tau^\prime_{\ell^\prime} = \{\ell^\prime, r^\prime\} =  \{\mathrm{Cyc}(\ell),\mathrm{Cyc} (r)\}\) is the corresponding arc in \(\Gamma^\prime\).

    If we have \(\omega\) a perfect orientation for \(\Gamma\), we will say that the \emph{inherited orientation} for \(\Gamma^\prime\) is such the agrees with the natural identification of edges and vertices across the graphs. We write \(\mathrm{Cyc}(\omega)\) for the inherited orientation and \(\mathrm{Cyc}(\Gamma^\omega)=\mathrm{Cyc}(\Gamma)^{\mathrm{Cyc}(\omega)}\). 

    For an internal vertex \(v\) in \(\Gamma\), let \(\mathrm{Cyc} (v)\) be the corresponding vertex in \(\mathrm{Cyc} (\Gamma)\). For an oriented vertex \(v^\omega\), define  \(\mathrm{Cyc}(v^\omega) = \mathrm{Cyc}(v)^{\mathrm{Cyc}(\omega)}\).
    
    For \(C\in\Gr{k}{2k}\) with columns \(C^i\), 
        \(
     \mathrm{Cyc}_k(C) \defeq 
    \left(
    \begin{array}{ccccc}
         \vertbar&\vertbar&\vertbar&&\vertbar  \\
         -(-1)^k C^{2k}&C^1&C^2&\dots&C^{2k-1}\\
         \vertbar&\vertbar&\vertbar&&\vertbar
    \end{array}
    \right).
    \)

\end{dfn}
\begin{rmk}
    It is easy to see that \(\mathrm{Cyc}_k^{-1} = \mathrm{Cyc}_k^{k-1}\).
\end{rmk}

\subsubsection{Conclusions}
\begin{coro}[\cite{companion}]
\label{coro:graph cell com diag coro}
  For an OG graph \(\Gamma\), \(C\in \Omega_\Gamma\), and \(G\) is either a \(\mathrm{Inc}\), \(\mathrm{Cyc}\), or \(\mathrm{Rot(\alpha)}\) for \(\alpha>0\), we have that \(G(C)\in\Omega_{G(\Gamma)}\).

  For an OG graph \(\Gamma\), \(C\in \Omega_{G(\Gamma)}\), and \(G\) is either a \(\mathrm{Inc}\), \(\mathrm{Cyc}\), we have that \(C = G(C^\prime)\) for some \(C^\prime\in\Omega_{\Gamma}\).

   For an OG graph \(\Gamma\), \(C\in \Omega_{Rot_{i,i+1}(\Gamma)}\), we have that \(\exists\alpha>0\) such that \(C = Rot_{i,i+1}(\alpha)(C^\prime)\) for some \(C^\prime\in\Omega_{\Gamma}\).
\end{coro}
\begin{coro}[\cite{companion}]
\label{arc supp prom}
  For \(\tau_\ell\) an external arc of \(\Gamma\) with support \(I\), and \(G\) either a \(\mathrm{Inc}\), \(\mathrm{Rot}\), or \(\mathrm{Cyc}\) move such that \(G\neq\mathrm{Inc}_i\) for \(i\in i\),  we have that \(G(\tau_\ell)\) is an external arc of \(G (\Gamma)\) with support \(G  (I)\). Additionally, if \(G = \mathrm{Inc}_i\), we have that \(\{i,i+1\}\) (mod \(2k\) is an external arc of \(\mathrm{Inc}_i(\Gamma)\).
\end{coro}

\begin{coro}[\cite{companion}]
\label{reduced move coro}
    Let \(\tau\) be the permutation of a reduced OG graph $\Gamma$, and \(G\) be a \(\mathrm{Inc}\), \(\mathrm{Rot}\), or \(\mathrm{Cyc}\) move. Also assume that \(G \neq \mathrm{Rot}_{i,i+1}\) with \(\tau_i\) and \(\tau_{i+1}\) being the same  or crossing arcs. We have that \(G (\Gamma)\) is reduced.  
\end{coro}

\begin{coro}[\cite{companion}]
\label{moves commute with equiv coro}
    Let \(\Gamma\) and  \(\Gamma^\prime\) be equivalent reduced OG graphs, \(\tau\) their corresponding permutation, and \(G\) be either a \(\mathrm{Inc}\), \(\mathrm{Rot}\), or \(\mathrm{Cyc}\) move. Assume as well that \(G\neq \mathrm{Rot}_{i,i+1}\) for \(\tau_i\) and \(\tau_{i+1}\) equal or crossing arcs. We have that \(G (\Gamma)\) is equivalent to \(G(\Gamma^\prime)\). 
\end{coro}

\begin{coro}[\cite{companion}]
\label{coro:moves dont change angles}
    Let \(\Gamma\) be an OG graph, \(G\) some move, \(v^\omega\) an oriented vertex of \(\Gamma\), and \(C\in \Omega_\Gamma\). If \(G(\Gamma)\) is reduced, \(C(v^\omega,\Gamma) =  (G(C))(G(v^\omega),G(\Gamma))\).
\end{coro}

\section{Calculations}
We present here the proofs of various claims and examples for which we rely on direct calculations. These calculations can be readily verified using the ancillary Wolfram Mathematica notebook. Similar computations, such as the calculation of more involved twistor-solutions, can also be performed using the ancillary Wolfram Mathematica package.

\label{apx:calcs}
\subsection{Proof of Lemma~\ref{triag}}
\label{sec:triag}
We will now prove Lemma~\ref{triag} by finding the unique preimage of an element \(Y\in\mathcal{O}_3(\Lambda)\) in \(\OGnon{3}{6}\) under the amplituhedron map. Recall  Definition~\ref{def:triag}:
\[
\Delta_\pm :=
\begin{pmatrix}
\langle3\,5\rangle&\pm \langle4\,6\rangle&-\langle1\,5\rangle&\mp \langle2\,6\rangle&\langle1\,3\rangle&\pm \langle2\,4\rangle\\
\langle1\,2\rangle&0&-\langle2\,3\rangle&\langle2\,4\rangle&-\langle2\,5\rangle&\langle2\,6\rangle\\
0&-\langle1\,2\rangle&\langle1\,3\rangle&-\langle1\,4\rangle&\langle1\,5\rangle&-\langle1\,6\rangle\\
\end{pmatrix}.
\]
By Lemma~\ref{lem:twist in C}, we have that
\(
\Twist{Y}\,\eta \subset C.
\)
We have \(\mathrm{dim} \, \orth{ \Twist{Y}} = 2k-2 = 4,\,\, \mathrm{dim} \Twist{Y}\,\eta  = 2,\,\, \mathrm{dim} \, C = 3 \). As the twistors form a projective vector we can assume \(\twist{Y}{1}{2} = 1\) without loss of generality. By Proposition~\ref{prop:twistor lambda} we have the following equality of spaces:
\[
\Twist{Y}\ =
\begin{pmatrix}
1&0&-\twist{Y}{2}{3}&-\twist{Y}{2}{4}&-\twist{Y}{2}{5}&-\twist{Y}{2}{6}\\
0&1&\twist{Y}{1}{3}&\twist{Y}{1}{4}&\twist{Y}{1}{5}&\twist{Y}{1}{6}\\
\end{pmatrix}.
\]
Thus, we have that
\[ 
 \orth {\Twist{Y}}
 =
\begin{pmatrix}
\twist{Y}{2}{3} & -\twist{Y}{1}{3} & 1 & 0 & 0 & 0 \\
\twist{Y}{2}{4} & -\twist{Y}{1}{4} & 0 & 1 & 0 & 0 \\
\twist{Y}{2}{5}& -\twist{Y}{1}{5} & 0 & 0 & 1 & 0 \\
\twist{Y}{2}{6} & -\twist{Y}{1}{6} & 0 & 0 & 0 & 1 \\
\end{pmatrix}.
\]
By Lemma~\ref{lem:twist in C} and Proposition~\ref{prop:twistor lambda}, we have that \(C\subset \orth{\Twist{Y}}\). By Proposition~\ref{prop:mom cons}, we have that \(\Twist{Y}\,\eta \subset\orth{\Twist{Y}}\), which is two dimensional by Proposition~\ref{prop:twistor lambda}. We need to find one additional vector from the span of \(\orth {\Twist{Y}}\) to add to \(\Twist{Y} \, \eta\) to find \(C\). This vector needs to be orthogonal to \(\Twist{Y} \, \eta\), and to itself, with respect to the inner product defined by \(\eta\).

Consider the following basis for \(\orth {\Twist{Y}}\):

\[
\begin{pmatrix}
-\twist{Y}{3}{5}&0&\twist{Y}{1}{5}&0&-\twist{Y}{1}{3}&0\\
0&-\twist{Y}{4}{6}&0&\twist{Y}{2}{6}&0&-\twist{Y}{2}{4}\\
1&0&-\twist{Y}{2}{3}&\twist{Y}{2}{4}&-\twist{Y}{2}{5}&\twist{Y}{2}{6}\\
0&-1&\twist{Y}{1}{3}&-\twist{Y}{1}{4}&\twist{Y}{1}{5}&-\twist{Y}{1}{6}\\
\end{pmatrix}
\]

The bottom two rows are \(\Twist{Y} \, \eta\). The top two rows are orthogonal to the bottom two, with respect to the inner product defined by \(\eta\), by the Pl\"ucker relations. The top two are manifestly orthogonal to each-other. The norm (with respect to the inner product defined by \(\eta\)) of the top row is \(-S_{\{1,\,3,\,5\}}(\Lambda,Y)\) and that of the second row is \(S_{\{2,\,4,\,6\}}(\Lambda,Y)\), as complementary Mandelstam variables are equal, in this basis, the form \(\eta\) restricted to \(\orth {\Twist{Y}}\) equals

\[
\begin{pmatrix}
-S_{\{1,\,3,\,5\}}(\Lambda,Y)&0&0&0\\
0&S_{\{1,\,3,\,5\}}(\Lambda,Y)&0&0\\
0&0&0&0\\
0&0&0&0\\
\end{pmatrix}.
\]

Keep in mind that \(-S_{\{1,\,3,\,5\}} (\Lambda,Y)\geq 0\) as it is a sum of squares. It is clear now, that for \(-S_{\{1,\,3,\,5\}}(\Lambda,Y) > 0 \) the only two possible choices for a subspace of dimension \(3\) that is orthogonal to itself with respect to the inner product defined by \(\eta\), is \( \Twist{Y}\,\eta\) together with either the sum, or the difference of the top two rows. In other words, if \(S_{\{1,\,3,\,5\}}(\Lambda,Y) \neq 0 \), given a matrix \(Y\), the two possible candidates for its preimage are:
\[
\Delta_\pm(\Lambda,\,Y) :=
\begin{pmatrix}
\twist{Y}{3}{5}&\pm \twist{Y}{4}{6}&-\twist{Y}{1}{5}&\mp \twist{Y}{2}{6}&\twist{Y}{1}{3}&\pm \twist{Y}{2}{4}\\
\twist{Y}{1}{2}&0&-\twist{Y}{2}{3}&\twist{Y}{2}{4}&-\twist{Y}{2}{5}&\twist{Y}{2}{6}\\
0&-\twist{Y}{1}{2}&\twist{Y}{1}{3}&-\twist{Y}{1}{4}&\twist{Y}{1}{5}&-\twist{Y}{1}{6}\\
\end{pmatrix}
\]

\begin{prop}
\label{prop:135}
For \([C,\Lambda,Y]\in\mathcal U_3^\geq\), we have \(S_{\{1,\,3,\,5\}}(\Lambda,Y) < 0\) for \(C\) in the interior
of the top cell of \(\OGnon{3}{6}\).
\end{prop}
\begin{proof}
Clearly it is non-positive. Assume towards contradiction that 
\[S_{\{1,\,3,\,5\}}(\Lambda,Y)\defeq -\twist{Y}{1}{3}^2-\twist{Y}{3}{5}^2-\twist{Y}{1}{5}^2 = 0.\] Thus \(\twist{Y}{1}{3},\,\twist{Y}{3}{5},\) and \(\twist{Y}{1}{5}\) are zero. Complimentary Mandelstam variables are equal by Proposition~\ref{prop:S orth}. Meaning that \(S_{\{2,\,4,\,6\}}(\Lambda,Y)\) must also vanish, and thus
\(\twist{Y}{2}{4},\,\twist{Y}{4}{6},\) and \(\twist{Y}{2}{6}\) are zero as well. Furthermore, by Proposition~\ref{prop:S orth},
\begin{align*}
S_{\{3,\,5\}}(\Lambda,Y) =& S_{\{1,\,2,\,4,\,6\}}(\Lambda,Y)\\
    \twist{Y}{3}{5}^2 =&\twist{Y}{1}{2}^2+\twist{Y}{1}{4}^2 +\twist{Y}{1}{6}^2-\twist{Y}{2}{4}^2-\twist{Y}{2}{6}^2-\twist{Y}{4}{6}^2\\
    0 =&\twist{Y}{1}{2}^2+\twist{Y}{1}{4}^2 +\twist{Y}{1}{6}^2
\end{align*}
which is impossible as consecutive twistors are non-zero for \(C\) in the interior the positive orthogonal Grassmannian by Proposition~\ref{prop:Consecutive Twistors}. We have reached a contradiction.
\end{proof}
\begin{proof}[Proof of Lemma~\ref{triag}]
 We need to show  that for a given \([C,\Lambda, Y] \in \mathcal U_k^\geq\) (meaning \(Y\) is in \(\mathcal{O}_3 (\Lambda)\)), with \(C\) in the interior of the top cell of \(OG_{\geq 0 } \left(3,\, 6\right)\), we have  \(C = \Delta_+(\Lambda, Y)\).

    By the previous discussion, \(C\) is either equivalent to \(\Delta_+(\Lambda, Y)\) or to \(\Delta_-(\Lambda, Y)\). We will show that only \(\Delta_+(\Lambda, Y)\) is positive:
    Consider the minor $\{2,4,6\}$ of \(\Delta_\pm(\Lambda, Y)\). By the Pl\"ucker relations, it is precisely \(\mp \twist{Y}{1}{2}\,S_{\{2,4,6\}}(\Lambda,Y) = \mp \twist{Y}{1}{2}\,S_{\{1,3,5\}}(\Lambda,Y)\)
    by Proposition~\ref{prop:S orth}.

    Now consider the minor $\{1,3,5\}$ of \(\Delta_\pm(\Lambda, Y)\), it is precisely \(-\twist{Y}{1}{2}\,S_{\{1,3,5\}}(\Lambda,Y).\) 
    By the previous claim, \(S_{1,3,5}<0\), and recall we have \(\twist{Y}{1}{2}=1\). As 
    \(
    \frac{\mathrm{det} (\Delta_-(\Lambda,Y)^{\{1,\,3,\,5\}})}{\mathrm{det} (\Delta_-(\Lambda,Y)^{\{2,\,4,\,6\}})} = -1
    ,\) 
    we can conclude \(\Delta_-(\Lambda,Y)\notin \OGnon{3}{6}\) by Theorem~\ref{thm:comp pluckers}. Therefore we must have \(C = \Delta_+ (\Lambda, Y)\).
\end{proof}
\subsection{Proof of Proposition~\ref{prop:sep k=3}}
\label{apx:sep k=3}
\begin{proof}[Proof of Proposition~\ref{prop:sep k=3}]
By rotational symmetry and the \(\mathrm{Cyc}\) move, it is enough to prove for one internal vertex. To show \(\evalat{\frac{\partial}{\partial v}S}{\Gamma_0}\) is positive, it is enough to prove that \(\evalat{\frac{\partial}{\partial v^\omega}S}{\Gamma_0}\) is positive on \(\Gamma_0\) for some \(\omega\) by Observation~\ref{obs:vertex derivative orientation}. Recall that by Corollary~\ref{coro: codim 1 boudnary classification} we have exactly one boundary graph of codimension \(1\) for each vertex in a BCFW graph, and that the top cell of \(\OGnon{3}{6}\) is BCFW.

We can write the graph as \(\Gamma = \mathrm{Rot}_{1,6}\mathrm{Rot}_{4,5}\mathrm{Rot}_{2,3}\mathrm{Inc}_{1,2}^3(O)\),
Which gives rise to the following parameterization of cell, with the boundary we are interested in corresponding to \(\gamma=0\):
\begin{align*}
C(\alpha,\beta,\gamma)&=\mathrm{Rot}_{1,6}(\gamma)\mathrm{Rot}_{4,5}(\beta)\mathrm{Rot}_{2,3}(\alpha)\begin{pmatrix}
1&1&0&0&0&0\\ 
0&0&1&1&0&0\\ 
0&0&0&0&1&1\\ 
\end{pmatrix}\\
=&\begin{pmatrix}
\cosh{\gamma}&\cosh{\alpha}&\sinh{\alpha}&0&0&-\sinh{\gamma}\\ 
0&\sinh{\alpha}&\cosh{\alpha}&\cosh{\beta}&\sinh{\beta}&0\\ 
-\sinh{\gamma}&0&0&\sinh{\beta}&\cosh{\beta}&\cosh{\gamma}\\ 
\end{pmatrix}
\end{align*}
By Theorem~\ref{thm:comp pluckers} the Pl\"ucker coordinates of \(C\) corresponding to complimentary sets of indices are equal. It is easy to see that the only Pl\"ucker coordinate of \(C\) which has first order dependence on $\gamma$ near $\gamma=0$ is  \(\Delta_{\{1,2,3\}}(C) = \Delta_{\{4,5,6\}}(C)=\sinh(\gamma)\).
Let \([C,\Lambda,Y]\in\mathcal U_3^\geq \). Note that the vertex-separator equals
\(
S(\Lambda, Y) = S_{\{1,2,3\}}(\Lambda,Y) = \twist{Y}{1}{2}^2-\twist{Y}{1}{3}^2+\twist{Y}{2}{3}^2,
\)
which is trivially positively-radical.
We will consider the first order approximation of $S$ near \(\gamma=0\). By the above, it is equivalent to study the first order approximation of $S$ near \(\Delta_{\{1,2,3\}}(C)=0\).
Expanding the twistors using the Cauchy-Binet formula as seen in the proof of Proposition~\ref{prop:Consecutive Twistors}, and using Theorem~\ref{thm:comp pluckers} we obtain 
\begin{align*}
    \twist{Y}{1}{2} =& \Delta_{\{1,2,6\}}(C)\Delta_{\{1,2,3,4,5\}}(\Lambda)+\Delta_{\{1,2,5\}}(C)\Delta_{\{1,2,3,4,6\}}(\Lambda)\\
    &+\Delta_{\{1,2,4\}}(C)\Delta_{\{1,2,3,5,6\}}(\Lambda)+\Delta_{\{1,2,3\}}(C)\Delta_{\{1,2,4,5,6\}}(\Lambda)\\
    \twist{Y}{1}{3} =& -\Delta_{\{1,3,6\}}(C)\Delta_{\{1,2,3,4,5\}}(\Lambda)-\Delta_{\{1,3,5\}}(C)\Delta_{\{1,2,3,4,6\}}(\Lambda)\\
    &-\Delta_{\{1,3,4\}}(C)\Delta_{\{1,2,3,5,6\}}(\Lambda)+\Delta_{\{1,2,3\}}(C)\Delta_{\{1,3,4,5,6\}}(\Lambda)\\
    \twist{Y}{2}{3} =& \Delta_{\{1,4,5\}}(C)\Delta_{\{1,2,3,4,5\}}(\Lambda)+\Delta_{\{1,4,6\}}(C)\Delta_{\{1,2,3,4,6\}}(\Lambda)\\
    &+\Delta_{\{1,5,6\}}(C)\Delta_{\{1,2,3,5,6\}}(\Lambda)+\Delta_{\{1,2,3\}}(C)\Delta_{\{2,3,4,5,6\}}(\Lambda)
\end{align*}
Thus, the linear coefficient of \(\Delta_{\{1,2,3\}}(C)\) in the expansion of $S_{\{1,2,3\}}(\Lambda,Y)$ is
\begin{align*}
    & 2\Delta_{\{1,2,4,5,6\}}(\Lambda)\bigg(\Delta_{\{1,2,6\}}(C)\Delta_{\{1,2,3,4,5\}}(\Lambda)\\
    &\quad\quad+\Delta_{\{1,2,5\}}(C)\Delta_{\{1,2,3,4,6\}}(\Lambda)+\Delta_{\{1,2,4\}}(C)\Delta_{\{1,2,3,5,6\}}(\Lambda)\bigg)\\
    &+2\Delta_{\{1,3,4,5,6\}}(\Lambda) \bigg(\Delta_{\{1,3,6\}}(C)\Delta_{\{1,2,3,4,5\}}(\Lambda)\\
    &\quad\quad+\Delta_{\{1,3,5\}}(C)\Delta_{\{1,2,3,4,6\}}(\Lambda)+\Delta_{\{1,3,4\}}(C)\Delta_{\{1,2,3,5,6\}}(\Lambda)\bigg)\\
    &+2\Delta_{\{2,3,4,5,6\}}(\Lambda)\bigg( \Delta_{\{1,4,5\}}(C)\Delta_{\{1,2,3,4,5\}}(\Lambda)\\
    &\quad\quad+\Delta_{\{1,4,6\}}(C)\Delta_{\{1,2,3,4,6\}}(\Lambda)+\Delta_{\{1,5,6\}}(C)\Delta_{\{1,2,3,5,6\}}(\Lambda)\bigg)
\end{align*}
which is positive as \(\Lambda\in \mathrm{Mat}^>_{8\times 6}\), and those Pl\"uckers of \(C\in\OGnon{3}{6}\). Some of the Pl\"uckers do not vanish for all \(\gamma\) (for example \(\Delta_{\{1,2,6\}}(C)\)). 
Thus \(\evalat{\frac{\partial}{\partial v^\omega}S}{\Gamma_0}\) is positive on \(\Gamma_0\), and \(S\) is zero on \(\Gamma_0\). By Observation~\ref{obs:vertex derivative orientation}, \(\evalat{\frac{\partial}{\partial v}S}{\Gamma_0}\) is positive.
    \end{proof}

\subsection{Proof of Proposition~\ref{prop:k=4 sep sign}}
\label{apx:k=4 sep sign}
\begin{proof}[Proof of Proposition~\ref{prop:k=4 sep sign}]
    By Observations~\ref{obs:triplet sections} and~\ref{obs:Y connected smooth submanifold}, Lemma~\ref{lem:sign is const} applies. Thus, it is enough to show for one point \(y\in \mathbf Y^4\) and one \(\mathbf u \in T_{y}\mathbf Y^4\setminus T_y \mathbf Y^4_0\). Let us start by fixing \(\Lambda = \{i^{j-1}\}_{i\in[8],j\in[6]}\), which has all positive minors, and \(C(t)\in \OGnon{4}{8}\) by
    \begin{align*}
        C&= \mathrm{Rot}_{8,1}(\alpha)\mathrm{Rot}_{6,1}(\alpha)\mathrm{Rot}_{4,1}(\alpha)\mathrm{Rot}_{2,1}(\alpha)\mathrm{Inc}_{7}\mathrm{Inc}_{3}\mathrm{Rot}_{2,3}(t)\mathrm{Inc}_{3}\mathrm{Inc}_{1}(O)\\
        &=\left(
\begin{array}{cccccccc}
 \frac{5}{3} & \frac{5 \cosh (t)}{3} & \frac{4 \cosh (t)}{3} & -\frac{4 \sinh (t)}{3} & -\frac{5 \sinh (t)}{3} & 0 & 0 & -\frac{4}{3} \\
 0 & \frac{5 \sinh (t)}{3} & \frac{4 \sinh (t)}{3} & -\frac{4 \cosh (t)}{3} & -\frac{5 \cosh (t)}{3} & -\frac{5}{3} & -\frac{4}{3} & 0 \\
 0 & \frac{4}{3} & \frac{5}{3} & \frac{5}{3} & \frac{4}{3} & 0 & 0 & 0 \\
 -\frac{4}{3} & 0 & 0 & 0 & 0 & \frac{4}{3} & \frac{5}{3} & \frac{5}{3} \\
\end{array}
\right),
    \end{align*}
with \(O\) being the empty matrix and \(\alpha=\mathrm{Arcsinh}(\frac{4}{3})\). Notice that \(C(0)\in \Omega_{\Gamma_0}\). Furthermore, \(C(t)\in \Omega_{\Gamma_-}\) for \(t>0\) with \(t\) corresponding to the angle associated to \(v_-\) under the orientation \(\omega\) inherited by the construction of \(C(t)\) on \(\Gamma_-\).
We fix \(y = (\Lambda,\widetilde\Lambda (C(0)))\in Y_0^4\). Let us  also fix \(\mathbf u \in T_{y}\mathbf Y^4\setminus T_y \mathbf Y^4_0\) by the direction defined by increasing \(t\) in \((\Lambda,\widetilde\Lambda (C(t)))\). It is easy to check \(\mathbf u\neq0\). Since \(C(t)\in \Omega_{\Gamma_-}\) for \(t>0\), we have that \(\mathbf u\notin T_y \mathbf Y^4_0\).

    By Observation~\ref{obs:BCFW 4 native} we have that \(v_-\) is \(4\)-native, thus by Proposition~\ref{prop:sep 4-native} we have that \(\evalat{\frac{\partial}{\partial v_-}S_-}{\Gamma_0}\) is positive. By Observation~\ref{obs:vertex derivative orientation}, we have that \(\evalat{\frac{\partial}{\partial v_-^\omega }S_-}{\Gamma_0}\) is positive. By Definition~\ref{def:vertex derivative}, we have that \(\evalat{\frac{\partial}{\partial v_-^\omega }S_-}{\Gamma_0}\) is precisely \(\mathrm{d}(S_-\circ \widetilde\Lambda)(\mathbf u )\). We conclude that \(\mathrm{d}(S_-\circ \widetilde\Lambda)(\mathbf u )>0\).

    Let us now calculate \(\mathrm{d}(S_+\circ \widetilde\Lambda)(\mathbf u )= \evalat{\frac{\partial}{\partial t}S_+(\Lambda,\widetilde \Lambda(C(t))}{t=0}\). Since \(\mathrm{Arc}_{4,4}\) is an ancestry-sequence for \(v_+\) in \(\Gamma_+\), by Definition~\ref{def:native vertex}, we have that \(\mathrm{Arc_{4,4}}(S_{\{2,3,4\}})\) is a vertex-separator for \(v_+\). Let us fix \(S_+ = \mathrm{Arc_{4,4}}(S_{\{2,3,4\}})\). See Example~\ref{exm:s234} for an explicit expression for \(\mathrm{Arc_{4,4}}(S_{\{2,3,4\}})\).
    
    It is easy now to check that \(\cosh(\alpha_{4,4,4,2})(\Lambda, \widetilde \Lambda (C(t)) = \frac{5}{3}+\frac{164 t}{213}+O\left(t^2\right)\), and thus 
    \[S_+(\Lambda, \widetilde \Lambda (C(t)) = -10,617,209,487,360,000 \,t+O\left(t^2\right).\]
    Meaning that \(\mathrm{d}(S_+\circ \widetilde\Lambda)(\mathbf u )<0\). Therefore we have that \(\frac{\mathrm d S_+(\mathbf u)}{\mathrm d S_-(\mathbf u)}<0\), completing the proof.
\end{proof}

\subsection{Examples}
\label{apx:calc exm}
\begin{exm}
\label{exm:s234}
    Let us calculate \(\mathrm{Arc_{4,4}}(S_{\{2,3,4\}})\):\\
    First, recall that by Definition~\ref{def:vec angle exp} we have that 
    \begin{align*}
        &\mathbf{v}_{4,4,4 }^i=
    \begin{cases}
        S_{\{5,6,7\}}&i=4\\
        \langle 4\,7\rangle\langle 5\,7\rangle-\langle 4\,5\rangle\langle 5\,6\rangle+\langle 6\,7\rangle\sqrt{S_{\{4,5,6,7\}}}&i=5\\
        \langle 4\,5\rangle\langle 5\,6\rangle-\langle 4\,7\rangle\langle 6\,7\rangle-\langle 5\,7\rangle\sqrt{S_{\{4,5,6,7\}}}&i=6\\
        \langle 4\,6\rangle\langle 6\,7\rangle-\langle 4\,5\rangle\langle 5\,7\rangle+\langle 5\,6\rangle\sqrt{S_{\{4,5,6,7\}}}&i=7\\
        0&\text{otherwise}
    \end{cases}\\
    &\alpha_{4,4,4,1}= \mathrm{arccosh}\left(\frac{\langle 4\,7\rangle\langle 5\,7\rangle-\langle 4\,5\rangle\langle 5\,6\rangle+\langle 6\,7\rangle\sqrt{S_{\{4,5,6,7\}}}}{S_{\{5,6,7\}}}\right)\\
    &\alpha_{4,4,4,2}= \mathrm{arccosh}\left(\frac{ \langle 4\,5\rangle\langle 5\,6\rangle-\langle 4\,7\rangle\langle 6\,7\rangle-\langle 5\,7\rangle\sqrt{S_{\{4,5,6,7\}}}}{S_{\{5,6,7\}}\sqrt{\left(\frac{\langle 4\,7\rangle\langle 5\,7\rangle-\langle 4\,5\rangle\langle 5\,6\rangle+\langle 6\,7\rangle\sqrt{S_{\{4,5,6,7\}}}}{S_{\{5,6,7\}}}\right)^2-1}}\right).
    \end{align*}

    Recall that \(
    \mathrm{Arc}_{4,4}(S_{\{2,3,4\}})=  \mathrm{Rot}_{6,7}(\alpha_{4,4,4,2})\mathrm{Rot}_{5,6}(\alpha_{4,4,4,4,1})\mathrm{Inc}_{4}(S_{\{2,3,4\}}).
    \)
    By Propositions~\ref{s rot} and~\ref{s inc}, we have that
    \(
    \mathrm{Arc}_{4,4}(S_{\{2,3,4\}})=  \mathrm{Rot}_{6,7}(\alpha_{4,4,4,2})(S_{\{2,3,4,5,6\}}).
    \)
    By Definition~\ref{def mand}, we have that
    \begin{align*}
        S_{\{2,3,4,5,6\}} = &\, \langle2\,3\rangle^2-\langle2\,4\rangle^2+\langle2\,5\rangle^2-\langle2\,6\rangle^2+\langle3\,4\rangle^2\\&-\langle3\,5\rangle^2+\langle3\,6\rangle^2
        +\langle4\,5\rangle^2-\langle4\,6\rangle^2
        +\langle5\,6\rangle^2\\
        = &\,S_{\{2,3,4,5\}}-\langle2\,6\rangle^2+\langle3\,6\rangle^2-\langle4\,6\rangle^2+\langle5\,6\rangle^2.
    \end{align*}
    By Proposition~\ref{prop:prom rules}, we therefore get
    \begin{align*}
        \mathrm{Arc}_{4,4}(S_{\{2,3,4\}}) = &\,S_{\{2,3,4,5\}}+\cosh(\alpha_{4,4,4,2})^2\left(-\langle2\,6\rangle^2+\langle3\,6\rangle^2-\langle4\,6\rangle^2+\langle5\,6\rangle^2\right)\\
        +&\sinh(\alpha_{4,4,4,2})^2\left(-\langle2\,7\rangle^2+\langle3\,7\rangle^2-\langle4\,7\rangle^2+\langle5\,7\rangle^2\right)\\
        +&2\cosh(\alpha_{4,4,4,2})\sinh(\alpha_{4,4,4,2})\left(-\langle2\,6\rangle\langle2\,7\rangle+\langle3\,6\rangle\langle3\,7\rangle-\langle4\,6\rangle\langle4\,7\rangle+\langle5\,6\rangle\langle5\,7\rangle\right).
    \end{align*}
\end{exm}
    \begin{exm}
\label{exm:twistor sol}
    Let us calculate the twistor-solution \(\mathcal F(\mathrm{Arc}_{4,4}\mathrm{Arc}_{4,2}\mathrm{Arc}_{3,2}\mathrm{Arc}_{2,1}(O))\):\\
    By Proposition~\ref{prop:sol induction prop}, we have that \[
    \mathcal F(\mathrm{Arc}_{4,4}\mathrm{Arc}_{4,2}\mathrm{Arc}_{3,2}\mathrm{Arc}_{2,1}(O))=\mathrm{Arc}_{4,4}\mathcal F(\mathrm{Arc}_{4,2}\mathrm{Arc}_{3,2}\mathrm{Arc}_{2,1}(O)).
    \]
    By Lemma~\ref{triag}, we have that
    \[
     \mathcal F(\mathrm{Arc}_{4,2}\mathrm{Arc}_{3,2}\mathrm{Arc}_{2,1}(O)) =\Delta_+=\begin{pmatrix}
\langle3\,5\rangle& \langle4\,6\rangle&-\langle1\,5\rangle&- \langle2\,5\rangle&\langle1\,3\rangle& \langle2\,4\rangle\\
\langle1\,2\rangle&0&-\langle2\,3\rangle&\langle2\,4\rangle&-\langle2\,5\rangle&\langle2\,6\rangle\\
0&-\langle1\,2\rangle&\langle1\,3\rangle&-\langle1\,4\rangle&\langle1\,5\rangle&-\langle1\,6\rangle\\
\end{pmatrix}.
    \]
    By Definitions~\ref{def:arc}  we have that \( \mathrm{Arc}_{4,4}=  \mathrm{Rot}_{6,7}(\alpha_{4,4,4,2})\mathrm{Rot}_{5,6}(\alpha_{4,4,4,1})\mathrm{Inc}_{4}\), where \(\alpha_{4,4,4,i}\) are as we found in Example 
   ~\ref{exm:s234}.

We start by applying the \(\mathrm{Inc}_4\) move to get
\[
\mathrm{Inc}_{4}\Delta_+=\begin{pmatrix}
\langle3\,7\rangle& \langle6\,8\rangle&-\langle1\,7\rangle&0&0&- \langle2\,7\rangle&\langle1\,3\rangle& \langle2\,6\rangle\\
\langle1\,2\rangle&0&-\langle2\,3\rangle&0&0&\langle2\,6\rangle&-\langle2\,7\rangle&\langle2\,8\rangle\\
0&-\langle1\,2\rangle&\langle1\,3\rangle&0&0&-\langle1\,6\rangle&\langle1\,7\rangle&-\langle1\,8\rangle\\
0&0&0&1&1&0&0&0
\end{pmatrix}.
\]
Applying \(\mathrm{Rot}_{5,6}(\alpha_{4,4,4,1})\) we get (we present it transposed and contract the indices of the abstract twistors linearly to better fit in the page)
\begin{align*}
&\left(\mathrm{Rot}_{5,6}(\alpha_{4,4,4,1})\mathrm{Inc}_{4}\Delta_+\right)^\intercal=\\
&\qquad\qquad\left(
\begin{array}{cccc}
 -\langle 3\,7\rangle  & \langle 1\,2\rangle  & 0 & 0 \\
 \left\langle 6 c_1+5 s_1,8\right\rangle  & 0 & -\langle 1\,2\rangle  & 0 \\
 \langle 1\,7\rangle  & -\langle 2\,3\rangle  & \langle 1\,3\rangle  & 0 \\
 0 & 0 & 0 & 1 \\
-\langle 2\,8\rangle s_1  & \left\langle 2,6 c_1 s_1+5 s_1^2\right\rangle  &
   \left\langle 1,-6 c_1 s_1-5 s_1^2\right\rangle  & c_1 \\
 -\langle 2\,8\rangle c_1  & \left\langle 2,5 c_1 s_1+6 c_1^2\right\rangle  &
   \left\langle 1,-5 c_1 s_1-6 c_1^2\right\rangle  & s_1 \\
 -\langle 1\,3\rangle  & -\langle 2\,7\rangle  & \langle 1\,7\rangle  & 0 \\
 \left\langle 2,6 c_1+5 s_1\right\rangle  & \langle 2\,8\rangle  & -\langle
   1\,8\rangle  & 0 \\
\end{array}
\right),
\end{align*}
where \(s_i\defeq \sinh(\alpha_{4,4,4,i)}\) and \(c_i\defeq \cosh(\alpha_{4,4,4,i})\) for \(i=1,2\). Applying \(\mathrm{Rot}_{6,7}(\alpha_{4,4,4,2})\), we get

\begin{align*}
&\left(\mathrm{Rot}_{6,7}(\alpha_{4,4,4,2})\mathrm{Rot}_{5,6}(\alpha_{4,4,4,1})\mathrm{Inc}_{4}\Delta_+\right)^\intercal=\\
&   \left(
\begin{array}{cccc}
 \left\langle 3,-7 c_2-6 s_2\right\rangle  & \langle 1\,2\rangle  & 0 & 0 \\
 \left\langle 7 {\tilde c_1} s_2+6 {\tilde c_1} c_2+5 {\tilde s_1},8\right\rangle  & 0 & -\langle
   1\,2\rangle  & 0 \\
 \left\langle 1,7 c_2+6 s_2\right\rangle  & -\langle 2\,3\rangle  & \langle
   1\,3\rangle  & 0 \\
 0 & 0 & 0 & 1 \\
  -\langle 2\,8\rangle {\tilde s_1} & \left\langle 2,7 {\tilde c_1} s_2+6 {\tilde c_1} c_2+5
   {\tilde s_1}\right\rangle {\tilde s_1}   &  -\left\langle 1,7 {\tilde c_1} s_2+6 {\tilde c_1} c_2+5
   {\tilde s_1}\right\rangle {\tilde s_1} & {\tilde c_1} \\
 -\langle 2\,8\rangle {\tilde c_1} c_2 - \langle 1\,3\rangle s_2  & \left\langle 2,\fracless{7 c_2
   {\tilde c_1}^2 s_2+5 c_2 {\tilde c_1} {\tilde s_1}}{-7 c_2 s_2+6 c_2^2 {\tilde c_1}^2-6 s_2^2}\right\rangle  &
   \left\langle 1,\fracless{-7 c_2 {\tilde c_1}^2 s_2-5 c_2 {\tilde c_1} {\tilde s_1}}{+7 c_2 s_2-6 c_2^2 {\tilde c_1}^2+6
   s_2^2}\right\rangle  & c_2 {\tilde s_1} \\
 -\langle 1\,3\rangle c_2 -{\tilde c_1} s_2 \langle 2\,8\rangle  & \left\langle 2,\fracless{7
   {\tilde c_1}^2 s_2^2+6 c_2 {\tilde c_1}^2 s_2}{+5 {\tilde c_1} {\tilde s_1} s_2-6 c_2 s_2-7 c_2^2}\right\rangle  &
   \left\langle 1,\fracless{-7 {\tilde c_1}^2 s_2^2-6 c_2 {\tilde c_1}^2 s_2}{-5 {\tilde c_1} {\tilde s_1} s_2+6 c_2 s_2+7
   c_2^2}\right\rangle  & {\tilde s_1} s_2 \\
 \left\langle 2,7 {\tilde c_1} s_2+6 {\tilde c_1} c_2+5 {\tilde s_1}\right\rangle  & \langle 2\,8\rangle  &
   -\langle 1\,8\rangle  & 0 \\
\end{array}
\right),
\end{align*}
where \(\tilde c_1 = (\mathrm{Rot}_{6,7}(\alpha_{4,4,4,2}))(c_1)\) and \(\tilde s_1 =( \mathrm{Rot}_{6,7}(\alpha_{4,4,4,2}))(s_1)\), that is,
\[
   \tilde c_1 =  \frac{\langle 4,\, 7c_2 +  6s_2\rangle\langle 5,\, 7c_2 +  6s_2\rangle-\langle 4\,5\rangle\langle 5,\, 6c_2+ 7s_1\rangle+\langle 6\,7\rangle\sqrt{S_{\{4,5,6,7\}}}}{S_{\{5,6,7\}}}\] and \(\tilde s_1 = \sqrt{\tilde c_1^2-1}\), where Proposition~\ref{s rot} justifies not changing the Mandelstam variables.
\end{exm}

\end{appendices}

\end{document}